\documentclass[german, a4paper, 11pt]{report}
\usepackage[colorlinks=true,linkcolor=black,anchorcolor=black,citecolor=black,menucolor=black,runcolor=black,urlcolor=black,bookmarks=true]{hyperref}
\usepackage{ellipsis}
\usepackage{xargs}
\usepackage[disable]{todonotes}
\usepackage{fp}
\usepackage{suffix}
\usepackage{tikz}
\usetikzlibrary{matrix,arrows,snakes,fit}
\usepackage{fixltx2e}
\usepackage{color}
\usepackage[a4paper]{geometry}
\usepackage[fit]{truncate}
\usepackage{fancyhdr}
\usepackage{environ}
\usepackage{ifthen}
\usepackage{boolexpr}
\usepackage{etextools}
\usepackage{ascii}
\usepackage{comment}
\usepackage{stmaryrd}
\usepackage{graphicx}
\usepackage{picins}
\usepackage[T1]{fontenc}
\usepackage{amsmath}
\usepackage{amssymb}
\usepackage{amsthm}
\usepackage[intoc]{nomencl}
\usepackage{mathbbol}
\makenomenclature
\usepackage[utf8]{inputenc}
\usepackage[german,english]{babel}
\usepackage[page,toc,title]{appendix}

\allowdisplaybreaks[1]
\usepackage{cite}
\usepackage{paralist}
\makeatletter
\let\make@piccaption@orig\make@piccaption
\renewcommand\make@piccaption{%
  \renewcommand\figurename{Abb.}%
  \make@piccaption@orig}
\makeatother
\makeatletter

\@addtoreset{equation}{section}
\makeatother
\makeatletter
\def\cleardoublepage{\clearpage\if@twoside \ifodd\c@page\else
    \hbox{}
    \vspace*{\fill}
\begin{center}
\end{center}
    \vspace{\fill}                                              %
    \thispagestyle{empty}                                       %
    \newpage                                                    %
    \if@twocolumn\hbox{}\newpage\fi\fi\fi}                      %
\makeatother
\title{Deformation Quantization and Reduction by Stages\thanks{Based on the
    author's diploma thesis at the University of Freiburg, August
  2010, Advisor: N. Neumaier, S. Waldmann}}
\author{\textbf{Dominic Maier}\thanks{E-mail: dominic.maier@gmx.net}
\\[0.1cm]
Fakult{\"a}t f{\"u}r Mathematik und Physik\\
Albert-Ludwigs-Universit{\"a}t Freiburg\\
Physikalisches Institut\\
  }
\date{}

\newcommand{\Ins}[1]{\operatorname{i}(#1)}

\newcommand{\kIn}{\ensuremath{\imath}}
\newcommand{\kInE}{\ensuremath{\imath_1}}
\newcommand{\kjIn}{\ensuremath{\jmath}}

\newcommand{\gIn}{\ensuremath{\iota}}

\newcommand{\kRes}{\ensuremath{\imath^*}}
\newcommand{\kResE}{\ensuremath{\imath_1^*}}
\newcommand{\kjRes}{\ensuremath{\jmath^*}}

\newcommand{\qRes}{\ensuremath{\boldsymbol{\imath}^{\boldsymbol{*}}}}
\newcommand{\qResE}{\ensuremath{\boldsymbol{\imath}_1^{\boldsymbol{*}}}}
\newcommand{\qjRes}{\ensuremath{\boldsymbol{\jmath}^{\boldsymbol{*}}}}

\newcommand{\Faser}[1]{
\futuredef[^]{\HochzeichenJaNein}{\mathrm{t}_{#1}\HochzeichenJaNein}
}
\newcommand{\oFaser}[1]{
\futuredef[^]{\HochzeichenJaNein}{\overline{\mathrm{t}}_{#1}\HochzeichenJaNein}
}
\newcommand{\FaserG}[2][]{\Faser{#1\Gamma_{#2}}}
\newcommand{\oFaserG}[2][]{\oFaser{#1\Gamma_{#2}}}

\newcommand{\qJ}[1][]{\boldsymbol{J}_{\!#1}}
\newcommand{\tqJ}[1][1]{\tilde {\qJ[#1]}}
\newcommand{\CM}[1][M]{\ensuremath{C^{\infty}(#1)}}

\newcommand{\prol}[1][]{{\ensuremath{\operatorname{\mathsf{prol}}_{#1}}}}
\newcommand{\qrol}[1][]{{\ensuremath{\operatorname{\mathsf{qrol}}_{#1}}}}

\newcommand{\starred}[1][]{\ensuremath{\operatorname{\star}_{\tiny
      \mathrm{red}_{#1}}}}

\newcommand{\kkoszul}[1][]{\operatorname{\partial}\Gradsetzer{#1}}

\newcommand{\kkoszulE}{\operatorname{\partial_1}}

\newcommand{\qkoszul}[1][]{\operatorname{{\boldsymbol{\partial}}\Gradsetzer{#1}}}
\newcommand{\qkoszulE}{\operatorname{\boldsymbol{\partial}_1}}
\newcommand{\qkoszulZ}{\operatorname{\boldsymbol{\partial}_2}}

\newcommand{\Gradsetzer}[1]{
\switch %
\case{\ifempty{#1}{0}{1}} %
\otherwise {^{\scriptscriptstyle{(#1)}}}
\endswitch
}

\newcommand{\Id}[1][]{\operatorname{\id\Gradsetzer{#1}}}
\newcommand{\homotopie}{\mathsf{h}}
\newcommand{\h}[1][]{\operatorname{\homotopie\Gradsetzer{#1}}}

\newcommand{\hE}{\operatorname{\h_{1}}}
\newcommand{\hZ}{\operatorname{\h_{2}}}
\newcommand{\qh}[1][]{\boldsymbol{\homotopie}\Gradsetzer{#1}}
\newcommand{\I} {{\mathrm i}}

\newcommand{\Fdot}{\ensuremath{\text{.}}}

\newcommand{\Fcom}{\ensuremath{\text{,}}}

\newcommand{\kanSymbol}{\ensuremath{\mathrm{kan}}}

\newcommand{\JKan}[1][]{
\switch
\case{\ifempty{#1}{0}{1}} \ensuremath{J_{\kanSymbol}}
\otherwise {\ensuremath{J_{\kanSymbol}^{#1}}}
\endswitch
}

\newcommand{\piKan}[1][]{%
\ensuremath{
  \switch %
  \case{\ifempty{#1}{0}{1}} \pi_{\kanSymbol} %
  \otherwise { %
  \futuredef[^]{\HochzeichenJaNein}{%
    \switch[\pdfstrcmp{\HochzeichenJaNein}]%
    \case{{^}} {({\pi_{\kanSymbol}^{#1}})}\HochzeichenJaNein%
    \otherwise {\pi_{\kanSymbol}^{#1}}\HochzeichenJaNein%
    \endswitch%
  }
}%
\endswitch%
}%
}

\newcommand{\omegaKan}[1][]{\ensuremath{\omega_{\kanSymbol}^{#1}}}

\newcommand{\thetaKan}{\theta_{\kanSymbol}}
\WithSuffix\newcommand\overline\thetaKan{%
  \csname\NoSuffixName\overline\expandafter\endcsname\thetaKan
}

\newcommand{\kInKan}[1][]{
\ensuremath{
  \switch %
  \case{\ifempty{#1}{0}{1}} \kIn_{\kanSymbol}%
  \otherwise { %
  \futuredef[^]{\HochzeichenJaNein}{%
    \switch[\pdfstrcmp{\HochzeichenJaNein}]%
    \case{{^}} {(\kIn_{\kanSymbol}^{#1})}\HochzeichenJaNein%
    \otherwise {\kIn_{\kanSymbol}^{#1}}\HochzeichenJaNein%
    \endswitch%
  }
}%
\endswitch%
}%
}

\newcommand{\prolKan}[1][]{\ensuremath{{\prol}_{\kanSymbol}^{#1}}}
\newcommand{\hKan}[1][]{\ensuremath{{\h}_{\kanSymbol}^{#1}}}

\newcommand{\JB}[1][]{\ensuremath{J_{B}^{#1}}}
\newcommand{\omegaB}[1][]{\ensuremath{\omega_{B}^{#1}}}

\newcommand{\piB}[1][]{%
\ensuremath{
  \switch %
  \case{\ifempty{#1}{0}{1}} \pi_{B} %
  \otherwise { %
  \futuredef[^]{\HochzeichenJaNein}{%
    \switch[\pdfstrcmp{\HochzeichenJaNein}]%
    \case{{^}} {({\pi_{B}^{#1}})}\HochzeichenJaNein%
    \otherwise {\pi_{B}^{#1}}\HochzeichenJaNein%
    \endswitch%
  }
}%
\endswitch%
}%
}

\newcommand{\kInB}[1][]{\ensuremath{\kIn_{B}^{#1}}}
\newcommand{\prolB}[1][]{\ensuremath{{\prol}_{B}^{#1}}}
\newcommand{\hB}[1][]{\ensuremath{{\h}_{B}^{#1}}}
\newcommand{\kKoszulB}[1][]{\ensuremath{\partial_B^{#1}}}
\newcommand{\qKoszulB}[1][]{\ensuremath{\qkoszul_B^{#1}}}

\newcommand{\chevalley}{
\futuredef[^]{\HochzeichenJaNein}{\operatorname{\partial}_{\mathrm{\scriptscriptstyle{CE}}}\HochzeichenJaNein}}

\newcommand{\tn}[1]{#1}

\newrobustcmd{\Comp}[2]{%
\begingroup
\let\Compare\undefinded
\DeclareStringFilter\Compare{#1}%
\Compare?{#2}{0}{1}%
\endgroup
}%

\newcommandx{\eAnn}[3][1=(, 3=)]{\footnotesize \text{#1#2#3}}

\newcommand{\qIdeal}[1][] {\ensuremath{{\boldsymbol{\mathcal I}_{{#1}}}}}
\newcommand{\qbIdeal}[1][] {\ensuremath{{\boldsymbol{\mathcal B}_{{#1}}}}}
\newcommand{\kIdeal}[1][] {\ensuremath{{\mathcal I}_{#1}}}
\newcommand{\kbIdeal}[1][] {\ensuremath{{\mathcal B}_{#1}}}
\newcommand{\lieAlgebra}[1][G]{\ensuremath{\mathfrak{\MakeLowercase{#1}}}} 
\newcommand{\wirk}{\vartriangleright} 
\newcommand{\nSchnitt}[1][]{\ensuremath{n_{#1}}} %
\newcommand{\dpA}{\ensuremath{\colon}} 
\newcommand{\tangentIso}{\ensuremath{I}}
\newcommand{\TIso}[1]{{#1}^{\Yleft}}
\newcommand{\TIsoI}[1]{{#1}^{\Yright}}
\newcommand{\mengenSystem}[1][U]{\ensuremath{\mathcal{#1}}} 
\newcommand{\potMenge}[1][X]{\ensuremath{\mathcal{P}(#1)}} 
\newcommand{\neuerBegriff}[1]{\emph{#1}} 
\newcommand{\abschluss}[1]{\ensuremath{\overline{#1}}} 
\newcommand{\vbNIso}[1]{\nu_{#1}} 
\newcommand{\conj}{\ensuremath{\text{con}}}
\newcommand{\conjG}[1][G]{\ensuremath{\text{con}^{\scriptscriptstyle{#1}}}}
\newcommand{\tconjG}[1][G]{\ensuremath{\widetilde{\text{con}}^{\scriptscriptstyle{#1}}}}
\newcommand{\Lie}[1][G]{\ensuremath{\operatorname{\mathrm{Lie}}(#1)}}

\newcommand{\Ad}         {\operatorname{\mathrm{Ad}}}
\newcommand{\AdG}[1][G]
{\ensuremath{\mathrm{Ad}^{\scriptscriptstyle{#1}}}}
\newcommand{\tAdG}[1][G]{\ensuremath{\widetilde{\mathrm{Ad}}^{\scriptscriptstyle{#1}}}}
\newcommand{\Mred}[1][]{{M_{\mathrm{red}_{#1}}}}

\newcommand{\dPaar}[2]{\ensuremath{\langle #1,#2\rangle}} 

\newcommand{\at}[1] {\ensuremath{|_{#1}}} 
\newcommand{\At}[1] {\ensuremath{\big|_{#1}}} 
\newcommand{\im}{\ensuremath{\mathrm{im}}} 
\newcommand{\D} {\operatorname{\mathrm{d}}} 
\newcommand{\ddt}[1][0]{\left . \frac{\D}{\D t} \middle | \right
  ._{t=#1}} 

\newcommand{\id}{\operatorname{\mathsf{id}}}
\newcommand{\inneres}[1]{\ensuremath{\mathring{#1}}}

\newcommand{\supp} {\operatorname{\mathrm{supp}}} 
\newcommand{\carr} {\operatorname{\mathrm{carr}}} 
\newcommand{\projFaktor}{\operatorname{\mathrm{pr}}} 
\newcommand{\zeugendeU}[1]{\textcolor{black}{\ensuremath{\Gamma_{#1}}}}

\makeatletter
\newcommand\afterheading{\par\nobreak\@afterheading}
\makeatother

\newenvironment{satzEnum}{\mbox{}\begin{compactenum}[\itshape
     i.)]}{\end{compactenum}}
\newenvironment{lemmaEnum}{\mbox{}\begin{compactenum}[\itshape
     i.)]}{\end{compactenum}}
\newenvironment{beweisEnum}{\mbox{}\begin{compactenum}[\itshape {a}d
     i.)]}{\end{compactenum}}
\newenvironment{propositionEnum}{\mbox{}\begin{compactenum}[\itshape
     i.)]}{\end{compactenum}}
\newenvironment{korollarEnum}{\mbox{}\begin{compactenum}[\itshape
     i.)]}{\end{compactenum}}
\newenvironment{bemerkungEnum}{\mbox{}\begin{compactenum}[\itshape
     i.)]}{\end{compactenum}}

\newenvironment{definitionEnum}{\mbox{}\begin{compactenum}[\itshape
     i.)]}{\end{compactenum}}

\newcommand{\refitem}[1] {\textit{\ref{#1}.)}} 

\newtheorem{theorem}{Theorem}[section]
\theoremstyle{plain}\newtheorem{satz}[theorem]{Satz}
\theoremstyle{plain}\newtheorem{lemma}[theorem]{Lemma}
\newtheorem*{lemma*}{Lemma} 
\theoremstyle{plain}\newtheorem{proposition}[theorem]{Proposition}
\theoremstyle{plain}\newtheorem*{proposition*}{Proposition}
\theoremstyle{plain}\newtheorem{korollar}[theorem]{Korollar}

\theoremstyle{definition}\newtheorem{definition}[theorem]{Definition}
\theoremstyle{remark} \newtheorem{bemerkung}[theorem]{Bemerkung}
\theoremstyle{remark}\newtheorem{beispiel}[theorem]{Beispiel}

\newenvironment{klein}{}{}
\newenvironment{proofklein}{\small\begin{proof}}{\end{proof}\normalsize}

\newcommand{\Bigvee}{\bigvee\nolimits}
\newcommand{\Bigwedge}{\bigwedge\nolimits}

\begin{document}
\renewcommand{\thepage}{\roman{page}}
\setcounter{page}{1}
\begin{otherlanguage}{english}
\maketitle{}
\begin{abstract}
We discuss the Quantum-Koszul method for constructing star products on
reduced phase spaces in the symplectic, regular case. It is shown that
the reduction method proposed by Kowalzig, Neumaier and Pflaum for
cotangent bundles is a special case of the Quantum-Koszul method. We
give sufficient conditions that reduction by stages is possible in the
Quantum-Koszul framework and show that the star product obtained by two
steps is identical to that obtained by one step. In order to do so, we prove an
equivariant version of the compatible tubular neighborhood theorem.
\end{abstract}
\end{otherlanguage}
\selectlanguage{german}
\cleardoublepage
\thispagestyle{empty}
\vspace*{5cm}
\begin{center}
\Large\emph{{Für Patricia}}
\end{center}
\cleardoublepage
\thispagestyle{empty}
\begin{center}\end{center}
\vspace{1.0cm}
\begin{flushright}
   \parbox{10.0cm}{ {\emph{
         \noindent"`Die Abgründe der Ahnung, %
         ein sicheres Anschauen der Gegenwart, %
         mathematische Tiefe, %
         physische Genauigkeit, %
         Höhe der Vernunft, %
         Schärfe des Verstandes, %
         bewegliche, sehnsuchtsvolle Phantasie, %
         liebevolle Freude am Sinnlichen, nichts kann %
         entbehrt werden zum lebhaften, fruchtbaren Ergreifen %
         des Augenblicks, wodurch ganz allein ein Kunstwerk, %
         von welchem Gehalt es auch sei, entstehen kann."' %
         }}
\vspace{0.6cm} J. W. v. Goethe, Materialien zur Geschichte der Farbenlehre}
\end{flushright}
\cleardoublepage
\tableofcontents
\newpage
\newpage
\chapter*{Einleitung}
\label{sec:Einleitung}
\addcontentsline{toc}{chapter}{Einleitung}
\pagestyle{fancy}
\setlength\headheight{22.36pt}
\fancyhf{}
\fancyhead[OR]{\footnotesize\uppercase{Einleitung}}
\fancyhead[EL]{\footnotesize\uppercase{Einleitung}}
\fancyfoot[C]{\thepage}
Es erfordert keine große Kühnheit zu behaupten, die Entwicklung der
Quantenmechanik am Anfang des zwanzigsten Jahrhunderts habe unser
physikalisches und philosophisches Weltbild so gewandelt wie kaum eine
andere Entwicklung in den Naturwissenschaften. Auch über hundert Jahre
nach Plancks Arbeiten zur Schwarzkörperstrahlung
\cite{planck1901ueber,giulini2000dachte} und Einsteins
Lichtquantenhypothese \cite{einstein1906plancksche} ist die
Quantentheorie nicht nur eines der aktivsten Gebiete der Physik, in dem
es unzählige hochinteressante Anwendungsprobleme zu lösen gibt, sondern
es stellen sich auch heute noch viele interpretatorische und
konzeptionelle Probleme, die immer wieder Anlass zu neuen Diskussionen
und Überlegungen geben \cite{auletta2001foundations}.
\subsubsection*{Verhältnis von klassischer Physik und Quantenphysik}
\label{sec:VerhKlassischerUndQuantenphysik}
Eine wichtige begriffliche Fragestellung, die  zugleich
interpretatorische und auch anwendungsbezogene Relevanz hat, ist die
nach dem Verhältnis von klassischer Physik zur Quantenphysik,
insbesondere klassischer Punktmechanik zur nichtrelativistischen
Quantenmechanik und die Beziehung klassischer Feld- und
Eichtheorien zu deren Quanten"=Analoga
\cite{landsman:1998a,landsman2007between}. Besonders durch neuere
experimentelle Techniken gewinnen diese Fragen an Aktualität. Im
Speziellen sind an dieser Stelle etwa quantenoptische Experimente zu
nennen, siehe \cite{landsman:2006a} und \cite{roemer:2004a}. Auch wenn
man wohl behaupten darf, dass die Mehrheit der Physiker die
Quantentheorie als die fundamentalere Theorie ansieht, aus der durch
bestimmte Näherungen die Gesetze der klassischen Physik zumindest im
Prinzip gewonnen werden können sollten, ist dies eine ganz und gar
nichttriviale Behauptung, die zu vielen tiefgreifenden physikalischen
und erkenntnistheoretischen Diskussionen geführt hat, siehe
\cite{falkenburg2007particle}, \cite{landsman:1998a},
\cite{primas1983chemistry},\cite{de2002foundations},
\cite{ludwig1985axiomatic} und \cite{bunge1967foundations}.  Historisch
scheint die Beziehung zwischen klassischer Physik und Quantenphysik eine
treibende Kraft gewesen zu sein. So geht etwa das Korrespondenzprinzip
auf Niels Bohr zurück, der schon in dem nach ihm benannten Atommodell
für große Quantenzahlen die Beziehung zur klassischen Strahlungsformel
betont hat, vergleiche \cite{bohr1920uber},
\cite[Sec.~5.4.1]{falkenburg2007particle}. Heisenberg,
vergleiche \cite{heisenberg1958physikalische},
\cite[Sec.~5.4.2]{falkenburg2007particle} hat diesen Gedanken in wesentlich
verallgemeinerter Form wieder aufgegriffen und wohl mit der Intuition
im Hinterkopf gearbeitet, dass für $\hbar \to 0$ aus der
Quantenphysik die klassische emergieren sollte. Natürlich ist $\hbar$
eine physikalische Konstante und der Ausdruck $\hbar \to 0$ ist so zu
verstehen, dass $\hbar$  gegen die typische Wirkung des
betrachteten Systems klein ist. Bei Dirac \cite{dirac:1964a} findet sich später
dann eine formalisiertere Form des Korrespondenzprinzips, welche
vermutlich das Denken von Generationen von Physikern geprägt haben
mag. Wie das Groenewold"=van Hove"=Theorem \cite[Satz
5.2.3]{waldmann:2007a} jedoch zeigt, ist die von Dirac vorgeschlagene
Formalisierung nicht streng widerspruchsfrei durchzuführen.
\subsubsection{Deformationsquantisierung}
\label{sec:DeformationsquantisierungEinleitung}
Nachdem in den 70er Jahren des zwanzigsten Jahrhunderts die Theorie der
formalen assoziativen Deformationen nach Gerstenhaber
\cite{gerstenhaber:1964a,gerstenhaber:1966a,gerstenhaber:1968a,gerstenhaber:1974a}
zur Verfügung stand, war es jedoch möglich eine physikalisch und
mathematisch tragfähige Formalisierung des Korrespondenzprinzips
anzugeben. Mit den Arbeiten von Bayen, Flato, Fr{\o}nstal, Lichnerwowicz
und Sternheimer \cite{bayen.et.al:1977a,bayen.et.al:1978a} konnte nun
die Quantenmechanik als formale Deformation der klassischen Mechanik
angesehen werden. Man spricht dabei in naheliegender Weise von
\neuerBegriff{Deformationsquantisierung}. Diese bildet den groben
thematischen Rahmen, in den sich die vorliegende Arbeit eingliedert.
Die Deformationsquantisierung sieht dabei die Observablenalgebra und
deren Quantisierung als primäres Objekt an und versucht die
quantenmechanische Observablenalgebra durch eine
\neuerBegriff{Deformation} des klassischen, kommutativen Produkts zu
gewinnen. Die quantenmechanische Observablenalgebra wird hierbei nicht
etwa als völlig neues Objekt verstanden, sondern die klassische wird mit
einem neuen nichtkommutativen Produkt, dem \neuerBegriff{Sternprodukt},
versehen, welches dem Operatorprodukt entsprechen soll. Der Ansatz für
dieses Produkt besteht in einer Entwicklung nach Potenzen von $\hbar$,
wobei die nullte Ordnung das kommutative Produkt der Funktionen auf dem
klassischen Phasenraum ist. Weiter verlangt man, dass der Kommutator
bezüglich des Sternprodukts in erster Ordnung von $\hbar$ mit dem
$\I\hbar$"=fachen der klassischen Poisson"=Klammer übereinstimmt. Als
essentielle algebraische Bedingung an die höheren Ordnungen erweist sich
die Forderung nach \neuerBegriff{Assoziatitivät} des
Sternprodukts. Zustände werden bei diesem Zugang als ein vom
Observablenbegriff abgeleitetes Konzept betrachtet und als positive
Funktionale beschrieben. Um dem Superpositionsprinzip Rechnung zu
tragen, betrachtet man nach Konstruktion der Sternproduktalgebra
Darstellungen derselben auf (Prä-)Hil\-bert\-räumen.
Die Deformationsquantisierung bringt mehrere konzeptionelle Vorteile mit
sich. Als erstes ist sicherlich die Tatsache zu nennen, dass der
klassische Limes auf Observablenniveau schon intrinsisch bei der
Quantisierung miteingebaut ist. Auf Zustandsniveau ist dieser Punkt
allerdings komplizierter. Die Deformationsquantisierung erlaubt es, grob
gesagt, in einem sinnvoll gewählten Rahmen alle Quantentheorien auf
kontrollierte Art und Weise, durch Deformation in einem formalem
Parameter, dem später die Rolle der Planckschen Konstante $\hbar$
zukommt, alle möglichen Quantenmodelle zu generieren, die ein gegebenes
klassisches Modell auf Observablenebene als klassischen Limes haben und
diese auch zu klassifizieren. Welches dieser dann die Natur richtig
beschreibt entscheidet natürlich das Experiment. Auch wenn der
klassische Limes dann immer noch ein sehr tiefes Problem ist, stellt die
Deformationsquantisierung einen guten Ausgangspunkt dar um dieses besser
verstehen zu können. Weiter zeigt es sich, dass die
Deformationsquantisierung auch die Quantisierung geometrisch
komplizierter klassischer Systeme begrifflich signifikant klarer und
handhabbarer macht. Insbesondere bringt die explizite Verwendung
formaler Potenzreihen eine klare Trennung geometrischer und
funktionalanalytischer Schwierigkeiten mit sich. Diese Vorteile legen
auch nahe, dass die Deformationsquantisierung eine wichtige Rolle beim
Verständnis von Quantengravitationstheorien spielen könnte.  Nach dem
bisher Diskutierten ist klar, dass die zunächst rein mathematische Frage
nach Existenz und Klassifikation von Sternprodukten auf beliebigen
Phasenräumen, das heißt auf beliebigen symplektischen oder allgemeiner
Poisson"=Mannigfaltigkeiten, von eminenter physikalischer Bedeutung
ist. Mathematisch betrachtet war die Lösung dieser Fragen jedoch alles
andere als einfach, insbesondere im Falle von
Poisson"=Mannigfaltigkeiten. Die mathematische Tiefe des letztgenannten
Problems, das von Kontsevich 1997 gelöst und 2003
\cite{kontsevich:2003a} veröffentlicht wurde, zeigt sich nicht zuletzt
darin, dass dieses Resultat mit der Fields"=Medaille gewürdigt
wurde. Die auch im symplektischen Fall nichttriviale Frage nach der
Existenz von Sternprodukten konnte in voller Allgemeinheit schon früher
von DeWilde und Lecomte \cite{dewilde.lecomte:1983b}, Fedosov
\cite{fedosov:1985a,fedosov:1986a,fedosov:1989a} und Omori, Maeda und
Yoshida \cite{omori.maeda.yoshioka:1991a} positiv
beantwortet werden. Aufbauend auf dem Beweis von Fedosov gelang Nest und
Tsygan \cite{nest.tsygan:1995a,nest.tsygan:1995b}, Deligne
\cite{deligne:1995a} und Bertelson, Cahen und Gutt
\cite{bertelson.cahen.gutt:1997a} auch eine Klassifikation von
Sternprodukten auf symplektischen Mannigfaltigkeiten.
\subsubsection{Entwicklung der Physik}
\label{sec:EntwicklungDerPhysik}
Möchte man die treibenden Kräfte benennen, die überhaupt die rasante
Entwicklung der modernen Physik, sowohl in reiner Erkenntnisgewinnung
als auch in ihrer technischen Anwendung, in den letzten Jahrhunderten
ermöglichten, so stößt man unweigerlich zum einen auf das Experiment,
das bestimmte Phänomene der Natur isoliert und nicht wie früher alles
als Gesamtes betrachtet, zum anderen auf die Mathematisierung der
Naturwissenschaft sowie auf den Symmetriebegriff. Letzterer spielt nicht
nur bei der Lösung physikalischer Probleme, sondern auch bei der
fundamentalen Theoriebildung und der Abbildung von
Naturgesetzmäßigkeiten in mathematische Strukturen eine enorm wichtige
Rolle. Mit den beiden letztgenannten Triebfedern der
Physikgeschichte beschäftigt sich auch diese Arbeit.
\subsubsection{Mathematische Physik}
\label{sec:MathematischePhysik}
Die vorliegende Abhandlung ist in der mathematischen Physik anzusiedeln,
einem Gebiet der Physik, bei dem die oben genannte Mathematisierung in
besonderem Maße als Leitgedanke dient. Zu ihrer Aufgabe gehört es, zum
einen zu untersuchen, welche mathematische Sprache für bestimmte
physikalische Probleme angemessen sein könnte und diese auch
mitzuentwickeln, zum anderen jene dann rigoros auf physikalische
Probleme anzuwenden. Gerade durch letztgenannten Anspruch scheint sich
die mathematische Physik bisweilen mit physikalischen Problemen zu
beschäftigen, die seit Jahren auf einem eher heuristischen Niveau als
gelöst betrachtet werden. Trotzdem spielt die mathematische Physik in
der Gesamtentwicklung der Physik wohl eine außergewöhnlich wichtige
Rolle, auch wenn dies auf den ersten Blick möglicherweise nur schwer
erkennbar ist. Hierzu sind einige Punkte zu nennen. Erstens erscheint es
letztlich für ein wirkliches Verständnis der Natur für unabdingbar, dass
die Sprache der einen bestimmten Wirklichkeitsbereich beschreibenden
Theorie klar und angemessen gewählt ist, dass sie begrifflich logisch
konsistent ist und insbesondere keine Widersprüche aufweist. Sich aber
um derartige Fragen zu kümmern, dauert lange und ist mühsam, oft ohne
direkt neue physikalische Effekte aufzudecken. Zweitens scheint es für
den Gesamtfortschritt der Physik wichtig zu sein, die logischen
Beziehungen zwischen den verschiedenen Grundpfeilern einer Theorie
genauestens zu kennen, um bei Entdeckung neuer Phänomene klarer sehen zu
können, welche Teile möglicherweise zusammenstürzen, welche beibehalten
werden können und wo die Theorie erweitert werden kann. Außerdem spielt
eine mathematisch durchdachte und logisch stringente Formulierung einer
Theorie eine wichtige Rolle in der Vermittlung dieser Theorie sowohl an
zukünftige Forscher im Bereich der Physik als auch an solche, die die
Theorie in der Technik anwenden wollen. Drittens ist das immerwährende
Prüfen und Weiterentwicklen der Schnittstelle von Mathematik und Physik
auch rein praktisch relevant, und zwar um physikalische Probleme in einer
Form vorliegen zu haben, in der neue Entwicklungen aus der angewandten
Mathematik, etwa der Numerik, möglichst direkt verwendet werden
können. Man verdeutliche sich die oben genannten Punkte exemplarisch an
den historisch schon etwas weiter zurückliegenden Ereignissen der
Geschichte der Physik, die aber gerade aus diesem Grund eine Beurteilung
für die Gesamtentwicklung der Physik zulassen, nämlich der Formulierung
der Elektrodynamik nach Maxwell, die Formulierung der Quantenmechanik
nach von Neumann und Mackey und die Entwicklung der allgemeinen
Relativitätstheorie durch Einstein. Schließlich wollen wir auch nicht
verschweigen, welch großen Einfluss die Physik auf Entwicklungen in der
Mathematik ausüben kann. Man denke zum Beispiel an die Arbeiten von
Witten, der mit Hilfe von Methoden aus der Quantenfeldtheorie und
Eichtheorie bedeutende mathematische Beiträge zum Verständnis von
Knoteninvarianten oder der Morsetheorie lieferte.
Da nun, wie oben beschrieben, mit der Deformationsquantisierung
insbesondere konzeptionelle Einsichten gewonnen werden sollen, ist es
ohne Weiteres verständlich, dass sie mathematisch exakt formuliert sein
sollte und somit auch Gegenstand der mathematischen Physik ist.
\subsubsection{Symmetrie}
\label{sec:Symmetrie}
Der weiter oben schon angesprochene Begriff der Symmetrie wird uns in
dieser Arbeit besonders beschäftigten. Etwa bei klassischen mechanischen
Systemen kann man durch Ausnutzen bestimmter Symmetrien konkrete
Probleme vereinfachen und quantitative, aber auch vor allem qualitative
Einsichten in das Verhalten eines Systems gewinnen, vergleiche
\cite{cushman.bates:1997a}. Das wohl prominenteste Beispiel einer
solchen Verwendung von Symmetrien ist das Kepler Problem, das in jeder
Anfängervorlesung behandelt wird. Auch hier besteht der Standardzugang
darin, bestimmte Symmetrien beziehungsweise die damit verbundenen
Erhaltungsgrößen auszunutzen, um die Planetenbewegungen zu verstehen.
In der klassischen Physik ist das Eliminieren von Freiheitsgraden, die
sogenannte \neuerBegriff{Phasenraumreduktion} schon eingehend studiert
worden. Auch im Rahmen der geometrischen Mechanik wurden auf diesem
Gebiet seit den 60er Jahren des vergangenen Jahrhunderts viele
Fortschritte erzielt, siehe \cite{abraham.marsden:1985a},
\cite{arnold:1989a}, \cite{marsden.ratiu:1999a} oder
\cite{cushman.bates:1997a}. Eine Frage, die in den letzten Jahren in
diesem Kontext an Bedeutung gewann, ist, grob gesagt, das Problem, ob es
äquivalent ist, eine Symmetrie direkt auszunutzen oder in mehreren
Schritten Teilsymmetrien. Es konnte gezeigt werden, dass dies unter
bestimmten schwachen technischen Annahmen immer wahr ist, siehe etwa
\cite{marsden.misiolek:2007a} für eine eingehende Diskussion.
In Eichtheorien spielen Symmetrien eine fundamentalere Rolle. In diesen
Theorien tauchen mit den Eichfreiheitsgraden nämlich prinzipiell
unphysikalische Freiheitsgrade auf. Physikalisch relevant ist hier nur
der reduzierte Phasenraum. Dessen Geometrie ist aber im Allgemeinen
kompliziert, weshalb es naheliegt, darüber nachzudenken, wie man Systeme
mit geometrisch komplexeren Phasenräumen quantisieren kann. Die
Phasenräume in klassischen Eichtheorien sind im Allgemeinen unendlich
dimensional, es handelt sich dabei ja um das Kotangentialbündel des
Raumes der Eichpotentiale oder mathematisch gesprochen der
Zusammenhangseinsformen eines Hauptfaserbündels, dessen Totalraum
physikalisch eine verallgemeinerte Phase eines Teilchens beschreibt,
dessen Raum"~(Zeit"~)Ko\-or\-di\-na\-ten wiederum durch Punkte der
Basismannigfaltigkeit beschrieben werden. Die Eichgruppe, die auf dem
Raum der Eichpotentiale operiert, ist im Allgemeinen auch unendlich
dimensional. Da die durch die Unendlichdimensionalität herrührenden
funktionalanalytischen Probleme und die geometrischen Probleme jeweils
für sich genommen schon schwierig genug sind, bietet es sich hier an,
diese Schwierigkeitsbereiche getrennt zu betrachten. Den geometrischen
Aspekt kann man modellhaft in der Sprache von Phasenraumreduktion endlich
dimensionaler Mannigfaltigkeiten untersuchen. Betrachtet man
Eichtheorien auf endlichen Gittern, so liegt eine
endlichdimensionale Situation vor.
\subsubsection{Reduktion von Sternprodukten}
\label{sec:ReduktionVonSternprodukten}
Natürlich stellt sich sofort die Frage, wie das Ausnutzen von Symmetrien
mathematisch im Rahmen der Deformationsquantisierung beschrieben werden
kann. Nach den allgemeinen Existenzsätzen für Sternprodukte gibt es
immer auch auf dem reduzierten Phasenraum ein solches. Man kann jedoch
genauer fragen, ob es konstruktive Verfahren gibt, aus einem gegebenen
Sternprodukt eines auf dem reduzierten Phasenraum zu gewinnen und ob man
idealerweise analoge Beziehungen vom ursprünglichen zum reduzierten
Sternprodukt vorliegen hat, wie es schon klassisch für die
Poisson"=Klammer der Fall ist. Zunächst müssen dazu geeignete
Invarianzbegriffe für Sternprodukte gefunden werden, insbesondere stellt
sich die Frage, ob das Konzept der Impulsabbildung im Rahmen von
Sternprodukten eine natürliche Verallgemeinerung erfährt. Mit diesem
Thema haben sich Arnal, Cortet, Molin und Piczon in
\cite{arnal.cortet.molin.pinczon:1983a} auseinandergesetzt und
verschiedene Invarianzbegriffe für Sternprodukte diskutiert. Der
naheliegendste Invarianzbegriff ist, zu fordern, dass die Wirkung auf
den Observablen mit einem Element der Symmetriegruppe ein
Algebramorphismus bezüglich des Sternprodukts ist. Daraus folgt
insbesondere, dass sie auch ein Poisson"=Isomorphismus ist. Im
symplektischen Fall genügt die Existenz einer unter der Symmetriegruppe
invarianten, torsionsfreien kovarianten Ableitung, was etwa dann der
Fall ist, wenn die Symmetriegruppe eigentlich operiert, also
insbesondere bei allen kompakten Symmetriegruppen. Bertelson, Bieliavsky
und Gutt gelang es in \cite{bertelson.bieliavsky.gutt:1998a} mit Hilfe
der Fedosov"=Konstruktion die $G$"=invarianten Sternprodukte zu
klassifizieren. Überträgt man die definierende Eigenschaft einer
Impulsabbildung durch Ersetzen der Poisson"=Klammer durch den
Sternproduktkommutator, so gelangt man zum Begriff der
\neuerBegriff{starken Invarianz} eines Sternprodukts und durch Zulassen
von Quantenkorrekturen bei der Impulsabbildung zu dem der
\neuerBegriff{Quantenimpulsabbildung} (vgl.\
\cite{bordemann.brischle.emmrich.waldmann:1996a,xu:1998a}). Xu konnte in
\cite{xu:1998a} hinreichende Bedingungen für die Existenz von
Quantenimpulsabbildungen angeben. Müller"=Bahns und Neumaier gelang es
in \cite{mueller-bahns.neumaier:2004a}, Kriterien für die Existenz von
Quantenimpulsabbildungen von Fedosov"=Sternprodukten zu
formulieren. Ebenso wurde dort untersucht, wann derartige Sternprodukte
stark invariant sind.
Eine Reduktionsmethode für Sternprodukte, für die es eine
Quantenimpulsabbildung gibt, wurde beispielsweise von Kowalzig, Neumaier
und Pflaum in \cite{kowalzig.neumaier.pflaum:2005a} für den wichtigen
Fall von Kotangentialbündeln mit Magnetfeldern vorgeschlagen. Eine
andere Methode besteht in der aus der Quantenfeldtheorie bekannten
BRST"=Quantisierungsmethode, die von Bordemann, Herbig und Waldmann in
\cite{bordemann.herbig.waldmann:2000a} für den Fall endlich vieler
Freiheitsgrade auch im Rahmen der Deformationsquantisierung beschrieben
werden konnte.  Waldmann und Gutt haben den davon im regulären Fall
relevanten Teil in \cite{gutt2010involutions} formuliert. Dies werden
wir später \emph{Quanten"=Koszul}"=Schema nennen. Für abstraktere Betrachtungen
zu diesem Thema ist schließlich noch die Arbeit von Bordemann
\cite{bordemann:2005a} zu nennen. Die Frage, ob und wie in den genannten
Methoden eine Reduktion in Stufen funktioniert, wurde in der Literatur
bisher noch nicht untersucht und ist Gegenstand dieser Arbeit.
\section*{Ergebnisse}
\label{sec:Ergebnisse}
In der vorliegenden Arbeit werden zwei Probleme aus dem Themengebiet der
Konstruktion von Sternprodukten auf dem reduzierten Phasenraum gelöst.
Erstens wird die von Kowalzig, Neumaier und Pflaum in \cite[Seite
547]{kowalzig.neumaier.pflaum:2005a} aufgeworfene Frage beantwortet, was
die von diesen Autoren dargelegte Methode zur Reduktion von
Sternprodukten auf Kotangentialbündeln mit BRST"=Quantisierung in der
Formulierung von Bordemann, Herbig und Waldmann
\cite{bordemann.herbig.waldmann:2000a} zu tun hat. Es wird in
Satz~\ref{satz:VergleichMitNikolaisArbeitAllgemein} gezeigt, dass
erstgenannte Methode einen echten Spezialfall der
Quanten"=Koszul"=Reduktionsmethode und damit auch der BRST"=Methode
darstellt. Genauer stellt es sich heraus, dass sich die etwas technisch
erscheinende Methode von Kowalzig, Neumaier und Pflaum aus der
Quanten"=Koszul"=Methode durch naheliegende Wahlen von, an das
Kotangentialbündel angepassten, Eingangsdaten  ergibt.
Zweitens gelingt es in dieser Arbeit, die aus der klassischen
Hamiltonschen geometrischen Mechanik bekannte Technik der symplektischen
Phasenraumreduktion in Stufen, unter gewissen Einschränkungen an die
betrachteten Symmetriegruppen, die etwa im Falle kompakter Gruppen immer
erfüllt sind, auf die Reduktionsmethode von Sternprodukten nach dem
Quanten"=Koszul"=Schema zu übertragen, siehe
Satz~\ref{thm:UebereinstimmungDerSternprodukte}.
Als Nebenprodukt dieser Ergebnisse aus der theoretisch"=mathematischen
Physik geht noch ein rein mathematisches Resultat hervor, nämlich wie
man im Rahmen eigentlicher Gruppenwirkungen durch bekannte
Integrationstechniken einen äquivarianten Spray konstruieren, vergleiche
Satz~\ref{satz:invarianterSpray}, und wie man mit Hilfe dessen den
Existenzsatz für kompatible Tubenumgebungen auf eine äquivariante
Situation verallgemeinern kann, siehe Satz~\ref{satz:GkompatibleTuben}.
\subsubsection{Einordnung der Resultate}
\label{sec:EinordnungDerResultate}
Das erstgenannte Resultat ist aus mehreren Gründen von Interesse. Zum
einen ist es natürlich im Allgemeinen für ein gutes Verständnis eines
Problems immer wichtig, auch die Beziehungen zwischen verschiedenen
Lösungsansätzen zu sehen. Zum anderen verhält es sich im betrachteten
Problem so, dass Kowalzig, Neumaier und Pflaum mit ihrer Methode für den
Fall des Kotangentialbündels etwa Aussagen über die Frage, ob
Quantisierung mit Reduktion vertauscht, treffen konnten. Ihre Methode
scheint allerdings technisch und schwer vom betrachteten Beispiel aus zu
verallgemeinern, die Quanten"=Koszul"=Methode jedoch erweckt einen
konzeptionell wesentlich klareren Eindruck. Somit gewinnen die durch die
erste Methode gewonnenen Erkenntnisse nochmals dadurch an Bedeutung, dass
sie ein Spezialfall einer anderen konzeptionell klareren Methode
darstellt. Die oben angesprochene Frage, ob Reduktion mit Quantisierung
vertauscht, ist insbesondere im Rahmen von Eichtheorien von Interesse,
wo das übliche Vorgehen wegen der komplizierten geometrischen Struktur
des reduzierten Phasenraums darin besteht, den unreduzierten Phasenraum
zu quantisieren und anschließend die Eichbedingungen zu verarbeiten
statt, wie es physikalisch und begrifflich eher angemessen wäre, den
reduzierten Phasenraum direkt zu quantisieren, da nur dieser von
direkter physikalischer Relevanz ist. Somit können Ergebnisse, die --
wenn auch nur in endlich dimensionalen Modellen -- zeigen, dass unter
bestimmten Bedingungen Reduktion mit Quantisierung vertauscht,
interessante Hinweise sein, die das übliche Vorgehen bei der
Quantisierung von Eichtheorien rechtfertigen und verständlicher machen.
Zum zweiten Resultat ist zu sagen, dass es zum einen in gewisser
intuitiver Weise für eine gute Quanten"=Reduktionsmethode als wichtig
erscheint, dass eine Reduktion in Stufen zu einer Reduktion in einem
Schritt äquivalent sein sollte, zum anderen ist dies insbesondere bei
der Quantisierung von Eichtheorien von Interesse, wo bestimmte
Eichfreiheitsgrade möglicherweise relativ einfach, andere sehr schwierig
in Griff zu bekommen sind. Wenn man nun weiß, dass diese ohne Bedenken
getrennt behandelt werden können, hat man schon eine erste Vereinfachung
des Problems erreicht und kann sich vollständig auf die schwerer zu
beherrschenden Eichsymmetrien konzentrieren.  Die gewonnenen Ergebnisse
wurden auf klassischer Seite von einem symplektischen Standpunkt aus
formuliert, da hier die klassische Reduktion in Stufen schon gut
ausgearbeitet in der Literatur zur Verfügung stand und da weiter für
diesen Fall die Klassifikation und Existenz von invarianten
Sternprodukten, Quantenimpulsabbildung und anderes wohl besser verstanden
ist, als im Poisson"=Fall. Wir vermuten jedoch stark, dass alle
Überlegungen entsprechend auch für den Poisson"=Fall richtig sein
werden. Das beschriebene Resultat liefert sofort Fragestellungen für
mögliche weitere Untersuchungen und Verallgemeinerungen. Zuerst
erscheint es interessant, genauer zu sehen, ob man die Voraussetzungen an
die Symmetriegruppen nicht noch weiter abschwächen kann, oder aber
zeigen, dass dies wirklich harte Bedingungen sind. Ferner mag es auch
von Interesse sein, die Quanten"=Koszul"=Reduktion in Stufen auch für singuläre
Fälle zu untersuchen. Diese haben wir in der vorliegenden Arbeit nicht
betrachtet, da zum Zeitpunkt ihrer Verfassung noch keine vollständige
Verallgemeinerung der Quanten"=Koszul-Methode bzw.\ der BRST"=Methode im
Rahmen der Deformationsquantisierung für beliebige
stratifizierte Räume  bekannt war. Es gibt eine neuere Arbeit
von Bordemann, Herbig und Pflaum \cite{bordemann2007homological} in der
die BRST"=Methode für den singulären Fall untersucht wird, allerdings nur
in dem Sinne, dass eine regulärer Phasenraum zu einem singulären
Phasenraum reduziert wird. Der Startpunkt ist also hier eine reguläre
Situation. Für eine Untersuchung von Reduktion in Stufen erscheint es
jedoch geboten, auch in einer singulären Situation starten zu
können. Schließlich stellt sich natürlich die Frage, wie unsere
Untersuchungen auf den unendlich dimensionalen feldtheoretischen Fall
verallgemeinert werden können.
Das dritte, rein mathematische Resultat erwies sich für die Lösung der
oben genannten Probleme als sehr nützlich  und wird vermutlich auch in
anderen Bereichen der theoretischen und mathematischen Physik
Anwendungen finden.
\section*{Aufbau der Arbeit}
\label{sec:AufbauDerArbeit}
Die vorliegende Arbeit gliedert sich in vier Kapitel und mehrere
Anhänge.
Wir beginnen in {\bfseries Kapitel~\ref{cha:Deformationsquantisierung}}
mit einer kurzen Einführung in die Grundideen und Konzepte der
Deformationsquantisierung und führen insbesondere den Begriff des
Sternprodukts ein. Weiter gehen wir dort auf Symmetriebegriffe der
klassischen Mechanik ein, speziell wird der Begriff der Impulsabbildung
definiert und der wichtige Satz über die symplektische Reduktion
angegeben. Schließlich behandeln  wir in diesem Kapitel auch Symmetriebegriffe
für Sternprodukte, insbesondere wie man das Konzept der
Impulsabbildung für diesen Rahmen verallgemeinern kann.
In {\bfseries
  Kapitel~\ref{cha:Koszul-Reduktion}} widmen wir uns  den
Grundlagen der Quanten"=Koszul"=Reduktion. Zuerst formulieren wir die
geometrische Reduktionstheorie algebraisch und nehmen anschließend dies
als Ausgangspunkt um das Reduktions"=Schema der
Quanten"=Koszul"=Reduktion, wie es von Gutt und Waldmann in
\cite{gutt2010involutions} dargestellt wurde, formulieren zu
können.
Diese Methode betrachten wir in {\bfseries
  Kapitel~\ref{cha:QuantenKoszulAufKotangentialbuendel}} am Beispiel
eines Kotangentialbündels und können dort zeigen, dass die Methode von
Kowalzig, Neumaier und Pflaum aus \cite{kowalzig.neumaier.pflaum:2005a}
ein Spezialfall der Quanten"=Koszul"=Methode darstellt (vgl.\ Satz
\ref{satz:VergleichMitNikolaisArbeitAllgemein}). Dabei richten wir unser
Augenmerk zuerst auf den Fall verschwindender Impulswerte ohne
Magnetfeld, den wir anschließend durch Verschiebungstricks von
Tubenumgebungen auch auf nichtverschwindende, invariante Impulswerte
und in Anwesenheit eines Magnetfeldes verallgemeinern können. Dies stellt
das erste Hauptresultat der vorliegenden Arbeit dar.
Das zweite Hauptergebnis folgt in {\bfseries
  Kapitel~\ref{cha:Koszul-Reduktion_in_Stufen}}, in dem wir zeigen
können, dass unter bestimmten Annahmen an die Symmetriegruppe auch die
Quanten"=Koszul"=Reduktion in Stufen durchführbar ist und in zwei Stufen
das gleiche Ergebnis liefert wie in einem Schritt. Bevor wir zu diesem
Punkt kommen, beleuchten wir noch die klassische symplektische Reduktion
in Stufen, wie sie in \cite{marsden.misiolek:2007a} beschrieben wird. Es
ist ausreichend, dies für verschwindende Impulswerte zu tun, was einige
technische Schwierigkeiten beseitigt und in dieser Variante noch nicht
in der Literatur aufgeschrieben wurde. Anschließend wenden wir uns der
Frage nach der Existenz einer Quantenimpulsabbildung für das im ersten
Reduktionsschritt erhaltene Sternprodukt zu und geben hinreichende
Existenzbedingungen sowie eine explizite Konstruktion an, siehe
Satz~\ref{thm:ReduziertesSternproduktInvariant}. Nachdem wir noch die
stark invariante Situation beleuchtet haben, formulieren wir schließlich
in Satz~\ref{thm:UebereinstimmungDerSternprodukte} das oben
angesprochene Hauptergebnis.
Die Anhänge beschäftigen sich schließlich mit diversen mathematischen
Grundlagen. In {\bfseries Anhang~\ref{cha:elementMathAnmerkungen}}
bringen wir zunächst einige elementare mathematische Bemerkungen. In
{\bfseries Anhang~\ref{cha:Differentialoperatoren}} folgt eine äußerst
knappe Einführung in die Grundbegriffe von Differentialoperatoren auf
Mannigfaltigkeiten. In {\bfseries Anhang~\ref{cha:Hauptfaserbuendel}}
geben wir anschließend eine kurze Zusammenfassung der Theorie von
Hauptfaserbündeln. {\bfseries Anhang~\ref{cha:invariante_Strukturen}}
beschäftigt sich mit $G$"=invarianten Strukturen auf
$G$"=Mannigfaltigkeiten. Hier legen wir zunächst einige topologische
Grundlagen und betrachten dann, wie man bestimmte Strukturen auf
Mannigfaltigkeiten für eigentliche Gruppenwirkungen durch
Integrationstechniken invariant machen kann. Insbesondere betrachten wir
Riemannsche Fasermetriken und Spray"=Vektorfelder.  In {\bfseries
  Anhang~\ref{cha:Tubensatze}} bringen wir detaillierte Beweise für
Tubensätze, die für den Hauptstrang dieser Arbeit immer wieder benötigt
werden. Insbesondere beweisen wir die angekündigte Verallgemeinerung des
Satzes über kompatible Tubenumgebungen in Anwesenheit eigentlicher
Gruppenwirkungen, siehe Satz~\ref{satz:GkompatibleTuben}.
\cleardoublepage
\pagestyle{fancy}
\setlength\headheight{22.36pt}
\fancyhf{}
\fancyhead[OR]{\footnotesize\uppercase{Konventionen}}
\fancyhead[EL]{\footnotesize\uppercase{Konventionen}}
\fancyfoot[C]{\thepage}
\section*{Konventionen}
\label{cha:Notation}
An dieser Stelle möchten wir kurz einige Konventionen einführen, die wir
in dieser Arbeit immer wieder gebrauchen wollen.
Zunächst sei gesagt, dass wir die \emph{Einsteinsche Summenkonvention}
verwenden. Vektorräume sind immer über dem Körper der reellen Zahlen
$\mathbb{R}$ oder der komplexen Zahlen $\mathbb{C}$ zu verstehen, ebenso
Lie"=Algebren, sofern nicht anders gesagt. Auftretende
Mannigfaltigkeiten werden stets als glatt, reell und zusammenhängend
angenommen, dies gilt insbesondere für Lie"=Gruppen. Falls es uns
wichtig erscheint, werden wir dies stellenweise jedoch auch nochmal im
Haupttext erwähnen. Wirkungen von Lie"=Gruppen auf Mannigfaltigkeiten
sind immer Linkswirkungen und glatt. Weiter verwenden wir die Begriffe
Wirkung und Operation einer Gruppe synonym. $\mathbb{K}$"=lineare
Abbildungen zwischen zwei $\mathbb{K}$"=Vektorräumen ($\mathbb{K} \in
\{\mathbb{R},\mathbb{C}\}$) wollen wir stets
$\mathbb{K}[[\lambda]]$"=linear auf deren formale Potenzreihen im
formalen Parameter $\lambda$ fortsetzen und mit demselben Symbol
bezeichnen, ohne dies jedes Mal explizit zu erwähnen.
Ferner benutzen wir in dieser Arbeit die in der Mathematik
gängige Notation, wollen jedoch noch auf einige ausgewählte
Bezeichnungen aufmerksam machen.
Sei $M$ eine Menge und $\{m_i\}_{i \in I}$ eine Familie von Elementen
von $M$. Falls aus dem Kontext die Indexmenge $I$ klar ist, so begehen
wir den Notationsmissbrauch einfach $\{m_i\}$ statt $\{m_i\}_{i \in I}$
zu schreiben. Ebenfalls verzichten wir gelegentlich auf das
Kompositionszeichen $\circ$ von Abbildungen, um Formeln nicht zu lang
oder unübersichtlich werden zu lassen, sofern keine Verwechslungen zu
befürchten sind. Betrachten wir eine absteigende Sequenz von Abbildungen
der Form \def\tA[#1]{A_{#1}}
   \begin{equation*}
      \begin{tikzpicture}[baseline=(current
         bounding box.center),description/.style={fill=white,inner sep=2pt}]
         \matrix (m) [matrix of math nodes, row sep=3.0em, column
         sep=3.5em, text height=1.5ex, text depth=0.25ex]
         { \dots &K^{n-1} & K^{n} & K^{n+1} & \dots \\
         }; %
         \path[<-] (m-1-1) edge node[auto]{$f$}(m-1-2); %
         \path[<-] (m-1-2) edge node[auto]{$f$}(m-1-3); %
         \path[<-] (m-1-3) edge
         node[auto]{$f$}(m-1-4); %
         \path[<-] (m-1-4) edge node[auto]{$f$}(m-1-5);
      \end{tikzpicture}\Fcom
   \end{equation*}
   so schreiben wir statt $f \colon K^{n} \to K^{n-1}$ auch
   $f^{\scriptscriptstyle(n)}$, um den Definitionsbereich zu
   betonen. Analoges gilt für aufsteigende Sequenzen.
   Ist $V$ ein Vektorraum, so bezeichnen wir seinen algebraischen
   Dualraum mit $V^*$ und die adjungierte Abbildung einer linearen
   Abbildung $l \dpA V \to W$ zwischen zwei Vektorräumen $V$ und $W$
   notieren wir mit $l^* \dpA W^* \to V^*$. Die duale Paarung bezeichnen
   wir mit $\dPaar{\cdot}{\cdot}$. Außerdem verwenden wir die
   Bezeichnung $\bigwedge^\bullet V$ für $\bigoplus_k \bigwedge^k V$,
   falls wir die Gradierung hervorheben wollen. Erhöht oder verringert
   eine Abbildung von $\bigwedge^\bullet V$ in sich den Grad um $k$, so
   machen wir dies durch $\bigwedge^\bullet V \to \bigwedge^{\bullet \pm
     k}V$ deutlich. Für die Lie"=Algebra einer Lie"=Gruppe $G$ schreiben
   wir $\lieAlgebra[g]$ oder auch $\Lie[G]$.  Ist $G$ eine Lie"=Gruppe
   und $\Phi \dpA G \times M \to M$ eine $G$"=Wirkung, so bezeichnen wir
   das fundamentale Vektorfeld bezüglich dieser Wirkung zu einem Vektor
   $\xi \in \lieAlgebra[g]$ mit $\xi^{\scriptscriptstyle{G}}_M$, oder,
   falls wir die Gruppe nicht betonen wollen, einfach mit $\xi_M$. Ist
   \(W\) ein Vektorraum und \(G\) eine Gruppe, so heißt eine Wirkung
   \(\rho \colon G \times W \to W\) \neuerBegriff{linear}, falls
   \(\rho(g,\cdot) \colon W \to W\) für alle \(g \in G\) linear ist. In
   diesem Fall induziert \(\rho\) durch \(G \times W^* \ni
   (g,\alpha)\mapsto \rho^{*}(g,\alpha) :=
   (\rho(g^{-1},\cdot))^{*}(\alpha) \in W^*\) eine lineare Wirkung
   \(\rho^{*}\) von $G$ auf \(W^{*}\). Ist \(M_i\), \(i=1,2\), eine
   Menge und \(G\) eine Gruppe, die sowohl auf \(M_1\) als auch auf
   \(M_2\) wirke, so liefert dies eine Wirkung von \(G\) auf \(M_1
   \times M_2\) via \(g(m_1,m_2) := (gm_1,gm_2)\) für alle $g \in G,
   (m_1,m_2) \in M_1 \times M_2$. Sind $W_1$ und $W_2$ Vektorräume, die
   eine lineare $G$"=Wirkung tragen, so induziert auch dies eine
   $G$"=Wirkung auf $W_1 \otimes W_2$ via $G \times (W_1 \otimes W_2)
   \ni (g,w_1 \otimes w_2) \mapsto gw_1 \otimes gw_2 \in W_1 \otimes
   W_2$ und linearer Fortsetzung. Analoges gilt auch für mehrfache
   Tensorprodukte und äußere Potenzen.
Die glatten \emph{komplexwertigen} Funktionen auf einer Mannigfaltigkeit
$M$ bezeichnen wir immer mit $C^\infty(M)$, ist $f \in C^\infty(M)$, so
sei $\overline{f}$ die komplex konjugierte Funktion. Ist $G$ eine
Lie"=Gruppe, die auf $M$ und damit vermöge Zurückziehen auch auf
$C^\infty(M)$ wirkt, dann schreiben wir $C^\infty(M)^G$ für die
$G$"=invarianten, glatten komplexwertigen Funktionen auf $M$. Ist $\phi
\colon M \to M'$ eine glatte Abbildung zwischen Mannigfaltigkeiten, so
notieren wir den pull"=back als $\phi^* \colon C^\infty(M') \to
C^\infty(M)$. Es sollte immer aus dem Kontext klar sein, wann $^*$
den pull"=back und wann die duale Abbildung bezeichnet. Den Träger einer
Funktion $f \colon M \to V$ von einem topologischen Raum $M$ in einen
Vektorraum $V$ bezeichnen wir mit $\supp f = \abschluss{\carr f}$, wobei
$\carr f := \{p \in M \mid f(p) \neq 0\}$. Für die glatten Schnitte auf
einem Vektorbündel $E \to M$ schreiben wir $\Gamma^\infty(E)$.
\cleardoublepage
\pagestyle{fancy}
\setlength\headheight{22.36pt}
\fancyhf{}
\fancyhead[OR]{\footnotesize\rightmark}
\fancyhead[EL]{\footnotesize\leftmark}
\fancyfoot[C]{\thepage}
\renewcommand{\thepage}{\arabic{page}}
\setcounter{page}{1}
\chapter{Deformationsquantisierung}
\label{cha:Deformationsquantisierung}
Eine physikalische Theorie verwendet in der Regel die Sprache der
Mathematik, um Teile der Wirklichkeit zu beschreiben. Es ist nicht
verwunderlich, dass eine Theorie $\mathcal{T}_1$, die einen echt
größeren Wirklichkeitsbereich beschreibt als eine Theorie $\mathcal{T}_2$,
auch eine allgemeinere -- und nicht etwa eine völlig andere --
mathematische Struktur zu dessen Beschreibung verwendet. Ein schönes
Beispiel liefert etwa das Verhältnis der klassischen Punktmechanik zur
speziellen Relativitätstheorie oder die letztere im Verhältnis zur
allgemeinen Relativitätstheorie. Bei der klassischen Mechanik und der
Quantentheorie sieht man dies in verbandstheoretischer oder
algebraischer Formulierung sehr schön \cite{piron1976foundations}.
Die oben genannte echte Erweiterung des Wirklichkeitsbereichs ist oft,
und so auch in den genannten Beispielen, mit einer Änderung der Skala einer
physikalischen Größe verbunden. Es stellt sich heraus, dass sich die
Erweiterung des mathematischen Formalismus gewissermaßen durch eine
Deformation eines Repräsentanten dieser Größe innerhalb der mathematischen
Struktur der Theorie kontrollieren lässt.
Oftmals zeigt sich so a posteriori die allgemeinere Theorie nach
Abstecken des mathematischen Rahmens gewissermaßen als zwingend, da sie
stabil unter Deformationen ist. Es ist als nicht zu unterschätzende
Eigenschaft einer guten physikalischen Theorie anzusehen, dass sich ihre
charakteristische mathematische Struktur nicht durch kleine Variation
von bestimmten Konstanten der Theorie ändert, da deren Werte letztlich
gemessen werden und somit auch unweigerlich mit Messfehlern behaftet
sind.
Die \neuerBegriff{Deformationsquantisierung} beschäftigt sich nun mit
der Anwendung der oben beschriebenen Heuristik, die oft auch
\neuerBegriff{Deformationsphilosophie} genannt wird
(vgl.\ \cite{MR2285047}), auf das Verhältnis von klassischer Mechanik und
Quantenmechanik oder auch auf das von klassischen zu
Quantenfeldtheorien.  Wir beschränken uns in dieser Arbeit auf den
ersten Fall, der im Übrigen auch in Termen von Gittereichtheorien
interpretiert werden kann. Im Folgenden wollen wir kurz erläutern, in
welchem Rahmen deformiert wird und führen den grundlegenden Begriff des
Sternprodukts ein. Dazu präsentieren wir zuerst in aller Kürze eine
mathematische Beschreibung der klassischen Mechanik, anschließend
erläutern wir die der Quantenmechanik und schließlich legen wir dar, was
man in diesem Zusammenhang mit Deformationen meint. Es würde jedoch den
Rahmen dieser Arbeit bei weitem sprengen, wenn wir versuchen wollten,
eine detaillierte Einführung in das Quantisierungsproblem und die
Deformationsquantisierung zu geben. Wir verweisen hierfür auf das sehr
schön geschriebene Lehrbuch von Waldmann \cite{waldmann:2007a}.
\section{Sternprodukte}
\label{sec:Sternprodukte}
Die Observablenalgebra in der klassischen Mechanik beschreibt man
typischerweise und idealisiert\footnote{Vgl.\ dazu
  \cite{waldmann:2007a}, insbesondere Bemerkung 1.3.7.} durch die
kommutative Poisson"=$^*$"=Algebra der glatten komplexwertigen
Funktionen auf dem Phasenraum $\mathbb{R}^{2n}$, der mit der kanonischen
Poisson"=Klammer versehen wird, oder allgemeiner durch die auf einer
Poisson"=Mannigfaltigkeit $(M,\{\cdot,\cdot\})$.  Die Poisson"=Klammer
bestimmt zusammen mit der Hamiltonfunktion die Zeitentwicklung des
Systems. \neuerBegriff{Observabel} im eigentlichen Sinne sind die
reellwertigen Funktionen $f = \overline{f} \in C^\infty(M)$. Die
möglichen Messwerte sind durch die Funktionswerte gegeben. Die reinen
Zustände  sind die Punkte auf dem Phasenraum und allgemeine
Zustände werden durch positive Borelmaße mit kompaktem Träger
beschrieben. Mit Hilfe des Rieszschen Darstellungssatzes (vgl.\
\cite{rudin:1987a}) können diese alternativ als positive Funktionale
angesehen werden.
In der Quantenmechanik beschreibt man die Observablen üblicherweise als
selbstadjungierte Opertoren auf einem separablen, komplexen, im
Allgemeinen unendlich dimensionalen Hilbertraum $\mathcal{H}$. Die
reinen Zustände eines Systems werden durch Strahlen in $\mathcal{H}$
repräsentiert, allgemeine Zustände beschreibt man durch Dichtematrizen,
d.\,h.\ durch positive Spurklasse"=Operatoren auf $\mathcal{H}$ mit Spur
Eins, vgl.\ \cite{blank2008hilbert}. Typische Observablen wie der Orts"=
oder Impulsoperator sind unbeschränkt, insbesondere nicht auf ganz
$\mathcal{H}$ definiert. Man kann sich jedoch überlegen, dass es
ausreicht, nur die beschränkten Operatoren auf $\mathcal{H}$ zu
betrachten, siehe etwa \cite[Ch.~IV~2.3]{prugovecki1981quantum}.
Es stellt sich weiter als sinnvoll heraus, die im Allgemeinen
nichtkommutative $^*$"=Algebra $\mathfrak{B}(\mathcal{H})$ aller
beschränkten Operatoren auf $\mathcal{H}$ mit Operatoradjunktion als
$^*$"=Involution als Observablenalgebra anzusehen und die
selbstadjungierten, die hier gleich den Hermiteschen Operatoren sind,
als die \neuerBegriff{observablen} Elemente darin zu
bezeichnen. Zustände können dann als positive, normierte Funktionale auf
$\mathfrak{B}(\mathcal{H})$ angesehen werden. Um den wohlverstandenen
funktionalanalytischen Fragen der Einfachheit halber im Rahmen dieser
Motivation aus dem Weg zu gehen, wollen wir im Folgenden $\mathcal{H}$
nur als Prä"=Hilbertraum auffassen, was auch ausreichend ist, solange
man den Spektralkalkül nicht benötigt.
Ein verbreitetes Vorgehen, ein System der klassischen Mechanik mit
Phasenraum $\mathbb{R}^{2n}$ zu "`quantisieren"' besteht darin, der
Ortsvariablen $q^k$ den Ortsoperator $Q^k$ zuzuordnen und der
Impulsvariablen $p_k$ den Impulsoperator $P_k$, dabei werden die Orts-
und Impulsoperatoren üblicherweise wie folgt definiert.
\begin{align}
   \label{eq:Ortssoperator}
   Q^k \colon C_0^\infty(\mathbb{R}^n) \ni \psi \mapsto (\mathbb{R}^{n} \ni
   q \mapsto (Q^k\psi)(q) = q^k \psi(q)) \in C_0^\infty(\mathbb{R}^n)
\end{align}
und
\begin{align}
   \label{eq:Impulsoperator}
   P_k \colon C_0^\infty(\mathbb{R}^n) \ni \psi \mapsto (\mathbb{R}^{n} \ni
   q \mapsto (P_k\psi)(q) = \frac{\hbar}{\I} \frac{\partial
     \psi}{\partial q^k}(q)) \in C_0^\infty(\mathbb{R}^n)\Fdot
\end{align}
Dabei bezeichnet $C^\infty_0(\mathbb{R}^n)$ den Prä"=Hilbertraum der
glatten Funktionen mit kompaktem Träger auf $\mathbb{R}^n$.
Nun ist schon bei einfachsten klassischen Observablen, etwa der Form
$q^kp_k$, nicht klar, wie diese grobe Regel anzuwenden ist, denn auf der
einen Seite vertauschen $q^k$ und $p_k$, die korrespondierenden
Operatoren $Q^k$ und $P_k$ tun dies jedoch nicht. Man muss also eine
Ordnungsvorschrift wählen. Die einfachste Möglichkeit dies zu tun, ist
die sogenannte \neuerBegriff{Standardordnung}. Hat man eine in $q^k$ und
$p_l$ polynomiale Funktion vorliegen, so schreibt man erst alle Impulse
nach rechts und ersetzt anschließend die Orts- und Impulsvariablen wie
oben beschrieben. Man kann zeigen, dass so eine $\mathbb{C}$"=lineare
Bijektion $\rho_{\mathrm{std}}\colon \mathrm{Pol}(\mathbb{R}^{2n}) \to
\mathrm{DiffOp}_{\mathrm{Pol}}(C^\infty_0(\mathbb{R}^n))$ von den
polynomialen Funktionen $\mathrm{Pol}(\mathbb{R}^{2n})$ auf
$\mathbb{R}^{2n}$ in die Differentialoperatoren auf
$C^\infty_0(\mathbb{R}^n)$ mit polynomialen Koeffizienten definiert
wird, siehe \cite[Prop. 5.2.7]{waldmann:2007a}. Analog liefern viele
andere Ordnungsvorschriften genauso $\mathbb{C}$"=lineare Bijektionen
$\rho \colon \mathrm{Pol}(\mathbb{R}^{2n}) \to
\mathrm{DiffOp}_{\mathrm{Pol}}(C^\infty_0(\mathbb{R}^n))$. Die grobe
Idee besteht nun darin, mit Hilfe eines derartigen $\rho$ die
Komposition von Differentialoperatoren auf den Raum $
\mathrm{Pol}(\mathbb{R}^{2n})$ zurückzuziehen und die so erhaltene
assoziative Multiplikation mit der punktweisen, kommutativen
Multiplikation zu vergleichen. Für $\rho_{\mathrm{std}}$ erhält man so
das sogenannte \neuerBegriff{standardgeordnete Sternprodukt}
$\star_{\mathrm{std}}$.
\begin{proposition}
   \label{prop:StandardgeordnetesSternprodukt}
   Das standardgeordnete Sternprodukt $\star_{\mathrm{std}}$ ist für
   alle $f,f' \in \mathrm{Pol}(\mathbb{R}^{2n})$ durch
   \begin{align}
      \label{eq:StandardgeordnetesSternprodukt}
      f \star_{\mathrm{std}} f' =
      \rho_{\mathrm{std}}^{-1}(\rho_{\mathrm{std}}(f)\rho_{\mathrm{std}}(f'))
      = \sum_{r = 0}^\infty \frac{1}{r!} \left(\frac{\hbar}{\I} \right)^r
      \sum_{i_1,\dots,i_r} \frac{\partial^r f}{\partial p_{i_1}
        \dotsm \partial p_{i_r}} \frac{\partial^r f'}{\partial q^{i_1}
        \dotsm \partial q^{i_r}}
   \end{align}
   gegeben. Dabei sind $(q^1,\dots,q^n,p_1,\dots,p_n)$ Koordinaten von
   $\mathbb{R}^{2n}$, wobei die $q^i$ als Orts- und die $p_i$ als
   Impulskoordinaten interpretiert werden.
\end{proposition}
Einen Beweis dieser Tatsache findet man in \cite[Prop. 5.2.17]{waldmann:2007a}.
\begin{bemerkung}
   \label{bem:PolynomialSternprodukt}
   Offensichtlich kann man die Definition von $\star_{\mathrm{std}}$ ohne
   Weiteres auf Funktionen ausdehnen, die nur in den Impulsvariablen
   polynomial sind.
\end{bemerkung}
\begin{bemerkung}
   \label{bem:Hermitesch}
   Aus physikalischer Sicht hat die Standardordnung einen wesentlichen
   Nachteil, denn sie bildet observable Elemente $f = \overline{f} \in \mathrm{Pol}(\mathbb{R}^{2n})$ im
   Allgemeinen nicht wieder auf observable Elemente ab. So ist zum
   Beispiel der Operator $\rho_{\mathrm{std}}(q^1p_1)$ kein
   Hermitescher Operator bezüglich der üblichen
   Prä"=Hilbertraumstruktur von $C^\infty_0(\mathbb{R}^n)$. Für
   $\star_{\mathrm{std}}$ spiegelt sich dieser Defekt darin wieder, dass
   im Allgemeinen die Gleichung $\overline{f \star_{\mathrm{std}} f'} =
   \overline{f'} \star_{\mathrm{std}}  \overline{f}$ verletzt ist. Wir
   werden später, nachdem der allgemeine Begriff des Sternprodukts
   eingeführt wurde, Beispiele sehen, wo diese Problematik nicht auftritt.
\end{bemerkung}
Diese Überlegungen münden letztlich in den Begriff des Sternprodukts,
den wir nun definieren wollen. Für weitere Details zur beschriebenen
Motivation müssen wir leider wieder auf \cite{waldmann:2007a}
verweisen. Ebenso verweisen wir für die im Folgenden benötigten
Grundlagen zur Differentialgeometrie, symplektischen und
Poisson"=Geometrie, die wir alle in dieser Arbeit voraussetzen wollen,
auf die zahlreiche Literatur. Es seien exemplarisch \cite{lee:2003a},
\cite{ortega.ratiu:2004} und \cite{waldmann:2007a} genannt.
\begin{definition}[Sternprodukt]
   \label{def:Sternprodukt}
   Ein \neuerBegriff{formales Sternprodukt} oder kurz
   \neuerBegriff{Sternprodukt} $\star$ für eine Poisson"=Mannigfaltigkeit
   $(M,\{\cdot,\cdot\})$ ist eine $\mathbb{C}[[\lambda]]$"=bilineare Abbildung
   \begin{align}
      \label{eq:Sternprodukt}
      \star \colon C^\infty(M)[[\lambda]] \times C^\infty(M)[[\lambda]] \to
      C^\infty(M)[[\lambda]]
   \end{align}
   der Form
   \begin{align}
      \label{eq:FormSternprodukt}
      \star = \sum_{r=0}^\infty \lambda^r \star_r
   \end{align}
   mit $\mathbb{C}$"=bilinearen Abbildungen $\star_r \colon C^\infty(M)
   \times C^\infty(M) \to C^\infty(M)$ welche auf die übliche Weise
   $\mathbb{C}[[\lambda]]$"=bilinear fortgesetzt werden, so dass für
   alle $f,f' \in C^\infty(M)$ folgende Bedingungen erfüllt sind.
   \begin{definitionEnum}
      \item %
         $\star$ ist assoziativ.
      \item %
         $f\star_0 f' = f\cdot f'$.
      \item %
         $f \star_1 f' - f' \star_1 f =  \I \{f,f'\}$.
      \end{definitionEnum}
      Falls zusätzlich $\star_r$ für alle $r \in \mathbb{N}$ ein
      Bidifferentialoperator ist, heißt $\star$
      \neuerBegriff{differentiell}.  Gilt $\overline {(f \star f')} =
      \overline{f'} \star \overline{f}$ für alle $f,f' \in C^\infty(M)$,
      nennt man $\star$ \neuerBegriff{Hermitesch}. Dabei definiert man
      $\overline{\lambda} = \lambda$.
\end{definition}
\begin{bemerkung}
   \label{bem:Unteralgebra}
   Wir wollen im Folgenden auch Produkte, die nur auf formalen
   Potenzreihen von Poisson-$^*$-Unteralgebren von
   $C^\infty(M)[[\lambda]]$ definiert sind, als Sternprodukte bezeichnen,
   sofern sie die sonstigen Eigenschaften aus Definition
   \ref{def:Sternprodukt} erfüllen.
\end{bemerkung}
\begin{bemerkung}
   \label{bem:KonvergenzProbleme}
   Interpretiert man $M$ als Phasenraum eines physikalischen Systems, so
   kann man sich überlegen, dass es sinnvoll ist anzunehmen, dass
   $\star_r$ für alle $r \in \mathbb{N}$ die Dimension von
   $\frac{1}{[\mathrm{Wirkung}]^r}$ hat. Grob gesagt legt das Beispiel
   des Kotangentialbündels nahe, die Koordinatenfunktionen auf einem
   beliebigen Phasenraum mit $\sqrt{[\mathrm{Wirkung}]}$ als
   physikalische Dimension zu interpretieren, da es auf dem
   Kotangentialbündel genausoviele Koordinaten gibt, die man als
   Ortskoordinaten interpretieren kann, wie solche, die man als
   Impulskoordinaten ansehen kann und deren Produkte somit die Einheit
   einer Wirkung tragen. Dann überlegt man sich im Wesentlichen durch
   Betrachten der Differentiationsordnung in der Fedosovkonstruktion
   (vgl.\ \cite[Abschnitt 6.4]{waldmann:2007a}), dass bei einem
   Fedosov"=Sternprodukt $\star_r$ die Dimension
   $\frac{1}{[\mathrm{Wirkung}]^r}$ hat. Da jedes differentielle
   Sternprodukt zu einem Fedosovsternprodukt äquivalent ist (vgl.\
   \cite[Satz 6.4.27]{waldmann:2007a}), ist es naheliegend, dies auch
   für beliebige differentielle Sternprodukte so zu interpretieren.
   Weiter kann man $\lambda$ als $\hbar$ auffassen, welches ja die
   Dimension einer Wirkung hat. In dieser Sichtweise wollen wir für
   $\star$ das Symbol $\star^{\scriptscriptstyle \hbar}$
   verwenden. Somit ist für alle $r \in \mathbb{N}$ die Größe
   ${\star^{\scriptscriptstyle{\hbar}}_r}\hbar^r$ dimensionslos und man
   kann für zwei Funktionen $f,f' \in C^\infty(M)$ sinnvoll nach der
   Konvergenz der Reihe $\sum_{r = 0}^\infty \hbar^r f
   {\star^{\scriptscriptstyle{\hbar}}_r} f'$ fragen. Es ist bisher im
   Allgemeinen nicht klar, ob und wie man eine sinnvolle Topologie
   angeben kann, in der diese Konvergenz zu verstehen ist. Außerdem wäre
   es wünschenswert, eine Unteralgebra $\mathcal{A}$ von
   $C^\infty(M)[[\lambda]]$ auszuzeichnen, bei der für alle $f,f' \in
   \mathcal{A}$ die Reihe $f\star^{\scriptscriptstyle{\hbar}}f'$ in
   einem bestimmten Sinne konvergiert. Siehe
   \cite{beiser.roemer.waldmann:2007a} und \cite{beiser:2005a} für
   nichttriviale Beispiele. Für weitere Betrachtungen zum
   Konvergenzproblem sei auch auf das Buch von Fedosov
   \cite{fedosov:1996a} verwiesen. In den Beispielen
   \ref{bsp:WeylSternprodukt} und \ref{bsp:WickSternprodukt} werden wir
   noch einfache, aber physikalisch relevante Situationen angeben, in
   denen das Konvergenzproblem einfach zu lösen ist.
\end{bemerkung}
\begingroup
\emergencystretch=0.6em
\begin{definition}[Äquivalenz von Sternprodukten]
   \label{def:Aequivalenz}
   Sei $\star$ bzw.\ $\star'$ ein formales Sternprodukt auf der
   Poisson"=Mannigfaltigkeit $(M,\{\cdot,\cdot\})$ bzw.\ $(M',\{\cdot,\cdot\}')$.
   \begin{definitionEnum}
   \item %
      Die formalen Sternprodukte $\star$ und $\star'$ heißen
      \neuerBegriff{isomorph}, falls es einen
      $\mathbb{C}[[\lambda]]$"=Algebraisomorphismus $S \colon
      (C^\infty(M)[[\lambda]],\star) \to
      (C^\infty(M')[[\lambda]],\star')$ gibt. Kann $S$ als
      $\mathbb{C}[[\lambda]]$"=lineare Fortsetzung eines
      Poisson"=Algebraisomorphismus $s \colon
      (C^\infty(M),\{\cdot,\cdot\}) \to (C^\infty(M'),\{\cdot,\cdot\}')$
      gewählt werden, so wollen wir $\star$ und $\star'$
      \neuerBegriff{Poisson"=isomorph} nennen. Sind
      $(M,\{\cdot,\cdot\})$ und $(M',\{\cdot,\cdot\}')$ symplektisch, so
      sagen wir auch $\star$ und $\star'$ sind
      \neuerBegriff{symplektomorph}.
   \item %
      Gilt $(M,\{\cdot,\cdot\}) = (M',\{\cdot,\cdot\}')$, so heißen
      $\star$ und $\star'$ \neuerBegriff{äquivalent}, falls es einen
      $\mathbb{C}[[\lambda]]$"=Algebraisomorphismus $S \colon
      C^\infty(M)[[\lambda]] \to C^\infty(M')[[\lambda]]$ der Form $S =
      \id + \sum_{r = 1}^\infty \lambda^r S_r$ von
      $(C^\infty(M)[[\lambda]],\star)$ nach
      $(C^\infty(M')[[\lambda]],\star')$ gibt.
   \end{definitionEnum}
\end{definition}
\begin{bemerkung}
   \label{bem:MilnorsExercise}
   Eine Abbildung $s \colon C^\infty(M) \to C^\infty(M')$ ist genau dann
   ein Algebraisomorphismus, wenn es einen Diffeomorphismus $\phi \colon
   M' \to M$ mit $s = \phi^*$ gibt (\emph{Milnors Exercise}
   \cite{grabowski:2005a}). Insbesondere kann man sich so jeden
   Poisson"=Algebraisomorphismus $s \colon (C^\infty(M),\{\cdot,\cdot\})
   \to (C^\infty(M'),\{\cdot,\cdot\}')$ als pull-back einer
   Poisson"=Abbildung denken.
\end{bemerkung}
\begin{bemerkung}
   \label{bem:AequivalenzTransformation}
   Ist $\star'$ ein Sternprodukt für $(M',\{\cdot,\cdot\}')$, so liefert
   eine Abbildung
   \begin{equation}
      \label{eq:AequivalenzTransformation}
      S = S_0 +
      \sum^{\infty}_{r=1} \lambda^r S_r \colon C^\infty(M')[[\lambda]]
      \to C^\infty(M)[[\lambda]]
   \end{equation}
   mit einem Poisson"=Algebraisomorphismus $S_0 \colon C^\infty(M') \to
   C^\infty(M)$ und $\mathbb{C}$-linearen Abbildungen $S_r \colon
   C^\infty(M') \to C^\infty(M)$ sowie $S_r(1) = 0$ für $r >
   0$ durch
   \begin{equation}
      \label{eq:AequivalenzTransformation2}
      f \star \tilde f = S(S^{-1}(f) \star' S^{-1}(\tilde f))
   \end{equation}
   für alle $f,\tilde f \in C^\infty(M)[[\lambda]]$ ein Sternprodukt $\star$
   für $(M,\{\cdot,\cdot\})$, siehe dazu auch
   \cite[Prop.~6.1.7]{waldmann:2007a}. Falls $S = \phi^*$ für einen
   Poisson"=Diffeomorphismus $\phi \colon M \to M'$ gilt, so schreiben
   wir auch $\phi^*\star' := \star$ und sprechen von einem (mit $\phi$)
   \neuerBegriff{zurückgezogenen} Sternprodukt.
\end{bemerkung}
\endgroup
Wir kommen nun zu zwei einfachen Beispielen für Sternprodukte, das \neuerBegriff{Weyl-}
und das \neuerBegriff{Wick"=Sternprodukt}.

\begin{beispiel}[Weyl"=Sternprodukt]
   \label{bsp:WeylSternprodukt}
   Auf $M = \mathbb{R}^{2n}$ mit den Ortskoordinaten $q^k$ und den
   Impulskoordinaten $p_k$ definiert
   \begin{align}
      \label{eq:WeylSternprodukt}
      f \star_{\mathrm{\scriptscriptstyle{Weyl}}} f' := \mu \circ \exp(-\frac{\I\lambda}{2}
      \omega^{rs} \frac{\partial}{\partial x^r} \otimes
      \frac{\partial}{\partial x^s} ) f\otimes f'
   \end{align}
   für alle $f,f' \in C^\infty(\mathbb{R}^{2n})$ und
   $\mathbb{C}[[\lambda]]$"=lineare Fortsetzung ein Sternprodukt auf
   $(\mathbb{R}^{2n},\omega)$. Dabei ist $\omega = dq^i \wedge d p_i$
   die kanonische symplektische Form auf $\mathbb{R}^{2n}$ mit
   Fundamentalmatrix $(\omega_{ij})$ und deren
   inverser Matrix $(\omega^{ij})$. Die Abbildung $\mu \colon
   C^\infty(\mathbb{R}^{2n}) \otimes C^\infty(\mathbb{R}^{2n}) \to
   C^\infty(\mathbb{R}^{2n})$ bezeichnet die punktweise Multiplikation
   von Funktionen. Das Sternprodukt
   $\star_{\mathrm{\scriptscriptstyle{Weyl}}}$ heißt auch
   \neuerBegriff{Weyl"=Moyal"=Sternprodukt} oder kurz
   \neuerBegriff{Weyl"=Sternprodukt}. Das Weyl"=Sternprodukt ist
   differentiell und Hermitesch.  Für die Unteralgebra der in den
   Impulskoordinaten polynomialen Funktionen sieht man leicht, dass die
   Exponentialreihe abbricht.
\end{beispiel}
\begin{beispiel}[Wick"=Sternprodukt]
   \label{bsp:WickSternprodukt}
   Es sei $M = \mathbb{C}^n = \mathbb{R}^{2n}$ mit der kanonischen
   symplektischen Form $\omega$ versehen. Dann definiert
   \begin{align}
      \label{eq:WickSternprodukt}
      f \star_{\mathrm{\scriptscriptstyle{Wick}}} f' := \mu \circ \exp(2 \lambda
      \frac{\partial}{\partial z^k} \otimes \frac{\partial}{\partial
        \overline{z}^k}) f \otimes f'
   \end{align}
   für alle $f,f' \in C^\infty(\mathbb{C}^n)$ und
   $\mathbb{C}[[\lambda]]$"=lineare Fortsetzung ein Sternprodukt für
   $(\mathbb{R}^{2n},\omega)$. Dabei ist $\mu$ definiert wie in Beispiel
   \ref{bsp:WeylSternprodukt} und $z^i$, $\overline{z}^i$ sind die
   holomorphen und anti"=holomorphen Standardkoordinaten der reellen
   Mannigfaltigkeit $\mathbb{C}^n$. Man nennt
   $\star_{\mathrm{\scriptscriptstyle{Wick}}}$
   \neuerBegriff{Wick"=Sternprodukt}. Es ist differentiell und
   Hermitesch. Wählt man die in $z^i$ und $\overline{z}^i$ polynomialen
   Funktionen als Unteralgebra der glatten Funktionen auf
   $\mathbb{C}^n$, so hat man auch hier nach Ersetzen von $\lambda$
   durch $\hbar$ eine konvergente Situation vorliegen.
\end{beispiel}
Für weitere konkrete Beispiele von Sternprodukten verweisen wir auf
\cite{waldmann:2007a}, \cite{waldmann:1998a} und \cite{gutt:1983a}.
Nach dem bisher Gesagten ergibt sich freilich die Frage, ob Sternprodukte
immer existieren und ob man diese auf sinnvolle Weise klassifizieren
kann. Dies wurde in voller Allgemeinheit, d.\,h.\ für beliebige
Poisson"=Mannigfaltigkeiten $1997$ von Kontsevich beantwortet, jedoch
erst 2003 publiziert (vgl.\ \cite{kontsevich:2003a}). Wir fassen sein
Resultat in folgendem Satz zusammen.
\begin{satz}[Existenz und Klassifikation von Sternprodukten auf
   Poisson"=Mannigfaltigkeiten]
   \label{thm:KontsevichExistenzUndKlassifikation}
   Sei $(M,\{\cdot,\cdot\})$ eine Poisson"=Mannigfaltigkeit mit
   Poisson"=Tensor $\pi$ (vgl.\ \cite[Def. 4.1.7]{waldmann:2007a}), dann
   gilt Folgendes.
   \begin{satzEnum}
   \item %
      Es gibt ein differentielles Sternprodukt für
      $(M,\{\cdot,\cdot\})$.
   \item%
      Die Äquivalenzklassen von Sternprodukten auf $(M,\{\cdot,\cdot\})$
      stehen zu den Äquivalenzklassen von formalen Deformationen
      (vgl.\ \cite[Def. 4.2.35]{waldmann:2007a}) des Poisson"=Tensors
      $\pi$ modulo formalen Diffeomorphismen im Sinne von
      \cite[Def. 4.2.40]{waldmann:2007a} in Bijektion.
   \end{satzEnum}
\end{satz}
Das Problem dieser Aussage besteht darin, dass die formalen
Poisson"=Tensoren modulo formalen Diffeomorphismen im Allgemeinen sehr
unzugänglich sind.
Anders als im Poisson"=Fall ist die Existenz eines Sternprodukts auf
einer symplektischen Mannigfaltigkeit aufgrund des Darboux"=Theorems
(vgl.\ \cite[Satz 3.1.24]{waldmann:2007a}) zumindest lokal klar. Die
Schwierigkeit ist hier, zu zeigen, dass auch tatsächlich ein
globales Sternprodukt existiert. DeWilde und Lecomte gelang der erste
allgemeine Existenzbeweis für Sternprodukte auf symplektischen
Mannigfaltigkeiten schon 1983, siehe \cite{dewilde.lecomte:1983b} sowie
\cite{gutt.rawnsley:1999a} für eine neuere Darstellung. Einen anderen
unabhängigen und viel geometrischeren Beweis konnte Fedosov in einer
zunächst wenig beachteten Arbeit angeben, siehe
\cite{fedosov:1985a,fedosov:1986a,fedosov:1989a}. Ein dritter Beweis
stammt von Omori, Maeda und Yoshida
\cite{omori.maeda.yoshioka:1991a}. Mit Hilfe der Fedosov Konstruktion
gelang es die Äquivalenzklassen von Sternprodukten auf symplektischen
Mannigfaltigkeiten zu klassifizieren, insbesondere sind hier die
Arbeiten von Nest und Tsygan \cite{nest.tsygan:1995a,nest.tsygan:1995b},
Deligne \cite{deligne:1995a}, Bertelson, Cahen und Gutt
\cite{bertelson.cahen.gutt:1997a} und Weinstein und Xu
\cite{weinstein.xu:1998a} zu nennen.
\begin{satz}[Klassifikation von Sternprodukten auf symplektischen
   Mannigfaltigkeiten]
   \label{satz:KlassifikationSymplektisch}
   Sei $(M,\omega)$ eine symplektische Mannigfaltigkeit. Dann stehen die
   Äquivalenzklassen von differentiellen Sternprodukten zu den formalen
   Potenzreihen $\mathrm{H}^2_{\mathrm{dR}}(M,\mathbb{C})[[\lambda]]$,
   der zweiten de"=Rham"=Kohomologie von $M$ mit komplexen
   Koeffizienten in Bijektion.
\end{satz}
\begin{bemerkung}[Fedosov"=Konstruktion]
   \label{bem:Fedosov}
   Die Fedosovkonstruktion liefert für jedes geschlossene $\Omega \in
   \lambda \Gamma^\infty(\Bigwedge^2 T^*M)[[\lambda]]$ ein differentielles
   Sternprodukt $\star_{\Omega}$, das sogenannte
   \neuerBegriff{Fedosov"=Sternprodukt}. Es zeigt sich, dass erstens zwei
   Fedosov"=Sternprodukte $\star_{\Omega}$ und $\star_{\Omega'}$ genau
   dann äquivalent sind, wenn $\Omega - \Omega'$ exakt ist, d.\,h.\ in
   $H^2_{\mathrm{dR}}(M,\mathbb{C})[[\lambda]]$ die Gleichung $[\Omega]
   = [\Omega]$ gilt und zweitens dass jedes differentielle Sternprodukt
   zu einem Fedosovsternprodukt äquivalent ist. Für eine ausführliche
   Lehrbuchdarstellung dieses Sachverhaltes verweisen wir auf
   \cite[Abschnitt 6.4]{waldmann:2007a}. Die Fedosovkonstruktion hat außer
   $\Omega$ noch weitere Eingangsdaten, etwa eine torsionsfreie
   kovariante Ableitung $\nabla$. Allerdings ändert eine Variation
   dieser zusätzlichen Daten nicht die Äquivalenzklasse des
   erhaltenen Sternprodukts, weshalb wir sie hier unterschlagen wollen.
\end{bemerkung}
Satz \ref{satz:KlassifikationSymplektisch} besagt im Speziellen, dass
auf einer Mannigfaltigkeit mit verschwindender zweiter de"=Rham
Kohomologie, wie etwa dem $\mathbb{R}^n$, $n \in \mathbb{N}$, alle
Sternprodukte äquivalent sind. Insbesondere sind nach dem
Poincar\'e"=Lemma im symplektischen Fall auch alle Sternprodukte lokal
äquivalent, womit sich eine Nicht"=Äquivalenz zweier Sternprodukte als
globaler Effekt erweist. Ersteres legt auch nahe, die obigen Ergebnisse
als Klassifikation bis auf Wahl einer (verallgemeinerten)
Ordnungsvorschrift zu interpretieren. Innerhalb einer Äquivalenzklasse von
Sternprodukten kann es durchaus verschiedene Sternprodukte geben, die
physikalisch als nicht äquivalent anzusehen sind.
Bisher haben wir nur die Observablenalgebra betrachtet, jedoch noch
nicht darüber gesprochen, wie wir Zustände in der
Deformationsquantisierung beschreiben wollen. Auch wenn dies für die
vorliegende Arbeit nicht direkt von Belang ist, wollen wir auf diesen
Punkt der allgemeinen Übersicht halber in aller Kürze kurz eingehen.
Es stellt sich als sinnvoll heraus, in der Deformationsquantisierung,
wie auch in anderen Zugängen zur Quantenmechanik, Zustände als ein von
den Observablen abgeleitetes Konzept anzusehen, nämlich als positive
Funktionale auf der Observablenalgebra. Mit Hilfe des Begriffs konvexer
Zerlegbarkeit, kann man von gemischten und reinen Zuständen sprechen. Um
das für die Quantenmechanik wichtige Superpositionsprinzip auch in die
Deformationsquantisierung implementieren zu können, kann man
Darstellungen von Sternproduktalgebren betrachten.  Wir können an dieser
Stelle jedoch nicht weiter darauf eingehen. Siehe \cite{waldmann:2007a}
und die Referenzen dort für weitergehende Betrachtungen. Allgemein zum
Zustandskonzept und Superpositionsprinzip sehr lesenswert ist das Buch
von Araki \cite{araki1999mathematical}.
\section{Symmetrien und Invarianzen}
\label{sec:SymmetrienUndInvarianzen}
Die in dieser Arbeit betrachteten Symmetrien sind mathematisch immer
durch eine Gruppenwirkung einer Lie"=Gruppe $G$ auf einer
Mannigfaltigkeit $M$ gegeben, wobei letztere als Phasenraum eines
klassischen mechanischen Systems interpretiert werden kann. Wenn wir von
Lie"=Gruppe sprechen, meinen wir eine zusammenhängende
Lie"=Gruppe. Sofern nichts anderes gesagt, bezeichnet $\lieAlgebra$
stets die Lie"=Algebra von $G$.
\subsection{Impulsabbildungen und Phasenraumreduktion in der klassischen Mechanik}
\label{sec:ImpulsAbbKlassMechanik}
In diesem Abschnitt wiederholen wir den für diese Arbeit wichtigen
Begriff der Impulsabbildung und stellen die Phasenraum"=Reduktion nach
Marsden"=Weinstein vor. Wir orientieren uns hierbei unter anderem an
\cite{waldmann:2007a}, \cite{ortega.ratiu:2004} und
\cite{abraham.marsden:1985a}.
Es sei zunächst daran erinnert, dass man eine
Wirkung $\Phi \colon G \times M \to M$ einer Lie"=Gruppe $G$ auf einer
Poisson"=Mannigfaltigkeit $(M,\{\cdot,\cdot\})$ \neuerBegriff{kanonisch}
nennt, falls sie durch Poisson"=Abbildungen wirkt, d.\,h.\ falls
$\Phi_g^*\{f,f'\} = \{\Phi_g^*f,\Phi_g^*f'\}$ für alle $g \in G$ und $f,f'
\in C^\infty(M)$ gilt. Ist $M$ symplektisch, bedeutet dies, dass sie
durch Symplektomorphismen wirkt, man spricht in diesem Fall auch von
einer \neuerBegriff{symplektischen} Wirkung.

\begin{definition}[Impulsabbildung]
   \label{def:Impulsabbildung}
   Sei $(M,\{\cdot,\cdot\})$ eine Poisson"=Mannigfaltigkeit und $G$ eine
   Lie"=Gruppe, die  kanonisch auf $M$ operiere. Eine glatte
   Abbildung

   \begin{align}
      \label{eq:Impulsabbildung}
      J \colon M \to \lieAlgebra^*
   \end{align}
   heißt \neuerBegriff{Impulsabbildung}, falls für jedes $\xi \in
   \lieAlgebra$  die Gleichung
   \begin{align}
      \label{eq:ImpulsabbildungDefGleichung}
      \xi_M = \{\cdot,J(\xi)\}
   \end{align}
   erfüllt ist, wobei $J(\xi) \in C^\infty(M)$ als punktweise duale
   Paarung $J(\xi)(p) := \dPaar{J(p)}{\xi}$ für alle $p \in M$ und $\xi
   \in \lieAlgebra$ zu verstehen ist.
\end{definition}
\begin{bemerkung}
   \label{bem:ImpulsabbildungSichtweiseRaume}
   Äquivalent zu Definition \ref{def:Impulsabbildung} kann man eine
   Impulsabbildung $J$ vermöge $J(\xi)(p) = \dPaar{J(p)}{\xi}$ für alle
   $p \in M$   und $\xi
   \in \lieAlgebra$ auch als lineare Abbildung
   \begin{align}
      \label{eq:ImpulsabbildungAndereSichtweise}
      J \colon \lieAlgebra \to C^\infty(M)
   \end{align}
   auffassen, die $\overline{J(\xi)} = J(\xi)$ für alle $\xi \in
   \lieAlgebra$ und Gleichung \eqref{eq:ImpulsabbildungDefGleichung}
   erfüllt. Wir wollen im Folgenden beide Sichtweisen gleichberechtigt
   verwenden.
\end{bemerkung}

\begin{bemerkung}
   \label{bem:WirkungenundAequivarianteImpulsabb}
   $G$ operiert auf $\lieAlgebra$ vermöge der adjungierten, auf
   $\lieAlgebra^*$ vermöge der koadjungierten Wirkung und auf
   $C^\infty(M)$ durch pull"=backs. Wenn wir von der $G$-Äquivarianz der
   Impulsabbildung sprechen, beziehen wir uns immer auf diese
   (Links-)Wir\-kungen. Des Weiteren sei bemerkt, dass die
   $G$"=Äquivarianz bezüglich der einen Auffassung einer Impulsabbildung
   äquivalent zur
   $G$"=Äquivarianz in der anderen Auffassung ist.
\end{bemerkung}

\begin{bemerkung}
   \label{bem:Eindeutigkeit}
         Ist $M$ zusammenhängend, so ist $J$ eindeutig bis auf Addition von
         Elementen aus $\lieAlgebra^*$.
\end{bemerkung}

Eine Impulsabbildung braucht für eine gegebene Wirkung nicht zu
existieren, siehe etwa \cite[4.5.16]{ortega.ratiu:2004} für ein
einfaches Gegenbeispiel im symplektischen Fall. Die folgende Proposition,
vgl.\ \cite[Prop. 4.5.17]{ortega.ratiu:2004}, klärt die Situation für den
symplektischen Fall und liefert zwei einfach zu prüfende hinreichende
Existenzkriterien. Die entsprechende Frage für allgemeine
Poisson"=Mannigfaltigkeiten ist komplizierter.

\begin{proposition}
   \label{prop:ExistenzImpulsAbbildung}
   Sei $(M,\omega)$ eine symplektische Mannigfaltigkeit, auf der eine
   Lie"=Gruppe $G$ symplektisch operiere. Es gibt genau dann eine
   Impulsabbildung zu dieser Wirkung, wenn die lineare Abbildung
   \begin{align}
      \label{eq:AbbildungFuerExistenzVonImpulsabbildungen}
      \lieAlgebra/[\lieAlgebra,\lieAlgebra] &\to H^1_{\mathrm{dR}}(M,\mathbb{R})  \\
      [\xi] &\mapsto [\omega(\xi_M,\cdot)]
   \end{align}
   identisch Null ist. Dabei bezeichnet $H^1_{\mathrm{dR}}(M,\mathbb{R})$ die erste
   de"=Rham"=Kohomologie von $M$ mit Koeffizienten in $\mathbb{R}$.
\end{proposition}
\begin{korollar}
   \label{kor:HinreichendeBedFuerExistenzVonImpulsabb}
   Seien $(M,\omega)$ und $G$ wie in Proposition
   \ref{prop:ExistenzImpulsAbbildung}. Dann sind folgende Bedingungen
   hinreichend für die Existenz einer Impulsabbildung.
   \begin{korollarEnum}
      \item %
         \label{item:deRhamExFuerImpulsabb}
         $H^1_{\mathrm{dR}}(M,\mathbb{R}) = 0$.
      \item %
         \label{item:LieAlgebarExFuerImpulsabb}
         $[\lieAlgebra,\lieAlgebra] = \lieAlgebra$.
   \end{korollarEnum}
\end{korollar}

\begin{bemerkung}
   \label{bem:HinreichendeBedFuerExistenzVonImpulsabb}
   Die Bedingung~\refitem{item:deRhamExFuerImpulsabb} aus Korollar
   \ref{kor:HinreichendeBedFuerExistenzVonImpulsabb} ist insbesondere
   für jede einfach zusammenhängende Mannigfaltigkeit $M$ erfüllt.
   Lie"=Algebren, die Bedingung~\refitem{item:LieAlgebarExFuerImpulsabb}
   erfüllen heißen auch vollkommen oder perfekt, insbesondere ist jede
   halbeinfache Lie"=Algebra perfekt, siehe
   \cite[Satz~10.7]{schottenloher:1995a}. Für nichttriviale abelsche
   Lie"=Algebren ist~\refitem{item:LieAlgebarExFuerImpulsabb}
   offensichtlich nie erfüllt.
\end{bemerkung}

Die folgende Proposition (vgl.\ \cite{waldmann:2007a}) ist insbesondere
deshalb nützlich, da wir Lie"=Gruppen immer als zusammenhängend
voraussetzen.

\begin{proposition}
   \label{prop:AequivariantUndInfinitesimalAequivariant}
   Sei $(M,\{\cdot,\cdot\})$ eine Poisson"=Mannigfaltigkeit, $G$ eine
   zusammenhängende Lie"=Gruppe, die auf $M$ kanonisch operiere. Eine
   Impulsabbildung $J\colon M \to \lieAlgebra^*$ ist  genau dann
   $G$"=äquivariant, wenn die Gleichung $\{J(\xi),J(\eta)\} =
   J([\xi,\eta])$ für alle $\xi,\eta \in \lieAlgebra$ erfüllt ist.
\end{proposition}

Eine gegebene Impulsabbildung $J\colon M \to \lieAlgebra^*$ ist im Allgemeinen nicht $G$"=äquivariant. Für kompakte $G$ kann man aus $J$
jedoch immer durch Integration über $G$ eine andere, invariante
Impulsabbildung bauen, vgl.\ \cite[Prop. 4.5.19]{ortega.ratiu:2004}.

\begin{proposition}
   \label{prop:KompaktInvarianteImpulsabbildung}
   Sei $G$ eine kompakte Lie"=Gruppe, die auf einer
   Poisson"=Mannigfaltigkeit $(M,\{\cdot,\cdot\})$ kanonisch
   operiere. Falls eine Impulsabbildung existiert, so gibt es auch eine
   $G$"=äquivariante Impulsabbildung.
\end{proposition}
Eine andere Idee besteht darin, eine gegebene Impulsabbildung additiv so
abzuändern, dass sie äquivariant wird. Dies ist Gegenstand der folgenden
Proposition, siehe \cite[Satz 3.3.48]{waldmann:2007a}.

\begin{proposition}
   \label{prop:ExistenzAequivarianterImpulsabbildungen}
   Sei $G$ eine Lie"=Gruppe, die auf einer zusammenhängenden
   symplektischen Mannigfaltigkeit $(M,\omega)$ symplektisch operiere und
   $J\colon M \to \lieAlgebra^*$ eine Impulsabbildung. Dann ist für alle
   $\xi,\eta \in \lieAlgebra$ die Funktion $c(\xi,\eta) := J([\xi,\eta])
   - \{J(\xi),J(\eta)\}$ konstant auf $M$ und die Abbildung $c \colon
   \lieAlgebra \times \lieAlgebra \ni (\xi,\eta) \mapsto c(\xi,\eta) \in
   \mathbb{R}$ ist ein $2$"=Lie"=Algebra"=Kozyklus. Weiter gibt es genau
   dann eine $G$"=äquivariante Impulsabbildung $J'$, wenn das Element
   $[c]$ der zweiten Lie"=Algebra"=Kohomologie $H^2(\lieAlgebra)$
   trivial ist.
\end{proposition}
\begin{bemerkung}
   \label{bem:InvarianteImpulsabbildung}
   Falls es eine Impulsabbildung gibt, so ist eine hinreichende
   Bedingung für die Existenz einer $G$"=äquivarianten Impulsabbildung,
   dass die zweite Lie"=Algebra"=Kohomologie $H^2(\lieAlgebra)$
   verschwindet. Dies ist nach dem Whitehead"=Lemma
   (vgl.\ \cite[52]{guillemin.sternberg:1984b}) insbesondere für jede
   halbeinfache Lie"=Algebra erfüllt.
\end{bemerkung}

Siehe \cite{ortega.ratiu:2004}  und die dort
angegebenen Referenzen für weitere Kriterien für die Existenz
äquivarianter Impulsabbildungen.

Im Lichte der obigen Betrachtungen erscheint die folgende Definition
gerechtfertigt.

\begin{definition}
   \label{def:Hamiltonsch}
   Sei $(M,\{\cdot,\cdot\})$ eine Poisson"=Mannigfaltigkeit und
   $\Phi\colon G \times M \to M$ eine kanonische Wirkung einer Lie"=Gruppe $G$.
   \begin{definitionEnum}
      \item %
         $\Phi$ heißt \neuerBegriff{Hamiltonsch}, falls es eine
         Impulsabbildung gibt.
      \item %
         $\Phi$ heißt \neuerBegriff{stark Hamiltonsch}, falls es eine
         $G$"=äquivariante Impulsabbildung gibt.
   \end{definitionEnum}
\end{definition}

Um die physikalische Interpretation anzudeuten und zur abkürzenden
Schreibweise, geben wir noch folgende Definition.
\begin{definition}[Hamiltonsches System]
   \label{def:hamiltonschesSystem}

   Ist $(M,\{\cdot,\cdot\})$ eine Poisson"=Mannigfaltigkeit und $H \in
   C^\infty(M)$, so nennen wir das Tripel $(M,\{\cdot,\cdot\},H)$ ein
   \neuerBegriff{Hamiltonsches System}. Die Funktion $H$ heißt dann auch
   \neuerBegriff{Hamiltonfunktion}. Den Fluß des Hamiltonschen
   Vektorfelds $X_H := \{\cdot,H\}$ nennen wir
   \neuerBegriff{Zeitentwicklung} von $H$. Ist $f \in C^\infty(M)$
   invariant unter der Zeitenwicklung von $H$, so sagen wir $f$ ist eine
   \neuerBegriff{Erhaltungsgröße} (bezüglich der Zeitenwicklung von $H$).
\end{definition}
Die nächste Proposition, siehe etwa \cite[Satz
11.4.1]{marsden.ratiu:2000a}, verdeutlicht die physikalische Bedeutung
des Begriffs Impulsabbildung.

\begin{proposition}[Noether"=Theorem]
   \label{prop:NotherTheorem}
   Sei $(M,\{\cdot,\cdot\})$ eine Poisson"=Mannigfaltigkeit, $G$ eine
   zusammenhängende Lie"=Gruppe, die auf $M$ Hamiltonsch wirke und $J
   \colon M \to \lieAlgebra^*$ eine Impulsabbildung. Dann ist eine
   Hamiltonfunktion $H \in C^\infty(M)$ genau dann $G$"=invariant, wenn
   $J(\xi) \in C^\infty(M)$ eine Erhaltungsgröße bezüglich der
   Zeitentwicklung von $H$ ist.
\end{proposition}

Geometrisch grundlegend für das Verständnis der restlichen Arbeit ist
der Satz über die symplektische Reduktion nach Marsden, Weinstein und
Meyer, \cite{meyer:1973a} \cite{marsden.weinstein:1974a}, vgl.\ auch
\cite[Satz~3.3.55]{waldmann:2007a} für eine ausführliche
Lehrbuchdarstellung.

\begin{satz}[Marsden"=Weinstein"=Reduktion]
   \label{satz:MarsdenWeinstein}
   Sei $(M,\omega)$ eine symplektische Mannigfaltigkeit, $G$ eine
   Lie"=Gruppe mit zugehöriger Lie"=Algebra $\lieAlgebra$ und $\Phi \colon
   G \times M \to M$ eine stark Hamiltonsche Wirkung mit
   $G$"=äquivarianter Impulsabbildung $J \colon M \to \lieAlgebra^*$.
   Sei weiter $\mu \in \lieAlgebra^*$ ein regulärer Wert und die
   Impulsniveaufläche $C:= J^{-1}(\{\mu\}) \subseteq M$ nicht
   leer. Zudem wirke die Isotropiegruppe $G_\mu \subset G$ von $\mu$
   frei und eigentlich auf $M$. Unter diesen Voraussetzungen gibt es
   eine eindeutig bestimmte symplektische Form $\omega_{\mathrm{red}}$
   auf dem Quotientenraum $M_{\mathrm{red}}:= C/G_{\mu}$, so dass
   \begin{align}
      \label{eq:ReduzierteSymplektischeForm}
      \kRes \omega = \pi^* \omega_{\mathrm{red}}
   \end{align}
   gilt. Dabei ist $\kIn \colon C \hookrightarrow M$ die kanonische
   Inklusion und $\pi \colon C \to M_{\mathrm{red}}$ die kanonische Projektion auf
   den Quotienten $M_{\mathrm{red}}$. Insbesondere ist $\pi$ eine surjektive
   Submersion.
\end{satz}
Man beachte, dass die Voraussetzungen des Satzes insbesondere dann
erfüllt sind, wenn $G$ frei und eigentlich auf ganz $M$ wirkt. Falls
$\mu$ $G$"=invariant ist, gilt $G_{\mu} = G$.

\begin{definition}[Reduzierter Phasenraum]
   \label{def:RedPhasenraum}
   Man nennt die symplektische Mannigfaltigkeit $(\Mred,\omega_{\mathrm{red}})$
    \neuerBegriff{reduzierter Phasenraum} zu $(M,\omega,J,\mu)$.
\end{definition}
\begin{bemerkung}[Eichtheorie]
   \label{bem:Eichtheorie}
   In einer klassischen Eichtheorie betrachtet man den affinen, im
   Allgemeinen unendlich dimensionalen Vektorraum $\mathcal{A}$ der
   Zusammenhangseinsformen auf einem Hauptfaserbündel $P \to B$ mit
   Strukturgruppe $S$. Auf $\mathcal{A}$ operiert die Gruppe
   $\mathcal{G}$ der Hauptfaserbündelautomorphismen über der Identität
   auf $P \to B$ vermöge pull"=backs, oder allgemeiner eine Untergruppe
   davon. In physikalischer Interpretation ist $B$ zum Beispiel die
   Raumzeit, nur der physikalische Raum, oder allgemeiner eine
   Cauchy"=Hyperfläche $\Sigma$ in einer global hyperbolischen Raumzeit
   (vgl.\ \cite[2.2.4]{waldmann2009wave}). $P$ interpretiert man als Raum
   verallgemeinerter Phasenfaktoren, $\mathcal{A}$ als Raum der
   Eichpotentiale und $S$ als interne Symmetriegruppe.  Die Gruppe
   $\mathcal{G}$ ist dann die Eichgruppe und zwei Eichfelder entlang
   eines $\mathcal{G}$"=Orbits sind physikalisch äquivalent. Wir wollen
   annehmen, dass eine global hyperbolische Raumzeit in der Form $\Sigma
   \times \mathbb{R}$ vorliegt, d.\,h.\ wir schon eine Zeiteichfixierung
   durchgeführt haben und nehmen $B = \Sigma$ an. Dann kann man
   versuchen, die zeitliche Dynamik der Eichpotentiale als Hamiltonsches
   System zu beschreiben. Man betrachtet dabei das (i.\,Allg.\ unendlich
   dimensionale) Kotangentialbündel $T^*\mathcal{A}$, versehen mit der
   kanonischen symplektischen Form. Die Eichgruppe $\mathcal{G}$ wirkt
   dann auf $\mathcal{M} := T^*\mathcal{A}$ vermöge der gelifteten
   $\mathcal{G}$ Wirkung auf $\mathcal{A}$ und man kann als
   Impulsabbildung die kanonische wählen. Nun ist man also zumindest
   strukturell in der Situation von Satz
   \ref{satz:MarsdenWeinstein}. Allerdings sind $\mathcal{M}$ und
   $\mathcal{G}$ unendlich dimensional. Außerdem ist es im Allgemeinen
   nicht der Fall, dass $\mathcal{G}$ frei wirkt, so dass im Quotienten
   ein symplektisch stratifizierter Raum zu erwarten ist. Betrachtet man
   eine Gittereichtheorie, d.\,h.\ ersetzt konkret $\Sigma$ durch ein
   endliches Gitter, so werden $\mathcal{M}$ und $\mathcal{G}$ endlich
   dimensional, die Problematik, dass im Allgemeinen Singularitäten
   auftreten, bleibt jedoch bestehen. Für das genauere Studium von
   Gittereichtheorien in oben skizzierter Formulierung verweisen wir auf
   Arbeiten von Rudolph
   et.\,al.\ \cite{huebschmann2009gauge,fischer.rudolph.schmidt2007,charzynski.kijowski.rudolph2005},
   für eine Lehrbuchdarstellung singulärer (endlich"=dimensionaler)
   Marsden"=Weinstein"=Reduktion siehe \cite[Ch. 8]{ortega.ratiu:2004}.
   Für eine elementare Einführung in die Hamiltonsche, geometrische
   Sichtweise von Eichtheorien sei auf die Arbeit von Belot
   \cite{belot2003symmetry} verwiesen. Als weitere Literatur seien
    \cite{kijowski1979symplectic},\cite{marathe1992mathematical},
   \cite{schottenloher:1995a},
   \cite{mitter1981bundle},\cite{kijowski1984canonical},
   \cite{trautman1970fibre}, \cite{shabanov2000geometry},
   \cite{emmrich.roemer:1990a}, \cite{rudolph2002gauge},
   \cite{moncrief1980reduction} und \cite{daniel1980geometrical} genannt.
  \end{bemerkung}

  Auch wenn der folgende Satz für den weiteren Verlauf dieser Arbeit
  nicht benötigt wird, möchten wir ihn dennoch an dieser Stelle
  zitieren, vgl.\ \cite[Thm. 6.1.1]{ortega.ratiu:2004} und
  \cite[6.1.10]{ortega.ratiu:2004}, da er die physikalische
  Bedeutung des vorangegangenen Satzes innerhalb der klassischen
  Mechanik deutlich macht.

\begin{satz}[Reduktion und Rekonstruktion der Dynamik]
   \label{satz:ReduktionDerDynamik}

   Seien die Voraussetzungen wir in Satz~\ref{satz:MarsdenWeinstein}
   gegeben.  Weiter sei $H \in C^\infty(M)^G$ eine $G$"=invariante
   Hamiltonfunktion. Dann sind folgende Aussagen richtig.
   \begin{satzEnum}
   \item %
      Das Hamiltonsche Vektorfeld $X_H$ ist tangential an die
      Impulsniveaufläche $C$ und dessen Fluss $\Phi_t^{X_H}$
      bildet diese in sich ab.
   \item %
      Es gibt eine eindeutig bestimmte Hamiltonfunktion
      $H_{\mathrm{red}} \in C^\infty(\Mred)$, so dass
      \begin{align}
         \label{eq:ReduktionDerDynamik1}
         \kRes H = \pi^* H_{\mathrm{red}}
      \end{align}
      und
      \begin{align}
         \label{eq:ReduktionDerDynamik2}
         \pi \circ \Phi_t^{X_H} \circ \kIn = \Phi_t^{X_{H_\mathrm{red}}} \circ \pi
      \end{align}
      gilt, wobei $\Phi^{X_{H_\mathrm{red}}}$ den Fluss von $X_{H_{\mathrm{red}}}$
      bezeichnet.
   \item %
      Sei $p_0\in C$, $d(t)$ eine glatte Kurve in $C$ mit $d(0) = p_0$,
      $\xi(t)$ eine glatte Kurve in $\lieAlgebra_\mu := \{\xi \in
      \lieAlgebra \mid \mathrm{ad}^*_\xi \mu = 0\}$ mit
      \begin{align}
         \label{eq:RekonstruktionDerDynamik1}
         \xi(t)_M(d(t)) = X_H(d(t)) - \dot d(t)
      \end{align}
      und
      $g(t)$ eine glatte Kurve in $G_\mu$ mit
      \begin{align}
         \label{eq:RekonstruktionDerDynamik2}
         \dot g(t) = T_eL_{g(t)}\xi(t) \quad \text{und} \quad g(0) = e \Fcom
      \end{align}
      dann ist $c(t) = g(t) d(t)$ eine Integralkurve von $X_H$ mit $c(0)
      = p_0$.
      Falls $G$ abelsch ist, gilt
      \begin{align}
         \label{eq:RekonstruktionDerDynamik3}
         g(t) = \exp\left(\int_0^t \xi(s) \, ds\right) \Fdot
      \end{align}
      Ist nun $\gamma \colon TC \to \lieAlgebra_{\mu}$ eine
      Zusammenhangseinsform für das $G_\mu$"=Hauptfaserbündel $\pi
      \colon C \to \Mred$ und sei $d(t)$ die horizontale Hebung
      von $\Phi_t^{X_{H_{\mathrm{red}}}}(\pi(p_0))$ durch $p_0$, dann ist
      \begin{align}
         \label{eq:RekonstruktionDerDynamik4}
         \xi(t) = \gamma(X_H(d(t)))
      \end{align}
      eine Lösung von Gleichung \eqref{eq:RekonstruktionDerDynamik1}.
   \end{satzEnum}
\end{satz}

Gleichung \eqref{eq:RekonstruktionDerDynamik1} ist typischerweise
algebraischer Natur. Falls zum Beispiel $G$ eine Matrix"=Lie"=Gruppe ist,
handelt es sich einfach um eine Matrixgleichung.

Die Bedeutung dieser Reduktionssätze liegen innerhalb der klassischen
Mechanik zum einen in einem Gewinn des qualitativen Verständnisses eines
mechanischen Systems, zum anderen können sie hilfreich sein bei der
Lösung konkreter Bewegungsgleichungen eines solchen.

An dieser Stelle möchten wir noch bemerken, dass eine Impulsabbildung $J
\colon M \to \lieAlgebra^*$ submersiv ist, d.\,h.\ alle Werte auch reguläre
Werte sind, falls $G$ frei auf $M$ wirkt. Wir werden auf diesen Punkt
nun in etwas größerer Allgemeinheit eingehen.

\begin{definition}
   \label{def:lokalFreiUndInfiniFrei}
   Eine Wirkung einer Lie"=Gruppe $G$ auf einer Mannigfaltigkeit $M$
   heißt \neuerBegriff{lokal frei}, falls für alle $p \in M$ die
   Stabilisatorgruppe $G_p := \{g \in G \mid gp = p\}$ diskret ist,
   d.\,h.\ es für jedes $g \in G_p$ eine offene Umgebung $U$ von $g$ in $G$
   gibt, so dass $U \cap G_p = \{g\}$. Sie heißt
   \neuerBegriff{infinitesimal frei}, falls für alle $p \in M$ die
   \neuerBegriff{Symmetriealgebra} $\lieAlgebra_p := \{\xi \in \lieAlgebra \mid
   \xi_M(p) = 0\}$ trivial ist, d.\,h.\ $\lieAlgebra_p = \{0\}$ gilt.
\end{definition}
\begin{bemerkung}
   \label{bem:lokalFrei}
   \begin{bemerkungEnum}
   \item %
      Jede freie Wirkung ist offensichtlich auch lokal frei.
   \item %
      Falls die Wirkung eigentlich ist, ist $G_p$ kompakt für alle $p
      \in M$, d.\,h.\ sie ist genau dann lokal frei, wenn $G_p$ endlich ist
      für alle $p \in M$.
   \end{bemerkungEnum}
\end{bemerkung}

\begin{proposition}
   \label{prop:lokalFrei}
   Sei $M$ eine Mannigfaltigkeit, die eine Wirkung $\Phi \colon M \times
   G \to G$ einer Lie"=Gruppe $G$ trage. Dann sind die folgenden Aussagen äquivalent.
   \begin{propositionEnum}
   \item %
      \label{item:lokalFrei1}
      $\Phi$ ist lokal frei.
   \item %
      \label{item:lokalFrei2}
      Für alle $p \in M$ gibt es eine offene Umgebung $V_p$ von $e$ in
      $G$ mit  für $V_p \cap G_p = \{e\}$.
   \item %
      \label{item:lokalFrei3}
      $\Phi$ ist infinitesimal frei.
   \end{propositionEnum}
\end{proposition}
\begin{proof}
      Die Implikation von~\refitem{item:lokalFrei1} nach~\refitem{item:lokalFrei2} ist trivial. Sei also $p \in M$ und
      $V_p$ eine offene Umgebung von $e$ in $G$ mit $V_p \cap G_p = \{e\}$. Sei
      weiter $g \in G_p$. Da die Wirkung stetig ist, ist $gV_p$ eine
      offene Umgebung von $g$. Offensichtlich gilt für $h \in G_p \cap
      gV_p$ auch $g^{-1}h \in G_p \cap V_p$, also $g^{-1}h = e$,
      d.\,h.\ $h= g$, womit sich~\refitem{item:lokalFrei1} ergibt.

      Sei nun $\Phi$ infinitesimal frei. Sei $\hat \Phi \colon G \times
      M \to M \times M$, $(g,p) \mapsto (gp,p)$. Es genügt zu zeigen,
      dass es für alle $p \in M$ eine Umgebung $V_p \times W_p$ von $(e,p)$
      gibt, so dass $\hat \Phi\at{V_p \times W_p}$ injektiv ist. Denn dann
      ist offenbar $V_p$ eine Umgebung wie in~\refitem{item:lokalFrei2}. Um dies zu sehen, genügt es nach dem
      Immersionssatz (vgl.\ \cite[Thm. 3.5.7]{marsden:2002}) zu zeigen,
      dass die Ableitung von $\hat \Phi$ an der Stelle $(e,p)$ injektiv
      ist. Dies ist aber nach Voraussetzung und Proposition
      \ref{prop:produkte} klar, denn für $\xi \in \lieAlgebra = T_eG$
      und $v_p \in T_pM$ gilt
      \begin{align*}
         0 = T_{(e,p)}\hat \Phi \TIsoI{(\xi,v_p)} &= \TIsoI{(T_e
           \Phi_{\cdot}(p), 0)} + \TIsoI{(v_p,v_p)} = \TIsoI{(\xi_M(p) +
           v_p, v_p)} \\
         &\implies v_p = 0 \, \text{und} \, \xi_M(p) = 0 \Fdot
      \end{align*}
      Daraus folgt nach Voraussetzung schon $\xi = 0$, also auch
      $\TIsoI{(\xi,v_p)} = 0$.
      Sei nun angenommen $\Phi$ nicht infinitesimal frei. Dann gibt es
       ein $p \in M$ und ein $0 \neq\xi \in \lieAlgebra$ mit
      $T_e\Phi(p) = \xi_M(p) = 0$. Es gibt also eine Umgebung $V$
      von $e$ in $G$, so dass $\Phi_g(p) = p$ für alle $g \in V$,
      d.\,h.~\refitem{item:lokalFrei2} kann nicht gelten.
\end{proof}

\begin{proposition}
   \label{lem:LineareAlgebra}
   Seien $V$ und $W$ endlich dimensionale Vektorräume über $\mathbb{K}$ und $A \colon V
   \to W$ linear, so ist $A$ genau dann surjektiv, wenn aus $\dPaar{Av}{w^*}
   = 0 \; \forall v \in V$ schon $w^* = 0 \in W^*$ folgt.
\end{proposition}
\begin{proof}
   Die Aussage ist äquivalent dazu, dass $A$ genau dann injektiv ist,
   wenn die duale Abbildung $A^*$ surjektiv ist. Damit ist die
   Behauptung klar.

\end{proof}

\begin{lemma}

   \label{lem:JSurjektiv}
   Sei $(M,\omega)$ eine symplektische Mannigfaltigkeit und $\Phi \colon
   G \times M \to M$ eine Hamiltonsche Wirkung einer Lie"=Gruppe $G$ mit
   Impulsabbildung $J \colon M \to \lieAlgebra^*$. Dann ist $p \in M$
   genau dann ein regulärer Punkt von $J$, wenn $\lieAlgebra_p = \{0\}$
   gilt. Insbesondere ist $\Phi$ genau dann lokal frei, wenn $J$
   submersiv ist, was genau dann der Fall ist, wenn jeder Wert von $J$
   regulär ist.
\end{lemma}
\begin{proof}
   Für $\mu \in \lieAlgebra^*$ bezeichne in diesem Beweis $I_\mu\colon
   \lieAlgebra^* \to T_\mu\lieAlgebra^*$ den kanonischen Isomorphismus
   der Mannigfaltigkeit $\lieAlgebra^*$ mit ihrem Tangentialraum
   $T_{\mu}\lieAlgebra^*$ an der Stelle $\mu$.
   \begin{align*}
      p \in M \; \text{regulär} &\iff T_pJ \colon T_pM \to
      T_{J(p)}\lieAlgebra^* \; \text{ist surjektiv} \\
      &\iff I_{J(p)}^{-1} \circ T_pJ \colon T_pM \to \lieAlgebra^* \;
      \text{ist surjektiv} \\
      &\iff \dPaar{I_{J(p)}^{-1} \circ T_p J v}{\xi} = 0 \; \forall v \in T_pM \implies \xi =
      0 &&\eAnn{nach Proposition \ref{lem:LineareAlgebra}}\\
      &\iff \dPaar{d_pJ(\xi)}{v} = 0 \;\forall v \in T_pM \implies  \xi = 0 \\
      &\iff \omega_p(\xi_M(p),v) = 0 \; \forall v \in T_pM \implies \xi = 0
      \\
      &\iff \xi_M(p) = 0 \implies \xi = 0 &&\eAnn{da $\omega_p$ nicht ausgeartet} \\
      &\iff \lieAlgebra_p = \{0\}
   \end{align*}
   Die Umformung im vierten Schritt sieht man leicht in Koordinaten. Ist
   etwa $(U,x)$ eine Karte von $M$ um $p$, $\{e_i\}$ eine Basis von
   $\lieAlgebra$ und $\{e^i\}$ die zugehörige duale Basis, und schreiben
   wir $v = v_k \frac{\partial}{\partial x^k}\at{p}$ sowie $\xi = \xi^i
   e_i$, dann rechnet man leicht
   \begin{align*}
      \dPaar{I_{J(p)}^{-1} \circ T_p J v}{\xi} = \dPaar{\frac{\partial
          J_i}{\partial x^k}\at{p}v_k e^i}{\xi^j e_j}
      &= v_k \frac{\partial J_i}{\partial x^k}\at{p} \xi^i \\
      &= v_k \frac{\partial}{\partial x^k}\at{p}(\xi^i J_i)
      = v_k
      \frac{\partial}{\partial x^k}(J(\xi))
      = \dPaar{d_pJ(\xi)}{v}
   \end{align*}
   nach.
\end{proof}

\subsection{Invarianzbegriffe in der Deformationsquantisierung}
\label{sec:InvarianzbegriffeInDerDeformationsquantisierung}

Wir kommen nun zu Invarianzbegriffen von Sternprodukten
(vgl.\ \cite{arnal.cortet.molin.pinczon:1983a},
\cite{xu:1998a,bordemann.brischle.emmrich.waldmann:1996a}).

\begin{definition}[Invarianz von Sternprodukten]
   \label{def:InvariantesSternprodukt}
   Sei $(M,\{\cdot,\cdot\})$ eine Poisson"=Mannigfaltigkeit, die eine
   kanonische Wirkung $\Phi \colon G \times M \to M$ einer Lie"=Gruppe
   $G$ trage. Ein Sternprodukt $\star$ für $(M,\{\cdot,\cdot\})$ heißt
   \neuerBegriff{$G$"=invariant}, wenn für alle $g \in G$ und $f,f' \in
   C^\infty(M)$ die Gleichung
   \begin{align}
      \label{eq:InvariantesSternprodukt}
      \Phi_g^*(f \star f') = \Phi_g^*f \star \Phi_g^* f'
   \end{align}
   gilt.  Ist $J \colon M \to \lieAlgebra^*$ eine $G$"=äquivariante
   Impulsabbildung, so heißt $\star$ \neuerBegriff{stark invariant},
   falls für alle $\xi \in \lieAlgebra$ und $f \in C^\infty(M)$
   \begin{align}
      \label{eq:SternproduktKovariant}
      \I \lambda \{J(\xi),f\} = [J(\xi),f]_{\star}
   \end{align}
   gilt, wobei $[\cdot,\cdot]_\star$ der Kommutator bezüglich des
   Sternprodukts $\star$ ist.

\end{definition}
\begin{bemerkung}
   \label{bem:BeziehungStarkInvariantInvariant}
   Die starke Invarianz eines Sternprodukts impliziert dessen Invarianz
   im Falle einer zusammenhängenden Lie"=Gruppe.  Der Umkehrschluss gilt
   im Allgemeinen nicht. Siehe etwa \cite[p. 112,
   136]{bordemann.herbig.waldmann:2000a}.
\end{bemerkung}

Als nächstes verallgemeinern wir das Konzept der Impulsabbildung, indem
wir noch Terme höherer $\lambda$"=Ordnungen zulassen.
\begin{definition}[Quantenimpulsabbildung]
   \label{def:Quantenimpulsabbildung}
   Sei $(M,\{\cdot,\cdot\})$ eine Poisson"=Mannigfaltigkeit, $G$ eine
   Lie"=Gruppe, die auf $M$ wirke. Weiter sei $\star$ ein
   $G$"=invariantes Sternprodukt auf $M$. Dann heißt eine lineare
   Abbildung $\qJ \colon \lieAlgebra \to C^\infty(M)[[\lambda]]$ der
   Form
   \begin{align}
      \qJ = \sum_{r=0}^\infty  \lambda^r J_r
   \end{align}
   mit linearen  Abbildungen $J_r \colon \lieAlgebra \to C^\infty(M)$
   \neuerBegriff{Quantenhamiltonfunktion} (oder auch
   \neuerBegriff{Quanten"=Hamiltonian}), falls $\qJ$ die Deformation
   einer Impulsabbildung ist, d.\,h.\ $J := J_0$ eine Impulsabbildung ist
   und für alle $f \in C^\infty(M)$ sowie $\xi \in \lieAlgebra$ die
   Gleichung
   \begin{align}
      \label{eq:QuantenHamiltonian}
      \I \lambda \{J(\xi),f\} = [\qJ(\xi),f]_\star
   \end{align}
   gilt. Falls zusätzlich für alle $\xi,\eta \in
   \lieAlgebra$ die Gleichung
   \begin{align}
      \label{eq:Quantenimpulsabbildung}
      \I \lambda \qJ([\xi,\eta]) = [\qJ(\xi),\qJ(\eta)]_\star
   \end{align}
   wahr ist, nennt man $\qJ$
   \neuerBegriff{Quantenimpulsabbildung}. Eine $G$"=äquivariante
   Quantenhamiltonfunktion heißt auch \neuerBegriff{$G$"=äquivariante
     Quantenimpulsabbildung},
   vgl.\ \cite[p. 492]{kowalzig.neumaier.pflaum:2005a}.
\end{definition}
\begin{bemerkung}
   \label{bem:StarkInvariant}
   Das Sternprodukt $\star$ ist offensichtlich genau dann stark
   invariant, wenn die klassische Impulsabbildung $J$ auch eine
   $G$"=äquivariante Quantenimpulsabbildung ist.
\end{bemerkung}

\begin{bemerkung}
   \label{bem:QuantenhamiltonfunktionAndereSichtweise}
   Wie bei der klassischen Impulsabbildung, vgl.\ Bemerkung~\ref{bem:ImpulsabbildungSichtweiseRaume}, kann man eine
   Quantenhamiltonfunktion über die Beziehung $\qJ(\xi)(p) =
   \dPaar{\qJ(p)}{\xi}$ für alle $p \in M$ und $\xi \in \lieAlgebra$ als
   Abbildung
   \begin{align}
      \label{eq:QuantenhamiltonfunktionAndereSichtweise}
      \qJ \colon M \to \lieAlgebra_{\mathbb{C}}^*[[\lambda]]
   \end{align}
   der Form $\qJ = \sum_{r=0}^\infty \lambda^r J_r$ mit glatten
   Abbildungen $J_r\colon M \to \lieAlgebra_{\mathbb{C}}^*$ auffassen,
   so dass $J_0$ eine Impulsabbildung ist und $\qJ$ die Gleichung
   \eqref{eq:QuantenHamiltonian} erfüllt. Dabei bezeichnet
   $\lieAlgebra_{\mathbb{C}} = \lieAlgebra \otimes_{\mathbb{R}}
   \mathbb{C}$ die komplexifizierte Lie"=Algebra. Auch hier wollen wir
   beide Sichtweisen austauschbar verwenden.
\end{bemerkung}

\begin{bemerkung}
   \label{bem:PhysikalicheBedeutung}
   \begin{bemerkungEnum}
   \item %
      Ist $H = J(\xi)$ für ein $\xi \in \lieAlgebra$, so folgt aus der
      starken Invarianz des Sternprodukts $\star$, dass die
      Quantenzeitentwicklung jeder Observablen bezüglich $H$ mit der
      klassischen übereinstimmt. Dies ist vom Beispiel des harmonischen
      Oszillators wohlbekannt. Für eine eingehendere Diskussion der
      physikalischen Bedeutung stark invarianter Sternprodukte verweisen
      wir auf die Pionier"=Arbeit \cite{bayen.et.al:1978a} von Bayen
      et.\,al.\
      \item %
         Wenn Symmetrien eines klassischen Systems nicht als Symmetrien
         des korrespondierenden Quantensystems implementiert werden
         können, spricht man von \neuerBegriff{Anomalien}. Römer und
         Paufler bemerken etwa in \cite{romer2000anomalies}, dass für
         $\qJ = J$ eine Verletzung von Gleichung
         \eqref{eq:SternproduktKovariant} feldtheoretisch Schwinger
         Termen entsprechen sollte und eine Verletzung von Gleichung
         \eqref{eq:QuantenHamiltonian} anomalen Ward"=Identitäten. Sie
         vermuten weiter, dass diese und andere Anomalien nicht
         auftreten, wenn man auch Quantenimpulsabbildungen zulässt.
   \end{bemerkungEnum}
\end{bemerkung}

Die folgende Proposition zeigt, dass in der vorliegenden Situation von
zusammenhängenden Lie"=Gruppen der Begriff der $G$"=äquivarianten
Quantenimpulsabbildung aus Definition \ref{def:Quantenimpulsabbildung}
mit dem einer Quantenimpulsabbildung, die $G$"=äquivariant ist, übereinstimmt.
\begin{proposition}
   \label{bem:QuantenImpulsabbHamiltonian}
   Die Voraussetzungen seien wie in Definition
   \ref{def:Quantenimpulsabbildung} gegeben. Ist $\qJ$ eine
   $G$"=äquivariante Quantenhamiltonfunktion, so ist $\qJ$ schon eine
   Quantenimpulsabbildung.

   Da wir $G$ als zusammenhängend annehmen, gilt auch die Umkehrung.
\end{proposition}
\begin{proof}
   Sei $\qJ$ eine $G$"=äquivariante Quantenhamiltonfunktion. Dann gilt für
   alle $g \in G$ und $\eta \in \lieAlgebra$
   \begin{align*}
      g \qJ(\eta) = \qJ(g \eta)\Fcom
   \end{align*}
   also durch Ableiten nach $g$ für alle $\xi \in \lieAlgebra$
   \begin{align*}
      \xi_M(\qJ(\eta)) = -\qJ([\xi,\eta])
   \end{align*}
   d.\,h.\
   \begin{align*}
      [\qJ(\xi),\qJ(\eta)]_{\star} = \I \lambda \{J(\xi),\qJ(\eta)\} =
      -\I \lambda \xi_M(\qJ(\eta)) = \I \lambda \qJ([\xi,\eta]) \quad
      \forall \xi \in \lieAlgebra \Fdot
   \end{align*}
   Sei umgekehrt $\qJ$ eine Quantenimpulsabbildung und seien $\xi,\eta \in
   \lieAlgebra$. Dann gilt
   \begin{align*}
      -\xi_M(\qJ(\eta)) = \{J(\xi),\qJ(\eta)\} =
      \frac{1}{\I \lambda}[\qJ(\xi),\qJ(\eta)]_{\star} = \qJ([\xi,\eta]) \Fdot
   \end{align*}
   Da $G$ zusammenhängend ist, folgt die $G$"=Äquivarianz von $\qJ$.
\end{proof}

Über die Eindeutigkeit von Quantenimpulsabbildungen kann man analoge
Aussagen treffen wie bei der klassischer Impulsabbildungen.

\begin{proposition}
   \label{prop:EindeutigkeitvonQuanteimpulsabbildungen}
   Seien die Voraussetzungen wie in  Definition
   \ref{def:Quantenimpulsabbildung} gegeben und zusätzlich
   $\{\cdot,\cdot\}$ symplektisch und $M$ zusammenhängend.
   \begin{propositionEnum}
   \item %
      Eine Quantenhamiltonfunktion $\qJ$ zu einer gegebenen
      Impulsabbildung $J$ ist bis auf Addition von Elementen aus
      $\bigwedge_{\mathbb{C}}\lieAlgebra_{\mathbb{C}}^*[[\lambda]]$ eindeutig
      bestimmt.
   \item %
      Eine Quantenimpulsabbildung $\qJ$ zu einer gegebenen
      $G$"=äquivarianten Impulsabbildung $J$ ist bis auf Addition von
      $G$"=invarianten Elementen aus $\bigwedge_{\mathbb{C}}
      \lieAlgebra_{\mathbb{C}}^*[[\lambda]]$ eindeutig bestimmt. Diese
      sind genau die formalen Potenzreihen von Lie"=Algebra
      $1$"=Kozykeln. Insbesondere sind Quantenimpulsabbildungen eindeutig,
      wenn die erste Lie"=Algebra"=Kohomologie verschwindet.
   \end{propositionEnum}
\end{proposition}
\begin{proof}
   Da das Zentrum von $C^\infty(M)[[\lambda]]$ bezüglich $\star$ nur aus
   den konstanten Funktionen besteht
   (vgl.\ \cite[Bsp. 6.3.15]{waldmann:2007a}) ist die Aussage klar, siehe
   auch \cite{xu:1998a}.
\end{proof}

\begin{bemerkung}
   \label{bem:KlassenFeinerMitSymmetrie}
   Wir haben im vorangegangenen Abschnitt gesehen, dass die
   Klassifikation von Sternprodukten auf symplektischen
   Mannigfaltigkeiten mit Hilfe der de"=Rham"=Kohomologie im Allgemeinen
   physikalisch zu grob ist. In vielen Situationen wird man jedoch gute
   physikalische Gründe haben, bestimmte Invarianzforderungen an das
   Sternprodukt zu stellen. Dies verkleinert im Allgemeinen die
   Sternprodukt"=Äquivalenzklassen oder zeichnet gar ein Sternprodukt
   aus. Für eine weitergehende Diskussion der Klassifikation von
   invarianten Sternprodukten verweisen wir auf die Diplomarbeit von
   Schaumann \cite{schaumann:2010}.
   Dort wird insbesondere als schönes Beispiel gezeigt, dass das
   Weyl"=Moyal"=Sternprodukt das einzige $\mathbb{Sp}_{2n} \ltimes
   \mathbb{R}^{2n}$"=invariante, differentielle $\mathbb{Sp}_{2n}$"=stark
   invariante Sternprodukt für den $\mathbb{R}^{2n}$ mit kanonischer
   symplektischer Form ist. Dabei bezeichnet $\mathbb{Sp}_{2n}$ die
   symplektische Gruppe und $\ltimes$ das semi"=direkte"=Produkt von
   Gruppen. Für Details verweisen wir auf \cite{schaumann:2010}.
\end{bemerkung}

Mit Hilfe der Fedosov Konstruktion erhält man das folgende Resultat über
die Existenz von invarianten Sternprodukten auf symplektischen
Mannigfaltigkeiten, vgl.\ \cite[Prop. 6.1]{xu:1998a}.

\begin{satz}[Existenz invarianter Sternprodukte]
   \label{satz:ExistenzInvarianterSternprodukte}
   Sei $(M,\omega)$ eine symplektische Mannigfaltigkeit und $G$ eine
   Lie"=Gruppe, die auf $M$ symplektisch operiere. Es gibt genau dann
   ein $G$"=invariantes Sternprodukt $\star$ für $(M,\omega)$, wenn es
   eine $G$"=invariante, torsionsfreie kovariante Ableitung $\nabla$ auf
   $M$ gibt. Sind $\nabla$ und die formale Potenzreihe geschlossener
   Zweiformen $\Omega \in \lambda\Gamma^\infty(\Bigwedge^2
   T^*M)[[\lambda]] $ $G$"=invariant, so ist auch das induzierte
   Fedosovsternprodukt $\star_{\Omega}$ $G$"=invariant.
\end{satz}

\begin{korollar}
   \label{kor:ExistenzInvarianterSternprodukte}
   Die Lie"=Gruppe $G$ wirke symplektisch und eigentlich auf der
   symplektischen Mannigfaltigkeit $(M,\omega)$. Dann existiert ein
   $G$"=invariantes Sternprodukt $\star$ für $(M,\omega)$.
\end{korollar}
\begin{proof}
   Da die $G$"=Wirkung eigentlich ist, existiert nach Korollar
   \ref{kor:InvarianterZusammenhang} eine $G$"=invariante, torsionsfreie
   kovariante Ableitung auf $M$, so dass die Behauptung mit Satz
   \ref{satz:ExistenzInvarianterSternprodukte} klar ist.
\end{proof}

\begin{definition}
   \label{def:Gaequivalent}
   Zwei $G$"=invariante Sternprodukte $\star$ und $\star'$ für eine
   Poisson"=Mannigfaltigkeit $(M,\{\cdot,\cdot\})$, auf der eine
   Lie"=Gruppe $G$ operiere, heißen \neuerBegriff{$G$"=äquivalent},
   falls es eine $G$"=äquivariante Äquivalenztransformation zwischen
   $\star$ und $\star'$ gibt.
\end{definition}

Bertelson, Bieliavsky und Gutt konnten in
\cite{bertelson.bieliavsky.gutt:1998a} mit Hilfe der
Fedosov"=Konstruktion das folgende Klassifikationsresultat für
$G$"=invariante Sternprodukte zeigen.

\begin{satz}
   \label{satz:KlassifikationGinvarianterSternprodukte}
   Sei $(M,\omega)$ eine symplektische Mannigfaltigkeit auf der eine
   Lie"=Gruppe $G$ derart wirke, dass eine $G$"=invariante,
   torsionsfreie Ableitung $\nabla$ auf $M$ existiere. Dann gibt es
   eine, von $\nabla$ unabhängige, Bijektion zwischen der Menge der
   Äquivalenzklassen von $G$"=invarianten Sternprodukten im Sinne der
   $G$"=Äquivalenz in die Menge der formalen Potenzreihen von Elementen
   der zweiten $G$"=invarianten de"=Rham Kohomologie von $M$.
\end{satz}

Müller"=Bahns und Neumaier\cite{mueller-bahns.neumaier:2004a} haben
Existenzkriterien für Quantenhamiltonfunktionen angegeben. Wir fassen
diese im nächsten Satz zusammen.

\begin{satz}
   \label{satz:ExistenzVonQuantenImpulsabb}
   Sei $(M,\omega)$ eine symplektische Mannigfaltigkeit auf der eine
   Lie"=Gruppe $G$ symplektisch operiere. Weiter sei $\nabla$ eine
   $G$"=invariante, torsionsfreie kovariante Ableitung und $\Omega \in \lambda
   \Gamma^\infty(\Bigwedge^2 T^*M)$ eine formale Reihe $G$"=invarianter
   Zweiformen. Dann besitzt das zugehörige $G$"=invariante
   Fedosov"=Sternprodukt $\star_{\Omega}$ genau dann eine
   Quantenhamiltonfunktion, wenn es eine formale Potenzreihe $\qJ =
   \sum_{r = 0}^\infty \lambda^r J_r$ mit glatten Abbildungen $J_r
   \colon M \to \lieAlgebra^*$ gibt, die für alle $\xi \in \lieAlgebra$
   \begin{align}
      \label{eq:BedingungQuantenhamilton}
      d \qJ(\xi) = (\omega + \Omega)(\xi_M,\cdot)
   \end{align}
    erfüllt.
\end{satz}
\begin{korollar}
   \label{kor:StarkInvariantFedosov}
   Seien die Voraussetzungen wie in Satz
   \ref{satz:ExistenzVonQuantenImpulsabb} und  $J \colon M \to \lieAlgebra^*$
   eine klassische Impulsabbildung, so ist $\star_{\Omega}$ genau dann
   stark invariant, wenn für alle $\xi \in \lieAlgebra$
   \begin{align}
      \label{eq:StarkInvariantFedosov}
      \Omega(\xi_M,\cdot) = 0
   \end{align}
   gilt.
\end{korollar}

\begin{bemerkung}
   \label{bem:HinreichendeBedinungFuerStarkInvariant}
   Wählt man $\Omega = 0$, so ist Gleichung
   \eqref{eq:StarkInvariantFedosov} trivial erfüllt. Es ist also eine
   hinreichende Bedingung für die Existenz eines stark invarianten
   Sternprodukts, dass eine $G$"=invariante, torsionsfreie kovariante Ableitung
   existiert, was z.\,B.\ immer dann der Fall ist, wenn die $G$"=Wirkung
   eigentlich ist.
\end{bemerkung}

Xu konnte in \cite[Prop. 6.3]{xu:1998a} die aus der klassischen
Situation schon bekannte Proposition
\ref{prop:ExistenzAequivarianterImpulsabbildungen} für
Quantenimpulsabbildungen verallgemeinern.

\begin{proposition}
   \label{prop:QuantenImpulsabbildungUndLieAlgebraCohomology}
   Sei $(M,\omega)$ eine zusammenhängende symplektische Mannigfaltigkeit
   und $G$ eine zusammenhängende Lie"=Gruppe, die auf $M$ wirke. Sei
   weiter $\star$ ein Sternprodukt für $(M,\omega)$ und $\qJ$ eine
   Quantenimpulsabbildung. Für alle $\xi,\eta \in \lieAlgebra$ ist die
   Funktion $c(\xi,\eta) := \qJ([\xi,\eta]) -
   [\qJ(\xi),\qJ(\eta)]_{\star}$ konstant als Element von
   $C^\infty(M)[[\lambda]]$ und die Abbildung $\lieAlgebra \times
   \lieAlgebra \ni (\xi,\eta) \mapsto c(\xi,\eta) \in
   \mathbb{C}[[\lambda]]$ ist ein Lie"=Algebra 2"=Kozyklus. Es gibt
   genau dann eine $G$"=äquivariante Impulsabbildung, wenn $[c]$ als
   Element der zweiten Lie"=Algebra"=Kohomologie
   $H^2(\lieAlgebra,\mathbb{C}[[\lambda]]) \simeq
   H^2(\lieAlgebra)\otimes \mathbb{C}[[\lambda]]$ verschwindet.
\end{proposition}

\begin{bemerkung}
   \label{bem:Quantenimpulsabbildungundliealgebracohomology}
   Falls es eine Quantenimpulsabbildung $\qJ$ gibt, ist  nach
   Proposition~\ref{prop:QuantenImpulsabbildungUndLieAlgebraCohomology}
   das Verschwinden der zweiten Lie"=Algebra"=Kohomologie hinreichend
   für die Existenz einer $G$"=äquivarianten Impulsabbildung. Dies ist,
   wie in Bemerkung \ref{bem:InvarianteImpulsabbildung} schon gesagt,
   insbesondere für halbeinfache Lie"=Algebren der Fall.

\end{bemerkung}

Die aus der klassischen Situation bekannte Proposition
\ref{prop:KompaktInvarianteImpulsabbildung} wollen wir nun auf den
Quantenfall verallgemeinern. Wie es scheint, ist dies noch nicht in der
Literatur bekannt.

\begingroup
\begin{proposition}
   \label{prop:IntegrierenVonQuantenimpulsabbildungen}
   Sei $G$ eine kompakte, zusammenhängende Lie"=Gruppe, die auf einer
   Poisson"=Mannigfaltigkeit $(M,\{\cdot,\cdot\})$ operiere und es gebe
   eine Quantenhamiltonfunktion $\qJ$ zu einer gegebenen
   $G$"=äquivarianten Impulsabbildung $J$ und einem $G$"=invarianten
   differentiellen Sternprodukt $\star$, dann exisitert auch eine $G$"=äquivariante
   Quantenhamiltonfunktion.
\end{proposition}
\begin{proof}
   Für jedes $g \in G$ definieren wir $\qJ[g]$ durch $\qJ[g](\xi ) :=
   g(\qJ(g^{-1}\xi))$. Dann ist auch $\qJ[g]$ eine
   Quantenhamiltonfunktion. Denn in der Tat rechnet man nach
   \begin{align*}
      [\qJ[g](\xi),f]_{\star} &= [g(\qJ(g^{-1}\xi)),f]_\star =
      g[\qJ(g^{-1}\xi),g^{-1} f]_\star \\ &= \I \lambda g
      \{J(g^{-1}\xi),g^{-1} f\} = \I \lambda g \{g^{-1}J(\xi),g^{-1} f\}
      = \I \lambda \{J(\xi),f\} \Fdot
   \end{align*}
   Dabei wurde im zweiten Schritt die Invarianz von $\star$ verwendet,
   im vierten, dass $J$ $G$"=äquivariant ist und im letzten schließlich,
   dass $G$ auf $C^\infty(M)$ durch Poisson"=Abbildungen operiert.  Sein
   nun $dg$ das normierte Haar"=Maß auf $G$, dann definieren wir
   $\langle \qJ \rangle$ für jedes $\xi \in \lieAlgebra$ durch $\langle
   \qJ \rangle (\xi) := \int_G \qJ[g](\xi) \, dg$. Da $\star$
   differentiell ist, folgt mit der Linearität des $\star$"=Kommutators
   für alle $f \in C^\infty(M)$
   \begin{align*}
      [\langle \qJ \rangle (\xi),f]_{\star} =
      \int_G[\qJ[g](\xi),f]_{\star} \, dg = \I \lambda \int_G
      \{J(\xi),f\} \, dg = \I \lambda \{J(\xi),f\} \Fdot
   \end{align*}
   Somit ist $\langle \qJ \rangle$ eine
   Quantenhamiltonfunktion. Dass $\langle \qJ \rangle$ tatsächlich
   $G$"=äquivariant ist, ist mit Korollar \ref{kor:integrieren_inv_funkt} klar.
\end{proof}

\endgroup

Wir wollen nun noch sehen, wie sich Quantenimpulsabbildungen unter
$G$"=äquivarianten Sternprodukthomomorphismen verhalten.

\begin{proposition}
   \label{prop:QuantenImpulsabbildungUndSternproduktHomos}
   Seien $(M,\{\cdot,\cdot\})$ und $(M',\{\cdot,\cdot\}')$
   Poisson"=Mannigfaltigkeiten mit Sternprodukten $\star$ und
   $\star'$. Weiter wirke eine Lie"=Gruppe $G$ sowohl auf $M$ als auch
   auf $M'$, sodass $\star$ und $\star'$ bezüglich dieser $G$"=invariant
   sind. Ferner sei $S\colon (C^\infty(M)[[\lambda]],\star) \to
   (C^\infty(M')[[\lambda]],\star')$ ein $G$"=äquivarianter
   Isomorphismus von Algebren. Ist nun $\qJ$ eine Quantenimpulsabbildung
   für $\star$, so ist $S(\qJ)$ eine Quantenimpulsabbildung für
   $\star'$, wobei $S(\qJ)$ als $S(\qJ)(\xi) := S(\qJ(\xi))$ für alle
   $\xi \in \lieAlgebra$ zu verstehen ist.
\end{proposition}
\begin{proof}
   Klar.
\end{proof}

Mit Hilfe von Proposition
\ref{prop:QuantenImpulsabbildungUndSternproduktHomos} gelingt es
Schaumann eine Klassifikation der Sternprodukte mit
Quantenimpulsabbildungen in Termen der $G$"=äquivarianten
de"=Rham"=Kohomologie anzugeben, siehe \cite{schaumann:2010}.  Wir
wollen an dieser Stelle jedoch nicht näher darauf eingehen.

\chapter{Koszul-Reduktion}
\label{cha:Koszul-Reduktion}

Ziel dieses Kapitels ist es, ein Sternprodukt für den reduzierten
Phasenraum zu konstruieren. Das vorgestellte Schema basiert auf der
Arbeit \cite{bordemann.herbig.waldmann:2000a} von Bordemann, Herbig und
Waldmann, welche in dem zitierten Papier die aus der Quantenfeldtheorie
bekannte Quantisierungsmethode nach Becchi, Rouet, Stora und Tyutin,
kurz BRST"=Methode genannt, in den Rahmen der Deformationsquantisierung
übertragen haben. In unserer Situation ist es jedoch
nicht notwendig, den gesamten BRST"=Komplex zu betrachten. Aus Sicht der
homologischen Algebra genügt es, wie wir sehen werden, die untersten
Grade des klassischen und des Quanten"=Koszul"=Komplexes zu
verstehen. Dieser Weg wurde auch in einer Arbeit von Gutt und Waldmann
\cite{gutt2010involutions} beschritten.

In diesem Kapitel sei $(M,\omega)$ eine symplektische Mannigfaltigkeit
mit induzierter Poisson"=Klammer $\{\cdot,\cdot\}$. Weiter sei $G$ eine
zusammenhängende Lie"=Gruppe, die stark Hamiltonsch sowie frei und
eigentlich auf $M$ operiere und $J\colon M \to \lieAlgebra^*$ eine
$G$"=äquivariante Impulsabbildung. Die Zwangsfläche $C := J^{-1}(0)$ sei
dabei nicht leer. Ferner bezeichnen wir mit $\kIn \colon C
\hookrightarrow M$ die Inklusion und mit $\pi \colon C \to \Mred$ die
kanonische Projektion auf den reduzierten Phasenraum $\Mred =
C/G$. Schließlich sei $\star$ ein Sternprodukt für $(M,\omega)$, für das
eine $G$"=äquivariante Quantenimpulsabbildung $\qJ \colon \lieAlgebra
\to C^\infty(M)[[\lambda]]$ existiert.

\section{Symplektische Reduktion und Koszul-Komplex}
\label{sec:SymplektischeReduktionUndKoszulKomplex}

\subsection{Algebraische Beschreibung von $C^\infty(\Mred)$}
\label{sec:AlgeraischeBeschreibungDerFunktionenAufDemReduziertenPhasenraum}

In diesem Abschnitt formulieren wir die geometrische
Marsden"=Weinstein"=Reduktion algebraisch, um so ein analoges Vorgehen in
der Quanten"=Situation motivieren zu können.  Zuerst möchten wir die
Poisson"=Algebra $(C^\infty(\Mred),\{\cdot,\cdot\}_{\mathrm{red}})$ der
glatten Funktionen auf dem reduzierten Phasenraum so beschreiben, dass
die Beziehung zur Poisson"=Algebra der glatten Funktionen auf dem
ursprünglichen Phasenraum algebraisch klar sichtbar wird und dann als
Grundlage für die spätere Konstruktion des Sternprodukts für den
reduzierten Phasenraum dienen kann.

\begin{proposition}
   \label{prop:InvarianteFunktionen}
   Die Abbildung
   \begin{align}
      \label{eq:FunktionenAufRedPhasenraumZwangsflaecheInvariant}
      \pi^* \colon \CM[\Mred] \to \CM[C]^G
   \end{align}
   ist ein
   Algebraisomorphismus zwischen den glatten Funktionen
   $C^\infty(\Mred)$ auf dem reduzierten Phasenraum und den
   $G$"=invarianten Funktionen $C^\infty(C)^G$ auf der Zwangsfläche.
\end{proposition}
\begin{proof}
   Zunächst ist klar, dass $\pi^*(C^\infty(\Mred))$ tatsächlich in
   $C^\infty(C)^G$ liegt. Ebenfalls unmittelbar einsichtig ist die
   Verträglichkeit mit der Algebramultiplikation.  Die Injektivität von
   $\pi^*$ folgt sofort aus der Surjektivität von $\pi$. Um die
   Surjektivität von $\pi^*$ zu zeigen, sei $f \in \CM[C]^G$ beliebig. Da $f$
   $G$"=invariant und $\pi$ eine surjektive Submersion ist, gibt es
   eine eindeutig bestimmte glatte Funktion $\tilde f \in \CM[\Mred]$
   mit $f = \pi^* \tilde f$.

\end{proof}

Sei $\kIdeal := \ker \kRes$ das Verschwindungsideal von $C$ und
$\kbIdeal := \{f \in \CM \mid \{f,h\} \in \kIdeal \;\, \forall h \in
\kIdeal\}$ der Lie"=Idealisator von $\kIdeal$.

  Die folgende wohlbekannte Proposition gibt eine algebraische
  Charakterisierung des Tangentialraumes einer Untermannigfaltigkeit mit
  Hilfe des Verschwindungsideals.
  \begin{proposition}
     \label{prop:TangentialRaumUntermannigfaltigkeit}
     Sei $\kIn \colon C \hookrightarrow M$ eine Untermannigfaltigkeit einer Mannigfaltigkeit
     $M$, dann gilt für alle $c \in C$
     \begin{align}
        \label{eq:TangentialRaumUntermannigfaltigkeit}
        T_c\kIn T_cC = \{v \in T_cM \mid v(f) = 0 \quad \forall f \in
        \kIdeal\} \Fdot
     \end{align}
  \end{proposition}
  \begin{proof}
     Siehe \cite[Prop. 8.5]{lee:2003a}.
  \end{proof}
  Als nächstes geben wir ein geometrisches Kriterium an, wann $\kIdeal$
  eine Poisson"=Unteralgebra von $C^\infty(M)$ ist.  Dazu sei zunächst
  daran erinnert, dass man eine Untermannigfaltigkeit $\kIn \colon C
  \hookrightarrow M$ \neuerBegriff{koisotrop} nennt, wenn in jedem Punkt
  $c \in C$ das symplektische Komplement $T_cC^\bot := \{w_c \in T_cM
  \mid \omega_c(T_c\kIn v_c,w_c) = 0 \quad \forall v_c \in T_cC\}$ von
  $T_cC$ in $T_c\kIn T_cC$ liegt
  (vgl.\ \cite[Ch. 5.3]{abraham.marsden:1985a},
  \cite[4.1.5]{ortega.ratiu:2004}).
  \begin{proposition}
     \label{prop:Koisotrop}
     Sei $(M,\omega)$ eine symplektische Mannigfaltigkeit mit
     induzierter Poisson"=Klammer $\{\cdot,\cdot\}$ und $C \subset
     M$ eine Untermannigfaltigkeit. Dann sind die folgenden Aussagen richtig.
     \begin{propositionEnum}
     \item %
        \label{item:CharaktSymplektKomplement}
        Sei $c \in C$, dann gilt für das symplektische Komplement $T_cC^\bot
        \subset T_cM$
        \begin{align}
           \label{eq:SymplektischeKomplement}
           T_cC^\bot = \{X_f(c) | f \in \kIdeal\} \Fdot
        \end{align}
     \item %
        \label{item:PoissonUnteralgebraKoisotrop}
        $C$ ist genau dann koisotrop, wenn $\kIdeal$ eine
        Poisson"=Unteralgebra von $C^\infty(M)$ ist.
     \end{propositionEnum}
  \end{proposition}
  \begin{proof}

     \begin{beweisEnum}
        \item %
           Zunächst beachte man, dass die rechte Seite von Gleichung
           \eqref{eq:SymplektischeKomplement} tatsächlich ein Vektorraum
           ist, da  $C^\infty(M) \ni f \mapsto X_f \in
           \Gamma^\infty(TM)$ offensichtlich linear ist.
          Sei $f \in \kIdeal$, dann folgt für alle $c \in C$ und $v
          \in T_c\kIn T_cC$
          \begin{align*}
             \omega(X_f(c),v) = d f\at{c}(v) = v(f) = 0 \Fcom
          \end{align*}
          also $X_f(c) \in T_cC^\bot$. Für die umgekehrte Inklusion sei
          $(U,x)$ eine Untermannigfaltigkeitskarte von $C$ um $c \in C$
          und $\chi \in C^\infty(M)$ eine glatte Abschneidefunktion mit
          $\supp \chi \subset U$ und $\chi = 1$ auf einer Umgebung von
          $c$. Dann gilt für $j \in \{\dim C + 1,\dots,\dim M\}$ und
          $f^j := \chi \cdot x^j \in C^\infty(M)$ die Gleichung $\kIn^*
          f^j = 0$ und $\{df^j\at{c}\} = \{dx^j\at{c}\}$ ist linear
          unabhängig. Somit ist $\{X_{f^j}(c)\}$ linear
          unabhängig. Andererseits gilt
          (vgl.\ \cite[Prop. 5.3.2]{abraham.marsden:1985a}) $\dim
          T_cC^\bot = \dim T_cM - \dim T_c\kIn T_cC = \dim M - \dim C$.
       \item %
          Sei $C \subset M$ koisotrop, d.\,h.\ für alle $c \in C$ gilt
          $T_cC^\bot \subset T_c\kIn T_cC$. Seien dann $f,f' \in
          \kIdeal$, so gilt nach Proposition
          \ref{prop:TangentialRaumUntermannigfaltigkeit} und
          Teil~\refitem{item:CharaktSymplektKomplement}
          $\kIn^*\{f,f'\}(c) = X_{f'}(c)(f) = 0$ für alle $c \in C$.

          Sei umgekehrt $\kIdeal$ eine Poisson"=Unteralgebra und $v \in
          T_cC^\bot$ beliebig. Dann gibt es nach Teil~\refitem{item:CharaktSymplektKomplement} ein $f \in \kIdeal$
          mit $v = X_f(c)$. Es folgt $X_f(c)(f') = \kIn^*\{f',f\}(c) =
          0$ für alle $f' \in \kIdeal$ und $c \in C$, also mit
          Prop. \ref{prop:TangentialRaumUntermannigfaltigkeit} $v =
          X_f(c) \in T_c\kIn T_cC$.
     \end{beweisEnum}
  \end{proof}

  Da in der von uns betrachteten Situation die Wirkung eigentlich ist,
  ist die Bahn $G\cdot p$ für alle $p \in M$ eine eingebettete
  Untermannigfaltigkeit von $M$ und der Tangentialraum $T_{p'}(G\cdot
  p)$ an $p' \in M$ lässt sich als
  \begin{align}
     \label{eq:TangentialraumAnBahn}
     T_{p'}(G\cdot p) = \{\xi_M(p') \mid \xi \in \lieAlgebra \}
  \end{align}
schreiben, siehe etwa \cite[Bem. 3.3.26,
  Prop. 3.3.24]{waldmann:2007a}.

Offensichtlich gilt für jedes $p \in M$, $v \in T_pM$ und $\xi \in \lieAlgebra$
  \begin{align*}
     \omega_p(\xi_M(p),v) = \dPaar{d_pJ(\xi)}{v} = \dPaar{I_{J(p)}^{-1}
       \circ T_pJv}{\xi} \Fcom
  \end{align*}
  wobei $I_{J(p)}\colon \lieAlgebra^* \to T_{J(p)}\lieAlgebra^*$ die
  kanonische Identifikation der Mannigfaltigkeit $\lieAlgebra^*$ mit
  ihrem Tangentialraum $T_{J(p)}\lieAlgebra^*$ an der Stelle $J(p)$
  bezeichnet. Die elementare Rechnung zum zweiten Gleichheitszeichen
  wurde im Beweis von Lemma \ref{lem:JSurjektiv} ausgeführt.
  Damit liegt $v$ genau dann in $\ker T_pJ$, wenn es ein Element von
  $(T_p(G\cdot p))^\bot$ ist. Wegen $T_c\kIn T_cC = \ker T_cJ$
  (vgl.\ \cite[Lemma 8.15]{lee:2003a}), gilt also für $c \in C$
  \begin{align*}
     T_c(G\cdot c)^\bot = T_c\kIn T_cC \Fdot
  \end{align*}
  Mit der Kettenregel, der Äquivarianz von $J$ und Definition von $C$
  gilt für $c \in C$
  \begin{align*}
  T_cJ \xi_M(c) = T_cJ T_e \phi_{\cdot}(c) \xi = T_e(J \circ \phi_{\cdot}(c)) = T_e(\phi_{\cdot}(\underbrace{J(c)}_{=0}))
  = 0 \Fcom
\end{align*}
wobei $\phi \colon G \times M \to M$ die gegebene $G$"=Wirkung auf $M$
ist. Mit $T_c\kIn T_cC = \ker T_cJ$ und Gleichung
\eqref{eq:TangentialraumAnBahn} folgt daraus
\begin{align*}
   T_c(G\cdot c) \subset T_c\kIn T_cC\Fcom
\end{align*}
  und demnach auch
  \begin{align}
     \label{eq:BeiUnsKoisotrop}
     T_cC^\bot = T_c(G \cdot c)^{\bot \bot} = T_c(G \cdot c) \subset
     T_c\kIn T_cC \Fdot
  \end{align}
  Für einen Beweis der dabei im zweiten Schritt angewendeten Rechenregel
  verweisen wir auf \cite[Prop. 5.3.2]{abraham.marsden:1985a}.
  Wir können also die folgende Proposition festhalten.

\begin{proposition}
   \label{prop:BeiUnsKoisotrop}
   $C = J^{-1}(0)$ liegt koisotrop in $M$.
\end{proposition}
Wir sind nun in der Lage, die angestrebte algebraische Formulierung der
Observablenalgebra des reduzierten Phasenraumes anzugeben. Dies ist
Gegenstand der nächsten Proposition.

{
\emergencystretch=1.6em
\begin{proposition}
   \label{prop:FunktionenAufDemReduziertenPhasenraum}
 \begin{propositionEnum}
 \item %
    $\kIdeal$ ist ein Poisson"=Ideal\footnote{d.\,h.\ $\kIdeal$ ist sowohl
      ein Ideal bezüglich Algebramultiplikation als auch bezüglich der
      durch $\{\cdot,\cdot\}$ gegebenen Lie"=Klammer, es gilt also für
      alle $f \in C^\infty(M)$ und $f' \in \kIdeal$ $f\cdot f' \in \kIdeal$
      und $\{f,f'\} \in \kIdeal$,
      vgl.\ \cite[S. 347]{cushman.bates:1997a}} in $\kbIdeal$, wodurch
    $\kbIdeal/\kIdeal$ via $\{[f],[h]\}' := [\{f,h\}]$ zu einer
    Poisson"=Algebra mit Poisson"=Klammer $\{\cdot,\cdot\}'$ wird.
   \item %
      \label{item:EinschraenkungSurjektiv} %
       Die Abbildung  $\kIn^* \colon \CM[M] \to C^\infty(C)$ ist
       surjektiv.
    \item %
       \label{item:CharakterisierungDesKlassischenIdealisators}
       Es gilt
       \begin{align}
          \label{eq:CharakterisierungDesKlassischcenIdealisators}
          \kbIdeal = \{f \in C^\infty(M) \mid \kRes f \in
          \pi^*C^\infty(\Mred)\} \Fdot
       \end{align}

       Ist $\prol \colon C^\infty(C) \to C^\infty(M)$ ein
       Rechtsinverses zu $\kRes$, dann gilt insbesondere $\prol
       \pi^*C^\infty(\Mred) \subset \kbIdeal$.
    \item %
      Die Abbildung
      \begin{align}
         \label{eq:FunktionenAufDemReduziertenPhasenraum}
         \mathsf{iso}\colon \kbIdeal/\kIdeal \ni [f] \mapsto \kIn^* f \in
         \pi^*\CM[\Mred] = \CM[C]^G
      \end{align}
      ist ein Algebraisomorphismus. Dabei ist die Algebramultiplikation auf
      $\kbIdeal/\kIdeal$ durch $[f][f'] := [ff']$ für $f,f' \in
      \kbIdeal$ gegeben.

      Ist $\prol \colon
      C^\infty(C) \to C^\infty(M)$ ein Rechtsinverses zu $\kIn^* \colon
      C^\infty(M) \to C^\infty(C)$, so gilt für jedes $f \in {C^\infty(C)}^G$
      \begin{align}
         \label{eq:FunktionenAufDemReduziertenPhasenraumIsoInverses}
         \mathsf{iso}^{-1}(f) = [\prol(f)] \in \kbIdeal/\kIdeal \Fdot
      \end{align}
   \item %
      Die Abbildung
      \begin{align}
         \label{eq:FunktionenAufDemReduziertenPhasenRaumEndgIso}
         \mathsf{iso}^{-1} \circ \pi^* \colon C^\infty(\Mred) \to \kbIdeal/\kIdeal
      \end{align}
      ist ein Isomorphismus von Poisson"=Algebren. Ist $\prol \colon
      C^\infty(C) \to C^\infty(M)$ ein Rechts\-in\-verses zu $\kIn^* \colon
      C^\infty(M) \to C^\infty(C)$, so gilt insbesondere für $\varphi,\varphi' \in C^\infty(\Mred)$
      \begin{align}
         \label{eq:ReduziertePoissonKlammer}
         \pi^*\{\varphi,\varphi'\}_{\mathrm{red}} = \kIn^*\{\prol \pi^* \varphi, \prol \pi^* \varphi'\} \Fdot
      \end{align}
   \end{propositionEnum}
\end{proposition}
}
\begin{proof}
   \begin{beweisEnum}
   \item %
      Da nach Proposition
      \ref{prop:Koisotrop}~\refitem{item:PoissonUnteralgebraKoisotrop}
      und Proposition \ref{prop:BeiUnsKoisotrop} das Ideal $\kIdeal$
      eine Poisson"=Unteralgebra von $C^\infty(M)$ ist, gilt
      offensichtlich $\kIdeal \subset \kbIdeal $. Der Rest ist somit
      klar.
   \item %
      Man kann jedes $f \in C^\infty(C)$ lokal glatt fortsetzen und
      vermöge einer Zerlegung der Eins eine glatte Fortsetzung von $f$
      auf eine offene Umgebung von $C$ erlangen. Mit Hilfe einer glatten
      Abschneidefunktion erhält man letztlich eine Fortsetzung  zu
      einer glatten Funktion auf ganz $M$.
   \item %
      Zunächst bemerken wir, dass für $c \in C$ die  Gleichung
      \begin{align}
         \label{eq:FundiHamVektorFelder}
         \{\xi_M(c) \mid \xi \in \lieAlgebra\} = T_c(G \cdot c) =
         T_cC^\bot = \{X_h(c) \mid h \in \kIdeal\} \Fdot \tag{$*$}
      \end{align}
      gilt.  Sei nun $f \in C^\infty(M)$. Wir betrachten die folgenden
      Äquivalenzumformungen.
      \begin{align*}
         \kIn^*f \in C^\infty(C)^G &\iff \Phi_g^*(\kIn^* f) = \kIn^* f
         &&\forall g \in G \\
         &\iff f(g c) = f(c) &&\forall g \in G, c\in C \\
         &\iff \xi_M(c)(f) = 0 &&\forall \xi \in \lieAlgebra, c \in
         C \\
         &\iff X_h(c)(f) = 0 &&\forall h \in \kIdeal, c \in C\\
         &\iff \{f,h\}(c) = 0 &&\forall h \in \kIdeal, c\in C \\
         &\iff f \in \kbIdeal \Fdot
      \end{align*}
      Dabei wurde in der dritten Äquivalenz verwendet, dass $G$
      zusammenhängend ist und in der vierten die Gleichung
      \eqref{eq:FundiHamVektorFelder}.

   \item %
      Aufgrund der Definition von $\kIdeal = \ker \kRes$ ist nach
      Gleichung~\eqref{eq:CharakterisierungDesKlassischcenIdealisators}
      klar, dass $\mathsf{iso}$ wohldefiniert und injektiv ist. Auch die
      Verträglichkeit mit der Algebra- und Vektorraumstruktur ist
      klar. Die Surjektivität folgt aus Teil~\refitem{item:EinschraenkungSurjektiv} und~\refitem{item:CharakterisierungDesKlassischenIdealisators}. Sei
      nun $\prol$ ein Rechtsinverses für $\kIn^*$. Dann gilt für $f \in
      {C^\infty(C)}^G$ nach der obigen Äquivalenz $\prol f \in
      \kbIdeal$, da $\kRes \prol f = f \in C^\infty(C)^G$. Weiter ist
      dann offensichtlich die  Gleichung
      \begin{align*}
         \mathsf{iso}([\prol(f)]) = \kIn^* \prol(f) = f
      \end{align*}
      richtig, und für $f \in \kbIdeal$ gilt
      \begin{align*}
         [\prol(\mathsf{iso}([f]))] = [\prol(\kIn^* f)] = [f] \Fcom
      \end{align*}
      da $\prol \kIn^* f - f \in \kIdeal$.
   \item %
      Nach dem bisher Gezeigten ist klar, dass $\mathsf{iso}^{-1} \circ
      \pi^*$ ein Isomorphismus von Algebren ist. Wir müssen also nur
      prüfen, ob die Poisson"=Klammern respektiert werden. Dazu sei
      $\prol$ ein Rechtsinverses zu $\kIn^*$.

      Zunächst sehen wir mit Gleichung
      \eqref{eq:TangentialRaumUntermannigfaltigkeit}, dass für das
      Hamiltonsche Vektorfeld $X_f$ einer Funktion $f \in \kbIdeal$ für
      alle $c \in C$ die Beziehung $X_f(c) \in T_c\kIn T_cC$ gilt, denn
      für $f' \in \kIdeal$ ist $X_f(c)(f') = \kIn^*\{f',f\}(c) = 0$. Es
      gibt somit ein Vektorfeld $\tilde X_f$ auf $C$, das zu $X_f$
      $\kIn$"=verwandt ist.  Weiter seien $c \in C$ und $v \in T_cC$
      beliebig. Dann gilt für $\varphi \in C^\infty(\Mred)$
      \begin{align*}
         \omega_{\mathrm{red}}\at{\pi(c)}(T_c\pi {\tilde X}_{\prol \pi^*
           \varphi}(c), T_c\pi v) &= \omega\at{c}(T_c \kIn{\tilde X}_{\prol \pi^*
           \varphi}(c),T_c\kIn v) \\
         &= \omega\at{c}(X_{\prol \pi^*\varphi}(c),T_c\kIn v)
         \\
         &= \dPaar{d_c(\prol \pi^* \varphi)}{T_c \kIn v} \\
         &= \dPaar{d_c(\prol
           \pi^* \varphi \circ \kIn)}{v} \\
         &= \dPaar{d_c(\kIn^*\prol (\varphi \circ
           \pi))}{v} \\
         &= \dPaar{d_c(\varphi \circ \pi)}{v} \\
         &= \dPaar{d_{\pi(c)}\varphi}{T_c\pi v} \\
         &= \omega_{\mathrm{red}}\at{\pi(c)}(X_\varphi(\pi(c)),T_c\pi v)\Fdot
      \end{align*}
      Wegen der Nicht"=Ausgeartetheit von $\omega_{\mathrm{red}}$ sind
      demnach die Vektorfelder ${\tilde X}_{\prol \pi^* \varphi}$ und $X_\varphi$
      $\pi$"=verwandt. Somit ergibt sich für $\varphi,{\varphi'} \in C^\infty(\Mred)$
      \begin{align*}
         \kIn^* \{\prol \pi^* \varphi,\prol \pi^* {\varphi'}\} %
         &= \kIn^*(\omega(X_{\prol
           \pi^* \varphi}, X_{\prol \pi^* {\varphi'}})) \\
         &= \kIn^*(\omega(T\kIn {\tilde
           X}_{\prol \pi^* \varphi},T \kIn {\tilde X}_{\prol \pi^* {\varphi'}}))\\
         &= (\kIn^*\omega)({\tilde X}_{\prol \pi^* \varphi},{\tilde X}_{\prol
           \pi^* {\varphi'}}) \\
         &= (\pi^*\omega_{\mathrm{red}})({\tilde X}_{\prol
           \pi^* \varphi},{\tilde X}_{\prol \pi^* {\varphi'}} )\\
         &= (\omega_{\mathrm{red}} \circ \pi)(T\pi {\tilde X}_{\prol \pi^* \varphi}, T \pi
         {\tilde X}_{\prol \pi^* {\varphi'}}) \\
         &= ({\omega_{\mathrm{red}}\circ \pi})(X_\varphi\circ
         \pi,X_{\varphi'}\circ \pi) \\
         &= \pi^*(\omega_{\mathrm{red}}(X_\varphi,X_{\varphi'})) =
         \pi^*\{\varphi,{\varphi'}\}_{\mathrm{red}} \Fdot
      \end{align*}
      Damit gilt natürlich auch
      \begin{align*}
         (\mathsf{iso}^{-1}\circ \pi^*)(\{\varphi,{\varphi'}\}_{\mathrm{red}}) %
         &= [\prol(\kIn^* \{\prol \pi^* \varphi, \prol \pi^* {\varphi'}\})] \\
         &= [\{\prol
         \pi^* \varphi, \prol \pi^* {\varphi'} \}]\\
         &= \{(\mathsf{iso}^{-1}\circ\pi^*) \varphi,(\mathsf{iso}^{-1} \circ\pi^*)
         {\varphi'}\} \Fcom
      \end{align*}
      d.\,h.\ $\mathsf{iso}^{-1}\circ \pi^*$ ist ein Isomorphismus von Poisson"=Algebren.
   \end{beweisEnum}
\end{proof}

\subsection{Augmentierter Koszul-Komplex}
\label{sec:AugementierterKoszulKomplex}

Da a priori keine \glqq{}Quanteneinschränkung\grqq{} gegeben ist,
wollen wir eine andere Charakterisierung von $\kIdeal$ angeben, die eine
naheliegende Analogie für den Quantenfall zulässt. Dazu sei
$\langle J \rangle$ das von der Impulsabbildung $J$ erzeugte Ideal. Ist
$\{e_\alpha\}$ eine Basis von $\lieAlgebra$, so ist $\langle J \rangle =
\{f^\alpha \dPaar{J}{e_\alpha} \mid f^\alpha \in C^\infty(M)\}$. Man
beachte, dass diese Definition unabhängig von der Wahl der Basis
$\{e_\alpha\}$ von $\lieAlgebra$ ist. Offensichtlich gilt $\langle J \rangle
\subset \kIdeal$. Wir wollen im Folgenden zeigen, dass sogar $\langle J
\rangle = \kIdeal$ erfüllt ist. Dazu müssen wir  beweisen, dass es für
jedes $f \in \kIdeal$ glatte Funktionen $f^\alpha \in C^\infty(M)$ gibt
mit $f = f^\alpha \dPaar{J}{e_\alpha}$.  Definiert man die Abbildung
$\partial \colon C^\infty(M) \otimes \lieAlgebra \to C^\infty(M)$, $f
\otimes \xi \mapsto f \dPaar{J}{\xi}$ und lineare Fortsetzung, so
schreibt sich dies als $f = \partial (f^\alpha \otimes e_\alpha)$.

\begin{definition}[Augmentierter Koszul-Komplex]
   \label{def:KoszulKomplex}
   Für $k \in \mathbb{N}$ definiert man die Abbildung
   \begin{align}
      \label{eq:KoszulFormel}
      \partial  \colon C^\infty(M) \otimes \Bigwedge^k
      \lieAlgebra \to C^\infty(M) \otimes \Bigwedge^{k-1} \lieAlgebra
   \end{align}
   durch lineare Fortsetzung von
\begin{align}
\label{eq:KoszuFormel2}
   f\otimes \xi \mapsto f \dPaar{J}{e_\alpha}\otimes
      \Ins{e^\alpha}\xi \Fcom
\end{align}
wobei $\Ins{e^\alpha}$ die antisymmetrische Einsetzderivation  bezeichnet.
Die so erhaltene Sequenz

   \def\tA[#1]{A_{#1}}
   \begin{equation}
      \begin{tikzpicture}[baseline=(current
         bounding box.center),description/.style={fill=white,inner sep=2pt}]
         \matrix (m) [matrix of math nodes, row sep=3.0em, column
         sep=3.5em, text height=1.5ex, text depth=0.25ex] %
         { C^\infty(C) & C^\infty(M) \otimes \Bigwedge^0 \lieAlgebra &
           C^\infty(M)
           \otimes \Bigwedge^1 \lieAlgebra & \dots \\
         }; %

         \path[<-] (m-1-1) edge node[auto]{$\kIn^*$}(m-1-2); %
         \path[<-] (m-1-2) edge node[auto]{$\partial$}(m-1-3); %
         \path[<-] (m-1-3) edge node[auto]{$\partial$}(m-1-4); %
      \end{tikzpicture},
   \end{equation}
   heißt \neuerBegriff{(augmentierter) Koszul"=Komplex} $K$.
\end{definition}
\begin{bemerkung}
   \label{bem:KoszulKomplex}
   Der augmentierte Koszul"=Komplex ist tatsächlich ein Komplex, d.\,h.\ es
   gilt $\partial^2 = 0$ und $\kRes \partial = 0$.
\end{bemerkung}

Eigentlich bezeichnet man nur den Teil ohne $C^\infty(C)$ als
Koszul"=Komplex \cite[XXI,§4]{lang:1997a},
\cite[3.4.6]{loday:1998a}. Durch $C^\infty(C)
\stackrel{\phantom{=}\kRes}{\longleftarrow}$ wird dieser augmentiert. Ist $U
\subset M$ eine offene Teilmenge, so bezeichnen wir im Folgenden mit
$\partial_U \colon C^\infty(U)\otimes \Bigwedge^\bullet \lieAlgebra \to
C^\infty(U)\otimes \Bigwedge^{\bullet + 1} \lieAlgebra$ die eindeutig
bestimmte Abbildung, so dass für alle $f \otimes \xi \in
C^\infty(M)\otimes \Bigwedge^\bullet \lieAlgebra$ die Gleichung
$\partial_U( f\at{U}\otimes \xi) = (\partial (f\otimes \xi))\at{U}$
gilt. Die Abbildung $\partial_U$ ist offenbar für alle $k \in
\mathbb{N}$ und $f \otimes \xi \in C^\infty(U) \otimes \Bigwedge^k
\lieAlgebra$ durch $\partial_U(f\otimes\xi) = f J_\alpha\at{U} \otimes
\Ins{e^\alpha}\xi$ und lineare Fortsetzung gegeben.

Die Gleichung $\langle J \rangle = \kIdeal$ ist dann äquivalent zu der
Aussage, dass der augmentierte Koszul"=Komplex an der Stelle
$C^\infty(M) \otimes \Bigwedge^0 \lieAlgebra$ exakt ist. Da es
keinen wesentlichen Mehraufwand bedeutet, vor allem da es sich für die
spätere Quantisierung als nützlich erweisen wird, werden wir den ganzen
augmentierten Koszul"=Komplex betrachten und durch Konstruktion einer
expliziten zusammenziehbaren Kettenhomotopie zeigen, dass alle
Homologien verschwinden. Mit anderen Worten, der augmentierte
Koszul"=Komplex ist eine Auflösung des $C^\infty(M)$"=Moduls
$C^\infty(C)$. Bevor wir uns dieser etwas technischen Aufgabe zuwenden,
sei noch folgende Äquivarianzeigenschaft von $\partial$ bemerkt.
\begin{proposition}[Äquivarianz von $\partial$]
   \label{prop:partialInvariant}
   Sei $H$ eine Lie"=Gruppe, die auf $M$ wirke und sei weiter eine
   lineare $H$"=Wirkung auf $\lieAlgebra[g]$ gegeben, so dass $J$
   $H$"=äquivariant ist, dann ist auch $\partial$
   $H$"=äquivariant.
\end{proposition}
\begin{proof}
   Sei $h \in H$, dann gilt
   \begin{align*}
       h (\partial (f \otimes \xi)) &= h ( J(e_{\alpha})  f
      \otimes \Ins{e^{\alpha}} \xi)\\
      &= h (J(e_{\alpha})  f) \otimes h \Ins{e^{\alpha}}\xi \\
      &= h J(e_{\alpha})  h f \otimes \Ins{h e^{\alpha}} (h \xi)
      \\
      &= J(h e_{\alpha})  h f \otimes \Ins{h e^\alpha} (h \xi) \\
      &= \partial(h f \otimes h \xi) = \partial (h (f \otimes \xi)) \Fdot
   \end{align*}
\end{proof}

\subsection{Konstruktion der Homotopie}
\label{sec:KonstruktionDerHomotopie}

Da in dieser Arbeit nur zusammenziehbare Kettenhomotopien vorkommen,
bezeichnen wir diese im Folgenden der Kürze halber als Homotopie.

Ist $U$ eine offene Umgebung von $C$ und $r \colon U \to C$ eine glatte
Retraktion, so liegt es nahe, diese zu verwenden um auf geometrischem
Wege ein Rechtsinverses zu $\kRes$ zu konstruieren, wie es schon in
Proposition \ref{prop:FunktionenAufDemReduziertenPhasenraum} rein
algebraisch verwendet wurde. Dies führt die nächste Definition genauer
aus. In unserer Situation existiert nach dem Tubensatz
\ref{satz:GkompatibleTuben} immer eine derartige glatte Retraktion. Ist
$H$ eine Lie"=Gruppe, die eigentlich auf $M$ wirkt und $C$ stabilisiert,
so kann $U$ $H$"=invariant und $r$ $H$"=äquivariant gewählt werden.

\begin{definition}[Geometrische Prolongationsabbildung]
   \label{def:geometrischeProlongation}
   Sei $C$ eine Untermannigfaltigkeit einer Mannigfaltigkeit $M$, $U$
   und $O$ offene Umgebungen von $C$ in $M$ mit $\abschluss{O}
   \subset U$ und $\psi_U \colon
   M \to \mathbb{R}$ eine glatte Funktion, so dass $\{\psi_U,1 -
   \psi_U\}$ eine glatte Zerlegung der Eins ist, die der offenen
   Überdeckung $\{U,M\setminus \abschluss{O}\}$ untergeordnet ist. Weiter sei $r
   \colon U \to C$ eine glatte Retraktion auf $C$. In diesem Fall nennen
   wir die Abbildung
   \begin{align}
      \label{eq:geometrischeProlongation}
      &\prol_{M,O,r,\psi_U}  \colon C^\infty(C) \to C^\infty(M)\\
      &\prol_{M,O,r,\psi_U}(f)(p) :=
      \begin{cases}
         \psi_U(p)f(r(p)) & \text{falls $p \in U$} \\
         0 & \text{sonst}
      \end{cases}
   \end{align}
   \neuerBegriff{geometrische Prolongationsabbildung} oder kurz
   \neuerBegriff{geometrische Prolongation}. Falls aus dem Zusammenhang
   klar ist, welche Daten $(M,O,r,\psi_U)$ gewählt wurden schreiben
   wir einfach $\prol$ statt $\prol_{M,O,r,\psi_U}$.

\end{definition}

\begin{bemerkung}
   \label{bem:geometrischeProlongation}
   Es gelten die Bezeichnungen aus Definition
   \ref{def:geometrischeProlongation} und $H$ sei eine Lie"=Gruppe, die
   auf $M$ wirke. Falls $C$, $U$, $O$ und $\psi_U$ $H$"=invariant sind
   und $r$ $H$"=äquivariant ist, ist auch $\prol$ $H$"=äquivariant. Nach
   Satz \ref{satz:GkompatibleTuben} und Satz
   \ref{satz:inv_zerlegung_der_eins} existiert immer eine glatte,
   $H$"=äquivariante Retraktion und eine glatte $H$"=invariante Funktion
   $\psi_U$ mit den Eigenschaften aus Definition
   \ref{def:geometrischeProlongation}.
\end{bemerkung}

\begin{proposition}
   \label{prop:EigenschaftenGeometrischeProlongation}
   Ist $\prol$ eine geometrische Prolongation, so gilt
   \begin{align}
      \kIn^* \prol = \id\label{eq:EigenschaftenGeometrischeProlongation}\Fcom
   \end{align}
   wobei $\kIn \colon C \hookrightarrow M$ die kanonische Inklusion der
   Untermannigfaltigkeit $C$ in die Mannigfaltigkeit $M$ ist.
\end{proposition}
\begin{proof}
   Klar.
\end{proof}

Um die Geometrie besser unter Kontrolle zu haben und die
Homotopie konstruieren zu können, werden wir in
einer Tubenumgebung arbeiten, die wir, wie der folgende Satz zeigt,
besonders angepasst wählen können.

\begin{satz}
   \label{satz:CxgA}
   Sei $M$ eine glatte Mannigfaltigkeit, und \(W\) ein
   endlichdimensionaler Vektorraum. Weiter sei \(G\) eine Lie"=Gruppe,
   die auf \(M\) eigentlich wirke und auf $W$ linear. Zudem sei eine
   $G$"=äquivariante, glatte Abbildung $J \colon M \to W$ gegeben und $0
   \in W$ ein regulärer Wert von $J$ mit $C := J^{-1}(0) \neq
   \emptyset$. Ist schließlich \(r\colon U' \to C\) eine glatte,
   \(G\)"=äquivariante Retraktion von einer offenen, \(G\)"=invarianten
   Umgebung \(U'\) von \(C\) in \(M\) auf \(C\) (die z.\,B.\ von einer
   Tubenumgebung herrühren könnte), so gibt es \(G\)"=invariante, offene
   Umgebungen \(U \subset U'\) von \(C\) in \(M\) und \(V\) von \(C
   \times \{0\}\) in \(C \times W\), so dass
   \begin{align}
      \label{eq:DefinitionDerGutenTubenAbbildung}
   F \colon U \to V, \quad u \mapsto  (r(u),J(u))
\end{align}
   ein
   \(G\)"=äquivarianter Diffeomorphismus ist.
   Für alle \((c,v) \in V\) gilt dann
   \begin{align}
      \label{eq:TubeneigenschaftDerSpeziellenTube1}
      F^{-1}(c,0) = c \quad \text{und} \quad  J \circ F^{-1}(c,v) = v \Fdot
   \end{align}
   \(V\) kann sternförmig in Faserrichtung gewählt werden, d.\,h.\ derart, dass
   aus \((c,v) \in V \) schon \( (c,tv) \in V\) für alle \(t \in [0,1]\)
   folgt.
\end{satz}
\begin{proof}
   Wir definieren $\widetilde{F} \dpA U' \to C \times W$ durch
   $\widetilde{F}(u) := (r(u),J(u))$.
   Zuerst zeigen wir, dass $T_c\tilde{F}$ für alle $c \in C$ injektiv
   ist. Dazu sei $v \in T_cM$ mit $T_c\tilde{F}v = 0$. Mit Proposition
   \ref{prop:produkte}~\refitem{item:ProdukteKomma} ist dies aber
   äquivalent zu $T_cr v = 0$ und $T_cJv = 0$. Da $0$ ein regulärer Wert
   von $J$ ist, bedeutet letzteres, dass $v \in T_c\iota T_cC$ ist,
   wobei $\iota \dpA C \hookrightarrow M$ die Inklusion bezeichnet. Da
   weiter $r \circ \iota = \id$ gilt, folgt aus $T_cr v = 0$ schon $v =
   0$. Somit ist $T_c\widetilde{F}$ injektiv und aus Dimensionsgründen
   auch bijektiv für alle $c \in C$. Damit ist mit dem Satz über die
   Umkehrfunktion $\widetilde{F}$ ein lokaler Homöomorphismus auf einer
   offenen Umgebung von $C$. Ferner ist offensichtlich wegen der
   $G$"=Äquivarianz von $J$ die Menge $C \times \{0\}$ $G$"=invariant
   und $s \dpA C \times \{0\} \to U'$ mit $s((c,0)) = c$ liefert einen
   stetigen, $G$"=äquivarianten Schnitt von $\widetilde{F}$ auf $C \times
   \{0\}$. Wegen der $G$"=Äquivarianz von $r$ und $J$ ist offenbar auch
   $\widetilde{F}$
   $G$"=äquivariant.

   Mit Korollar \ref{kor:AufUmgebungVonCHomoe} folgt dann, dass es eine
   $G$"=invariante, offene Umgebung $\widetilde{U}$ von $s(C \times
   \{0\}) = C$ gibt, so dass $\widetilde{F}\at{\tilde U}$ ein
   $G$"=äquivarianter Homöomorphismus auf eine offene, $G$"=invariante
   Teilmenge $\widetilde{V}$ von $C \times W$ ist. Offenbar gilt nach
   Konstruktion $J \circ \widetilde{F}^{-1}(c,v) = v$ für $(c,v) \in
   \tilde{V} \subset C \times W$.  Nach Proposition \ref{lem:1Ball} gibt
   es eine $G$"=invariante Riemannsche Fasermetrik $h$ auf $C \times W$
   mit $V := B_1(C) \subset \widetilde{V}$.  Somit sind $U :=
   \widetilde{F}\at{\widetilde{U}}^{-1}(V)$, $V$ und $F :=
   \widetilde{F}\at{U} \dpA U \to V$ wie gewünscht.
\end{proof}

Wir können insbesondere den folgenden Spezialfall des vorangehenden
Satzes formulieren.

\begin{korollar}
   \label{kor:Cxg}
   Sei $(M,\omega)$ eine symplektische Mannigfaltigkeit und $G$ eine
   Lie"=Gruppe mit zugehöriger Lie"=Algebra $\lieAlgebra[G]$. $G$ wirke
   auf $\lieAlgebra^*$ vermöge $\Ad^*$ und auf $M$ stark Hamiltonsch sowie
   eigentlich und sei $J \dpA M \to \lieAlgebra[g]^*$ die zugehörige
   $G$"=äquivariante Impulsabbildung. Sei schließlich $0 \in
   \lieAlgebra[g]^*$ ein regulärer Wert für $J$ und $C :=
   J^{-1}(\{0\}) \neq \emptyset$.

   Dann gibt es eine $G$"=invariante Umgebung $U$ von $C$ in $M$ und
   einen $G$"=äquivarianten Diffeomorphismus $F \dpA U \to V \subset C
   \times \lieAlgebra[g]^*$ auf eine offene $G$"=invariante Umgebung $V$
   von $C \times \{0\}$, wobei die $G$"=Wirkung auf $C \times
   \lieAlgebra[g]^*$  durch $g (c,v) := (g c,\Ad^*_{g^{-1}}
   \,v)$, für alle $g \in G$, $(c,v) \in C \times \lieAlgebra[g]^*$
   gegeben ist.  Man kann
   sogar erreichen, dass $J \circ F^{-1}(c,v) = v$ für $(c,v) \in
   V$ gilt, sowie dass für jedes $(c,v) \in V$ auch $(c,tv) \in V$ für jedes
   $t \in [0,1]$ ist.
\end{korollar}

\begin{bemerkung}[Gute Tubenumgebung]
   \label{bem:SpezielleTubenumgebung}
   $(C\times W \to C,V,U,F^{-1})$ aus Satz \ref{satz:CxgA} ist insbesondere
   eine $G$"=äquivariante Tubenumgebung von $C$ mit $G$"=äquivarianter
   Retraktion $r \dpA U \to C$. Eine derartige Tubenumgebung wollen wir
   im Folgenden auch \neuerBegriff{von $r$ induzierte \tn{gute}
     Tubenumgebung} oder einfach \neuerBegriff{\tn{gute} Tubenumgebung}
   nennen. Analog sprechen wir dann von \neuerBegriff{\tn{guter}
     Tubenabbildung} usw.\ Schließlich vereinbaren wir den
   Notationsmissbrauch auch die inverse Tubenabbildung $F$ als solche zu
   bezeichnen und sprechen auch von der \tn{guten} Tubenabbildung $F \dpA U
   \to V$.
\end{bemerkung}

Mit Hilfe einer guten Tubenumgebung werden wir nun die gesuchten
Homotopie konstruieren.

\begin{lemma}
   \label{lem:HomotopienInDerTubenumgebung}
   Sei $\Psi \colon U \to V \subset C \times \lieAlgebra^*$ eine
   \tn{gute} Tubenabbildung mit zugehöriger Retraktion $r \colon U \to
   C$. Weiter sei die Abbildung $\h_U \colon C^\infty(U) \otimes
   \Bigwedge^\bullet \lieAlgebra \to C^\infty(U) \otimes
   \Bigwedge^{\bullet + 1}\lieAlgebra$ für alle $k \in \mathbb{N}$
   $f \otimes \eta \in C^\infty(U) \otimes \Bigwedge^k \lieAlgebra$ und
   $(c,\mu) \in V$
   durch lineare Fortsetzung von
   \begin{align}
      \label{eq:HomotopienInDerTubenumgebung}
      \h_U(f \otimes \eta) \circ \Psi^{-1}(c,\mu) := \int_{0}^1
      t^k \partial^\alpha(f \circ \Psi^{-1})(c,t \mu) \, dt \otimes
      e_\alpha \wedge \eta
   \end{align}
   gegeben, wobei
   $\partial^\alpha := \partial_{e^\alpha}$ die Ableitung in Richtung
   $e^\alpha$ bezeichnet. Dann gilt
   \begin{align}
      \label{eq:HomotopienInDerTubenumgebung2}
      \h_U\partial_U x + \partial_U \h_U x = x \quad \text{für $x \in
        C^\infty(U) \otimes \Bigwedge^k \lieAlgebra$ und $k \geq 1$} \Fcom
   \end{align}
   \begin{align}
      \label{eq:HomotopienInDerTubenumgebung3}
      r^* \kIn^* f + \partial_U \h_U f = f \quad \text{für $f \in
        C^\infty(U)$}
   \end{align}
   und
   \begin{align}
      \label{eq:HomotopienInDerTubenumgebung4}
       \h_U r^* = 0 \Fdot
   \end{align}

   Trägt $M$ eine Wirkung einer Lie"=Gruppe $H$ unter der $C$ stabil ist,
   $\lieAlgebra$ eine lineare $H$"=Wirkung und ist $\Psi$
   $H$"=äquivariant, so gilt dies auch für $\h_U$.
\end{lemma}
\begin{proof}
   Der Beweis verläuft ähnlich wie der des klassischen
   Poincar\'e"=Lemmas, (vgl.\ etwa \cite[Lemma 11.4,
   Thm. 11.11]{lee:2003a}).  Für $\alpha \in \{1,\dots,\dim G\}$ sei
   $\mathrm{pr}_\alpha \colon C \times \lieAlgebra^* \to \lieAlgebra^*$
   die $\alpha$"=te Koordinatenfunktion bezüglich der Basis
   $\{e^\alpha\}$, d.\,h.\  für $(c,\mu_\alpha e^\alpha) \in C \times
   \lieAlgebra^*$ gilt die Gleichung $\mathrm{pr}_\alpha(c,\mu_\beta
   e^\beta) = \mu_\alpha$. Man beachte $J_\alpha \circ \Psi^{-1} =
   \mathrm{pr}_\alpha$. Für $k \in \mathbb{N}$ und $f \otimes \eta \in
   C^\infty(U) \otimes \Bigwedge^k \lieAlgebra$ gilt schließlich

   \begin{align*}
      \lefteqn{(\partial \h_U f\otimes \eta)(\Psi^{-1}(c,\mu))}%
      \\ &= J_\alpha(\Psi^{-1}(c,\mu))
      \Ins{e^\alpha} \h_U(f \otimes \eta)(\Psi^{-1}(c,\mu)) \\
      &= \mu_\alpha \left ( \int_0^1 t^k \partial^\beta(f \circ
         \Psi^{-1})(c,t\mu)
         \, dt \otimes \Ins{e^\alpha}(e_\beta \wedge \eta) \right) \\
      &= \mu_\alpha \left ( \int_0^1 t^k \partial^\beta(f \circ
         \Psi^{-1})(c,t\mu)
         \, dt \otimes \delta^\alpha_\beta \eta \right) \\
      &\phantom{=} - (1 - \delta_{0k})\mu_\alpha \left ( \int_0^1
         t^k \partial^\beta(f \circ \Psi^{-1})(c,t\mu) \, dt \otimes
         (e_\beta \wedge \Ins{e^\alpha} \eta) \right) \\
      &= \int_0^1 t^k \frac{d}{dt}(f \circ \Psi^{-1}(c,t\mu)) \, dt \otimes \eta \\
      &\phantom{=} - (1 - \delta_{0k})\int_0^1
      t^{k-1}\left(\mathrm{pr}_\alpha \cdot \partial^\beta(f \circ
         \Psi^{-1}) \right)(c,t\mu) \, dt \otimes
      e_\beta \wedge \Ins{e^\alpha} \eta \\
      &= \int_0^1 \frac{d}{dt} \left(\left(t^k
            f(\Psi^{-1}(c,t\mu))\right)\right)\, dt \otimes \eta -
      \int_0^1 \frac{d t^k}{dt} (f \circ \Psi^{-1})(c,t\mu) \, dt
      \otimes \eta \\
      &\phantom{=} - (1 - \delta_{0k}) \int_0^1 t^{k-1} \partial^\beta
      (\mathrm{pr}_\alpha \cdot f \circ \Psi^{-1})(c,t\mu) \, dt \otimes
      e_\beta \wedge\Ins{e^\alpha}
      \eta \\
      &\phantom{=} + (1 - \delta_{0k}) \int_0^1
      t^{k-1}(\delta_\alpha^\beta f \circ
      \Psi^{-1})(c,t\mu) \, dt \otimes e_\beta \wedge \Ins{e^\alpha} \eta \\
      &= ((f\circ \Psi^{-1})(c,\mu) - \delta_{0k}(f \circ
      \Psi^{-1})(c,0))\otimes \eta - (1 - \delta_{0k}) k \int_0^1 t^{k-1}
      (f\circ\Psi^{-1})(c,t\mu) \, dt \otimes \eta
      \\
      &\phantom{=} - (1 - \delta_{0k})((\h_U\partial (f\otimes
      \eta))\circ \Psi^{-1})(c,\mu) + (1 - \delta_{0k}) k  \int_0^1 t^{k-1} (f \circ
      \Psi^{-1})(c,t\mu)
      \, dt \otimes \eta \\
      &= (f \otimes\eta)\circ\Psi^{-1}(c,\mu) - (f \otimes \eta) \circ
      \Psi^{-1}(c,0) \delta_{0k} - (1 - \delta_{0k})((\h_U\partial
      (f\otimes \eta))\circ \Psi^{-1})(c,\mu) \Fdot
   \end{align*}
   Dabei wurde im dritten Schritt ausgenutzt, dass $\Ins{e^\alpha}$ eine
   Super"=Derivation ist, im drauffolgenden Schritt kam die Kettenregel
   für Ableitungen zur Anwendung, beim nächsten Gleichheitszeichen wurde
   zweimal die Produktregel für Ableitungen verwendet und schließlich im
   vorletzten Schritt haben wir die Gleichung $e_\alpha \wedge
   \Ins{e^\alpha} \xi = k \xi$ ausgenutzt, die für alle $k \in
   \mathbb{N}$ und alle $\xi \in \Bigwedge^k \lieAlgebra$ gilt.

   Aus obiger Rechnung folgt dann unmittelbar die Behauptung.  Ist nun
   $H$ eine Lie"=Gruppe wie in den Voraussetzungen, so beachte man, dass
   für alle $(c,\mu) \in C \times \lieAlgebra$ und $h \in H$ die
   Gleichung
   \begin{align*}
      (h(\partial_{e^\alpha}(f \circ \Psi^{-1})))(c,\mu)
      &= \partial_{e^\alpha}(f \circ \Psi^{-1})(h^{-1} (c,\mu))\\ &=
      \lim_{s \to 0} \frac{f \circ \Psi^{-1}(h^{-1} c,h^{-1}\mu + s e^\alpha) - f
        \circ
        \Psi^{-1}(h^{-1}c,h^{-1}\mu)}{s} \\
      &= \lim_{s \to 0}\frac{(h (f \circ \Psi^{-1}))(c,\mu + s
        he^\alpha) - (h (f\circ\Psi^{-1}))(c,\mu)}{s} \\
      &= (\partial_{h e^\alpha}(h(f\circ \Psi^{-1})))(c,\mu)
   \end{align*}
   gilt. Ferner gilt $\dPaar{h e_\alpha}{h e^\beta} = \dPaar{h
     h^{-1} e_\alpha}{e^\beta} = \delta_\alpha^\beta$, d.\,h.\ $\{h
   e_\alpha\}$ und $\{h e^\alpha\}$ sind zueinander duale Basen für
   $\lieAlgebra$ und $\lieAlgebra^*$. Mit diesen Bemerkungen sind die
   behaupteten Invarianzeigenschaften klar.
\end{proof}

Das folgende technische Lemma wird sich als nützlich erweisen, die in
Lemma \ref{lem:HomotopienInDerTubenumgebung} auf
einer Tubenumgebung konstruierte Homotopie zu globalisieren.

\begin{lemma}
   \label{lem:GlobalisierungsAbbildung}
   Sei $W \subset M\setminus C$ offen. Dann gibt es eine glatte Abbildung
   $\xi \colon W \to \lieAlgebra$ mit
   \begin{align}
      \label{eq:GlobalisierungsAbbildung}
      \dPaar{J\at{W}}{\xi} = 1 \Fdot
   \end{align}
   Sei $H$ eine Lie"=Gruppe, die auf $M$ eigentlich wirke und sei $W$
   $H$"=invariant sowie $J$ $H$"=äquivariant. Weiter sei
   $\lieAlgebra$ mit einer linearen $H$"=Wirkung versehen. Dann
   kann $\xi$ auch $H$"=äquivariant gewählt werden.
\end{lemma}
\begin{proof}
   Wir betrachten direkt den Fall mit Gruppenwirkung. Der andere ergibt
   sich trivialerweise mit $H = \{e\}$.  Sei $x \in W$ beliebig. Als
   erstes konstruieren wir eine glatte $H$"=äquivariante Abbildung
   $\eta^{[x]} \colon H\cdot x \to \lieAlgebra$ mit der Eigenschaft
   $\dPaar{J\at{H\cdot x}}{\eta^{[x]}} = 1$ wie folgt. Da die
   $H$"=Wirkung eigentlich ist, ist die Stabilisatorgruppe $H_x$ von $x$
   kompakt, womit ein $H$"=invariantes Skalarprodukt $\langle \cdot ,
   \cdot \rangle_x$ mit induzierter Norm $|\cdot|_x$ auf $\lieAlgebra$
   existiert (vgl.\ Lemma \ref{kor:invarianteFasermetrik}). Ist nun $\nu
   \in \lieAlgebra^*$, so sei $\nu^{(x)}$ das eindeutig bestimmte
   Element von $\lieAlgebra$ mit $\langle \nu^{(x)}, \cdot \rangle_x =
   \nu$. Für $h \in H_x$ ist dann ${J(x)}^{(x)} = {J(hx)}^{(x)} =
   h{J(x)}^{(x)}$, denn wie man leicht nachvollzieht gilt für jedes
   $\rho \in \lieAlgebra$
   \begin{align*}
      \langle {J(h x)}^{(x)}, \rho \rangle_x = \dPaar{J(hx)}{\rho} &=
      \dPaar{h J(x)}{\rho} \\ &=
      \dPaar{J(x)}{h^{-1} \rho} = \langle {J(x)}^{(x)}, h^{-1} \rho
      \rangle_x = \langle h {J(x)}^{(x)},\rho \rangle_x \Fdot
   \end{align*}
   Sind $h,h' \in H$ mit $hx = h'x$ so ist offenbar $\tilde h := h^{-1}
   h' \in H_x$ und es folgt $\tilde h J(x)^{(x)} = J(x)^{(x)}$, d.\,h.\
   $h J^{(x)}(x) = h' J^{(x)}(x)$. Somit drückt sich die glatte
   Abbildung $H \to \lieAlgebra$, $h \mapsto h
   {J(x)}^{(x)}/|J(x)^{(x)}|_x$ zur glatten Abbildung $H/H_x \to
   \lieAlgebra$, $[h] \mapsto h {J(x)}^{(x)}/|{J(x)}^{(x)}|_x$
   herunter. Da bekanntlich (vgl.\ \cite[Prop. 3.3.24]{waldmann:2007a})
   $H/H_x \to H\cdot x$, $[h] \mapsto hx$ ein Diffeomorphismus ist,
   folgt durch Komposition, dass $\eta^{[x]} \colon H\cdot x \to
   \lieAlgebra$, $hx \mapsto h J^{(x)}/|{J(x)}^{(x)}|_x$ eine
   wohldefinierte glatte Abbildung ist. Per Konstruktion ist
   $\eta^{[x]}$ $H$"=äquivariant. Nun können wir nach Bemerkung
   \ref{bem:geometrischeProlongation} eine $H$"=invariante offene
   Umgebung $\tilde{U}^{[x]}$ von $H \cdot x$ in $M$ sowie eine
   $H$"=äquivariante geometrische Prolongation $\prol^{[x]} \colon
   C^\infty(H \cdot x) \to C^\infty(M)$ darauf wählen. Wir definieren
   dann $\xi^{[x]} := \prol^{[x]} \eta^{[x]}$. Aus Stetigkeits- und
   Invarianzgründen gibt es eine offene, $H$"=invariante Umgebung
   $U^{[x]} \subset \tilde{U}^{[x]}$ von $H\cdot x$ mit
   $\dPaar{J}{\xi^{[x]}}\at{U^{[x]}} > 0$. Damit ist ${\xi^{[x]}}'
   \colon U^{[x]} \to \lieAlgebra$, $p \mapsto
   \frac{\xi^{[x]}(p)}{\dPaar{J(p)}{\xi^{[x]}}}$ eine wohldefinierte
   glatte und $H$"=äquivariante Abbildung und es gilt
   $\dPaar{J\at{U^{[x]}}}{{\xi^{[x]}}'} = 1$. Wir wählen dann eine der
   offenen $H$"=invarianten Überdeckung $\{U^{[x]}\}_{x\in W}$
   untergeordnete, glatte, $H$"=invariante, lokal endliche Zerlegung der
   Eins $\{\chi_x\}_{x \in W}$ (vgl. Satz
   \ref{satz:inv_zerlegung_der_eins}) und setzen $\xi := \sum_{x\in
     W}\chi_x {\xi^{[x]}}'$. Dies erfüllt offensichtlich die gewünschte
   Eigenschaft $\dPaar{J\at{W}}{\xi} = 1$.

\end{proof}

\begin{korollar}
   \label{kor:XiMitPsi}
   Es seien dieselben Voraussetzungen wie in Lemma
   \ref{lem:GlobalisierungsAbbildung} gegeben und zusätzlich eine glatte
   Funktion $\psi_W \colon M \to \mathbb{R}$ mit $\supp \psi_W \subset W$. Dann
   gibt es eine glatte Abbildung $\xi \colon M \to \lieAlgebra$ mit
   $\supp \xi \subset W$ und $\dPaar{J}{\xi} = \psi_W$. Ist $H$ eine
   Lie"=Gruppe wie in Lemma~\ref{lem:GlobalisierungsAbbildung} und $W$
   sowie $\psi_W$ $H$"=invariant, so ist $\xi$ $H$"=äquivariant wählbar.
\end{korollar}
\begin{proof}
   Nach Lemma \ref{lem:GlobalisierungsAbbildung} gibt es ein glattes
   $\tilde \xi \colon W \to \lieAlgebra$ mit $\dPaar{J\at{W}}{\xi} =
   1$. Dann definiere man $\xi := \psi_W \tilde \xi \colon M \to
   \lieAlgebra$.
\end{proof}
\begin{proposition}
   \label{prop:HomotopieEigenschaftVonXiW}
   Sei $W \subset M\setminus C$ offen und $\psi_W \in C^\infty(M)$ mit $\supp
   \psi_W \subset W$ sowie $\xi \colon M \to \lieAlgebra$ eine glatte
   Abbildung mit $\dPaar{J}{ \xi} = \psi_W$ und $\supp \xi \subset
   W$. Ist $\{e_\alpha\}$ eine Basis von $\lieAlgebra$, so gibt es eindeutig
   bestimmte Funktionen $\xi^\alpha \in C^\infty(M)$ mit $\xi =
   \xi^\alpha e_\alpha$.  Sei $\h_W \colon C^\infty(M) \otimes
   \Bigwedge^{\bullet} \lieAlgebra \to C^\infty(M) \otimes
   \Bigwedge^{\bullet +1}\lieAlgebra$, für alle $k \in \mathbb{N}$ und
   $f\otimes\eta \in C^\infty(M)\otimes \Bigwedge^k \lieAlgebra$ durch
   $\h_W(f\otimes \eta) := \xi^\alpha f\otimes e_\alpha \wedge \eta$ und lineare
   Fortsetzung definiert. Diese Definition ist unabhängig von der Wahl
   der Basis $\{e_\alpha\}$ und es gilt
   \begin{align}
      \label{eq:HomotopieEigenschaftVonXiW1}
      \partial \h_W x + (1 - \delta_{0k}) \h_W \partial x = \psi_W x \quad \forall x \in
      C^\infty(M) \otimes \Bigwedge^k \lieAlgebra, k \in \mathbb{N} \Fdot
   \end{align}
   Es trage $M$ die Wirkung einer Lie"=Gruppe $H$, welche auch auf
   $\lieAlgebra$ linear operiere. Falls $\psi_W$ $H$"=invariant und $\xi$
   $H$"=äquivariant ist, so ist auch $\h_W$ $H$"=äquivariant.
\end{proposition}
\begin{proof}
Sei $k \in \mathbb{N}$ und $f \otimes \eta \in C^\infty(M) \otimes \Bigwedge^k \lieAlgebra$, dann gilt
   \begin{align*}
      \lefteqn{\partial \h_W (f \otimes \eta) + (1 - \delta_{0k})
        \h_W \partial
        (f \otimes \eta)}\\
      &= \partial(\xi^\alpha f \otimes e_\alpha \wedge \eta) + (1 -
      \delta_{0k})\h_W(f J_\alpha \otimes \Ins{e^\alpha} \eta) \\
      &= J_\beta \xi^\alpha f \otimes \Ins{e^\beta}(e_\alpha \wedge
      \eta) + (1 - \delta_{0k}) f J_\alpha \xi^\beta \otimes (e_\beta
      \wedge
      \Ins{e^\alpha}\eta) \\
      &= J_\alpha\xi^\alpha f\otimes \eta -
      (1 - \delta_{0k})J_\beta\xi^\alpha f\otimes e_\alpha \wedge\Ins{ e^\beta} \eta +
      (1 - \delta_{0k}) f J_\alpha \xi^\beta \otimes (e_\beta
      \wedge
      \Ins{e^\alpha}\eta) \\
      &= \dPaar{J}{\xi}(f \otimes \eta) \\
      &= \psi_W (f\otimes \eta) \Fdot
   \end{align*}
   Um die Äquivarianzeigenschaften zu zeigen sei $h \in H$ , $k \in
   \mathbb{N}$ und $f\otimes \eta \in C^\infty(M) \otimes \Bigwedge^k
   \lieAlgebra$, dann folgt:
   \begin{align*}
      \h_W(h(f \otimes \eta)) &= \h_W(hf \otimes h \eta) = \xi^\alpha hf
      \otimes e_\alpha \wedge h\eta \\
      &= h(h^{-1} \xi^\alpha f \otimes h^{-1} e_\alpha \wedge \eta) = h
      \h_W(f \otimes \eta) \Fdot
   \end{align*}
\end{proof}

Wir kommen nun zum angekündigten Existenzsatz einer globalen Homotopie.

\begin{satz}[Homotopie]
   \label{satz:globalisierteHomotopie}
   Sei $\Psi \colon U \to V \subset C \times \lieAlgebra^*$ eine
   \tn{gute} Tubenabbildung mit zugehöriger Retraktion $r \colon U \to
   C$ und $\h_U \colon C^\infty(U) \otimes \Bigwedge^\bullet \lieAlgebra
   \to C^\infty(U) \otimes \Bigwedge^{\bullet + 1} \lieAlgebra$ wie in
   Lemma \ref{lem:HomotopienInDerTubenumgebung}. Sei weiter $O$ eine
   offene Umgebung von $C$ mit $\abschluss{O} \subset U$, $W :=
   M\setminus \abschluss{O}$ sowie $\{\psi_U,\psi_W\}$ eine der offenen
   Überdeckung $\{U,W\}$ untergeordnete glatte Zerlegung der Eins. Dann
   gibt es eine glatte Abbildung $\xi \colon M \to \lieAlgebra$ mit
   $\supp \xi \in W$, $\dPaar{J}{\xi} = \psi_W$ und $\xi\at{\supp \psi_W
     \cap U} = -\h_U(\psi_U\at{U})\at{\supp \psi_W \cap U}$. Weiter sei
   $\h_W$ davon induziert wie in Proposition~\ref{prop:HomotopieEigenschaftVonXiW}. Schließlich sei $\prol :=
   \prol_{M,O,r,\psi_U}$.  Dann erfüllt die Abbildung $\h \colon
   C^\infty(M) \otimes \Bigwedge^\bullet \lieAlgebra \to C^\infty(M)
   \otimes \Bigwedge^{\bullet + 1} \lieAlgebra$, welche für $k \in
   \mathbb{N}$ und $f \otimes \eta \in C^\infty(M) \otimes \Bigwedge^k
   \lieAlgebra$ durch
   \begin{align}
      \h(f \otimes \eta):= \psi_U \h_U(f\at{U} \otimes \eta) + \h_W(f
      \otimes \eta)
   \end{align}
   und lineare Fortsetzung definiert ist,
   die folgenden Gleichungen.
\begin{align}
      \label{eq:globalisiertHomotopie1}
      \h\partial x + \partial \h x = x \quad \text{für $x \in
        C^\infty(M) \otimes \Bigwedge^k \lieAlgebra$ und $k \geq 1$}\Fcom
   \end{align}
   \begin{align}
      \label{eq:globalisiertHomotopie2}
      \prol \kIn^* f + \partial \h f = f \quad \text{für $f \in
        C^\infty(M)$}
   \end{align}
   und
   \begin{align}
      \label{eq:globalisierteHomotopie3}
      \h \prol  f = 0 \quad \text{für $f \in C^\infty(C)$}\Fdot
   \end{align}
   Trägt $M$ eine Wirkung einer Lie"=Gruppe $H$, unter der $C$ stabil ist, und
   $\lieAlgebra$ eine lineare $H$"=Wirkung, so kann man $\psi_U$,
   $\psi_W$ $H$"=invariant und $\xi$ $H$"=äquivariant wählen. Ebenso kann
   nach Lemma \ref{lem:HomotopienInDerTubenumgebung} $\h_U$
   $H$"=äquivariant gewählt werden und nach Proposition
   \ref{prop:HomotopieEigenschaftVonXiW} auch $\h_W$. Unter diesen
   Annahmen ist $\h$ ebenfalls $H$"=äquivariant.
\end{satz}
\begin{proof}

   Sei $\hat \xi \colon U \cap W \to \lieAlgebra$ definiert durch $\hat
   \xi = -\h_U(\psi\at{U})\at{U \cap W}$. Dann gilt offensichtlich für
   $p \in U \cap W$
   \begin{align*}
      \dPaar{J(p)}{\hat \xi(p)} = J_\alpha(p) \xi^\alpha(p) &= -
      (\partial \h_U(\psi_U\at{U}))(p) \\
      &= -(\psi_U(p) - (r^*\kRes \psi_U\at{U})(p)) = 1 - \psi_U(p) =
      \psi_W(p)
   \end{align*}
   Nun wählen wir auf $W$ ein $\xi' \colon M \to \lieAlgebra$ mit
   $\dPaar{J}{\xi'} = \psi_W$. Sei weiter $\{\chi_1,\chi_2\}$ eine
   Zerlegung der Eins von $W$ mit $\supp \chi_1 \subset U \cap W$ und
   $\supp \chi_2 \subset W \setminus \supp \psi_U$, dann erfüllt $\xi :=
   \chi_1 \hat \xi + \chi_2 \xi'$ die gewünschten Eigenschaften.

   Für die äquivariante Situation beachte man für $h \in H$ und $p \in U$
   \begin{align*}
     (\h_U(\psi_U))(hp) &= \int_0^1(\partial_{e^\alpha}(\psi_U \circ
     \Psi^{-1}))(r(hp)t J(hp)) \, dt \, e_\alpha \\
     &= \int_0^1(\partial_{e^\alpha}(\psi_U \circ
     \Psi^{-1}))(h(r(p),tJ(p))) \, dt \, e_\alpha \\
     &= \int_0^1(h^{-1}(\partial_{e^\alpha}(\psi_U \circ \Psi^{-1})))(r(p),tJ(p))
     \, dt  \, e_\alpha \\
     &= \int_0^1 \partial_{h^{-1} e^\alpha}(h^{-1} (\psi_U \circ
     \Psi^{-1}))(r(p),tJ(p)) \, dt  \, e_\alpha \\
     &= h\left(\int_0^1 \partial_{h^{-1} e^\alpha}(\psi_U \circ \Psi^{-1}) \,
     dt \, h^{-1} e_\alpha\right) \\
     &= h(\h_U(\psi_U)(p)),
   \end{align*}
   womit $\hat \xi$ $H$"=äquivariant ist. Dabei wurde beim zweitletzten
   Schritt die Invarianz von $\psi_U$ und die Äquivarianz von $\Psi$
   verwendet:
   \begin{align*}
      h^{-1} (\psi_U \circ \Psi^{-1})(p) = \psi_U \circ \Psi^{-1}(hp) =
      \psi_U (h \Psi^{-1}(p))  = \psi_U \circ \Psi^{-1}(p) \quad \forall
      p \in V
   \end{align*}
   Für $k >
   0$ und $x \in C^\infty(M) \otimes \Bigwedge^k \lieAlgebra$ gilt
   \begin{align*}
      \h \partial x + \partial \h x = \psi_U \h_U (\partial x)\at{U} +
      \h_W \partial x + \psi_U \partial \h_U x\at{U} + \partial \h_W x = \psi_U x + \psi_W
      x = x
   \end{align*}
   und für $f \in C^\infty(M)$
   \begin{align*}
      \prol \kIn^* f + \partial \h f = \psi_U r^* \kIn^* f +
      \psi_U \partial \h_U f\at{U} + \partial \h_W f = \psi_U f + \psi_W
      f = f \Fdot
   \end{align*}
   Für den letzten Teil der Aussage  beachte man,  dass  wegen
   $\dPaar{J\at{W}}{\xi} = \psi_W$ schon $\supp \psi_W \subset \supp
   \xi$ gilt, d.\,h.\ für $p \notin \supp \psi_U \cap \supp \xi$ gilt $p
   \notin \supp \psi_U \cap \supp \psi_W$, also offensichtlich
   $\psi_U(p)\h_U(\psi_U)(p) = 0$. Man rechnet dann nach, dass für $p \in W$
   \begin{align*}
      \h_W(\prol f)(p) &= \xi^\alpha(p) \prol f(p) e_\alpha  \\
      &=
      \begin{cases}
          \psi_U(p) f(r(p)) \xi(p) & \text{für } p \in
         \supp \psi_U \cap \supp \xi \subset U \\
         0 &\text{sonst}
      \end{cases}  \\
      &=
      \begin{cases}
         - \psi_U(p)f(r(p)) \h_U(\psi_U) &\text{für } p \in \supp \psi_U
         \cap \supp \xi \\
         0 &\text{sonst}
      \end{cases}\Fdot
   \end{align*}
   Andererseits gilt
   \begin{align*}
      (\h_U\prol f)(p) &= \int_0^1 \partial^\alpha((\psi_U r^*f)\circ
      \Psi^{-1})(r(p),tJ(p)) \, dt\, e_\alpha \\
      &= \int_0^1 \partial^\alpha((\psi_U \circ \Psi^{-1})f\circ r \circ
      \Psi^{-1})(r(p),tJ(p)) \, dt \, e_\alpha \\
      &= \int_0^1 f(r(p)) \partial^\alpha(\psi_U \circ
      \Psi^{-1})(r(p),tJ(p)) \, dt\, e_\alpha \\
      &= f(r(p)) \h_U(\psi_U)(p) \Fdot
   \end{align*}
   Damit ist  klar, dass
   \begin{align*}
      \h \prol f = 0
   \end{align*}
   gilt.
\end{proof}
Wir haben nun also eine explizite Homotopie für den Koszul"=Komplex
konstruiert, insbesondere wissen wir damit, dass $\kIdeal = \ker \kRes =
\im \kkoszul[1] = \langle J \rangle$ gilt. Da wir später die
geometrischen Eingangsdaten für diese Homotopie explizit verwenden
werden, erscheint die folgende Definition sinnvoll.
\begin{definition}[Geometrische Homotopie"=Daten]
   \label{def:GeometrischeHomotopieDaten}
   Das Tupel $\mathsf{GH} := (\Psi \colon U \to V, O \subset U,
   \psi_U,\psi_W,\xi)$ aus Satz \ref{satz:globalisierteHomotopie} nennen
   wir auch \neuerBegriff{geometrische Homotopie"=Daten} (für $C$). Ist $H$ eine
   Lie"=Gruppe wie in Satz \ref{satz:globalisierteHomotopie}, und sind
   $\Psi$ und $\xi$ $H$"=äquivariant sowie $O$, $U$, $\psi_U$, $\psi_W$
   $H$"=invariant gewählt, so nennen wir $\mathsf{GH}$ $H$"=invariante
   geometrische Homotopie"=Daten (für $C$).
\end{definition}
\begin{bemerkung}
   \label{bem:induzierteDingens}
   Wie wir in diesem Abschnitt gesehen haben, induzieren die
   geometrischen Homotopie"=Daten eine zusammenziehbare Kettenhomotopie
   $\h$ für den Koszul"=Komplex sowie eine geometrische
   Prolongationsabbildung $\prol$ sodass $\h \prol = 0$ gilt. Ist $H$
   eine Lie"=Gruppe wie in Satz \ref{satz:globalisierteHomotopie} und
   $J$ $H$"=äquivariant, so haben wir gesehen, dass es immer möglich
   ist, $H$"=invariante Homotopiedaten zu wählen. Insbesondere die davon
   induzierten Abbildungen $\h$ und $\prol$ $H$"=äquivariant.
\end{bemerkung}
\subsection{Chevalley-Eilenberg-Komplex}
\label{sec:ChevalleyEilenberg}

Zum Abschluss dieses Abschnitts möchten wir darlegen, dass man den oben
vorgestellten Koszul"=Komplex als Spezialfall des
Chevalley"=Eilenberg"=Komplexes auffassen kann. Diese Sichtweise wird
sich bei der im nächsten Kapitel vorgestellten Konstruktion eines
Sternprodukts auf dem reduzierten Phasenraum als hilfreich
erweisen. Zunächst wollen wir jedoch eine allgemeine Definition des
Chevalley"=Eilenberg"=Komplexes geben und zeigen, dass es sich dabei
tatsächlich um einen Komplex handelt. Die von uns verwendete Definition
findet man etwa in dem Buch von Loday \cite[Def. 10.1.3]{loday:1998a},
mit dem kleinen Unterschied, dass wir nicht wie Loday über einem Körper,
sondern etwas allgemeiner, über einem kommutativen Ring $R$ mit
$\mathbb{Q} \subset R$ arbeiten wollen. Bevor wir zur Definition kommen,
sei noch kurz an den Begriff der Lie"=Algebra"=Darstellung einer
Lie"=Algebra $\lieAlgebra$ über $R$ auf ein $R$"=Modul $V$
erinnert. Darunter verstehen wir eine $R$"=lineare Abbildung $\rho
\colon \lieAlgebra \to \mathrm{\operatorname{End}}(V)$, so dass für alle
$\xi,\eta \in \lieAlgebra$ die Gleichung $\rho([\xi,\eta]) = \rho(\xi)
\circ \rho(\eta) - \rho(\eta)\circ \rho(\xi)$ gilt.

Damit ist $\rho$ also nichts anderes, als ein
$R$"=Lie"=Algebrahomomorphismus von $(\lieAlgebra,[\cdot,\cdot])$ in
die $R$"=Lie"=Algebra der $R$"=linearen Endomorphismen
$\mathrm{\operatorname{End}}(V)$ auf $V$, welche mit dem Kommutator als
Lie"=Klammer versehen ist.

\begin{definition}[Chevalley"=Eilenberg"=Komplex]
   \label{def:Chevalley-Eilenberg-Komplex}
   Sei $\lieAlgebra$ eine Lie"=Algebra über $R$, $V$ ein Modul über $R$
   und $\rho \colon \lieAlgebra \to \mathrm{\operatorname{End}}(V)$ eine
   Lie"=Algebra"=Darstellung. Dann heißt die
   Sequenz von Abbildungen \def\tA[#1]{A_{#1}}
\begin{equation}
      \begin{tikzpicture}[baseline=(current
    bounding box.center),description/.style={fill=white,inner sep=2pt}]
         \matrix (m) [matrix of math nodes, row sep=3.0em, column
         sep=3.5em, text height=1.5ex, text depth=0.25ex]
         {
V & V \otimes \Bigwedge^1
\lieAlgebra  & V
\otimes \Bigwedge^2 \lieAlgebra & \dots \\
}; %

\path[<-] (m-1-1) edge node[auto]{$\chevalley$}(m-1-2); %
\path[<-] (m-1-2) edge node[auto]{$\chevalley$}(m-1-3); %
\path[<-] (m-1-3) edge node[auto]{$\chevalley$}(m-1-4); %
 \end{tikzpicture},
\end{equation}
bei der die Abbildung $\chevalley$ für alle $n \in \mathbb{N}$ und $v
\otimes \xi_1 \wedge \dots \wedge \xi_n \in V \otimes \Bigwedge^k \lieAlgebra$ durch
\begin{multline}
   \label{eq:DefinitionChevalleyEilenbergOperator}
   \chevalley(v \otimes \xi_1 \wedge \dots \wedge \xi_n) := \sum_{1 \leq
     j \leq n} (-1)^{j+1} \rho(\xi_j)v \otimes \xi_1 \wedge \dots \wedge
   \widehat{\xi_j} \wedge \dots \wedge \xi_n \\
   + \sum_{1 \leq i < j \leq n}(-1)^{i + j - 1} v \otimes [\xi_i,\xi_j]
   \wedge \xi_1 \wedge \dots \wedge \widehat{\xi_i} \wedge \dots \wedge
   \widehat{\xi_j} \wedge \dots \wedge \xi_n
\end{multline}
sowie durch lineare Fortsetzung gegeben ist,
\neuerBegriff{Chevalley"=Eilenberg"=Komplex} zur Lie"=Algebra"=Darstellung
$\rho$. Dabei bedeutet $\widehat{\xi_i}$ das Weglassen der Variablen
$\xi_i$. Der Randoperator $\chevalley$ heißt
$\neuerBegriff{Chevalley"=Eilenberg"=Operator}$ (zur Lie"=Algebra
Darstellung $\rho$).
\end{definition}
Um eventuellen Unklarheiten vorzubeugen, sei an dieser Stelle bemerkt,
dass es sich bei dem  im Papier von Bordemann, Herbig und Waldmann
\cite{bordemann.herbig.waldmann:2000a} betrachteten
Chevalley"=Eilenberg"=Komplex um den zu obigem Kettenkomplex dualen
Kokettenkomplex handelt.

Schreibt man für $\xi,\eta \in \lieAlgebra$ und $v \in V$ $\xi \wirk v
:= \rho(\xi)(v)$ und $\xi \wirk \eta := [\xi,\eta]$ sowie
$(\xi_0,\xi_1,\dots,\xi_n) := (\xi_0 \otimes \xi_1 \wedge \dots \wedge
\xi_n) \in V \otimes \Bigwedge^n \lieAlgebra$, so lässt sich
$\chevalley$ offensichtlich in der folgenden etwas homogeneren Form schreiben.
\begin{align}
   \label{eq:HomogeneFormVonChevalleyEilenberg}
   \chevalley(\xi_0,\dots,\xi_n) = \sum_{0 \leq i < j \leq n}
   (-1)^{j+1}(\xi_0,\xi_1,\dots,\xi_{i-1},\xi_j \wirk \xi_i,
   \xi_{i+1},\dots,\widehat{\xi_j},\dots,\xi_n) \Fdot
\end{align}
Für $n \in \mathbb{N}$ und $(\xi_0,\dots,\xi_n) \in V \otimes \Bigwedge^n
\lieAlgebra$ und $\xi \in \lieAlgebra$ setzen wir
\begin{align}
   \label{eq:WirkungAufDachprodukt}
   \xi \wirk (\xi_0,\dots,\xi_n) := \sum_{i=0}^n (\xi_0,\dots,\xi \wirk \xi_i,\dots,\xi_n)\Fdot
\end{align}
Dies definiert offensichtlich eine Lie"=Algebra-Wirkung $\wirk$ auf $V
\otimes \Bigwedge^n \lieAlgebra$, d.\,h.\ es gilt insbesondere für
$\xi,\xi' \in \lieAlgebra$ und $\gamma \in V \otimes \Bigwedge^n
\lieAlgebra$ die Gleichung
\begin{align}
   \label{eq:LieAlegrbaWirkungseigenschaft}
   (\xi \wirk \xi') \wirk \gamma = [\xi,\xi'] \wirk \gamma = \xi \wirk
   (\xi' \wirk \gamma) - \xi' \wirk (\xi \wirk \gamma) \Fdot
\end{align}

Mit dieser Notation können wir nun ohne große Mühe zeigen, dass es
sich beim Chevalley"=Eilenberg"=Komplex tatsächlich um einen Komplex
handelt, d.\,h.\ $\chevalley^2 = 0$ gilt.
\begin{proposition}
   \label{prop:ChevalleyQuatratGleichNull}
   \begin{propositionEnum}
      \item %
         \label{item:ChevalleyReinziehen}
         Für jedes $n\in \mathbb{N}$ und $\alpha = (\xi_0,\dots,\xi_n)
         \in V \otimes \Bigwedge^n \lieAlgebra$ und jedes $\xi \in
         \lieAlgebra$ gilt
         \begin{align}
            \label{eq:ChevalleyReinziehen}
            \chevalley (\xi \wirk \alpha) = \xi \wirk \chevalley \alpha \Fcom
         \end{align}
         d.\,h.\ $\chevalley$ ist $\lieAlgebra$"=äquivariant.
      \item %
         \label{item:ChevalleyQuatratGleichNull}
         Es gilt
         \begin{align}
            \label{eq:ChevalleyQuatratGleichNull}
            \chevalley^2 = 0 \Fdot
         \end{align}
   \end{propositionEnum}
\end{proposition}
\begin{proof}
   Wir folgen bei diesem Beweis eng der Argumentation von Loday,
   vgl.\ \cite[Lem. 10.6.3]{loday:1998a}.
   \begin{beweisEnum}
   \item %
      Wir zeigen die Aussage durch vollständige Induktion. Für $n = 0$
      ist sie trivialerweise erfüllt. Wir wollen annehmen, dass sie für
      ein festes $n-1 \geq 0$ gelte und zeigen, dass sie dann auch schon für
      $n$ gilt.  Dazu seien $\xi \in \lieAlgebra$ und $\alpha :=
      (\beta,y)$ mit $\beta := (\xi_0,\dots,\xi_{n-1}) \in V \otimes
      \Bigwedge^{n-1}\lieAlgebra$ und $y := \xi_n \in \lieAlgebra$
      gegeben.
      Nach Definition von $\chevalley$ und $\wirk$ gelten offensichtlich
      die folgenden Verträglichkeitsbedingungen.
      \begin{align}
         \label{eq:ChevalleyMitKlammer}
         \chevalley(\beta,y) = (\chevalley \beta,y) + (-1)^{n+1} y \wirk
         \beta \tag{$*$},
      \end{align}
      \begin{align}
         \label{eq:WirkungMitKlammer}
         \xi \wirk (\beta,y) = (\xi \wirk \beta,y) + (\beta,\xi \wirk y)\Fdot
         \tag{$**$}
      \end{align}
      Damit sehen wir
      \begin{align*}
         \chevalley(\xi \wirk \alpha)%
         &= \chevalley((\xi \wirk \beta,y) + (\beta,\xi \wirk y))
         &&\eAnn{nach \eqref{eq:WirkungMitKlammer}}\\
         &= (\chevalley(\xi \wirk \beta),y) + (-1)^{n+1} y\wirk (\xi
         \wirk
         \beta) &&\eAnn{nach \eqref{eq:ChevalleyMitKlammer}}\\
         &\phantom{=}+ (\chevalley \beta,\xi \wirk y) + (-1)^{n+1}(\xi
         \wirk y) \wirk \beta
         \\
         &= (\xi \wirk \chevalley \beta,y) + (\chevalley \beta, \xi
         \wirk y)  &&\eAnn{Induktionsvoraussetzung}\\
         &\phantom{=}+ (-1)^{n+1}((\xi \wirk y) \wirk \beta + y \wirk
         (\xi
         \wirk \beta))\\
         &= \xi \wirk (\chevalley \beta,y) + (-1)^{n+1}\xi \wirk (y
         \wirk
         \beta) &&\eAnn{nach \eqref{eq:WirkungMitKlammer}  und \eqref{eq:LieAlegrbaWirkungseigenschaft}} \\
         &= \xi \wirk (\chevalley(\beta,y)) &&\eAnn{nach \eqref{eq:ChevalleyMitKlammer}}\\
         &= \xi \wirk \chevalley \alpha \Fdot
      \end{align*}
   \item %
      Wir zeigen durch vollständige Induktion $\chevalley^2 \alpha = 0$
      für alle $\alpha \in V \otimes \Bigwedge^n \lieAlgebra$. Für $n =
      0$ ist die Aussage trivial. Wir nehmen nun an, sie gelte für $n-1$
      und zeigen, dass daraus folgt, dass sie auch für $n \in
      \mathbb{N}$ gilt. Dazu seien wieder $\xi \in \lieAlgebra$ und
      $\alpha := (\beta,y)$ mit $\beta := (\xi_0,\dots,\xi_{n-1}) \in V
      \otimes \Bigwedge^{n-1}\lieAlgebra$ und $y := \xi_n \in
      \lieAlgebra$ gegeben.
         \begin{align*}
            \chevalley^2 \alpha %
            &= \chevalley\chevalley(\beta,y) \\
            &= \chevalley((\chevalley \beta, y) + (-1)^{n+1} y \wirk \beta)
            &&\eAnn{nach \eqref{eq:ChevalleyMitKlammer}}\\
            &= (\chevalley^2 \beta,y) + (-1)^{n} y \wirk \chevalley
            \beta &&\eAnn{nach
              \eqref{eq:ChevalleyMitKlammer} und \eqref{eq:ChevalleyReinziehen}}\\
            &\phantom{=}+ (-1)^{n+1} y \wirk \chevalley \beta \\
            &= (\chevalley^2\beta,y) = 0
            &&\eAnn{Induktionsvoraussetzung} \Fdot
         \end{align*}
   \end{beweisEnum}
\end{proof}
Nun kommen wir noch zu einer lokalen Formel für den
Chevalley"=Eilenberg"=Operator. Diese wird sich später als hilfreich
erweisen, den Quanten"=Koszul"=Operator so zu schreiben, wie er in
der Literatur  auftaucht.
\begin{proposition}
   \label{prop:ChevalleyEilenbergDifferential}
   Die Voraussetzungen seien wie in Definition
   \ref{def:Chevalley-Eilenberg-Komplex}, die Lie"=Algebra
   $\lieAlgebra$ sei zusätzlich frei und endlich erzeugt mit einer Basis
   $\{e_\alpha\}$ und zugehöriger dualer Basis $\{e^\alpha\}$, dann gilt für alle
   $n \in \mathbb{N}$ und $v \otimes \xi  \in V \otimes \Bigwedge^n
   \lieAlgebra$ die folgende Gleichung.
   \begin{align}
      \label{eq:LokaleDarstellungDesChevalleyEilenbergDifferentials}
      \chevalley(v \otimes \xi) = \rho(e_\alpha)(v)\otimes \Ins{e^\alpha} \xi -
      \frac{1}{2} c_{\alpha \beta}^\gamma  v \otimes e_\gamma \wedge
      \Ins{e^\alpha} \Ins{e^\beta}\xi \Fdot
   \end{align}
   Dabei sind $\{c_{\alpha \beta}^\gamma\}$ die Strukturkonstanten der
   Lie"=Algebra $\lieAlgebra$ bezüglich der Basis $\{e_\alpha\}$, die für
   $\alpha,\beta,\gamma \in \{1,\dots,\dim \lieAlgebra\}$ durch $c_{\alpha
     \beta}^\gamma := \dPaar{e^\gamma}{[e_\alpha,e_\beta]}$ gegeben sind.
\end{proposition}
\begin{proof}
   Der erste Summand der rechten Seite von Gleichung
   \eqref{eq:LokaleDarstellungDesChevalleyEilenbergDifferentials} stimmt
   offensichtlich mit dem ersten Summenden der rechten Seite der
   $\chevalley$ definierenden Gleichung
   \eqref{eq:DefinitionChevalleyEilenbergOperator} überein. Sei nun $\xi
   = \xi_1 \wedge \dots \wedge \xi_n$. Dann ist der zweite Summand
   offensichtlich gleich
   \begin{align*}
      &\lefteqn{- \frac{1}{2} c_{\alpha \beta}^\gamma v \otimes e_\gamma \wedge
      \Ins{e^\alpha} \Ins{e^\beta}\xi} \\
      &=-\frac{1}{2}v \otimes
        [e_\alpha,e_\beta] \wedge \Ins{e^\alpha}(\sum_{j=1}^n(-1)^{j-1}
        \xi_1 \wedge \dots \wedge
        \Ins{e^\beta}\xi_j \wedge \dots \xi_n)\\
      &=-\frac{1}{2} v \otimes \sum_{j=1}^n(-1)^{j-1} [e_\alpha,\xi_j]
      \wedge (\sum_{i = 1}^{j-1} (-1)^{i-1} \xi_1 \wedge \dots \wedge
      \Ins{e^\alpha}\xi_i \wedge \dots \wedge
      \widehat{\xi_j} \wedge \dots  \wedge\xi_n \\
      &\phantom{v \otimes \sum_{j=1}^n(-1)^j [e_\alpha,\xi_j] \wedge (}
      + \sum_{i = j+1}^n (-1)^{i} \xi_1 \wedge \dots \wedge
      \widehat{\xi_j} \wedge
      \dots \wedge \Ins{e^\alpha} \xi_i \wedge \dots \wedge \xi_n) \\
      &=-\frac{1}{2} v \otimes \sum_{j=1}^n \sum_{i=1}^{j-1}
      (-1)^{j+i} [\xi_i,\xi_j] \wedge \xi_1 \wedge \dots \wedge
      \widehat{\xi_i} \wedge \dots
      \wedge \widehat{\xi_j} \wedge \dots \wedge \xi_n \\
      &\phantom{=}- \frac{1}{2} v\otimes\sum_{j=1}^n \sum_{i=j+1}^n
      (-1)^{j+i -1} [\xi_i,\xi_j] \wedge \xi_1 \wedge \dots \wedge
      \widehat{\xi_j} \wedge \dots \wedge \widehat{\xi_i} \wedge \dots
      \wedge \xi_n\\
      &= \sum_{1 \leq i < j \leq n}(-1)^{i + j - 1} v \otimes
      [\xi_i,\xi_j] \wedge \xi_1 \wedge \dots \wedge \widehat{\xi_i}
      \wedge \dots \wedge \widehat{\xi_j} \wedge \dots \wedge \xi_n
      \Fdot
   \end{align*}
\end{proof}
Nun kommen wir zur schon angekündigten Interpretation des
Koszul-Komplexes als Chevalley"=Eilenberg-Komplex.
\begin{bemerkung}
   \label{bem:KoszulAlsChevalleyEilenbergKlassisch}
   Versieht man $\lieAlgebra$ mit der trivialen Lie"=Klammer als
   Lie"=Algebrastruktur, so definiert $\rho_J \colon \lieAlgebra \to
   \operatorname{\mathrm{End}}(C^\infty(M))$, $\rho_J(\xi)(f) :=
   f\cdot J$ offensichtlich eine Lie"=Algebra-Darstellung. Dann ist klar,
   dass der Chevalley"=Eilenberg"=Komplex zu $\rho_J$ mit dem
   Koszul"=Komplex übereinstimmt.
\end{bemerkung}
\section{Quanten-Koszul-Reduktion}
\label{sec:QuantenKoszul}

In diesem Kapitel wollen wir die Grundlagen der
Quanten"=Koszul"=Reduktion darstellen. Wie bereits gezeigt, gilt in der
klassischen Situation $C^\infty(\Mred) \simeq \kbIdeal/\kIdeal$ als
Poisson"=Algebren. Die Poisson"=Klammer auf $\kbIdeal/\kIdeal$ wird
dabei auf naheliegende Art und Weise von der gegebenen auf $C^\infty(M)$
induziert. Die grobe Idee, ein Sternprodukt auf
$C^\infty(\Mred)[[\lambda]]$ zu konstruieren, besteht nun darin für eine
Basis $\{e_\alpha\}$ von $\lieAlgebra$, das $\star$"=Linksideal $\qIdeal
:= \langle \qJ \rangle_\star := \{f^\alpha \star \dPaar{\qJ}{e_\alpha}
\mid f^\alpha \in C^\infty(M)[[\lambda]]\}$ und den
$\star$"=Lie"=Idealisator $\qbIdeal := \{f \in C^\infty(M)[[\lambda]]
\mid [f,f']_\star \in \qIdeal \quad \forall f' \in \qIdeal\}$ von
$\qIdeal$ in $C^\infty(M)[[\lambda]]$ zu betrachten, via $\star$ auf
$\qbIdeal/\qIdeal$ ein assoziatives Produkt zu induzieren,

einen geeigneten Isomorphismus von $\qbIdeal/\qIdeal$ nach
$C^\infty(\Mred)[[\lambda]]$ zu finden und so vermöge diesem ein
Sternprodukt für $(\Mred,\omega_{\mathrm{red}})$ zu induzieren. Man
beachte, dass die Definition von $\qIdeal$ offensichtlich unabhängig von
der Wahl der Basis $\{e_\alpha\}$ ist. Durch Verwendung geometrischer
Konstruktionen im Hintergrund dieser algebraischen Betrachtungen nährt
sich dabei die Hoffnung, dass dieses auch differentiell wird. Wir zeigen
in diesem Kapitel, dass sich diese Heuristik in der Tat durchführen
lässt.

\subsection{Quanten-Koszul-Komplex}
\label{sec:QuantenKoszulKomplex}

\begingroup
\emergencystretch=0.8em
In Analogie zum Vorgehen in der klassischen Situation suchen wir eine
Abbildung $\qkoszul \colon C^\infty(M)[[\lambda]] \otimes
\Bigwedge^{\bullet} \lieAlgebra \to C^\infty(M)[[\lambda]] \otimes
\Bigwedge^{\bullet "=1} \lieAlgebra $, so dass $\qIdeal =
\qkoszul(C^\infty(M)[[\lambda]] \otimes \lieAlgebra)$ gilt. Für $f \otimes
\xi \in C^\infty(M)[[\lambda]] \otimes \lieAlgebra$ definieren wir
$\qkoszul$ auf die naheliegende Weise durch
\begin{align}
   \qkoszul(f \otimes \xi) := f\star \qJ(e_\alpha) \xi^\alpha =
   f \star \qJ({e_\alpha}) \Ins{e^\alpha}\xi \Fdot
\end{align}
Dabei ist hier und im Folgenden $\{e_\alpha\}$ wieder
eine Basis von $\lieAlgebra$ mit dualer Basis $\{e^\alpha\}$ und ein $\xi \in
\lieAlgebra$ schreiben  wir als $\xi = \xi^\alpha
e_\alpha$. Die gegebene Definition ist unabhängig von der
Wahl dieser Basis. Als nächstes möchten wir die Definition von
$\qkoszul$ auf ganz $C^\infty(M)[[\lambda]] \otimes \Bigwedge^\bullet \lieAlgebra$
ausdehnen und zwar  so, dass $\qkoszul^2 = 0$ gilt. Kennt man
den Chevalley"=Eilenberg"=Komplex aus Definition
\ref{def:Chevalley-Eilenberg-Komplex}, so bereitet dies keine
Schwierigkeiten. Man betrachte nämlich die Lie"=Algebra
$\lieAlgebra_{\mathbb{C}}[[\lambda]]$ über dem kommutativen Ring
$\mathbb{C}[[\lambda]]$ mit der Lie"=Algebra"=Struktur $-\I \lambda
[\cdot,\cdot]$. Dann definiert die Abbildung
\begin{align}
   \label{eq:QuantenDarstellungMitJ}
   \rho_{\qJ} \colon \lieAlgebra_{\mathbb{C}}[[\lambda]] \to
   \operatorname{\mathrm{End}}_{\mathbb{C}[[\lambda]]}(C^\infty(M)[[\lambda]]),
   \quad \rho_{\qJ}(\xi)(f) := f \star \qJ(\xi) \quad \forall \xi \in
   \lieAlgebra, f \in C^\infty(M)[[\lambda]]
\end{align}
eine Lie"=Algebra"=Darstellung, da $\qJ$ eine Quantenimpulsabbildung
ist.

Ausgeschrieben heißt dies  für $\xi,\eta \in \lieAlgebra$ und
$f \in C^\infty(M)[[\lambda]]$
\begin{align}
   \label{eq:QuantenDarstellungMitJ1}
   \rho_{\qJ}(\xi)(\rho_{\qJ}(\eta)(f)) -
   \rho_{\qJ}(\xi)(\rho_{\qJ}(\eta)(f))  &= f \star \qJ(\eta) \star
   \qJ(\xi) - f \star \qJ(\xi) \star \qJ(\eta)\notag\\
   &= f \star (-\I \lambda
   \qJ([\xi,\eta])) = \rho_{\qJ}(-\I \lambda [\xi,\eta])(f) \Fdot
\end{align}
\endgroup
Man beachte, dass hier die Assoziativität von $\star$ verwendet wurde.
Betrachtet man die Definition des Chevalley"=Eilenberg Differentials
$\chevalley$ zur Darstellung $\rho_{\qJ}$, so induziert die kanonische
Einbettung, $C^\infty(M)[[\lambda]] \otimes \Bigwedge^k
\lieAlgebra \hookrightarrow C^\infty(M)[[\lambda]] \otimes
\Bigwedge^k_{\mathbb{C}} \lieAlgebra_{\mathbb{C}}[[\lambda]]$, die, da $\lieAlgebra$
endlichdimensional ist, offensichtlich sogar ein Isomorphismus ist,
einen zum Chevalley"=Eilenberg"=Komplex isomorphen Komplex
\def\tA[#1]{A_{#1}}
\begin{equation}
      \begin{tikzpicture}[baseline=(current
    bounding box.center),description/.style={fill=white,inner sep=2pt}]
         \matrix (m) [matrix of math nodes, row sep=3.0em, column
         sep=3.5em, text height=1.5ex, text depth=0.25ex]
         {
C^\infty(M)[[\lambda]] & C^\infty(M)[[\lambda]] \otimes \Bigwedge^1
\lieAlgebra  & C^\infty(M)[[\lambda]]
\otimes \Bigwedge^2 \lieAlgebra & \dots \\
}; %

\path[<-] (m-1-1) edge node[auto]{$\qkoszul$}(m-1-2); %
\path[<-] (m-1-2) edge node[auto]{$\qkoszul$}(m-1-3); %
\path[<-] (m-1-3) edge node[auto]{$\qkoszul$}(m-1-4); %
 \end{tikzpicture},
\end{equation}
den wir \neuerBegriff{Quanten"=Koszul-Komplex} mit
\neuerBegriff{Quanten"=Koszul"=Operator} $\qkoszul$ nennen
wollen. Wie man leicht sieht, stimmt $\qkoszul$ in niedrigstem Grad mit dem
oben schon definierten $\qkoszul$ überein. Es sei bemerkt, dass die
konkrete Gestalt von $\qkoszul$ im Folgenden nur im niedrigsten Grad
relevant sein wird, ansonsten benötigen wir nur, dass $\qkoszul^2 = 0$
gilt. Der Vollständigkeit halber und zum besseren Vergleich mit der
Literatur (vgl.\ \cite[Def. 3.3]{gutt2010involutions},
\cite[Def. 16]{bordemann.herbig.waldmann:2000a})
wollen wir an dieser Stelle die explizite Formel des oben
konstruierten Quanten"=Koszul-Operators niederschreiben.
\begin{proposition}
   \label{prop:KonkreteGestaltDesQuantenKoszuloperators}
   Für den Quanten"=Koszul"=Operator $\qkoszul$ gilt für $f \in
   C^\infty(M)[[\lambda]]$ und $\xi \in \lieAlgebra$ die folgende Formel.
   \begin{align}
      \label{eq:KonkreteGestaltDesQuantenKoszuloperators}
      \qkoszul(f \otimes \xi) = f \star \qJ(e_\alpha) \otimes \Ins{e^\alpha} \xi +
      \I \lambda \frac{1}{2} c_{\alpha \beta}^\gamma  f \otimes e_\gamma \wedge
      \Ins{e^\alpha} \Ins{e^\beta}\xi \Fdot
  \end{align}
   Dabei sind $\{c_{\alpha \beta}^\gamma\}$ die Strukturkonstanten der
   Lie"=Algebra $\lieAlgebra$ bezüglich der Basis $\{e_\alpha\}$, die für
   $\alpha,\beta,\gamma \in \{1,\dots,\dim G\}$ durch $c_{\alpha
     \beta}^\gamma := \dPaar{e^\gamma}{[e_\alpha,e_\beta]}$ gegeben sind.
\end{proposition}

In \cite{gutt2010involutions} und \cite{bordemann.herbig.waldmann:2000a}
steht in der Definition von $\qkoszul$ noch ein zusätzlicher Summand der
Form $\I \lambda \kappa \Ins{\Delta}$, wobei $\Delta :=
\dPaar{e^\alpha}{[e_\alpha,e_\beta]}e^\beta \in \lieAlgebra^*$ die modulare
Einsform von $\lieAlgebra$ bezeichnet und $\kappa$ ein Element von
$\mathbb{C}[[\lambda]]$ ist. Da Sternproduktkommutatoren auf konstanten
Funktionen verschwinden und nach Definition von $\qkoszul$ ist klar,
dass man diesen Term durch eine Redefinition der gegebenen
Quantenimpulsabbildung absorbieren kann. Um zu sehen, dass dies auch an
möglichen Äquivarianzeigenschaften dieser nichts ändert, geben wir die
folgende Proposition an.

\begin{proposition}
   \label{prop:ModulareEinsformInvariant}
   Sei $H$ eine Lie"=Gruppe und $\rho \dpA H \times \lieAlgebra[g] \to
   \lieAlgebra[g]$ eine lineare $H$"=Wirkung auf $\lieAlgebra[g]$, welche
   die Lie"=Algebrastruktur respektiert, d.\,h.\ es gilt $\rho_h([\xi,\xi']) = [\rho_h
   (\xi), \rho_h (\xi')]$ für alle $h \in H$ und $\xi,\xi' \in
   \lieAlgebra[g]$. Dann ist die modulare Einsform $\Delta \in
   \lieAlgebra[g]^*$ von $H$"=invariant bezüglich der $\rho^*$"=Wirkung
   und $\Ins{\Delta}$ ist $H$"=äquivariant.
\end{proposition}
\begin{proof}
   Wir schreiben $h \xi := \rho(h,\xi)$ für $h \in H$ und $\xi \in
   \lieAlgebra[g]$. Dann gilt für alle $h \in H$
   \begin{align*}
      h \Delta \stackrel{\rho_h \text{linear}}{=}
      \dPaar{e^\alpha}{[e_\alpha,e_\beta]} h e^\beta = \dPaar{h
        e^\alpha}{h [e_\alpha,e_\beta]} h e^\beta = \dPaar{h
        e^\alpha}{[h e_\alpha,h e_\beta]} h e^\beta = \Delta
   \end{align*}
   und
   \begin{align*}
      h\Ins{\Delta}(f \otimes \xi) = h f \otimes h \Ins{\Delta} \xi  &= (h
      f) \otimes (\Ins{h \Delta})(h \xi) \\ &=(h f) \otimes (\Ins{\Delta} (h
      \xi)) = \Ins{\Delta}(h(f\otimes \xi)) \Fdot
   \end{align*}
\end{proof}

\begin{proposition}

   \label{prop:QuantenKoszulInvariant}
   Sei $H$ eine Lie"=Gruppe, die auf $M$ wirke. Weiter sei eine lineare
$H$"=Wirkung auf $\lieAlgebra[g]$ gegeben, welche die
   Lie"=Algebrastruktur respektiere. $\Bigwedge^\bullet \lieAlgebra[g]$ sei mit
   der davon induzierten Wirkung versehen. Ist $\qJ$ $H$"=äquivariant, so
   ist auch $\qkoszul$ $H$"=äquivariant.
\end{proposition}

\begin{proof}
   Sei $h \in H$, $k \in \mathbb{N}$, $\xi \in \Bigwedge^k
   \lieAlgebra[g]$ und $f \in \CM[M]$.  Weiter sei $\{e_\alpha\}$ eine Basis
   von $\lieAlgebra[g]$ und $\{e^\alpha\}$ die zugehörige duale Basis.  Dann
   gilt
   \begin{align*}
      h (\qkoszul (f \otimes \xi)) &= h ( \qJ(e_{\alpha}) \star f
      \otimes \Ins{e^{\alpha}} \xi) + h (c_{\alpha \beta}^\gamma f
      \otimes
      e_\gamma \wedge \Ins{e^\alpha}\Ins{e^\beta} \xi)\\
      &= h (\qJ(e_{\alpha}) \star f) \otimes h \Ins{e^{\alpha}}\xi +
      c_{\alpha \beta}^\gamma hf \otimes
      he_\gamma \wedge h(\Ins{e^\alpha}\Ins{e^\beta} \xi) \\
      &= h \qJ(e_{\alpha}) \star h f \otimes \Ins{h e^{\alpha}} (h \xi)
      + c_{\alpha \beta}^\gamma hf \otimes he_\gamma \wedge
      \Ins{he^\alpha}\Ins{he^\beta} h\xi
      \\
      &= \qJ(h e_{\alpha}) \star h f \otimes \Ins{h e^\alpha} (h \xi) +
      c_{\alpha \beta}^\gamma hf \otimes e_\gamma \wedge
      \Ins{e^\alpha}\Ins{e^\beta} h\xi  \\
      &= \qkoszul(h f \otimes h \xi) = \qkoszul (h (f \otimes \xi))
      \Fdot
   \end{align*}
   Im zweitletzten Schritt wurde dabei ausgenutzt, dass $\{h e_\alpha\}$ und
   $\{h e^\alpha\}$ zueinander duale Basen von $\lieAlgebra[g]$ und
   $\lieAlgebra[g]^*$ sind und dass die Strukturkonstanten bezüglich
   dieser die gleichen sind wie bezüglich $\{e_\alpha\}$, denn
   \begin{align*}
      \dPaar{h e^\gamma}{[he_\alpha,he_\beta]} = \dPaar{h
        e^\gamma}{h[e_\alpha,e_\beta]} =
      \dPaar{e^\gamma}{[e_\alpha,e_\beta]}\Fdot
   \end{align*}
\end{proof}

\subsection{Quanteneinschränkung}
\label{sec:Quanteneinschraenkung}

\begingroup
\emergencystretch=0.8em Als nächstes wollen wir eine geeignete
Augmentationsabbildung $\qRes \colon C^\infty(M)[[\lambda]] \to
C^\infty(C)[[\lambda]]$ finden, so dass $\qRes$ eine Deformation von
$\kRes$ ist und der augmentierte Koszul"=Komplex an der Stelle
$C^\infty(M)[[\lambda]] \otimes
\Bigwedge^0 \lieAlgebra$ exakt ist.
\endgroup

Als Ausgangspunkt wählen wir Gleichung \eqref{eq:globalisiertHomotopie2}
und erhalten für $f \in C^\infty(M)[[\lambda]]$

\begin{align}
   \label{eq:IdeeQuantenEinschr}
   \qkoszul[1] \h[0] f + \prol \kRes f = (\qkoszul[1] - \kkoszul[1])\h[0] f + \kkoszul[1] \h[0]
   f + \prol \kRes f = (\id + (\qkoszul[1] - \kkoszul[1])\h[0])f \Fdot
\end{align}
Der Operator $\qkoszul[1] - \kkoszul[1]$ ist mindestens von der Ordnung
$\lambda$, womit $\id + (\qkoszul[1] - \kkoszul[1])\h[0]$ invertierbar ist.

Somit ergibt sich sofort die Gleichung
\begin{align}
   \label{eq:IdeeQuantenEinschr2} \qkoszul[1] \h[0] \frac{\id}{\id +
(\qkoszul[1] - \kkoszul[1])\h[0]} + \prol \kRes \frac{\id}{\id +
(\qkoszul[1] - \kkoszul[1])\h[0]} = \id \Fdot
\end{align}
Dies motiviert die Definition der Augmentierungsabbildung und einer
Deformation von $\h[0]$ wir folgt.
\begin{align}
   \label{eq:QuantenEinschrDef}
   \qRes := \kRes \frac{\id}{\id + (\qkoszul[1] - \kkoszul[1])\h[0]}
\end{align}
und
\begin{align}
   \label{eq:defHomotopie}
   \qh[0] := \h[0] \frac{\id}{\id + (\qkoszul[1] - \kkoszul[1])\h[0]} \Fdot
\end{align}
Mit diesen Definitionen erhalten wir die folgende Proposition.
\begin{proposition}
   \label{prop:QuantenHomotopieUntereOrdnung}
   Es gilt
   \begin{align}
      \id = \qkoszul[1] \qh[0] + \prol \qRes\label{eq:QuantenHomotopieUntereOrdnung}
   \end{align}
   und
   \begin{align}
      \label{eq:QuantenHomotopieUntereOrdnung2}
      \qh[0] \prol = 0 \Fdot
   \end{align}
   Sei $H$ eine Lie"=Gruppe, die auf $M$ wirke und sei eine lineare
   $H$"=Wirkung auf $\lieAlgebra[g]$ gegeben, so dass $\h$, $\kkoszul$
   und $\qkoszul$ $H$"=äquivariant sind, dann sind auch $\qh[0]$ und
   $\qRes$ $H$"=äquivariant.
\end{proposition}
\begin{proof}
Klar.
\end{proof}

\begin{bemerkung}
   \label{bem:Quanteneinschraenkung}
   Aus Gleichung \eqref{eq:QuantenHomotopieUntereOrdnung} ergibt sich
   insbesondere die wichtige Charakterisierung
   \begin{align*}
      \qIdeal = \ker \qRes \Fdot
   \end{align*}
\end{bemerkung}

Aus naheliegenden Gründen wollen wir $\qRes$ im Folgenden auch als
\neuerBegriff{Quanteneinschränkung} bezeichnen.

\begin{proposition}
   \label{prop:QuantenAugmentierung}
   Die Abbildung $\qRes = \sum_{r=0}^\infty \lambda^r \qRes_r\colon
   C^\infty(M)[[\lambda]] \to C^\infty(C)[[\lambda]]$ erfüllt die
   folgenden drei Bedingungen.
   \begin{align}
      \qRes_0 &= \kRes \Fdot    \label{eq:QuantenAugmentierung1}\\
      \qRes \qkoszul[1] &= 0 \Fdot \label{eq:QuantenAugmentierung2}\\
      \qRes \prol &= \id       \label{eq:QuantenAugmentierung3} \Fdot
   \end{align}
   Dabei impliziert die dritte unmittelbar, dass $\prol \qRes$ eine
   Projektion ist mit $\im{(\prol \qRes)} =
   \prol(C^\infty(M)[[\lambda]])$. Weiter gilt $\ker{(\prol \qRes)} =
   \qIdeal$, womit auch die direkte Summenzerlegung
   \begin{align}
      \label{eq:DirekteSummenzerlegung}
      C^\infty(M)[[\lambda]] = \qIdeal \oplus \prol(C^\infty(M)[[\lambda]])
   \end{align}
   gilt.
\end{proposition}
\begin{proof}
   Gleichung \eqref{eq:QuantenAugmentierung3} folgt sofort aus
   Gleichung~\eqref{eq:globalisierteHomotopie3}, denn per Definition von
   $\qRes$ gilt
\begin{align*}
   \qRes \prol = \kRes \frac{\id}{\id + (\qkoszul[1]
     - \kkoszul[1])\h[0]}\prol = \kRes \prol = \id \Fdot
\end{align*}
Um Gleichung \eqref{eq:QuantenAugmentierung2} zu zeigen, beachten wir
zunächst, dass nach Gleichung \eqref{eq:QuantenHomotopieUntereOrdnung}
\begin{align*}
   \qRes = \qRes(\qkoszul[1] \qh[0] + \prol \qRes) = \qRes \qkoszul[1] \qh[0]
   + \qRes
\end{align*}
gilt, woraus
\begin{align*}
   \qRes \qkoszul[1] \qh[0] = 0
\end{align*}
folgt und damit direkt auch
\begin{align}
   \label{eq:qResqKoszulhNull}
   \qRes \qkoszul[1] \h[0] = 0 \Fdot \tag{$*$}
\end{align}
Wegen $\qkoszul[1] \qkoszul[2] = 0$ gilt andererseits
\begin{align*}
   \qkoszul[1] \h[0] \kkoszul[1] &= \qkoszul[1](\id - \kkoszul[2] \h[1])
   \\ &= \qkoszul[1](\id + \qkoszul[2] \h[1] - \kkoszul[2]\h[1]) =
   \qkoszul[1](\id + (\qkoszul[2] - \kkoszul[2])\h[1]) \Fdot
\end{align*}
\begingroup
\emergencystretch=0.8em
Dabei haben wir im ersten Schritt die Homotopieeigenschaft von $\h$
ausgenutzt. So erhalten wir mit \eqref{eq:qResqKoszulhNull} die Gleichung
\begin{align*}
   0 = \qRes \qkoszul[1] \h[0] \kkoszul[1] = \qRes \qkoszul[1](\id +
   (\qkoszul[2] - \kkoszul[2])\h[1]) \Fcom
\end{align*}
woraus sich, unter Verwendung der Invertierbarkeit von $\id + (\qkoszul[2]
- \kkoszul[2])\h[1]$, Gleichung~\eqref{eq:QuantenAugmentierung2} ergibt.

Aus Gleichung \eqref{eq:QuantenAugmentierung3} folgt
$(\prol \qRes)^2 = \prol \qRes$ und $\im(\prol \qRes) =
\prol(C^\infty(M)[[\lambda]])$, d.\,h.\ $\prol \qRes$ ist eine Projektion
auf $\prol(C^\infty(M)[[\lambda]])$. Wir berechnen nun den Kern von
$\prol \qRes$. Dazu stellen wir zuerst fest, dass $\im \qkoszul[1] =
\im(\qkoszul[1] \qh[0])$ gilt. Die Inklusion $\im(\qkoszul[1]) \supset
\im(\qkoszul[1] \qh[0])$ ist trivial. Die andere ergibt sich unmittelbar aus der
Beziehung
\begin{align*}
   \qkoszul[1] \qh[0] \qkoszul[1] = (\id - \prol \qRes)\qkoszul[1] =
   \qkoszul[1]\Fcom
\end{align*}
wobei hier beim ersten Schritt Gleichung
\eqref{eq:QuantenHomotopieUntereOrdnung} verwendet wurde und beim
zweiten \eqref{eq:QuantenAugmentierung2}. Damit und mit Gleichung
\eqref{eq:QuantenHomotopieUntereOrdnung} folgt dann $\ker (\prol \qRes)
= \im(\qkoszul[1] \qh[0]) = \im(\qkoszul[1]) = \qIdeal$. Da wie eben
gezeigt $\prol \qRes$ eine Projektion ist, folgt damit auch sofort die
Zerlegung $C^\infty(M)[[\lambda]] = \ker(\prol \qRes) \oplus \im(\prol
\qRes) = \qIdeal \oplus \prol(C^\infty(M)[[\lambda]])$.
\endgroup
\end{proof}
\begin{proposition}[Eindeutigkeit der Quanteneinschränkung]
   \label{prop:EindeutigkeitDerQuanteneinschr}
   Sei eine $\mathbb{C}[[\lambda]]$"=lineare
   Abbildung $\boldsymbol{\mathrm{res}} = \sum_{r=0}^\infty \lambda^r
   \boldsymbol{\mathrm{res}}_r \colon C^\infty(M)[[\lambda]] \to
   C^\infty(C)[[\lambda]]$ gegeben, welche die folgenden Eigenschaften
   erfüllt\Fdot
   \begin{align}
      \boldsymbol{\mathrm{res}}_0 &= \kRes \label{eq:EindeutigkeitDerQuanteneinschr1}\Fdot\\
      \boldsymbol{\mathrm{res}} \circ \qkoszul[1] &= 0 \label{eq:EindeutigkeitDerQuanteneinschr2} \Fdot\\
      \boldsymbol{\mathrm{res}} \circ \prol &= \id  \label{eq:EindeutigkeitDerQuanteneinschr3}\Fdot
   \end{align}
   Dann gilt schon $\boldsymbol{\mathrm{res}} = \qRes$.
\end{proposition}
\begin{proof}
   Mit Gleichung \eqref{eq:EindeutigkeitDerQuanteneinschr3} ist klar,
   dass $\boldsymbol{p} := \prol \circ \boldsymbol{\mathrm{res}}$ eine
   Projektion ist, d.\,h.\ $\boldsymbol{p}^2 = \boldsymbol{p}$ gilt und
   $\im (\boldsymbol{p}) = \prol{(C^\infty(M))[[\lambda]]}$ erfüllt
   ist. Da per Definition von $\qIdeal$ die Gleichung $\qIdeal = \im
   \qkoszul[1]$ gilt, folgt mit
   \eqref{eq:EindeutigkeitDerQuanteneinschr2} dass $\qIdeal$ im Kern von
   $\boldsymbol{p}$ liegt. Da nun $\boldsymbol{p}$ eine Projektion ist,
   gilt einerseits $C^\infty(M)[[\lambda]] = \ker(\boldsymbol{p}) \oplus
   \im(\boldsymbol{p})$. Andererseits gilt nach Proposition
   \ref{prop:QuantenAugmentierung} die Zerlegung $C^\infty(M)[[\lambda]]
   = \qIdeal \oplus \prol(C^\infty(M)[[\lambda]])$ und wie oben gezeigt,
   $\im (\boldsymbol{p}) = \prol{(C^\infty(M))[[\lambda]]}$ sowie
   $\qIdeal \subset \ker(\boldsymbol{p})$. Somit ist klar, dass auch schon
   $\qIdeal = \ker(\boldsymbol{p})$ gilt.  Da eine Projektion durch
   ihren Kern und ihr Bild bekanntlich eindeutig festgelegt ist,
   erhalten wir zusammen mit Proposition \ref{prop:QuantenAugmentierung}
   die Beziehung $\prol \circ \boldsymbol{\mathrm{res}} = \boldsymbol{p}
   = \prol \qRes$. Durch Anwenden von $\qRes$ von links auf diese
   Gleichung folgt dann mit Gleichung \eqref{eq:QuantenAugmentierung3}
   $\boldsymbol{\mathrm{res}} = \qRes$.
\end{proof}
\begin{bemerkung}
   \label{bem:InduzierteQuantenDings}
   Aus der Definition von $\qRes$ ist ersichtlich, dass gewählte
   geometrische Homotopie"=Daten, zusammen mit $J$ und $\qJ$ immer eine
   Quanteneinschränkung $\qRes$ induzieren.  Trägt $M$ eine Wirkung einer
   Lie"=Gruppe $H$, unter der $C$ stabil ist, und $\lieAlgebra$ eine
   lineare $H$"=Wirkung, und sind $J$ und $\qJ$ $H$"=äquivariant so gibt
   es immer $H$"=invariante geometrische Homotopiedaten, siehe
   \ref{bem:induzierteDingens} und die davon induzierte
   Quanteneinschränkung ist $H$"=äquivariant. Nach Proposition
   \ref{prop:EindeutigkeitDerQuanteneinschr} können wir $\qRes$ auch von
   $J$, $\qJ$ und $\prol$ induziert ansehen, wobei $\prol$ seinerseits
   wieder von geometrischen Homotopiedaten induziert wird, wobei man zu
   dessen Konstruktion das Datum $\xi$ aus Definition
   \ref{def:GeometrischeHomotopieDaten} welches zur Globalisierung von
   der Homotopie benötigt wurde, nicht braucht.
\end{bemerkung}
\begin{proposition}[Lokalisierbarkeit von $\qRes$]
   \label{prop:LokalisierbarkeitDerQuanteneinschraenkung}
   \begin{propositionEnum}
   \item %
      \label{item:lokaleGestaltDerQuanteneinschraenkung}
      Sei $c \in C$, dann gibt es eine offene Umgebung $U' \subset U$
      von $c$, so dass $\Psi(U') \subset C \times \lieAlgebra^*$
      sternförmig bezüglich des Nullpunktes in jeder Faser ist und
      so dass für dieses und jedes kleinere $U'$ mit diesen Eigenschaften
      die folgende Aussage richtig ist. Für die lineare Abbildung
      $\qRes_{U'} \colon C^\infty(U')[[\lambda]]\to
      C^\infty(C)[[\lambda]]$,
      \begin{align}
         \label{eq:lokalisierteQuanteneinschraenkung}
         \qRes_{U'} := \kIn^*_{U'} \circ (\id + (\qkoszul_{U'}
         - \partial_{U'}) \h_{U',U})^{-1}
      \end{align}
      gilt für alle $f \in C^\infty(M)$ die Gleichung
      \begin{align}
         \label{eq:EinschraenkbarkeitDerQuanteneinschr}
         \qRes_{U'} f\at{U'} = \qRes f \Fdot
      \end{align}

      Dabei ist $\h_{U',U} \colon C^\infty(U') \otimes
      \Bigwedge^{\bullet}\lieAlgebra \to C^\infty(U') \otimes
      \Bigwedge^{\bullet + 1} \lieAlgebra$ für $k \in \mathbb{N}$,
      $f \otimes \xi \in C^\infty(U') \otimes \Bigwedge^k \lieAlgebra$
      und $p\in U'$ durch
      \begin{align}
         \label{eq:lokalisierteHomotopie}
         \h_{U',U}(f \otimes \xi)(p) = \int_0^1 t^k \partial^\alpha(f
         \circ \Psi\at{U'}^{-1})(r(p),t J(p)) \, dt \otimes e_\alpha \wedge \xi
      \end{align}
      gegeben und $\kIn_{U'}\colon C \cap U' \hookrightarrow U'$
      bezeichne die Inklusion.  $U'$ kann sogar als Umgebung von $C$
      gewählt werden.  Ist $H$ eine Lie"=Gruppe, die auf $M$ wirke,
      trage $\lieAlgebra$ eine lineare $H$"=Wirkung und seien alle
      vorkommenden Größen invariant bzw.\ äquivariant gewählt, so kann
      auch $U'$ $H$"=äquivariant gewählt werden.
   \item %
      \label{item:WennInUmgebungNullDannNull}
      Sei $U' \subset U$ eine offene Umgebung von $C$, so dass $\Psi(U')
      \subset C \times \lieAlgebra^*$ sternförmig bezüglich des
      Nullpunktes in jeder Faser ist.  Ist $f \in C^\infty(M)$ mit
      $f\at{U'} = 0$, dann gilt schon $\qRes f = 0$.
   \end{propositionEnum}

\end{proposition}
\begin{proof}
   \begin{beweisEnum}
   \item %
      Es ist klar, dass es eine offene Umgebung $U'$ von $C$ gibt,
      so dass $\Psi(U')$ sternförmig bezüglich des Nullpunktes in jeder
      Faser ist und sowohl $\psi_U\at{U'} = 1$ als auch $U' \cap W =
      \emptyset$ wahr sind. Dann gilt $\h(f)\at{U'} = \h_{U',U}f\at{U'}$
      und somit wegen der Lokalität von $\qkoszul - \partial$ auch
      $((\qkoszul - \partial)\h(f))\at{U'} = (\qkoszul_{U'}
      - \partial_{U'})\h_{U',U}(f\at{U'})$. Induktiv erhält man dann
      \begin{align*}
         (((\qkoszul - \partial)\h)^k f)\at{U'} = ((\qkoszul_{U'}
         - \partial_{U'})\h_{U',U})^k f\at{U'}
      \end{align*}
      für alle $k \in \mathbb{N}$. Der Rest ist klar.
   \item %
      Folgt sofort aus Teil~\refitem{item:lokaleGestaltDerQuanteneinschraenkung}.
   \end{beweisEnum}
\end{proof}

\subsection{Sternprodukt für den reduzierten Phasenraum}
\label{sec:SternproduktFuerDenReduziertenPhasenraum}

Wir sind nun in der Lage ein Sternprodukt für $\Mred$ zu
konstruieren. Zunächst geben wir jedoch eine andere Charakterisierung
für den $\star$"=Lie"=Idealisator $\qbIdeal$ an. Im Folgenden seien
$(\Psi \colon U \to V, O \subset U, \psi_U,\psi_W,\xi)$ $G$"=invariante
geometrische"=Homotopiedaten. Insbesondere sind die davon induzierte
Prolongation $\prol$ und die Quanteneinschränkung $\qRes$
$G$"=äquivariant. Wie im vorherigen Abschnitt gezeigt, können wir stets
solche wählen, da $J$ und $\qJ$ $G$"=äquivariant sind.
\begin{lemma}
   \label{lem:CharakterisierungQuantenLieIdealisator}
   Es gilt
   \begin{align}
      \label{eq:CharakterisierungQuantenLieIdealisator}
      \qbIdeal = \{f \in C^\infty(M)[[\lambda]] \mid \qRes f \in
      \pi^*C^\infty(\Mred)[[\lambda]] \} \Fdot
   \end{align}
   Insbesondere ist $\prol \pi^*C^\infty(\Mred)[[\lambda]] \subset \qbIdeal$.
\end{lemma}
\begin{proof}
   Nach Definition von $\qbIdeal$ ist genau dann $f \in \qbIdeal$, wenn
   $[f,f']_{\star} \in \qIdeal$ für alle $f' \in \qIdeal$ gilt. Dies ist
   äquivalent zu der Bedingung $f' \star f \in \qIdeal$ für alle $f' \in
   \qIdeal$, weil $\qIdeal$ ein $\star$"=Linksideal ist. Da weiter
   $\qIdeal$ von den Komponenten von $\qJ$ erzeugt wird, ist $f$ genau
   dann ein Element von $\qbIdeal$, wenn für alle $\xi \in \lieAlgebra$
   die Bedingung $[f,\qJ(\xi)]_\star \in \qIdeal$ gilt. Wegen $\qIdeal =
   \ker \qRes$ (vgl.\ Bem. \ref{bem:Quanteneinschraenkung}), ist dies
   weiter äquivalent zu $\qRes([f,\qJ(\xi)]_{\star}) = 0$ für alle $\xi \in
   \lieAlgebra$. Nun gilt aber für jedes $\xi \in \lieAlgebra$ die
   Gleichung
   \begin{align*}
      \qRes([f,\qJ(\xi)]_\star) = \qRes(\I \lambda \{f,J(\xi)\}) = \I
      \lambda \qRes(\xi_{M}f) = \I\lambda \xi_M \qRes f \Fdot
   \end{align*}

   Dabei wurde im ersten Schritt verwendet, dass $\qJ$ eine
   Quantenimpulsabbildung ist und im dritten dass $\qRes$ $G$"=äquivariant
   ist.

   Da nach Proposition \ref{prop:InvarianteFunktionen}
   $\pi^*C^\infty(\Mred) = C^\infty(C)^G$ gilt, und da $G$
   zusammenhängend ist, folgt die Behauptung.
\end{proof}
Wir kommen nun zum zentralen Satz
(vgl.\ \cite[Prop. 3.12]{gutt2010involutions}) dieses Kapitels. In diesem
wird das Sternprodukt auf dem reduzierten Phasenraum angegeben.
\begin{satz}[Sternprodukt auf dem reduzierten Phasenraum]
   \label{satz:SternproduktAufDemReduziertenPhasenraum}
   \begin{satzEnum}
   \item %
      $\qbIdeal$ ist eine Unteralgebra von
      $(C^\infty(M)[[\lambda]],\star)$ und $\qIdeal$ ist ein Ideal in
      $\qbIdeal$, wodurch $\qbIdeal/\qIdeal$ via $[f] \bullet [f'] := [f
      \star f']$ für $f,f' \in \qbIdeal$ zu einer assoziativen Algebra
      wird.
   \item %
      Die Abbildung
      \begin{align}
         \label{eq:SternproduktAufDemReduziertenPhasenraum}
         \boldsymbol{\mathrm{iso}} \colon \qbIdeal/\qIdeal \ni [f]
         \mapsto \qRes f \in \pi^*C^\infty(\Mred)[[\lambda]] =
         C^\infty(C)^G[[\lambda]]
      \end{align}
      ist ein Isomorphismus und es gilt für jedes $f \in
      C^\infty(C)^G[[\lambda]]$
      \begin{align}
         \label{eq:SternproduktAufDemReduziertenPhasenraum2}
         \boldsymbol{\mathrm{iso}}^{-1}(f) = [\prol(f)] \in
         \qbIdeal/\qIdeal \Fdot
      \end{align}
      Weiter induziert der Isomorphismus $\boldsymbol{\mathrm{iso}}^{-1}
      \circ \pi^* \colon C^\infty(\Mred)[[\lambda]] \to
      \qbIdeal/\qIdeal$ ein Sternprodukt $\starred$ auf $C^\infty(\Mred)[[\lambda]]$, welches insbesondere für alle
      $\phi,\phi' \in C^\infty(\Mred)[[\lambda]]$ durch die Gleichung
      \begin{align}
         \label{eq:SternproduktAufDemReduziertenPhasenraum3}
         \pi^*(\phi \star_{\mathrm{red}} \phi') = \qRes(\prol(\pi^*
         \phi) \star \prol(\pi^* \phi'))
      \end{align}
      bestimmt ist.
      \item %
        Ist $\star$ differenziell, so auch $\starred$.
   \end{satzEnum}
\end{satz}
\begin{proof}
   \begin{beweisEnum}
      \item %
         Nach Lemma \ref{lem:CharakterisierungQuantenLieIdealisator} ist
          $\qIdeal \subset \qbIdeal$ klar, die restlichen Aussagen
         sind dann unmittelbar einsichtig.
      \item %
         Nach Lemma \ref{lem:CharakterisierungQuantenLieIdealisator} ist
         $\qRes f \in \pi^*C^\infty(\Mred)[[\lambda]]$ für alle $f \in
         \qbIdeal$. Für jedes $f\in C^\infty(C)^G[[\lambda]]$ ist die Funktion
         $\qRes\prol(f) = f$ $G$"=invariant, also ein Element von
         $\pi^*C^\infty(\Mred)[[\lambda]] = C^\infty(C)^G[[\lambda]]$. Damit liegt,
         wieder nach Lemma
         \ref{lem:CharakterisierungQuantenLieIdealisator}, die Funktion
         $\prol(f)$ in $\qbIdeal$. Somit sind die Abbildungen
         $\boldsymbol{\mathrm{iso}}$ und $C^\infty(C)^G[[\lambda]] \ni f
         \mapsto [\prol f] \in \qbIdeal/\qIdeal$ wohldefiniert. Wir
         rechnen als nächstes nach, dass sie zueinander invers sind. Sei
         dazu zuerst $f \in \qbIdeal$ gegeben. Dann gilt
         \begin{align*}
            \qRes (\prol \qRes f - f) = 0 \Fcom
         \end{align*}
         also
         \begin{align*}
            [\prol \qRes f] = [f] \Fdot
         \end{align*}
         Sei umgekehrt $f \in \pi^*C^\infty(\Mred)[[\lambda]]$, so folgt
         \begin{align*}
            \boldsymbol{\mathrm{iso}}([\prol f]) = \qRes \prol f = f \Fdot
         \end{align*}
         Formel \eqref{eq:SternproduktAufDemReduziertenPhasenraum3} ist
         unmittelbar einsichtig, womit auch klar ist, dass  die Gleichung
         $[\phi,\phi']_{\star_{\mathrm{red}}} = \I \lambda
         \{\phi,\phi'\}_{\mathrm{red}}$ gilt.
   \item %
      Dies folgt direkt aus dem nachfolgenden Lemma \ref{lem:S} und der
      Tatsache, dass die Komposition von Multi"=Differentialoperatoren wieder
      ein Multi"=Differentialoperator ist.
   \end{beweisEnum}
\end{proof}

\begin{lemma}
   \label{lem:S}
   Es gibt eine formale Reihe $S := \id + \sum_{r=1}^\infty \lambda^r
   S_r$ von Differentialoperatoren $S_r \colon C^\infty(M) \to
   C^\infty(M)$, so dass
   \begin{align*}
      \qRes = \kRes \circ S
   \end{align*}
   gilt. Die Differentialoperatoren $S_r$ können dabei so gewählt
   werden, dass $S_r$ für $r \geq 1$ auf Konstanten verschwindet.

\end{lemma}
\begin{proof}

   Sei $U'$ eine offene Umgebung von $C$, $\kRes_{U'}$ und $\qRes_{U'}$
   wie in Proposition
   \ref{prop:LokalisierbarkeitDerQuanteneinschraenkung}.  Es genügt zu
   zeigen, dass es eine formale Reihe $S_{U'} = \id_{U'} +
   \sum_{r=1}^\infty\lambda^r{S_{U'}}_r$ von Differentialoperatoren ${S_{U'}}_r
   \colon C^\infty(U') \to C^\infty(U') $ gibt, die auf Konstanten
   verschwinden, so dass $\qRes_{U'}= \kRes_{U'} \circ S_{U'}$ gilt. Denn
   hat man ein derartiges $S_{U'}$ gefunden, definiere man für jedes $r
   \geq 1$ den Differentialoperator $S_r \colon C^\infty(M) \to
   C^\infty(M)$ durch $S_r := \chi_{U'}\cdot {S_{U'}}_r \circ \kIn_{U',M}^*$,
   wobei $\chi_{U'} \colon M \to \mathbb{R}$ eine glatte Funktion sei
   mit $\supp \chi_{U'} \subset {U'}$ und $\chi_{U'}\at{C} = 1$ und
   $\kIn_{U',M} \colon {U'} \hookrightarrow M$ die Inklusion
   bezeichne (vgl.\ Prop. \ref{prop:VortsetzungVonDiffops}).

   $S := \id + \sum_{r=1}^\infty S_r$ hat dann  die
   gewünschten Eigenschaften, denn für $f \in C^\infty(M)$  gilt
   \begin{align*}
      (\kRes \circ S_r)f  = \kRes (\chi_{U'} {S_{U'}}_r \circ \kRes_{U',M})(f) =
      (\kRes \chi_{U'}) \kRes_{U'} {S_{U'}}_r \kRes_{U',M} f = (\kRes_{U'} \circ
      {S_{U'}}_r) f\at{U'} \quad \forall r \geq 1
   \end{align*}
   und somit
   \begin{align*}
      (\kRes \circ S)f = \qRes_{U'} f\at{U'} = \qRes f \Fdot
   \end{align*}
   Um die lokale Situation beherrschen zu können, ist die folgende
   Beobachtung hilfreich. Sei $D \colon C^\infty(U') \otimes \lieAlgebra
   \to C^\infty(U')$ eine lineare Abbildung. Dann gibt es bekanntlich
   lineare Abbildungen $D_i \colon C^\infty(U') \to C^\infty(U')$, so dass
   $D(f \otimes \xi) = \Ins{e^i}(\xi) D_i(f)$ gilt für alle $f \otimes
   \xi \in C^\infty(U') \otimes \lieAlgebra$. Falls jedes $D_i$ ein
   Differentialoperator ist, wollen wir der Einfachheit halber auch $D$
   als Differentialoperator bezeichnen. Ist $D_i$ für alle $i$ von der
   Ordnung $k \in \mathbb{N}$, so wollen wir sagen, $D$ sei von der
   Ordnung $k$. Ist nun ein derartiges $D$ der Ordnung $k$
   gegeben, so gibt es einen Differentialoperator  $D' \colon C^\infty(U')
   \to C^\infty(U')$ von der Ordnung $k+1$, der
   \begin{align*}
      \kRes_{U'} \circ D \circ \h_{U',U} = \kRes_{U'} \circ D'
   \end{align*}
   erfüllt.
   Dies ist in der Tat der Fall, denn sei $(U_C,x)$ eine Karte für $C$,
   so ist $(U'' := \Psi^{-1}(U_C \times \lieAlgebra^*),(\overline x = x \circ
   r,J_1,\dots,J_{\dim G}))$ eine Untermannigfaltigkeitskarte von $C$ in
   $M$ und $D$ lässt sich in dieser Karte in der Form
   \begin{align*}
      D (f\at{U''} \otimes \xi) = \sum_{\nu = 0}^k \sum_{r+s = \nu}\Ins{e^j}(\xi)
      {D_{U''}^\nu}_{j j_1\dots j_{s}}^{i_1 \dots i_r}
      \frac{\partial^{r}}{\partial \overline{x}^{i_1} \dotsm \partial
\overline{x}^{i_r}} \frac{\partial^s }{\partial J_{j_1} \dotsm J_{j_s}}f
   \end{align*}
   mit glatten Funktionen ${D_{U''}^\nu}_{j j_1\dots j_s}^{i_1 \dots
     j_r} \in C^\infty(U'')$ angeben. Sei nun $f \in C^\infty(U')$ und $c \in
   C$, dann erhalten wir
   \begin{align*}
     \lefteqn{ (D \h_{U',U}(f))(c)}\\
     &= \sum_{\nu = 0}^k \sum_{r +s = \nu} {D_{U''}^\nu}_{j j_1\dots
        j_s}^{i_1 \dots j_r} \int_0^1 t^s \delta_{m,j}
      \frac{\partial^{r}}{\partial x^{i_1} \dotsm \partial x^{i_r}}
      \frac{\partial^s }{\partial \mu_{j_1} \dotsm \partial
        \mu_{j_s} \partial \mu_m} (f \circ
      \Psi^{-1})(r(c),t \cdot 0) \, dt \\
      &= \sum_{\nu = 0}^k\sum_{r +s = \nu} {D_{U''}^\nu}_{j j_1\dots j_s}^{i_1 \dots j_r}
      \frac{1}{s+1} \frac{\partial^{r}}{\partial \overline{x}^{i_1}
        \dotsm \partial \overline{x}^{i_r}} \frac{\partial^s }{\partial
        J_{j_1} \dotsm \partial J_{j_s} \partial J_j} f (c) \\
      &=: (D'f)(c) \Fdot
   \end{align*}
   Wobei wir $D'$ lokal durch das letzte Gleichheitszeichen definieren,
   wie man leicht sieht, ist $D'$ tatsächlich ein Differentialoperator
   (vgl.\ \cite[Prop. A.3.6]{waldmann:2007a}).  Durch ordnungsweises
   Anwenden dieser Beobachtung, sieht man, dass  für jede formale
   Potenzreihe $D \colon C^\infty(U')[[\lambda]] \otimes \lieAlgebra \to
   C^\infty(U')[[\lambda]]$ von Differentialoperatoren eine formale
   Potenzreihe $D' \colon C^\infty(U')[[\lambda]] \to
   C^\infty(U')[[\lambda]]$ von Differentialoperatoren  mit
   $\kRes_{U'} \circ D \circ \h_{U',U} = \kRes_{U'} \circ D'$ existiert.

   Die Abbildung $B_{U'} := \frac{1}{\lambda}(\qkoszul_{U'}
   - \partial_{U'})$ ist offensichtlich eine formale Reihe von
   Differentialoperatoren und verschwindet auf Konstanten. Nach obiger
   Beobachtung gibt es dann einen Differentialoperator $D_1'$ mit
   $\kRes_{U'} \lambda {B_{U'}}\circ \h_{U',U} = \kRes_{U'} \circ
   D_1'$. Wir setzen nun für $n \in \mathbb{N}$ $B_n := (\lambda
   {B_{U'}} \circ \h_{U',U})^n$ und behaupten, dass es für jedes $n \in
   \mathbb{N}$ eine formale Reihe von Differentialoperatoren $D_n'$
   gibt, so dass $\kRes_{U'} \circ D_n' = \kRes_{U'} \circ B_n$ gilt. Der Beweis erfolgt durch vollständige Induktion. Den
   Induktionsanfang haben wir dabei oben schon durchgeführt. Es gelte
   die Aussage nun für ein $n \in \mathbb{N}$, dann gilt sie auch für
   $n+1$. Dies sehen wir wie folgt ein. Es gilt
   \begin{align*}
      \kRes_{U'} \circ B_{n+1} &= \kRes_{U'} \circ B_n \circ
      \lambda B_{U'}  \circ h_{U',U} \\
      &= \kRes_{U'} \circ D_n' \circ \lambda B_{U'} \circ \h_{U',U} \Fcom
   \end{align*}
   wobei im letzten Schritt die Induktionsannahme einfloss. Als
   Komposition von formalen Potenzreihen von Differentialoperatoren ist
   $\lambda D'_n \circ B_{U'}$ wieder eine formale Potenzreihe von
   Differentialoperatoren, womit wir unsere obige Beobachtung anwenden
   können und das gewünschte $D_{n+1}'$ finden.
   Mit dem bisher gezeigten ist dann die Existenz von  $S_{U'}$ klar,
   woraus wie Eingangs bemerkt  sofort die Behauptung folgt.
\end{proof}

\begin{definition}
   \label{def:reduziertesSternprodukt}
   Das Sternprodukt $\starred$ aus Satz
   \ref{satz:SternproduktAufDemReduziertenPhasenraum} nennen wir auch
   das (nach dem Quanten"=Koszul"=Schema) \neuerBegriff{reduzierte
     Sternprodukt}.
\end{definition}

\begin{bemerkung}
   \label{bem:reduziertesSternprodukt}
   \begin{bemerkungEnum}
      \item %
         Das nach dem Quanten"=Koszul"=Schema reduzierte Sternprodukt
         ist dasselbe wie das, welches man durch die BRST"=Konstruktion
         nach \cite{bordemann.herbig.waldmann:2000a} erhält.
      \item %
         Satz \ref{satz:SternproduktAufDemReduziertenPhasenraum} liefert
         eine konkrete, relativ einfache Formel für das Sternprodukt auf
         dem reduzierten Phasenraum, sogar eine, die strukturell
         derjenigen für die Poisson"=Klammer auf dem reduzierten
         Phasenraum gleicht,
         vgl.\ Gleichung~\eqref{eq:ReduziertePoissonKlammer}.
      \item %
         Die Quanteneinschränkung $\qRes$ ist zwar durch Gleichung
         \eqref{eq:QuantenEinschrDef} relativ konkret gegeben, es stellt
         sich jedoch heraus, dass die dort auftretende geometrische
         Reihe selbst in einfachen Beispielen nur sehr schwer noch
         konkreter berechnet werden kann. Das in
\cite{cahen2007symplectic}
      betrachtete Beispiel verdeutlicht
         die zu erwartenden Schwierigkeiten.
      \item %
         Ist $\star$ Hermitesch, so stellt sich die Frage, ob auch
         $\starred$ Hermitesch ist. Falls $\star$ zusätzlich stark
         invariant ist und man als Quantenimpulsabbildung $\qJ = J +
         \frac{1}{2}\I \lambda \Delta $ wählt, ist das Sternprodukt
         $\starred$ wieder Hermitesch, wie Gutt und Waldmann in
         \cite[Prop. 4.7]{gutt2010involutions} zeigen konnten.

   \end{bemerkungEnum}
\end{bemerkung}

\subsection{Nichtverschwindende Impulswerte}
\label{sec:nichtVerschwindendeImpulswerte}

Die bisher beschriebene Quanten"=Koszul"=Reduktion ist nur für den Fall
verschwindender Impulswerte anwendbar. Für nicht"=verschwindende
$G$"=invariante Impulswerte $\mu \in \lieAlgebra^*$ liegt es nahe, die um
$-\mu$ verschobene Impulsabbildung zu betrachten und dann an der Null zu
reduzieren.

Genauer sei $(M,\omega)$ eine symplektische Mannigfaltigkeit. Weiter
wirke eine Lie"=Gruppe $G$ stark Hamiltonsch, frei und eigentlich auf
$M$, $J\colon M \to \lieAlgebra^*$ sei eine $G$"=äquivariante
Impulsabbildung und $\mu \in \lieAlgebra^*$ ein $G$"=invarianter
Impulswert mit $J^{-1}(\mu) \neq \emptyset$. Dann ist auch $J^\mu := J -
\mu \colon M \to \lieAlgebra^*$ eine $G$"=äquivariante Impulsabbildung
und es gilt $J^{-1}(\mu) = ({J^\mu})^{-1}(0)$, womit die reduzierten
Phasenräume inklusive ihrer symplektischen Struktur trivialerweise
übereinstimmen. Sei nun $\star$ ein Sternprodukt für $M$ und $\qJ$ eine
$G$"=äquivariante Quantenimpulsabbildung für $J$ und $\star$. Dann ist
$\qJ -\mu$ auch eine $G$"=äquivariante Quantenimpulsabbildung für
$J^\mu$ und $\star$, womit wir vermöge der Quanten"=Koszul"=Reduktion
ein Sternprodukt $\star_{\mathrm{red},J^\mu,0}$ für den bezüglich
$J^\mu$ am Impulswert $0$ reduzierten Phasenraum erhalten. Ist $\star$
mit $J$ stark invariant, so auch mit $J^\mu$, da Poisson"=Klammern und
$\star$"=Kommutatoren mit konstanten Funktionen verschwinden. Wir
definieren dann das reduzierte Sternprodukt $\star_{\mathrm{red},\mu}$
als $\star_{\mathrm{red},\mu} := \star_{\mathrm{red},J^\mu,0}$. Für
allgemeine, nicht notwendigerweise invariante Impulswerte kann man das
Shifting"=Theorem verwenden, an das wir zunächst kurz erinnern wollen
\begin{proposition}(Shifting-Theorem)
   \label{prop:ShiftingTrick}
   Sei $(M,\omega)$ eine symplektische Mannigfaltigkeit und $G$ eine
   Lie"=Gruppe, die stark Hamiltonsch, frei und eigentlich auf $M$
   operiere und $J \colon M \to \lieAlgebra^*$ eine $G$"=äquivariante
   Impulsabbildung. Weiter sei $\mu$ ein regulärer Wert von $J$ mit
   $J^{-1}(\mu) \neq \emptyset$, $G\mu$ der koadjungierte Orbit und
   $\omega_{G\mu}^+$ die Kirillov"=Kostant"=Souriau symplektische Form
   auf $G\mu$. Dann ist $J^\mu \colon M \times G\mu \to \lieAlgebra^*$,
   $(p,\nu) \mapsto J(p) - \nu$ eine $G$"=äquivariante Impulsabbildung
   auf der symplektischen Mannigfaltigkeit $(M \times G\mu, \omega -
   \omega_{G\mu}^+)$ für die kanonische Diagonalwirkung von $G$ auf $M
   \times G \mu$ und die reduzierten Phasenräume $(M \times G
   \mu)_{\mathrm{red}}$ und $M_{\mathrm{red},\mu}$ sind symplektomorph,
   wobei $M_{\mathrm{red},\mu}$ den bezüglich des Impulswerts $\mu$ reduzierten
   Phasenraum bezeichnet.

\end{proposition}
Der Beweis ist zwar nicht besonders schwierig, da er jedoch
Orbit"=Reduktion verwendet, auf welche wir nicht eingehen wollen und da
der Satz im weiteren Verlauf der Arbeit nicht von Bedeutung ist,
verweisen wir auf \cite[Thm. 6.5.2]{ortega.ratiu:2004}.

Man kann nun versuchen ein Sternprodukt $\star_{G \mu}$ auf $C^\infty(G
\mu)$ zu wählen, $\hat{\star} := \star \otimes \star_{G \mu}$ zu
betrachten und als reduziertes Sternprodukt $\hat{\star}_{\mathrm{red}}$ zu
nehmen. Damit dieses Vorgehen von einem konstruktivistischen Standpunkt
aus befriedigend erscheint, müsste jedoch noch untersucht werden, welche
Freiheiten man bei der Wahl von $\star_{G\mu}$ hat und wie sich diese
auf $\hat{\star}_{\mathrm{red}}$, beziehungsweise dessen Klasse
auswirken. Falls $\mu$ $G$"=invariant ist, gilt natürlich $G\mu = \{\mu\}$,
jegliche Wahlfreiheiten verschwinden und die beschriebene Anwendung des
Shifting"=Theorems liefert das weiter oben schon erklärte.

\subsection{Homotopie für den gesamten Quanten-Koszul-Komplex}
\label{sec:HomotopieFuerGesamtenQuatenKoszulKomplex}

Die folgenden Aussagen werden für die restliche Arbeit zwar nicht
benötigt, wir wollen jedoch der Vollständigkeit halber noch ausführen,
wie man auf relativ elementarem Wege eine Homotopie für den ganzen
augmentierten Quanten"=Koszul"=Komplex erhalten kann. Diese wird
interessant, wenn man den ganzen BRST"=Komplex betrachtet. Siehe
\cite{bordemann.herbig.waldmann:2000a} und
\cite{herbig2007variations} für den singulären Fall. Dort wird auch
gezeigt, wie man die Homotopie auf abstrakterem Wege, mittels
homologischer Störungstheorie bekommt.

Im Folgenden schreiben wir $\qkoszul[0] := \qRes$ und betrachten den
Komplex
\def\tA[#1]{A_{#1}}
\begin{equation}
      \begin{tikzpicture}[baseline=(current
    bounding box.center),description/.style={fill=white,inner sep=2pt}]
         \matrix (m) [matrix of math nodes, row sep=3.0em, column
         sep=3.5em, text height=1.5ex, text depth=0.25ex]
         {
C^\infty(C)[[\lambda]] & C^\infty(M)[[\lambda]] \otimes \Bigwedge^0
\lieAlgebra  & C^\infty(M)[[\lambda]]
\otimes \Bigwedge^1 \lieAlgebra & \dots \\
}; %

\path[<-] (m-1-1) edge node[auto]{$\qkoszul[0]$}(m-1-2); %
\path[<-] (m-1-2) edge node[auto]{$\qkoszul[1]$}(m-1-3); %
\path[<-] (m-1-3) edge node[auto]{$\qkoszul[2]$}(m-1-4); %
 \end{tikzpicture}.
\end{equation}
Wir sind nun in der Lage, eine Homotopie für diesen Komplex anzugeben. Um
die Notation für die folgenden Betrachtungen zu vereinfachen, schreiben
wir $\h[-1] := \prol$ und $\kkoszul[0] := \kRes$, $\kkoszul[k] := 0$, $ \h[k] :=
0$ für $k \leq -1$, $\Id[k] := 0$ für $k \leq -2$, $\Id[-1] :=
\id_{C^\infty(C)}$ und $\Id[k] := \id_{C^\infty(M) \otimes \Bigwedge^k
  \lieAlgebra}$ für $k \geq 0$. Damit kombinieren sich die
Gleichungen $\kRes \prol = \id$, \eqref{eq:globalisiertHomotopie1} und
\eqref{eq:globalisiertHomotopie2} für alle $k \in \mathbb{Z}$ zur Gleichung
\begin{align}
   \label{eq:KlassischeHomotopieAllesInEinem}
   \h[k-1] \kkoszul[k] + \kkoszul[k+1] \h[k] = \Id[k] \Fdot
\end{align}
Da $\h[0] \prol = 0$ gilt, können wir $\qh[0]$  noch etwas
anders schreiben, so dass, grob gesagt, die Rollen von $\qkoszul$ und $\qh$
symmetrischer verteilt werden. Es gilt nämlich für $n \in \mathbb{N}$
\begin{align*}
  \lefteqn{ \h[0](\h[-1](\qkoszul[0] - \kkoszul[0]) + (\qkoszul[1]
   - \kkoszul[1]) \h[0])^n} \\
 &= \h[0](\prol(\qkoszul[0] - \kkoszul[0]) +
   (\qkoszul[1] - \kkoszul[1])\h[0])(\h[-1](\qkoszul[0] - \kkoszul[0]) +
   (\qkoszul[1] - \kkoszul[1])\h[0])^{n-1} \\ &= \dots = \h[0]((\qkoszul[1]
   - \kkoszul[1])\h[0])^n \Fcom
\end{align*}
also
\begin{align}
   \label{eq:QuantenHomotopieGradNullSymmetrieisert}
   \qh[0] = \h[0] \frac{\Id[0]}{\Id[0] + \h[-1](\qkoszul[0] - \kkoszul[0])  +
     (\qkoszul[1] - \kkoszul[1])\h[0]} \Fdot
\end{align}
Den Nenner können wir mit Hilfe der klassischen Homotopieeigenschaft,
siehe Gleichung \eqref{eq:globalisiertHomotopie2}, wie folgt noch etwas
kompakter schreiben.
\begin{align}
   \label{eq:QuantenHomotopieGradNullSymmetrieisertKompakter}
  \lefteqn{ \id + \h[-1](\qkoszul[0] - \kkoszul[0]) + (\qkoszul[1]
   - \kkoszul[1])\h[0]}\notag\\ &= \h[-1]\kkoszul[0] + \kkoszul[1]\h[0] +
   \h[-1](\qkoszul[0] - \kkoszul[0]) + (\qkoszul[1] - \kkoszul[1])\h[0] \notag\\
   &= \h[-1]\qkoszul[0] + \qkoszul[1] \h[0] \Fdot
\end{align}
Demnach gilt
\begin{align}
   \label{eq:QuantenHomotopieGradNullSymmetrieisertKompakter2}
   \qh[0] = \h[0]\frac{\Id[0]}{\h[-1]\qkoszul[0] + \qkoszul[1] \h[0]}\Fdot
\end{align}

Dies liefert uns eine grobe Idee, wie man  eine Homotopie für den
Quanten"=Koszul-Komplex definieren könnte.
\begin{proposition}
   \label{prop:Quantenhomotopie1}
   Die Abbildung $\h[k-1]\qkoszul[k] + \qkoszul[k+1]\h[k]$ ist für
   alle $k \in \mathbb{Z}$ invertierbar.
\end{proposition}
\begin{proof}
   Da $\qkoszul$ eine Deformation von $\partial$ ist, fängt
   \begin{align*}
      \lefteqn{\h[k-1]\qkoszul[k] + \qkoszul[k+1]\h[k] - \Id[k]}\\
      &= \h[k-1]\kkoszul[k] + \kkoszul[k-1]\h[k] + \h[k-1](\qkoszul[k]
      - \kkoszul[k]) + (\qkoszul[k+1] - \kkoszul[k+1])\h[k] - \Id[k]\\
      &= \h[k-1](\qkoszul[k] - \kkoszul[k]) + (\qkoszul[k+1] -
      \kkoszul[k+1])\h[k]
    \end{align*}
mindestens in Ordnung $\lambda$ an, woraus die Behauptung folgt. Dabei
wurde im zweiten Schritt der obigen Umformung Gleichung
\eqref{eq:KlassischeHomotopieAllesInEinem} verwendet.
\end{proof}
Wir definieren nun für $k \in \mathbb{Z}$ die Abbildung
\begin{align}
   \label{eq:DefinitionQuantenHomotopie}
   \qh[k] := \h[k] \frac{\id}{\h[k-1] \qkoszul[k] + \qkoszul[k+1]\h[k]}
\end{align}
und zeigen in der nächsten Proposition, dass dies eine Homotopie
 für den augmentierten Quanten"=Koszul"=Komplex liefert.
\begin{proposition}
   \label{prop:QuantenHomotopie}
   Es gilt für alle $k \in \mathbb{Z}$ die Homotopiegleichung
   \begin{align}
      \label{eq:QuantenHomotopie}
      \qh[k-1]\qkoszul[k] + \qkoszul[k+1]\qh[k] = \Id[k]\Fdot
   \end{align}
   Für  $k = -1$ gilt explizit
   \begin{align}
      \label{eq:QuantenHomotopieGradMinusEins}
      \qh[-1] = \prol
   \end{align}
   und für $k = 0$
   \begin{align}
      \label{eq:QuantenHomotopieGradNull}
      \qh[0] = \h[0] \frac{\Id[0]}{\Id[0] + (\qkoszul[1] - \kkoszul[1])\h[0]} \Fdot
   \end{align}
\end{proposition}
\begin{proof}
   Sei $k \in \mathbb{Z}$. Wir setzen $A_k := \h[k-1]\qkoszul[k] +
   \qkoszul[k+1]\h[k] - \Id[k]$. Nach Proposition
   \ref{prop:Quantenhomotopie1} ist $\Id[k] + A_k$ invertierbar. Offensichtlich
   gilt wegen $\qkoszul^2 = 0$ und der Definition von $A_k$
   \begin{align*}
      A_k \qkoszul[k+1] = (\h[k-1]\qkoszul[k] + \qkoszul[k+1]\h[k] -
      \Id[k])\qkoszul[k+1] = 0 + \qkoszul[k+1]\h[k]\qkoszul[k+1] - \qkoszul[k+1]
   \end{align*}
   und
   \begin{align*}
      \qkoszul[k]A_k = \qkoszul[k](\h[k-1]\qkoszul[k] + \qkoszul[k+1]\h[k]
      - \Id[k]) = \qkoszul[k] \h[k-1]\qkoszul[k] + 0 - \qkoszul[k] \Fdot
   \end{align*}
Somit erhalten wir die Vertauschungsrelation
\begin{align*}
   A_k \qkoszul[k+1] = \qkoszul[k+1] A_{k+1}\Fcom
\end{align*}
woraus sich unmittelbar die Beziehung
\begin{align*}
   \frac{\Id[k]}{\Id[k] + A_k} \qkoszul[k+1] = \qkoszul[k+1]
   \frac{\Id[k+1]}{\Id[k+1] + A_{k+1}}
\end{align*}
ergibt.
Folglich ist klar, dass
\begin{align*}
   \qh[k-1]\qkoszul[k] + \qkoszul[k+1]\qh[k] &= \h[k-1]
   \frac{\Id[k-1]}{\Id[k-1] + A_{k-1}}\qkoszul[k] + \qkoszul[k+1]\h[k]
   \frac{\Id[k]}{\Id[k] + A_k} \\
   &=
   \h[k-1] \qkoszul[k]\frac{\Id[k]}{\Id[k] + A_{k}} + \qkoszul[k+1]\h[k] \frac{\Id[k]}{\Id[k] + A_k}\\
   &= (  \h[k-1]\qkoszul[k]  + \qkoszul[k+1]\h[k])\frac{\Id[k]}{\Id[k] + A_k} \\
   &= \Id[k]
\end{align*}
gilt. Die Homotopieeigenschaft ist damit gezeigt.

Für $k= -1$ finden wir
\begin{align*}
  \qh[-1] = \h[-1]\frac{\Id[-1]}{\h[-2]\qkoszul[-1] +
    \qkoszul[0]\h[-1]} = \prol \frac{\Id[-1]}{0 + \qRes\prol} = \prol \Fdot
\end{align*}
Der Rest ist nach Gleichung
\eqref{eq:QuantenHomotopieGradNullSymmetrieisertKompakter2} klar.
\end{proof}

\chapter{Quanten-Koszul-Reduktion auf Kotangentialbündeln}
\label{cha:QuantenKoszulAufKotangentialbuendel}

Wir wollen nun die in Abschnitt \ref{sec:QuantenKoszul} dargestellte
Reduktionsmethode für Sternprodukte auf den Spezialfall von
Kotangentialbündeln anwenden und das so erhaltene Sternprodukt auf dem
reduzierten Phasenraum mit dem von Kowalzig, Neumaier und Pflaum in
\cite{kowalzig.neumaier.pflaum:2005a} konstruierten vergleichen.

\section{Kotangentialbündel mit Magnetfeld}
\label{sec:KotangentialBuendelMitMagnetfeld}

Wir betrachten das Kotangentialbündel $\mathfrak{p} \dpA M = T^*Q \to Q$
einer Mannigfaltigkeit $Q$, versehen mit der symplektischen Form $\omegaB
= \omegaKan + \mathfrak{p}^* B$, wobei $\omegaKan \in
\Gamma^\infty(\Bigwedge^2 T^*(T^*Q))$ die kanonische symplektische Form
bezeichne und $B \in \Gamma^\infty(\Bigwedge^2 T^*Q)$ eine beliebige geschlossene
Zweiform ist. Es sei daran erinnert, dass $\omegaKan$ durch $\omegaKan := -
d \thetaKan $ definiert ist. Dabei ist $\thetaKan \in
\Gamma^\infty(T^*(T^*Q))$ die kanonische Einsform, welche für $\alpha_q \in
T^*_{q}Q$ und $w_{\alpha_q} \in T_{\alpha_q}(T^*Q)$ wiederum durch
\begin{align}
   \label{eq:KanonischeEinsform}
   \thetaKan\at{\alpha_q}(w_{\alpha_q}) = \alpha_q(T_{\alpha_q}\mathfrak{p}(w_{\alpha_q}))
\end{align}
gegeben ist. Eine Karte $(U,(x^1,\dots,x^n))$  für $Q$ induziert
bekanntlich (vgl.\ \cite[Abschnitt 3.2.1]{waldmann:2007a},
\cite[Prop. 4.4]{lee:2003a}) eine Karte
$(\mathfrak{p}^{-1}(U),(q^1,\dots,q^n,p_1,\dots,p_n))$ für $T^*Q$ mit
\begin{align}
   \label{eq:KotangentialbuendelKoordinaten1}
   q^i(\alpha_q) = x^i(q)
\end{align}
und
\begin{align}
   \label{eq:KotangentialbuendelKoordinaten2}
   p_i(\alpha_q) = \alpha_q \left(\frac{\partial}{\partial x^i}\at{q} \right)
\end{align}
für $i \in \{1,\dots,n\}$, $q \in Q$ und $\alpha_q \in T^*_qQ$. Man rechnet
leicht nach (vgl.\ \cite[Lem.~3.2.2]{waldmann:2007a}), dass $\thetaKan$ in diesen Koordinaten die folgende Gestalt
annimmt.
\begin{align}
   \label{eq:KanonischeEinsformInKoordinaten}
  \thetaKan\at{\mathfrak{p}^{-1}(U)} = p_i d q^i \Fdot
\end{align}
Durch Anwenden der Definition von $\omegaKan$ folgt damit auch unmittelbar
\begin{align}
   \label{eq:KanonischeSymplektischeFormInKoordinaten}
   \omegaKan\at{\mathfrak{p}^{-1}(U)} = dq^i \wedge
   d p_i \Fdot
\end{align}

Eine physikalische Interpretation der oben geschilderten Situation wäre
ein System von endlich vielen geladenen Teilchen im Sinne der klassischen
Punktmechanik, deren Bewegungen etwa aufgrund von
Zwangsbedingungen auf den Konfigurationsraum $Q$ eingeschränkt sind und
die sich zudem in einem Magnetfeld $B$ befinden.

Wir nehmen an, dass $Q$ eine Wirkung $\phi \colon G \times Q \to Q$
trägt und heben diese zu einer Wirkung $\Phi \colon G \times T^*Q \to
T^*Q$, mit $\Phi_g = \Phi(g,\cdot) := \mathsf{T}^*(\phi_{g^{-1}})$ für alle $g
\in G$. Dabei bezeichnet für $\phi \colon Q \to Q$ die Abbildung
$\mathsf{T}^*\phi \colon T^*Q \to T^*Q$ den Kotangentiallift, welcher
bekanntlich  für alle $q \in Q$, $\alpha_{\phi(q)} \in
T^*_{\phi(q)}Q$ und $v_q \in T_qQ$ durch
${\mathsf{T}^*\phi(\alpha_{\phi(q)})}{v_q} :=
\alpha_{\phi(q)}(T_q\phi(v_q))$ definiert ist. Da für $\phi \colon Q \to
Q$ der Kotangentiallift $\mathsf{T}^*\phi$ ein Vektorbündelisomorphismus
über $\phi^{-1}$ ist, siehe etwa \cite[Satz 3.2.11]{waldmann:2007a}, ist
klar, dass $\Phi$ eine Gruppenwirkung definiert. In der Tat ist $\Phi$
auch glatt, wie man etwa leicht in Bündelkarten sieht.

Man beachte, dass sich die Eigenschaften einer Wirkung, frei und
eigentlich zu sein, auf die geliftete Wirkung vererben, siehe
Proposition \ref{prop:eigentlichAnsteckend}.

Falls $B = 0$ ist, gibt es eine kanonische $G$"=äquivariante
Impulsabbildung $\JKan \colon T^*Q \to \lieAlgebra^*$, welche für $q \in
Q$ und $\alpha_q \in T_q^*Q$ durch $\JKan(\alpha_q)(\xi) :=
\dPaar{\alpha_q}{\xi_Q(q)}$ gegeben ist, vgl.\ etwa \cite[Satz 3.3.39]{waldmann:2007a}. Die Frage, unter welchen
Voraussetzungen es eine klassische Impulsabbildung für die Situation
mit $B \neq 0$ gibt und wie sie in diesem Falle aussieht, klärt das
folgende wohlbekannte Lemma, welches auch in
\cite{kowalzig.neumaier.pflaum:2005a} gezeigt wurde.

\begin{lemma}
   \label{lem:ImpulsAbbildungMitMagnetfeld}
   \begin{lemmaEnum}
   \item %
      $G$ wirkt genau dann symplektisch bezüglich $\omega_B$, wenn $B$
      $G$"=invariant ist.

   \item %
      Sei $B$ $G$"=invariant. Dann gibt es genau dann eine klassische
      Impulsabbildung bezüglich $\omegaB$, wenn eine glatte Abbildung
      $j_0 \colon Q \to \lieAlgebra^*$ existiert, so dass für alle $\xi
      \in \lieAlgebra$ die Gleichung

      \begin{align}
         \label{eq:ImpulsAbbildungMitMagnetfeld}
         d \dPaar{j_0}{\xi} = B(\xi_Q,\cdot)
      \end{align}
      erfüllt ist. In diesem Fall
      definiert  $\dPaar{\JB}{\xi} = \dPaar{\JKan}{\xi} +
      \mathfrak{p}^*\dPaar{j_0}{\xi}$ für $\xi \in \lieAlgebra$ eine
      klassische Impulsabbildung  $\JB \colon T^*Q \to
      \lieAlgebra^*$ bezüglich $\omegaB$. Weiter ist $\JB$ genau dann $G$"=äquivariant, wenn $j_0$
      $G$"=äquivariant ist.

   \end{lemmaEnum}
\end{lemma}
\begin{proof}

   \begin{beweisEnum}
      \item %
         Für $g \in G$ gilt
         \begin{align*}
            \Phi_g^* \omegaB = \Phi_g^* \omegaKan  + \Phi_g^* \mathfrak{p}^* B =
            \omegaKan + \mathfrak{p}^* \phi_g^* B
         \end{align*}
         und somit
         \begin{align*}
            \Phi_g^* \omegaB = \omegaB \iff \mathfrak{p}^* B = \mathfrak{p}^*
            \phi_g^* B \iff B = \phi_g^* B \Fdot
         \end{align*}
         Dabei wurde bei der letzten Äquivalenz die Submersivität und
         Surjektivität von
         $\mathfrak{p}$ verwendet.
      \item %
         Für $\xi \in \lieAlgebra$, $q \in Q$, $\alpha_q \in T^*_qQ$ und
         $v_{\alpha_q}\in T_{\alpha_q}T^*Q $ gilt
         \begin{align*}
            \omegaB(\alpha_q)(\xi_{T^*Q}(\alpha_q),v_{\alpha_q}) &=
            \omegaKan(\alpha_q)(\xi_{T^*Q}(\alpha_q),v_{\alpha_q}) +
            B(q)(T_{\alpha_q} \mathfrak{p} \xi_{T^*Q}(\alpha_q),T_{\alpha_q} \mathfrak{p} v_{\alpha_q})\\
            &= \dPaar{d_{\alpha_q} \JKan(\xi)}{v_{\alpha_q}} +
            B(q)(\xi_Q(q),T_{\alpha_q}
            \mathfrak{p} v_{\alpha_q}) \\
            &= \dPaar{d_{\alpha_q} \JKan(\xi)}{v_{\alpha_q}} + \dPaar{B(q)(\xi_Q(q),\cdot)}{T_{\alpha_q} \mathfrak{p} v_{\alpha_q}} \\
            &= \dPaar{d_{\alpha_q} \JKan(\xi)}{v_{\alpha_q}} +
            \dPaar{(\mathfrak{p}^*(B(\xi_Q,\cdot)))(\alpha_q)}{v_{\alpha_q}}
            \Fdot
         \end{align*}
         Somit gilt weiter
         \begin{align*}
            \lefteqn{\text{$J \colon T^*Q \to \lieAlgebra^*$ ist eine
                Impulsabbildung bzgl. $\omegaB$} } \\
            &\iff X_{J(\xi)} = \xi_{T^*Q} \quad \forall \xi \in
            \lieAlgebra \quad \text{bzgl. $\omegaB$}
            \\
            &\iff d J(\xi) = \omega_B(\xi_{T^*Q},\cdot) \quad \forall \xi \in \lieAlgebra\\
            &\iff d(J(\xi) - \JKan(\xi)) = \mathfrak{p}^*(B(\xi_Q,\cdot))
            \quad \forall \xi \in \lieAlgebra \Fdot
         \end{align*}
         Ist dann $j_0 \colon Q \to \lieAlgebra^*$ eine glatte Abbildung
         mit $d\dPaar{j_0}{\xi} = B(\xi_Q,\cdot)$, so ist nach diesen
         Überlegungen klar, dass $\JB := \JKan + \mathfrak{p}^*j_0$ eine
         Impulsabbildung bezüglich $\omegaB$ ist. Ist umgekehrt $J$ eine
         Impulsabbildung bezüglich $\omegaB$, so folgt aus obigem für
         $j_0 := n^*(J - \JKan)$ die Gleichung $d\dPaar{j_0}{\xi} =
         B(\xi_Q,\cdot)$, wobei $n \colon Q \to T^*Q$ den Nullschnitt
         bezeichne. Da $\JKan$ schließlich $G$"=äquivariant ist, ist im
         Falle der Existenz eines $j_0$ die $G$"=Äquivarianz von $\JB =
         \JKan + \mathfrak{p}^*j_0$ äquivalent zur $G$"=Äquivarianz von
         $j_0$.
   \end{beweisEnum}
\end{proof}

\begin{bemerkung}
   \label{bem:KanonischeImpulsabbildung}
   Ist $B = 0$, so kann man in Lemma
   \ref{lem:ImpulsAbbildungMitMagnetfeld}  $j_0 = 0$
   wählen und erhält $\JB = \JKan$.
\end{bemerkung}

Wir wollen im Folgenden immer die Existenz einer klassischen
$G$"=äquivarianten Impulsabbildung $\JB$ bezüglich $\omegaB$ der Form
$\dPaar{\JB}{\xi} = \dPaar{\JKan}{\xi} + \mathfrak{p}^*\dPaar{j_0}{\xi}$
für $\xi \in \lieAlgebra$ mit einem $G$"=äquivarianten $j_0$ wie in
Lemma \ref{lem:ImpulsAbbildungMitMagnetfeld} annehmen und stets nur
solche betrachten, wenn wir von $\JB$ sprechen.

Bezeichne $VQ$ das Vertikalbündel des Hauptfaserbündels $\varpi \dpA Q \to Q/G$,
dessen Fasern für alle $q \in Q$ durch $V_qQ := \{\xi_Q(q) \mid \xi \in \lieAlgebra
\}$ gegeben sind.

Weiter sei $\gamma \colon TQ \to \lieAlgebra$ eine
Zusammenhangseinsform für $\varpi \dpA Q \to Q/G$
(vgl.\ Prop.~\ref{prop:ExistenzHauptfaserbuendelzusammenhang} und Prop. \ref{prop:HauptfaserBuendelZusammenhang}). Dies
liefert eine Zerlegung des Tangentialbündels
\begin{align}
   \label{eq:ZerlegungTangentialbuendel}
   TQ = HQ \oplus  VQ  \Fcom
\end{align}
wobei $HQ$ das von $\gamma$ induzierte Horizontalbündel bezeichne,
dessen Fasern  für alle $q \in Q$ durch $H_qQ =
\{v_q \in T_qQ \mid \gamma(v_q) = 0 \in \lieAlgebra\}$ gegeben
sind. Setzen wir dann $H^{*^\gamma}_qQ := \{\alpha \in T_{q}^*Q \mid
\alpha\at{V_qQ} = 0\}$ und $V^{*^\gamma}_qQ := \{\alpha_q \in T_q^*Q
\mid \alpha_q\at{H_qQ} = 0\}$, so sind $H^{*^\gamma}Q := \bigcup_{q \in
  Q}H^{*^\gamma}_qQ$ und $V^{*^\gamma}Q := \bigcup_{q \in
  Q}V^{*^\gamma}_qQ$ Untervektorbündel von $T^*Q$ und es gilt $T^*Q =
H^{*^\gamma}Q \oplus V^{*^\gamma}Q $. Man beachte, dass die Zerlegungen
$TQ = HQ \oplus VQ$ und $T^*Q = H^{*^\gamma}Q \oplus V^{*^\gamma}Q$
$G$"=invariant sind. Dabei operiert $G$ auf $TQ$ via $G \times TQ \ni (g,v)
\mapsto T\phi_gv \in TQ$.

Die Vektorbündelstruktur von $T^*Q$ erlaubt es uns, von polynomialen Funktionen
\begin{align*}
   \mathcal{P}(Q) := \{f \in C^\infty(T^*Q) \mid f\at{T^*_qQ} \text{ ist
     polynomial für alle $q \in Q$ } \} \subset C^\infty(T^*Q)
\end{align*}
in den Impulskoordinaten, oder kurz von polynomialen Funktionen, zu
sprechen. Um diese Zusatzstruktur nutzen zu können, erscheint es
sinnvoll, sich auf Sternprodukte $\star$ zu beschränken, für die
$\mathcal{P}(Q)[[\lambda]]$ eine $\star$"=Unteralgebra von
$C^\infty(T^*Q)[[\lambda]]$ bildet.

\cite{kowalzig.neumaier.pflaum:2005a} folgend werden wir in diesem
Kapitel nur Sternprodukte $\star$ für $(M,\omega_B)$ mit den folgenden
Eigenschaften betrachten.

\begin{compactitem}
\item $\star$ ist $G$"=invariant, d.\,h.\ invariant bezüglich der
   gelifteten Wirkung $\Phi$ auf $T^*Q$.
\item $\qJ[B]$ ist eine $G$"=äquivariante Quantenimpulsabbildung für die
   klassische Impulsabbildung $J_B$ und das Sternprodukt $\star$.
\item $\mathcal{P}(Q)[[\lambda]]$ ist eine $\star$"=Unteralgebra von
   $C^\infty(T^*Q)[[\lambda]]$.
\item $\qJ[B](\xi) \in \mathcal{P}(Q)[[\lambda]]$ für alle $\xi \in \lieAlgebra$.
\end{compactitem}

Kowalzig, Neumaier und Pflaum konnten mit einigem Aufwand zeigen, dass
die letzte Forderung schon aus den restlichen folgt, siehe
\cite[Cor.~4.14~ii)]{kowalzig.neumaier.pflaum:2005a}. Man beachte, dass
für die klassische Impulsabbildung nach Konstruktion schon $J_B(\xi) \in
\mathcal{P}(Q)$ für alle $\xi \in \lieAlgebra$ erfüllt ist.

\section{Reduziertes Sternprodukt für verschwindenden Impulswert und verschwindendes Magnetfeld}
\label{sec:SpezialFallImpulsMagnetfeldNull}

In diesem Abschnitt wenden wir uns dem Spezialfall verschwindender
Impulswerte $\mu = 0$ und verschwindender Magnetfelder zu. Wir nehmen
also $B = 0$ sowie $J_B = \JKan$ an und schreiben $\qJ[\mathrm{kan}] =
\qJ[B]$.

 Wegen $\JKan(\xi)(\alpha_q) =
\alpha_q(\xi_Q(q))$ für $q \in Q$ und $\alpha_q \in T_q^*Q$ gilt
offensichtlich
\begin{align}
   \label{eq:ImpulsniveauFlaecheSpezialfall}
   C_{\mathrm{kan}} := \JKan^{-1}(0) &= \{\alpha \in T^*Q \mid
   \JKan(\alpha)(\xi) = 0 \,
   \forall \xi \in \lieAlgebra\} \notag \\
   &= \{\alpha \in T^*Q \mid \alpha(\xi_Q(\mathfrak{p}(\alpha))) = 0 \,
   \forall \xi \in \lieAlgebra\} = H^{*^\gamma}Q \Fdot
\end{align}
Ist $\mu \in \lieAlgebra^*$, so wird durch
\begin{align}
   \label{eq:Zusammenhangseinsform}
   \dPaar{\Gamma_\mu(q)}{v_q}
   := \dPaar{\mu}{\gamma(v_q)} \quad \text{für $q \in Q$, $v_q \in T_qQ$}
\end{align}
eine glatte Einsform
\begin{align}
   \Gamma_\mu \colon Q \ni q \mapsto \Gamma_\mu(q) \in T^*Q
\end{align}
definiert und es gilt
$\Gamma_{\mu}(q) \in V^{*^\gamma}_qQ$ für alle $q \in Q$.
Wir betrachten als nächstes die Abbildung
\begin{align}
   \label{eq:globaleTube}
   \Psi' \dpA C_{\mathrm{kan}} \times \lieAlgebra^* \to T^*Q = C_{\mathrm{kan}} \oplus
   V^{*^\gamma}Q, \quad (c,\mu) \mapsto c + \Gamma_\mu(\mathfrak{p}(c))
   \Fdot
\end{align}

\begin{proposition}
   \label{prop:globaleTubenabbildung}
   $\Psi'$ ist eine  globale, glatte, $G$"=äquivariante, \tn{gute}
   Tubenumgebung von $C_{\mathrm{kan}}$ in $T^*Q$.
\end{proposition}

\begin{proof}
$\Psi'$ ist offensichtlich glatt. Seien $(c,\mu), (c',\mu') \in C_{\mathrm{kan}} \times
\lieAlgebra^*$ mit $c + \Gamma_\mu(\mathfrak{p}(c)) = c' + \Gamma_{\mu'}(\mathfrak{p}(c))$,
so erhält man sofort $c = c'$ und $\dPaar{\mu -
  \mu'}{\gamma(v_{\mathfrak{p}(c)})} = 0$ für alle $v_{\mathfrak{p}(c)} \in
T_{\mathfrak{p}(c)}Q$. Mit $\gamma(\xi_Q(\mathfrak{p}(c))) = \xi$ für $\xi \in \lieAlgebra$

folgt weiter $\dPaar{\mu - \mu'}{\xi} = 0$ für alle $\xi \in
\lieAlgebra$, also auch $\mu = \mu'$. Somit ist klar, dass $\Psi'$
injektiv ist. Sieht man $C_{\mathrm{kan}} \times \lieAlgebra^*$ als Vektorbündel über
$Q$ an, so ist $\Psi'$ eine Vektorbündelmorphismus über der
Identität, womit $\Psi'$ wegen der Injektivität und aus
Dimensionsgründen schon ein Diffeomorphismus sein muss.

Dass es sich tatsächlich auch um eine Tubenabbildung handelt, ist klar,
da $\Psi'(c,0) = c$ für alle $c \in C_{\mathrm{kan}}$.
Um die Äquivarianz nachzuprüfen, sei $g \in G$. Dann gilt für $q \in Q$
und $v_{gq} \in T_{gq}Q$
\begin{align*}
   \dPaar{\Gamma_{g\mu}(gq)}{v_{gq}} &= \dPaar{g \mu}{\gamma(v_{gq})} =
   \dPaar{\mu}{g^{-1} \gamma(v_{gq})}\\ &=
   \dPaar{\mu}{\gamma(g^{-1} v_{gq})} =
   \dPaar{\Gamma_{\mu}(q)}{g^{-1} v_{gq}} = \dPaar{g
     \Gamma_\mu(q)}{v_{gq}} \Fdot
\end{align*}
Demnach ist für alle $c \in C_{\mathrm{kan}}$, $g \in G$ und $\mu \in \lieAlgebra^*$ die Gleichung
\begin{align*}
   \Psi'(g(c,\mu)) &= \Psi'((gc,g\mu)) = gc + \Gamma_{g\mu}(\mathfrak{p}(gc)) \\ &= gc
   + \Gamma_{g\mu}(g\mathfrak{p}(c)) = g(c + \Gamma_\mu(\mathfrak{p}(c))) = g\Psi'((c,\mu))
\end{align*}
richtig.
Sei nun $c \in C_{\mathrm{kan}}$ und $\mu \in \lieAlgebra^*$. Dann gilt für alle $\xi \in
\lieAlgebra$
\begin{align*}
   \dPaar{\JKan(\Psi'(c,\mu))}{\xi} &= \dPaar{\Psi'(c,\mu)}{\xi_Q(\mathfrak{p}(c))} \\ &=
   \dPaar{c}{\xi_Q(\mathfrak{p}(c))} + \dPaar{\Gamma_\mu(\mathfrak{p}(c))}{\xi_Q(\mathfrak{p}(c))}
   = \dPaar{\mu}{\gamma(\xi_Q(\mathfrak{p}(c)))} = \dPaar{\mu}{\xi} \Fdot
\end{align*}
Dabei wurde im dritten Schritt $c \in C_{\mathrm{kan}} =
H_{\mathfrak{p}(c)}^{*^\gamma}Q$ und $\xi_Q(\mathfrak{p}(c)) \in
V_{\mathfrak{p}(c)}Q$ ausgenutzt.

Demnach ist $\Psi'$ auch tatsächlich eine \tn{gute} Tubenabbildung.
\end{proof}

Wir schreiben im Folgenden $\Psi := \Psi'^{-1}$.

Die Tatsache, dass wir in diesem Abschnitt nicht eine beliebige
symplektische Mannigfaltigkeit mit irgendeiner $G$"=Wirkung
betrachteten, sondern ein Kotangentialbündel mit gelifteter
$G$"=Wirkung, verhalf, nach Wahl einer Zusammenhangseinsform, die
besonders einfache globale, \tn{gute} Tubenabbildung $\Psi'$ zu
konstruieren. Insbesondere war die Vektorbündelstruktur von $T^*Q$
wichtig. Diese wollen wir auch in den weiteren Betrachtungen massiv
ausnutzen. Einerseits erlaubt sie uns, wie schon erwähnt, von
polynomialen Funktionen, zu sprechen. Die gewählte Zusammenhangseinsform
liefert andererseits eine Zerlegung des Tangentialbündels von $Q$ und
damit, wie wir in Kürze genauer darlegen werden, auch Zerlegungen von
davon induzierten Tensorbündelkonstruktionen. Diese Zerlegungen
übertragen sich bekanntlich auf die polynomialen Funktionen
$\mathcal{P}(Q)$, da diese zu den Schnitten der symmetrischen Potenzen
von $TQ$ in Bijektion stehen. An diese Tatsache möchten wir in der
folgenden Proposition erinnern, vgl.\ auch \cite[2.2.23]{waldmann:2007a}.
Dazu sei
\begin{align}\label{eq:SymmetrischeAlgebraDefGleichung}
   \Gamma^\infty(\Bigvee TQ) := \bigoplus_{k=0}^\infty
   \Gamma^\infty({\Bigvee}^k TQ)
\end{align}
die symmetrische Algebra des $C^\infty(Q)$"=Moduls $\Gamma^\infty(TQ)$.

\begin{proposition}
   \label{prop:universelleImpulsabbildung}

   Die lineare Abbildung
   \begin{align}
      \Gamma^\infty( TQ) \to \mathcal{P}(Q), s \mapsto
      (T^*Q \ni \alpha \mapsto \dPaar{\alpha}{s(\mathfrak{p}(\alpha))})
   \end{align}
   besitzt eine eindeutige Fortsetzung zu einem Algebraisomorphismus
   \begin{align}
      \mathsf{P}\dpA \Gamma^\infty(\Bigvee TQ) \to \mathcal{P}(Q)\Fdot
   \end{align}
   Es gilt insbesondere für $k \in \mathbb{N}\setminus\{0\}$ und $s \in
   \Bigvee^k \Gamma^\infty(T^*Q)$
   \begin{align}
      \label{eq:universelleImpulsabbildung}
      \mathsf{P}(s)(\alpha_q) = \frac{1}{k!}
s(q)(\alpha_q,\dots,\alpha_q)
   \end{align}
   für alle $q \in Q$ und $\alpha_q \in T_q^*Q$ sowie
   \begin{align}
      \label{eq:universelleImpulsabbildungGrad0}
      \mathsf{P}(\chi) = \mathfrak{p}^* \chi
   \end{align}
   für $\chi \in \CM[Q]$.

\end{proposition}
\begin{proof}
Klar.

\end{proof}

\begin{bemerkung}
   \label{bem:GInvarianzDesIsoszwischenPolysUndSyms}
   Die Wirkung von $G$ auf $TQ$ via $G \times TQ \ni (g,v) \mapsto
   T\phi_gv \in TQ$ induziert für jedes $k \in \mathbb{N}$ eine Wirkung
   auf $\bigvee^kTQ$ vermöge linearer Fortsetzung von $g(v_q^1 \vee
   \dots \vee v_q^n) := g v_q^1 \vee \dots \vee g v_q^n$ für alle $g \in
   G$, $q \in Q$ und $v_q^1,\dots,v_q^n \in T_qQ$, insbesondere gilt
   $(gv)(\alpha_q^1,\dots,\alpha_q^n) = v(g^{-1} \alpha_q^1,\dots,g^{-1}
   \alpha_q^n)$. Dies liefert bekanntlich wieder eine $G$"=Wirkung auf
   den Schnitten $\Gamma^\infty(\bigvee^kTQ)$ durch $(gs)(q) :=
   g(s(g^{-1}q))$ für alle $g \in G$, $q \in Q$ und $s \in
   \Gamma^\infty(\bigvee^k TQ)$. Damit sieht man leicht, dass $\mathsf{P}$
   $G$"=äquivariant ist.
\end{bemerkung}

\begin{proposition}
   \label{prop:Zerlegung}
   $\mathrm{PV}(\Gamma^\infty(\Bigvee TQ)) := \bigoplus_{m = 0,n =
     1}^\infty \Gamma^\infty(\Bigvee^m HQ \otimes \Bigvee^n VQ)$ wird
   via $(h^m\otimes v^n) \vee (h'^{m'} \otimes v'^{n'}) := (h^m \vee
   h'^{m'}) \otimes (v^n \vee v'^{n'})$, für alle $m,m',n,n' \in
   \mathbb{N}$ und $h^m\in \Gamma^\infty(\Bigvee^m HQ)$, $h'^{m'} \in
   \Gamma^\infty(\Bigvee^{m'} HQ)$, $v^n \in \Gamma^\infty(\Bigvee^n
   VQ)$ und $v'^{n'}\in \Gamma^\infty(\Bigvee^{n'} VQ)$ und
   $C^\infty(Q)$"=linearer Fortsetzung zu einer assoziativen,
   $\mathbb{N}\times \mathbb{N}^+$"=gradierten Algebra, wobei $x \in
   \mathrm{PV}(\Gamma^\infty(\Bigvee TQ))$ genau dann
   \neuerBegriff{homogen} vom Grade $(m,n) \in \mathbb{N} \times
   \mathbb{N}^+$ ist, wenn $x \in \Gamma^\infty(\Bigvee^m HQ \otimes
   \Bigvee^n VQ)$. Ebenso ist $\bigoplus_{m,n = 0}^\infty
   \Gamma^\infty(\Bigvee^m HQ \otimes \Bigvee^n VQ)$ eine assoziative,
   $\mathbb{N}\times \mathbb{N}$"=gradierte Algebra.
\end{proposition}
\begin{proof}
Klar.
\end{proof}

\begingroup
\emergencystretch=0.8em
\begin{definition}
   \label{def:TeilweiseUndVollstaendigVertikaleSchnitte} %
   Wir wollen Elemente der gradierten Algebra
   $\mathrm{H}(\Gamma^\infty(\Bigvee TQ)) :=
   \bigoplus_{k=0}^\infty\Gamma^\infty(\Bigvee^k HQ)$ als
   \neuerBegriff{total horizontale Schnitte} und Elemente von
   $\mathrm{PV}(\Gamma^\infty(\Bigvee TQ))$ als \neuerBegriff{teilweise
     vertikale Schnitte} bezeichnen. Ein  vom Grade $(m,n) \in
   \mathbb{N} \times \mathbb{N}^+$ homogenes Element $T \in \Gamma^\infty(\Bigvee^m
   HQ \otimes \Bigvee^n VQ) $ bezeichnen wir auch als
   \neuerBegriff{horizontal vom Grad $m$} und \neuerBegriff{vertikal vom
     Grad $n$}.
\end{definition}
\endgroup

Mit den Regeln der Tensorrechnung,
vgl.\ Proposition \ref{prop:RechenregelnSchnitteUndSummen}, sehen wir dann
\begingroup
\allowdisplaybreaks[1]
\begin{align}
   \label{eq:ZerlegungSymmetrischeSchnitte}
   \Gamma^\infty(\Bigvee TQ) &= \bigoplus_{k=0}^\infty\Gamma^\infty(\Bigvee^k TQ) \notag \\
   &=\bigoplus_{k=0}^\infty \Gamma^\infty(\Bigvee^k(HQ \oplus VQ)) \notag \\
   &\simeq \bigoplus_{k=0}^\infty
   \Gamma^\infty(\bigoplus_{l=0}^k(\Bigvee^{k-l}HQ \otimes \Bigvee^l
   VQ)) \notag \\
   &\simeq \bigoplus_{k=0}^\infty\bigoplus_{l=0}^k
   \Gamma^\infty(\Bigvee^{k-l}HQ \otimes \Bigvee^l
   VQ) \notag \\
   &= \bigoplus_{m,n=0}^\infty \Gamma^\infty(\Bigvee^{m}HQ \otimes
   \Bigvee^n
   VQ) \notag \\
   &= \bigoplus_{k=0}^\infty \Gamma^\infty(\Bigvee^k HQ) \oplus
   \bigoplus_{m=0,n=1}^\infty \Gamma^\infty(\Bigvee^{m}HQ \otimes \Bigvee^n VQ) \notag
   \\
   &= \mathrm{H}(\Gamma^\infty (\Bigvee TQ)) \oplus \mathrm{PV}(\Gamma^\infty
   (\Bigvee TQ)) \Fdot
\end{align}
\endgroup
Die Projektionen zu dieser Zerlegung wollen wir im Folgenden mit
$\mathrm{H} \dpA \Gamma^\infty (\Bigvee TQ) \to \mathrm{H}(\Gamma^\infty
(\Bigvee TQ))$ und $\mathrm{PV} \dpA \Gamma^\infty (\Bigvee TQ) \to
\mathrm{PV}(\Gamma^\infty (\Bigvee TQ))$ bezeichnen. Offensichtlich ist
die Zerlegung \eqref{eq:ZerlegungSymmetrischeSchnitte} $G$"=invariant,
womit auch die Projektionen $\mathrm{H}$ und $\mathrm{PV}$
$G$"=äquivariant sind.  Die obige Zerlegung induziert mit Hilfe des
Isomorphismus $\mathsf{P}$ aus Proposition~\ref{prop:universelleImpulsabbildung} eine Zerlegung der polynomialen
Funktionen $\mathcal{P}(Q)$ in den Raum der \neuerBegriff{total
  horizontalen} polynomialen Funktionen $\mathrm{h}(\mathcal{P}(Q))$ und
den Raum der \neuerBegriff{teilweise vertikalen} polynomialen Funktionen
$\mathrm{pv}(\mathcal{P}(Q))$.
\begin{align}
  \mathcal{P}(Q) = \mathrm{h}(\mathcal{P}(Q)) \oplus \mathrm{pv}(\mathcal{P}(Q))
\end{align}
Die Projektionen bezüglich dieser Zerlegung sind dann durch
\begin{align}
   \label{eq:pvundh}
   \mathrm{pv} = \mathsf{P} \circ \mathrm{PV} \circ \mathsf{P}^{-1}
   \quad \text{und} \quad \mathrm{h} = \mathsf{P} \circ \mathrm{H} \circ
   \mathsf{P}^{-1}
\end{align}
gegeben. Mit Bemerkung \ref{bem:GInvarianzDesIsoszwischenPolysUndSyms}
folgt direkt auch die $G$"=Äquivarianz von $\mathrm{h}$ und
$\mathrm{pv}$. Insbesondere gilt $\mathrm{h}(\mathcal{P}(Q)^G) =
\mathrm{h}(\mathcal{P}(Q))^G$.  Ist $F \in
\mathsf{P}(\Gamma^\infty(\Bigvee^m HQ\otimes \Bigvee^n VQ)) \subset
\mathcal{P}(Q)$ mit $(m,n) \in \mathbb{N}\times
\mathbb{N}$, so nennen wir $F$ \neuerBegriff{horizontal vom Grad $m$}
und \neuerBegriff{vertikal vom Grad $n$}. Ist $\{e_\alpha\}_{1 \leq
  \alpha \leq \dim G}$ eine Basis von

$\lieAlgebra[g]$ mit zugehöriger dualer Basis $\{e^\alpha\}_{1 \leq
  \alpha \leq
  \dim G}$, dann bilden die
zugehörigen fundamentalen Vektorfelder $\{{(e_\alpha)}_Q\}_{1 \leq
  \alpha \leq
  \dim G}$ bekanntlich einen globalen Rahmen für das Vertikalbündel $VQ$.

Für jedes $T \in \mathrm{PV}(\Gamma^\infty(\Bigvee TQ))$ gibt es somit
eindeutig bestimmte $\mathrm{R}^i(T) \in \mathrm{PV}(\Gamma^\infty(\Bigvee TQ))$
mit  $T = \sum_{i=1}^{\dim G} \mathrm{R}^i(T)\vee {(e_i)}_Q$. Für $i \in \{1,\dots,\dim
G\}$
setzen wir $\mathrm{R}^i$ auf $\mathrm{H}(\Gamma^\infty(\Bigvee TQ))$ durch $0$ zu einer
Abbildung $\mathrm{R}^i\dpA \Gamma^\infty(\Bigvee TQ) \to \Gamma^\infty(\Bigvee TQ)$
fort. Offenbar kann man dann jedes $T \in \Gamma^\infty(\Bigvee TQ)$ in
der Form
\begin{align}
   \label{eq:Aufspaltung1}
   T = \mathrm{H}(T) + \mathrm{R}^i(T)\vee (e_i)_Q
\end{align}
schreiben.  Weiter setzen wir
\begin{align}
   \label{eq:ri}
   \mathrm{r}^i := \mathsf{P} \circ \mathrm{R}^i \circ \mathsf{P}^{-1}
   \dpA \mathcal{P}(Q) \to \mathcal{P}(Q)\Fcom
\end{align}
womit wir
mit Gleichung~\eqref{eq:Aufspaltung1} sehen, dass man jedes $F \in
\mathcal{P}(Q)$ in der Form
\begin{align}
   \label{eq:Aufspaltung2}
   F = \mathrm{h}(F) + \mathrm{r}^i(F)\mathsf{P}((e_i)_Q)
\end{align}
schreiben kann.

Wir wählen nun $(\Psi,T^*Q,1,\emptyset,\emptyset)$ als geometrische Homotopiedaten und können damit die
davon induzierte Homotopie  $\h_{\mathrm{kan}}$ aus dem
Quanten-Koszul-Schema mit Hilfe der $\mathrm{r}^i$ ausdrücken.

\begin{lemma}
   \label{lem:HomotopiePolynomial}
   Ist $F \in \mathcal{P}(Q)$, so gilt
   \begin{align}
      \label{eq:HomotopiePolynomial}
       \hKan(F)  = \mathrm{r}^i(F) \otimes e_i\Fdot
   \end{align}
\end{lemma}

\begin{proof}

   Sei $v \in \Gamma^\infty(HQ)$ und $F :=
   \mathsf{P}(v)$. Dann gilt  für $c \in C_{\mathrm{kan}}$
   und $\mu \in \lieAlgebra^*$
    \begin{align*}
       F \circ \Psi^{-1}(c,\mu) &= F(c + \Gamma_\mu(\mathfrak{p}(c))) = \mathsf{P}(v)(c +
       \Gamma_\mu(\mathfrak{p}(c))) \\ &= v\at{\mathfrak{p}(c)}(c +
       \Gamma_\mu(\mathfrak{p}(c))) = v\at{\mathfrak{p}(c)}(c) +
       \underbrace{v\at{\mathfrak{p}(c)}(\Gamma_\mu(\mathfrak{p}(c)))}_{=
         \dPaar{\mu}{\gamma(v(\mathfrak{p}(c)))}
       = 0} = v\at{\mathfrak{p}(c)}(c) \Fdot
   \end{align*}

   Also haben wir
   \begin{align*}
      \frac{\partial}{\partial \mu_i} (F \circ \Psi^{-1})(c,\mu) = 0
   \end{align*}
   für alle $i \in \{1,\dots,\dim G\}$, $c \in C_{\mathrm{kan}}$ und $\mu \in
   \lieAlgebra^*$, womit $\hKan(F) = 0 = \mathrm{r}^i(F) \otimes e_i$
   folgt.

   Für $F = \mathsf{P}((e_i)_Q)$ gilt
   \begin{align*}
      F \circ \Psi^{-1}(c,\mu) &= F(c + \Gamma_\mu(\mathfrak{p}(c))) =
      \mathsf{P}((e_i)_Q)(c + \Gamma_\mu(\mathfrak{p}(c))) \\ &=
      (e_i)_Q(\mathfrak{p}(c))(c+ \Gamma_\mu(\mathfrak{p}(c))) =
      \underbrace{(e_i)_Q(\mathfrak{p}(c))(c)}_{=0, \, \text{da $c \in
          H^{*^{\scriptscriptstyle{\gamma}}}_{\mathfrak{p}(c)}Q$}} +
      (e_i)_Q(\mathfrak{p}(c))(\Gamma_\mu(\mathfrak{p}(c))) \\  &=
      \dPaar{\mu}{\gamma((e_i)_Q(\mathfrak{p}(c)))} = \dPaar{\mu}{e_i} =
      \mu_i \Fcom
   \end{align*}
   wenn $\mu = \mu_i e^i$.  Wir erhalten somit weiter
   \begin{align*}
      \frac{\partial}{\partial \mu_\alpha}(F \circ \Psi^{-1})(c,\mu) =
      \delta_{\alpha i}
   \end{align*}
   für alle $\alpha \in \{1,\dots,\dim G\}$.
   Sei nun $(n,m) \in \mathbb{N} \times \mathbb{N}$ und $T =
   R^{k_1,\dots,k_m}_{l_1,\dots,l_n,i} v_{k_1} \vee \dots
   \vee v_{k_m} \vee (e_{l_1})_Q \vee
   \dots (e_{l_n})_Q \vee (e_i)_Q = \mathrm{R}^i(T) \vee (e_i)_Q$ mit $
   R^{k_1,\dots,k_m}_{l_1,\dots,l_n,i} \in C^\infty(Q)$, o.\,E.\ symmetrisch in den
   unteren Indices und $v_{k_j} \in \Gamma^\infty(HQ)$. Dann gilt für $F = \mathsf{P}(T)$
   \begin{align*}
      F \circ \Psi^{-1}(c,\mu) &= F(c + \Gamma_\mu(\mathfrak{p}(c)))\\
      &= \mathsf{P}(T)(c +
      \Gamma_\mu(\mathfrak{p}(c)))\\
      &= R^{k_1,\dots,k_m}_{l_1,\dots,l_n,i}(\mathfrak{p}(c)) \mathsf{P}(v_{k_1})(c +
      \Gamma_\mu(\mathfrak{p}(c)))\dotsm \mathsf{P}(v_{k_m})(c+
      \Gamma_\mu(\mathfrak{p}(c)))\cdot \\
      &\phantom{= } \mathsf{P}((e_{l_1})_Q)(c + \Gamma_\mu(\mathfrak{p}(c)))
      \dotsm \mathsf{P}((e_{l_n})_Q)(c + \Gamma_\mu(\mathfrak{p}(c)))\cdot \mathsf{P}((e_i)_Q)(c +
      \Gamma_\mu(\mathfrak{p}(c))) \\
      &=
      R^{k_1,\dots,k_m}_{l_1,\dots,l_n,i}(\mathfrak{p}(c))v_{k_1}\at{\mathfrak{p}(c)}(c)\dotsm
      v_{k_m}\at{\mathfrak{p}(c)}(c) \cdot \mu_{l_1} \dotsm
      \mu_{l_n} \cdot \mu_i \Fdot
   \end{align*}
   Folglich ergibt sich
   \begin{align*}
      \frac{\partial}{\partial \mu_\alpha} (F \circ \Psi^{-1})(c,\mu)&=
      R^{k_1,\dots,k_m}_{l_1,\dots,l_n,i}(\mathfrak{p}(c)) \cdot v_{k_1}\at{\mathfrak{p}(c)}(c)\dotsm
      v_{k_m}\at{\mathfrak{p}(c)}(c) \cdot \\
      &\phantom{= }(\delta_{\alpha l_1} \mu_{l_2}
      \dotsm \mu_{l_n} \cdot \mu_i + \dots + \mu_{l_1}\dotsm \delta_{l_n
      \alpha} \mu_i + \mu_{l_1} \dotsm \mu_{l_n} \delta_{i \alpha}) \\
    &= (n+1)  R^{k_1,\dots,k_m}_{l_1,\dots,l_n,\alpha}(\mathfrak{p}(c)) \cdot v_{k_1}\at{\mathfrak{p}(c)}(c)\dotsm
      v_{k_m}\at{\mathfrak{p}(c)}(c) \cdot \mu_{l_1} \dotsm \mu_{l_n} \\
      &= (n+1) \mathsf{P}(\mathrm{R}^\alpha(T))(c + \Gamma_\mu) \Fdot
   \end{align*}
   Des Weiteren gilt offensichtlich
   \begin{align*}
      \mathsf{P}(\mathrm{R}^\alpha(T))\circ\Psi^{-1}(c,t\mu) = t^n \mathsf{P}(\mathrm{R}^\alpha(T))\circ \Psi^{-1}(c,\mu)
   \end{align*}
   und damit
   \begin{align*}
      \int_0^1 \frac{\partial}{\partial \mu_\alpha}(F \circ
      \Psi^{-1})(c,t\mu) \, dt
      &= (n+1) \int_0^1 t^n \, dt \, \mathsf{P}(\mathrm{R}^\alpha(T))(\Psi^{-1}(c,\mu)) \\
      &= \mathsf{P}(\mathrm{R}^\alpha(T))(\Psi^{-1}(c,\mu)) \Fdot
   \end{align*}
   Schließlich sehen wir so
\begin{align*}
   \mathrm{r}^i(F) \otimes e_i = \mathrm{r}^i(\mathsf{P}(T)) \otimes e_i
   = \mathsf{P}(\mathrm{R}^i(T)) \otimes e_i  = \hKan(F) \Fdot
\end{align*}

\end{proof}
Wir definieren nun
\begin{align}
   \label{eq:DeltaNikolai}
   \Delta_\star \dpA \mathcal{P}(Q)[[\lambda]] &\to
   \mathcal{P}(Q)[[\lambda]], \notag\\
   F &\mapsto \frac{1}{\I \lambda} \sum_{i=1}^{\dim G} \left(\mathrm{r}^i(F)\JKan(e_i)
   - \mathrm{r}^i(F) \star
   \qJ[\mathrm{kan}](e_i)\right) \Fdot
\end{align}
Die Abbildung $\Delta_\star$ ist in der Tat wohldefiniert, denn nach
Voraussetzung sind die Funktionen $\mathrm{r}^i(F)$, $\JKan(e_i)$ sowie
$\qJ[\mathrm{kan}](e_i)$ polynomial und $\mathcal{P}(Q)[[\lambda]]$
bildet eine Unteralgebra von $C^\infty(M)[[\lambda]]$ bezüglich $\star$.
Mit Lemma \ref{lem:HomotopiePolynomial} sowie den Definitionen von
$\partial_{\mathrm{kan}}$ und $\qkoszul_{\mathrm{kan}}$ sehen wir für $F
\in \mathcal{P}(Q)$ sofort
\begin{align}
   \label{eq:DeltaUndPartial}
   \Delta_\star F = \frac{1}{\I \lambda} (\partial_{\mathrm{kan}} - \qkoszul_{\mathrm{kan}}) \hKan(F) \Fdot
\end{align}
Dabei bezeichnet $\partial_{\mathrm{kan}}$
bzw.\ $\qkoszul_{\mathrm{kan}}$ den von $\JKan$ bzw.\ $\qJ[\mathrm{kan}]$
induzierten Koszul- bzw.\ Quanten"=Koszul"=Operator.

Der folgende Satz zeigt, dass $(T^*Q)_{\mathrm{red}}$ wieder die
Struktur eines Kotangentialbündels besitzt.
\begin{samepage}
   \begingroup
   \emergencystretch=0.8em
   \begin{satz}[Kotangentialbündelreduktion für Impulswert $0$]

      \label{satz:ReduktionKotangentenbuendel}
      \begin{satzEnum}
      \item %
         Es gibt eine eindeutig bestimmte glatte, $G$"=invariante
         Abbildung
         \begin{align}
            u \dpA \JKan^{-1}(0) \to T^*(Q/G) \quad \text{mit} \quad
            \dPaar{u(\alpha_q)}{T_q\varpi v_q} = \dPaar{\alpha_q}{v_q}
         \end{align}
         für alle $q \in Q$, $\alpha_q \in \JKan^{-1}(0) \cap T_q^*Q$
         und $v_q \in T_qQ$. Weiter gibt es eine eindeutig bestimmte
         glatte Abbildung $\overline{u} \dpA (T^*Q)_{\mathrm{red}}\to
         T^*(Q/G)$, so dass das folgende Diagramm kommutiert.
         \def\tA[#1]{A_{#1}}
         \begin{equation}
            \label{eq:DiagrammU}
            \begin{tikzpicture}[baseline=(current
               bounding box.center),description/.style={fill=white,inner sep=2pt}]
               \matrix (m) [matrix of math nodes, row sep=3.0em, column
               sep=3.5em, text height=1.5ex, text depth=0.25ex]
               {
                 \JKan^{-1}(0) &   \\
                 (T^*Q)_{\mathrm{red}} & T^*(Q/G) \\
               }; %

            \path[->] (m-1-1) edge node[auto] {$\piKan$} (m-2-1); %
            \path[->] (m-2-1) edge node[auto]{$\overline u$}(m-2-2); %
            \path[->] (m-1-1) edge node[auto]{$u$}(m-2-2); %
         \end{tikzpicture} \Fdot
      \end{equation}
      Dabei bezeichnet $\piKan$ die kanonische Projektion.
\item %
   Die Abbildung $\overline u$ ist ein Symplektomorphismus.
   \end{satzEnum}
\end{satz}
\end{samepage}

\begin{proof}
   \begin{beweisEnum}
   \item %
      Wegen der Submersivität von $\varpi$ ist die Eindeutigkeit
      klar.

      Wir beachten die Gleichheit $\ker T_q\varpi = \{\xi_Q(q) \mid \xi \in
      \lieAlgebra\}$. Dies ist klar, denn offensichtlich ist $\xi_Q(q)
      \in \ker T_q\varpi$ für alle $\xi \in \lieAlgebra$ und da
      $T_q\varpi$ surjektiv ist, folgt $\dim \ker T_q\varpi = \dim Q -
      \dim Q/G = \dim G = \dim  \{\xi_Q(q) \mid \xi \in
      \lieAlgebra\}$.

      Falls also $T_q \varpi v_q = T_q \varpi
      \tilde{v}_q $, so folgt nach dieser Beobachtung schon, dass es ein
      $\xi \in \lieAlgebra$ gibt mit $v_q - \tilde{v}_q =
      \xi_Q(q)$. Dann ist aber für jedes $\alpha_q \in \JKan^{-1}(0) \cap
      T_q^*Q$ schon $\dPaar{\alpha_q}{v_q - \tilde{v}_q} =
        \dPaar{\alpha_q}{\xi_Q(q)} = 0$, da $\xi_Q(q) \in V_qQ$. Demnach
      existiert eine glatte Abbildung $u$ mit den gewünschten
      Eigenschaften.
      $u$ ist sogar $G$"=invariant, denn
      $\dPaar{u(g\alpha_q)}{T_{gq}\varpi v_{gq}} =
      \dPaar{\alpha_q}{g^{-1}v_{gq}} =
      \dPaar{u(\alpha_q)}{T_{gq}\varpi(g^{-1} v_{gq})} =
      \dPaar{u(\alpha_q)}{T_{gq}\varpi v_{gq}}$ für alle $\alpha_q \in
      \JKan^{-1}(0) \cap T_q^*Q$ und $v_{gq}\in T_{gq}Q$. Damit induziert
      $u$ eine eindeutig bestimmte glatte Abbildung, so dass Diagramm
      \eqref{eq:DiagrammU} kommutiert.
   \item %
      Sei $\alpha_{\varpi(q)} \in T^*_{\varpi(q)}(Q/G)$, dann ist
      $(T_q\varpi)^* \alpha_{\varpi(q)} \in \JKan^{-1}(0)$, denn für
      jedes $\xi \in \lieAlgebra$ gilt
      \begin{align*}
          \dPaar{\JKan((T_q\varpi)^*\alpha_{\varpi(q)})}{\xi} =
          \dPaar{(T_q\varpi)^*\alpha_{\varpi(q)}}{\xi_Q(q)} =
          \dPaar{\alpha_{\varpi(q)}}{T_q \varpi \xi_Q(q)} = 0 \Fdot
       \end{align*}
      Weiter gilt für alle $w_q \in T_qQ$
       \begin{align*}
          \dPaar{u((T_q\varpi)^* \alpha_{\varpi(q)})}{T_q\varpi w_q} =
          \dPaar{(T_p\varpi)^* \alpha_{\varpi(q)}}{w_q} =
          \dPaar{\alpha_{\varpi(q)}}{T_p \varpi w_q} \Fcom
       \end{align*}
       also wegen der Submersivität von $\varpi$ auch $u((T_q\varpi)^*
       \alpha_{\varpi(q)}) = \alpha_{\varpi(q)}$.

       Demnach ist $u$ und somit auch $\overline u$ surjektiv .  Die
       Abbildung $\overline{u}$ ist injektiv, denn ist $u(\alpha_q) =
       u(\beta_{q'})$, so folgt nach Definition von $u$ schon $\varpi(q) =
       \varpi(q')$, d.\,h.\ es gibt ein $g \in G$ mit $q' = gq$, und
      \begin{align*}
         \dPaar{\alpha_q}{w_q} &= \dPaar{u(\alpha_q)}{T_q\varpi w_q} =
         \dPaar{u(\beta_{q'})}{T_q\varpi w_q} =
         \dPaar{u(\beta_{q'})}{T_q(\varpi \circ \phi_g)w_q} \\&=
         \dPaar{u(\beta_{q'})}{T_{\phi_g(q)}\varpi T_q \phi_g w_q} =
         \dPaar{u(\beta_{q'})}{T_{gq}\varpi g w_q} =
         \dPaar{\beta_{q'}}{g w_q}\\ &= \dPaar{g^{-1} \beta_{q'}}{ w_q}
      \end{align*}
      für alle $w_q \in T_qQ$ also $\piKan(\alpha_q) =
      \piKan(\beta_{q'})$.  Um zu zeigen, dass $\overline{u}$
      symplektisch ist, genügt es die kanonischen $1$"=Formen zu
      betrachten.  Sei $\thetaKan$ die kanonische $1$"=Form auf $T^*Q$
      und $\overline\thetaKan$ diejenige auf $T^*(Q/G)$ sowie $\kIn \dpA
      \JKan^{-1}(0) \hookrightarrow T^*Q$ die Inklusion.  Nach
      Konstruktion der reduzierten symplektischen Form ist es also
      hinreichend, $\piKan^* \overline{u}^*\overline\thetaKan = \kIn^*\thetaKan$
      zu zeigen. Sei $q \in Q$, $\alpha_q \in T_q^*Q$ und
      $v_{\alpha_q} \in T_{\alpha_q}TQ$, dann gilt
      \begin{align*}
         \dPaar{\piKan^*\overline{u}^*\overline\thetaKan\at{\alpha_q}}{v_{\alpha_q}}
         &=
         \dPaar{\overline\thetaKan\at{\overline{u}(\pi(\alpha_q))}}{T_{\alpha_q}(\overline{u}\circ
           \piKan) v_{\alpha_q}}\\
         &=
         \dPaar{\overline{u}(\piKan(\alpha_q))}{T_{\alpha_q}(\overline{\mathfrak{p}}
           \circ
           \overline{u} \circ \piKan) v_{\alpha_q}}\\
         &= \dPaar{\overline{u}(\piKan(\alpha_q))}{T_{\alpha_q}(\varpi
           \circ \mathfrak{p} \circ \kIn ) v_{\alpha_q}}\\
         &= \dPaar{\alpha_q}{T_{\alpha_q} (\mathfrak{p} \circ \kIn) v_{\alpha_q}} =
         \dPaar{\kIn^*\thetaKan}{v_{\alpha_q}} \Fdot
      \end{align*}
      Die Abbildung $\overline{u}$ ist demnach symplektisch und
      insbesondere immersiv, woraus aus Dimensionsgründen und dem Satz
      über die Umkehrfunktion schon folgt, dass $\overline{u}$ ein
      Diffeomorphismus, d.\,h.\ auch ein Symplektomorphismus sein muss.
   \end{beweisEnum}
\end{proof}
\endgroup

\begingroup
\emergencystretch=0.8em In der folgenden Diskussion seien nun die
Abbildungen $u$ und $\overline{u}$ wie in Satz
\ref{satz:ReduktionKotangentenbuendel}. Wir setzen $^{\mathsf{h}} \dpA
\Gamma^\infty(T{(Q/G)}) \to \Gamma^\infty(TQ)$ zu einem Homomorphismus
$^{\mathsf{h}}\dpA \Gamma^\infty(\Bigvee T{(Q/G)}) \to
\Gamma^\infty(\Bigvee TQ)$ bezüglich $\vee$ fort. Insbesondere sei
$\chi^{\mathsf{h}} = \varpi^* \chi$ für $\chi \in \CM[Q/G]$. Schließlich
sei
\begin{align}
   \label{eq:DefintionVonl}
   l \dpA \mathcal{P}(Q/G) \to \mathrm{h}(\mathcal{P}(Q)^G)
\end{align}
durch
\begin{align}
   \label{eq:DefintionVon2}
   l(F) := \mathsf{P}(({\overline{\mathsf{P}}^{-1}(F))^{\mathsf{h}}})
\end{align}
für $F \in \mathcal{P}(Q/G)$ gegeben. Dabei bezeichnet die Abbildung
$\overline{\mathsf{P}} \dpA \Gamma^\infty(\Bigvee T(Q/G)) \to
\mathcal{P}(Q/G)$ den Algebraisomorphismus aus Proposition
\ref{prop:universelleImpulsabbildung} für $Q/G$. Man prüft leicht mit
Hilfe von Proposition \ref{prop:AngepassteBasen} nach, dass $l$ bijektiv
ist, siehe auch \cite[Prop. 3.3 iv)]{kowalzig.neumaier.pflaum:2005a}.

\begin{proposition}
   \label{prop:EinschraenkungFuerKotangentialbuendel}
   Sei $\kIn \colon C_{\mathrm{kan}} \hookrightarrow T^*Q$ die
   kanonische Inklusion, dann gilt
   \begin{align}
      \kIn^*\at{\mathcal{P}(Q)^G[[\lambda]]} = u^* \circ l^{-1} \circ
      \mathrm{h}\at{\mathcal{P}(Q)^G[[\lambda]]} \Fdot
   \end{align}
\end{proposition}
\begin{proof}
   Für $q \in Q$ und $\alpha_q \in C_{\mathrm{kan}} \cap T_q^*Q =
   H^{*^\gamma}_qQ$ ist $\mathsf{P}((e_i)_Q)(\alpha_q) =
   \dPaar{\alpha_q}{(e_i)_Q(q)} = 0$, also folgt für alle $F \in
   \mathcal{P}(Q)^G$, da $\mathsf{P}$ ein Algebrahomomorphismus ist,
   $\kIn^*(F) = \kIn^*(\mathrm{h}(F))$. Als nächstes zeigen wir, dass
   $u^*\at{\mathcal{P}(Q/G)} = \kIn^* \circ l$, d.\,h.\ auch $u^* \circ
   l^{-1} = \kIn^*\at{\mathrm{h}(\mathcal{P}(Q)^G)}$ gilt, woraus dann
   offensichtlich die Behauptung folgt. Dazu bemerken wir zuerst, dass
   sowohl $u^*$ als auch $\kIn^* \circ l$ Algebrahomomorphismen sind,
   womit wir diese Gleichung nur auf Erzeugern prüfen müssen. Sei also
   $\chi \in \CM[Q/G]$, dann gilt für alle $q \in Q$ und $\alpha_q \in
   C_{\mathrm{kan}} \cap T_q^*Q$
\begin{align*}
   u^*(\overline{\mathsf{P}}(\chi))(\alpha_q) &= (u^*\circ{\overline{\mathfrak{p}}}^*
   \chi)(\alpha_q) = \chi((\overline{\mathfrak{p}}  \circ u)(\alpha_q)) =
   \chi(\varpi \circ \mathfrak{p}( \alpha_q))\\ &= (\varpi^*
   \chi)(\mathfrak{p}(\alpha_q)) = \mathfrak{p}^*\varpi^* \chi(\alpha_q)
   = \mathsf{P}(\chi^{\mathsf{h}})(\alpha_q) =
   l(\mathsf{\overline{P}}(\chi))(\alpha_q) \Fdot
\end{align*}
Dabei bezeichnet $\overline{\mathfrak{p}} \colon T^*(Q/G) \to Q/G$ die Fußpunktprojektion.

Ist weiter $\overline{X} \in \Gamma^\infty(T{(Q/G)})$, dann gilt für alle
$q \in Q$ und $\alpha_q \in C_{\mathrm{kan}} \cap T_q^*Q$
\begin{align*}
   u^*(\overline{\mathsf{P}}(\overline{X}))(\alpha_q) &=
   \overline{\mathsf{P}}(\overline{X})(u(\alpha_q)) =
   \dPaar{u(\alpha_q)}{\overline{X}(\varpi(q))}\\ &=
   \dPaar{u(\alpha_q)}{T_q \varpi \overline{X}^{\mathsf{h}}(q)} =
   \dPaar{\alpha_q}{\overline{X}^{\mathsf{h}}(q)} =
   \mathsf{P}(\overline{X}^{\mathsf{h}}(q))(\alpha_q) = l(\overline{\mathsf{P}}(\overline{X}))(\alpha_q) \Fdot
\end{align*}
Da auch $\overline{\mathsf{P}}$ ein Algebramorphismus ist, folgt die
Behauptung sofort.
\end{proof}
\endgroup

Mit Proposition \ref{prop:EinschraenkungFuerKotangentialbuendel} und
Gleichung \eqref{eq:DeltaUndPartial} erhalten wir eine explizite Formel für die
Quanteneinschränkung im Falle von Kotangentialbündeln und den oben
getroffenen Wahlen für die geometrischen Homotopie"=Daten.

\begin{proposition}
   \label{prop:QuantenEinschraenkungFuerKotangentialbauendel}
Für jedes $F \in \mathcal{P}(Q)^G[[\lambda]]$ gilt
\begin{align}
   \label{eq:QuantenEinschraenkungFuerKotangentialbauendel}
   \qRes(F) = \left (u^* \circ l^{-1} \circ \mathrm{h} \circ \left(\frac{\id}{\id - \I
        \lambda \Delta_\star}\right) \right)(F) \Fdot
\end{align}
\end{proposition}
Um das reduzierte Sternprodukt nach Gleichung
\eqref{eq:SternproduktAufDemReduziertenPhasenraum3} genauer bestimmen zu
können, müssen wir auch noch für $\prol \circ \piKan^*$ einen an die
Kotangentialbündelgeometrie angepassteren Ausdruck finden.
\begin{proposition}
    \label{prop:ProlongationFuerKotangentialbuendel}
   Für alle $f \in \overline{u}^*(\mathcal{P}(Q/G))$
   ist die Gleichung
   \begin{align}
      \label{eq:ProlongationFuerKotangentialbuendel}
      (\prol \circ \piKan^*)(f) = (l \circ (\overline{u}^{-1})^*)(f)
   \end{align}
   erfüllt.
\end{proposition}
\begin{proof}
   Wir zeigen $l = \prol \piKan^* \overline{u}^*\at{\mathcal{P}(Q/G)}$.

   Sei $\overline{X} \in
   \Gamma^\infty(T{(Q/G)})$, dann gilt für $q \in Q$ und $\mu \in \lieAlgebra^*$
   \begin{align*}
      \mathsf{P}(\overline{X}^{\mathsf{h}})(\Gamma_\mu(q)) =
      \dPaar{\overline{X}^{\mathsf{h}}(q)}{\Gamma_{\mu}(q)} =
      \dPaar{\mu}{\gamma(\overline{X}^{\mathsf{h}}(q))} = 0 \Fdot
   \end{align*}
   Nun lässt sich jedes $f \in \mathrm{h}(\mathcal{P}(Q)^G)$ als $f =
   \mathsf{P}(\overline{T}^{\mathsf{h}})$ mit $\overline{T} \in
   \Gamma^\infty(\bigvee T(Q/G))$ schreiben.  Da $^{\mathsf{h}}$ und
   $\mathsf{P}$ Homomorphismen bezüglich $\vee$ sind, folgt für jedes $c
   \in C_{\mathrm{kan}}$, $\mu \in \lieAlgebra^*$ und jedes $f \in
   \mathrm{h}(\mathcal{P}(Q)^G)$ die Gleichung
   \begin{align*}
      f(\Psi'(c,\mu)) = f(c + \Gamma_{\mu}(\mathfrak{p}(c))) = f(c)
   \end{align*}

   und somit
   \begin{align*}
       \prol \kIn^* f = f
   \end{align*}
   für alle $f \in \mathrm{h}(\mathcal{P}(Q)^G)$.
   Und somit ergibt sich mit der im Beweis von Proposition
   \ref{prop:EinschraenkungFuerKotangentialbuendel} hergeleiteten
   Beziehung $u^*\circ l^{-1} = \kIn^*\at{\mathrm{h}(\mathcal{P}(Q)^G)}$
   \begin{align*}
      \prol u^* l^{-1} = \id \Fcom
   \end{align*}
   was gleichbedeutend ist mit $l = \prol u^*\at{\mathcal{P}(Q/G)} = \prol \piKan^*
   \overline{u}^*\at{\mathcal{P}(Q/G)}$.
\end{proof}

Trägt man nun all die in diesem Kapitel gemachten Beobachtungen
zusammen, so erhält man eine explizitere Formel für das reduzierte
Sternprodukt.

\begin{satz}
   \label{satz:ReduziertesSternproduktFuerKotangentialBuendel}
   Für das nach Koszul"=Schema reduzierte Sternprodukt
   $\star_{\mathrm{red}}$ gilt für alle $f,f' \in \overline{u}^*(\mathcal{P}(Q/G))$
   \begin{align}
      \label{eq:ReduziertesSternproduktFuerKotangentialBuendel}
      \piKan^*(f \star_{\mathrm{red}} f') = \left(u^* \circ l^{-1} \circ
      \mathrm{h} \circ
      \frac{\id}{\id - \I \lambda \Delta_\star} \right) \left ((l \circ
         {\overline{u}^{-1}}^*)(f)  \star (l \circ {\overline{u}^{-1}}^*)(f')
          \right ) \Fdot
   \end{align}
\end{satz}

Satz \ref{satz:ReduziertesSternproduktFuerKotangentialBuendel} zeigt,
dass die Quanten"=Koszul"=Reduktionsmethode mit den oben getroffenen
geometrischen Wahlen im Spezialfall von Kotangentialbündeln mit
kanonischer symplektischer Form, gelifteter Wirkung und verschwindendem
Impulswert sowie ohne Magnetfeld genau das Sternprodukt ist, das
Kowalzig, Neumaier und Pflaum in \cite{kowalzig.neumaier.pflaum:2005a}
konstruiert haben.

\section{Reduziertes Sternprodukt für nichtverschwindenden Impulswert
  und mit Magnetfeld}
\label{sec:nichtVerschwindendeImpulswerteInAnwesenheitVonMagnetfeldern}

In diesem Abschnitt wollen wir die Betrachtungen des vorangegangenen
verallgemeinern und sowohl nichtverschwindende, $G$"=invariante
Impulswerte $\mu$ als auch $G$"=invariante Magnetfelder $B$
zulassen. Ziel ist es, zu zeigen, dass die Konstruktion von Kowalzig,
Neumaier und Pflaum auch für diese Version einen Spezialfall der
Quanten"=Koszul-Reduktion darstellt. Um dies zu erreichen, müssen wir
spezielle geometrische Daten für die Homotopie wählen, genauer werden
wir solche durch Zurückziehen der geometrischen Homotopie"=Daten aus
Abschnitt \ref{sec:SpezialFallImpulsMagnetfeldNull} erhalten. Dies
betrachten wir zunächst in einer etwas allgemeineren Situation.

\begin{satz}[Zurückziehen geometrischer Homotopie"=Daten]

   \label{satz:ZurueckZiehenVonGeomHomotopieDaten}
   Seien $(M,\omega)$ und $(M',\omega')$ symplektische
   Mannigfaltigkeiten und $G$ eine Lie"=Gruppe, die auf $M$ und $M'$ frei
   und eigentlich wirke, sowie $J \colon M \to \lieAlgebra^*$ und $J'
   \colon M' \to \lieAlgebra^*$, zugehörige $G$"=äquivariante
   Impulsabbildungen mit $C := J^{-1}(0) \neq \emptyset$ und $C' := J'^{-1}(0)
   \neq \emptyset$. Weiter seien $\mathsf{GH'} := (\Psi' \colon U' \to
   V', O' \subset U',{\psi'}_{U'},{\psi'}_{W'},\xi')$ $G$"=invariante
   geometrische Homotopie"=Daten für $C' \subset M'$.  Sei nun $s \colon
   M \to M'$ ein $G$"=äquivarianter Diffeomorphismus mit $J =
   s^*J'$. Wir definieren  $U := s^{-1}(U')$, $O := s^{-1}(O')$, $W
   := M\setminus \abschluss{O}$ und $\psi_U := s^*\psi_U'$ sowie $\psi_W
   := s^*\psi_W'$. Zudem sei $V := (s\at{C}^{-1} \times \id)(V') \subset
   C \times \lieAlgebra^*$ und $\Psi := (s\at{C}^{-1} \times \id) \circ
   \Psi' \circ s\at{U} \colon U \to V$ sowie schließlich $\xi = s^*
   {\xi'}$.

   Dann sind $(\Psi \colon U \to V, O \subset
   U,{\psi}_{U},{\psi}_{W},\xi)$ $G$"=invariante geometrische
   Homotopie"=Daten für $C$ und es gilt für die davon induzierte
   geometrische Homotopie $\h$ die Beziehung
   \begin{align}
   \label{eq:ZurueckGezogeneHomotopie}
      \h = s^* \circ \h' \circ (s^{-1})^* \Fcom
   \end{align}
   wobei $\h'$ die von $\mathsf{GH'}$ induzierte geometrische Homotopie
   bezeichne. Des Weiteren ist für die induzierte geometrische
   Prolongation $\prol$ die Gleichung
   \begin{align}
      \label{eq:ZurueckGezogenesProl}
      \prol = s^* \circ \prol' \circ (s\at{C}^{-1})^*
   \end{align}
   erfüllt, wobei $\prol'$ die von $\mathsf{GH'}$ induzierte
   geometrische Prolongation sei.

   Ist $\Psi'$ global, $O = U' = M'$, also $W' = \emptyset$,
   ${\psi'}_{U'} \equiv 1$ und $\xi' = \emptyset$, so gilt entsprechendes auch
   für die zurückgezogenen geometrischen Homotopie"=Daten.
\end{satz}
\begin{proof}
   Es ist klar, dass $O$ und $U$ offen und $G$"=invariant sind. Da $s$
   ein Homöomorphismus ist, folgt $\abschluss{O} =
   \abschluss{s^{-1}(O')} = s^{-1}(\abschluss{O'}) \subset s^{-1}(U') =
   U$. Weiter sieht man leicht die Inklusion $C = s^{-1}(C')
   \subset O$ und die Beziehung $W = M\setminus \abschluss{O} =
   M\setminus s^{-1}(\abschluss{O'}) = s^{-1}(M'\setminus
   \abschluss{O'}) = s^{-1}(W')$. $\{\psi_{U},\psi_W\}$ ist
   offensichtlich eine glatte, $G$"=invariante Zerlegung der Eins mit
   $\supp \psi_U = s^{-1}(\supp {\psi'}_{U'}) \subset U$ und $\supp
   \psi_W = s^{-1}(\supp{\psi'}_{W'}) \subset W$. Ebenfalls bemerkt man
   sofort, dass $V$ offen und $G$"=invariant ist sowie dass $\Psi$ ein
   $G$"=äquivarianter Diffeomorphismus ist. Nach Konstruktion gilt für $c
   \in C$ die Gleichung
   \begin{align*}
      \Psi(c) = ((s\at{C}^{-1} \times \id))(\Psi'(s(c))) =
      ({s\at{C}}^{-1}(s(c)),0) = (c,0) \Fdot
   \end{align*}
 Des Weiteren
   ist $V$ offenbar sternförmig in Faserrichtung und für $(c,\mu) \in V$
   gilt
   \begin{align*}
      J(\Psi^{-1}(c,\mu)) = J' \circ s({s\at{U}}^{-1} \circ
      \Psi'^{-1} \circ (s\at{C} \times \id)(c,\mu)) =
      J'({\Psi'}^{-1}(s(c),\mu)) = \mu \Fdot
   \end{align*}
   Also ist $\Psi$ eine  $G$"=äquivariante \tn{gute} Tubenabbildung.  Sei  nun
   \begin{align*}
      \h_U \colon \CM[U] \otimes \Bigwedge^\bullet \lieAlgebra \to
      \CM[U] \otimes \Bigwedge^{\bullet + 1} \lieAlgebra
    \end{align*}
    die Homotopie auf $U$, die, wie wir nochmal erinnern möchten, für $k
    \in \mathbb{N}$ und $f \otimes \eta \in C^\infty(U)\otimes
    \Bigwedge^k \lieAlgebra$ folgendermaßen gegeben ist.
    \begin{align*}
       \h_U(f \otimes \eta) \circ \Psi^{-1}(c,\mu) = \int_0^1
       t^k \partial^\alpha(f \circ \Psi^{-1})(c,t \mu) \, dt \otimes
       e_\alpha \wedge \eta \Fdot
    \end{align*}

    Dann erhalten wir für $k \in \mathbb{N}$, $f \otimes \eta \in
    C^\infty(U)\otimes \Bigwedge^k \lieAlgebra$ und $(c,\mu) \in V$
        \begin{align*}
       \lefteqn{((s\at{U})^* \circ \h'_{U'} \circ (s\at{U}^{-1})^*)(f
         \otimes
         \eta)(\Psi^{-1}(c,\mu))} \\
       &= {\h'}_{U'}(f\circ s\at{U}^{-1} \otimes \eta)(s\at{U} \circ
       \Psi^{-1}(c,\mu)) \\
       &= \h'_{U'}(f\circ s\at{U}^{-1} \otimes \eta)({\Psi'}^{-1}
       \circ
       {(s\at{C} \times \id)}(c,\mu)) \\
       &= \int_0^1 t^k \partial^\alpha (f \circ
       s\at{U}^{-1}\circ{\Psi'}^{-1})(s(c),t\mu) \, dt \otimes e_\alpha \wedge \eta \\
       &= \int_0^1 t^k \partial^\alpha (f \circ \Psi^{-1} \circ
       (s\at{C}^{-1} \times \id))(s(c),t\mu) \, dt \otimes e_\alpha \wedge \eta \\
       &= \int_0^1 t^k \partial^\alpha (f \circ \Psi^{-1}
       )((s\at{C}^{-1} \times \id)(s(c),t\mu)) \, dt \otimes e_\alpha \wedge \eta \\
       &= \h_U(f \otimes \eta)(\Psi^{-1}(c,\mu)) \Fdot
     \end{align*}
     Zudem
   gilt
    \begin{align*}
       \dPaar{J}{\xi} = \dPaar{s^*J'}{s^*\xi'} = s^*\psi_W = \psi_W
    \end{align*}
    und
    \begin{align*}
       \supp \xi = s^{-1}(\supp \xi') \subset W \Fdot
    \end{align*}
    Sei nun $p \in \supp({\psi}_{U}) \cap W$, dann ist offenbar auch
    $s(p) \in \supp({\psi'}_{U'}) \cap W'$ und es folgt
     \begin{align*}
        \xi(p) = \xi'(s(p))&= -{\h'}_{U'}({\psi'_{U'}\at{U'}})(s(p)) \\
        &= -((s\at{U}^{-1})^* \circ \h_{U} \circ
        s\at{U}^*({\psi'_{U'}\at{U'}}))(s(p)) = -\h_U(\psi_U\at{U})(p)
        \Fcom
     \end{align*}
     d.\,h.\
     \begin{align*}
        \xi\at{{\supp(\psi_U) \cap W}} =
        -\h_U(\psi_U\at{U})\at{{\supp(\psi_U) \cap W}} \Fdot
     \end{align*}
     Weiter gilt für $\phi \in \CM \otimes \Bigwedge^k \lieAlgebra^*$

     \begin{align*}
        s^* \circ \h_{W'}' \circ (s^{-1})^* (\phi) = s^*(\xi' \wedge
        (s^{-1})^*(\phi)) = s^*\xi' \wedge s^*(s^{-1})^*\phi = \xi
        \wedge \phi = {\h}_{W}(\phi) \Fcom
     \end{align*}
     folglich auch
     \begin{align*}
        s^* \circ \h' \circ (s^{-1})^*(\phi) &= s^*(\psi'_{U'}
        \h_{U'}'(((s^{-1})^*\phi)\at{U'})) +
        s^*\h_{W'}' (s^{-1})^* \phi \\
        &= \psi_U \h_U(\phi\at{U}) + \h_{W}(\phi) \\
        &= \h(\phi)\Fdot
     \end{align*}
Schließlich sehen wir für alle $f \in C^\infty(C)$ und $(c,\mu) \in V$
\begin{align*}
   (s^* \circ {\prol'} \circ ({s\at{C}}^{-1})^*)(f)(\Psi^{-1}(c,\mu) )
   &=
   ({\prol'}((s\at{C}^{-1})^*f))({\Psi'}^{-1}(s(c),\mu)) \\
   &=
   \psi'_{U'}({\Psi'}^{-1}(s(c),\mu)) \cdot f((s\at{C})^{-1}(s(c)))
   \\
   &=
     \psi'_{U'}(s(\Psi^{-1}(c,\mu))) \cdot f(c) \\
&= \psi_U ({\Psi}^{-1}(c,\mu)) \cdot f(c)\\
   &= \prol(f)({\Psi}^{-1}(c,\mu)) \Fdot
\end{align*}
\end{proof}

\begin{bemerkung}
\label{bem:Zurueckziehen}
Sind in Satz \ref{satz:ZurueckZiehenVonGeomHomotopieDaten} $\kIn \colon
J^{-1}(0) \hookrightarrow M$ und $\kIn' \colon J'^{-1}(0)
\hookrightarrow M'$ die
Inklusionen der Zwangsflächen in den Phasenraum, so gilt offensichtlich
\begin{align}
   \label{eq:InklusionenZurueckZiehen}
   \kIn^* = s\at{C}^* \circ \kIn'^* \circ (s^{-1})^* \Fdot
\end{align}
\end{bemerkung}

Wir betrachten nun die in Abschnitt
\ref{sec:KotangentialBuendelMitMagnetfeld} beschriebene Situation in
voller Allgemeinheit. Insbesondere sei $\mu \in \lieAlgebra^*$ ein
$G$"=invarianter Impulswert und  $B \in \Gamma^\infty(\bigwedge^2 T^*Q)$
$G$"=invariant sowie geschlossen.

Da wir  einen im Allgemeinen nichtverschwindenden $G$"=invarianten
Impulswert $\mu$ vorliegen haben, müssen wir die in Abschnitt
\ref{sec:nichtVerschwindendeImpulswerte} dargelegten Überlegungen
anwenden. Wir definieren also $J^\mu_B := J_B - \mu$ und $\qJ[B]^\mu :=
\qJ[B] - \mu$. Die davon induzierten Koszul-
bzw.\ Quanten"=Koszul"=Operatoren wollen wir mit $\partial^\mu_B$
bzw.\ $\qkoszul^\mu_B$ bezeichnen. Die Inklusion der Zwangsfläche
$C_B^\mu := {J^\mu_B}^{-1}(0)$ notieren wir als $\kIn_B^\mu \colon
C_B^\mu \hookrightarrow T^*Q$ und die kanonische Projektion auf den reduzierten
Phasenraum als $\pi^\mu_B \colon T^*Q \to C^\mu_B/G$. Schließlich
schreiben wir $(\qRes)_B^\mu$ für die Quanteneinschränkung.

Ist $\alpha \in \Gamma^\infty(T^*Q)$, so definieren wir die
\neuerBegriff{Fasertranslation} $\Faser{\alpha}\colon T^*Q \to T^*Q$
durch $\Faser{\alpha}(\zeta_q) := \zeta_q + \alpha(q)$
(vgl.\ \cite[Def. 3.2.14,~ii.)]{waldmann:2007a}). Dies ist offensichtlich
ein Diffeomorphismus von $T^*Q$ mit inverser Abbildung
$(\Faser{\alpha})^{-1} = \Faser{-\alpha}$ und es gilt
$\Faser{\alpha}^*(\mathcal{P}(Q)) \subset \mathcal{P}(Q)$.  Sofern
$\alpha$ $G$"=invariant ist, folgt, wie man leicht nachrechnet, die
$G$"=Äquivarianz von $\Faser{\alpha}$ und damit unmittelbar die von
$\Faser{\alpha}^*$. Ist weiter $j \colon Q \to \lieAlgebra^*$ eine
glatte Abbildung, so definieren wir die Einsform $\Gamma_{j} \in
\Gamma^\infty(T^*Q)$ durch $\Gamma_{j} := \dPaar{j}{\gamma}$. Für ein
$G$"=äquivariantes $j$ ist, wie eine einfache Rechnung zeigt,
$\Gamma_{j}$ $G$"=invariant. Da nun $\mu \in \lieAlgebra^*$ ein
$G$"=invarianter Impulswert ist sowie $j_0 \colon Q \to \lieAlgebra^*$
eine $G$"=äquivariante Abbildung, sieht man nach obiger Anmerkung
direkt, dass auch $\FaserG{j_0 - \mu}$ $G$"=äquivariant ist.

\begin{proposition}
   \label{lem:sBeiUns}
   Es gilt $\FaserG{j_0 - \mu}^* \JKan = \JB[\mu]$.
\end{proposition}
\begin{proof}
   Sei $\xi \in \lieAlgebra$, $q \in Q$ und $\alpha_q \in T^*_qQ$, dann gilt
   \begin{align*}
      (\FaserG{j_0 - \mu}^*( \JKan(\xi)))(\alpha_q) &=
      \JKan(\xi)(\FaserG{j_0 - \mu}(\alpha_q)) \\
      &= \JKan(\xi)(\alpha_q + \Gamma_{j_0 - \mu}(q)) \\
      &= \dPaar{\xi_Q(q)}{\alpha_q} + \dPaar{\xi_Q(q)}{\Gamma_{j_0 -
          \mu}(q)}\\
      &= \JKan(\xi)(\alpha_q) + \dPaar{j_0(q) -
        \mu}{\gamma(\xi_Q(q))}\\
      &= \JKan(\xi)(\alpha_q) + \dPaar{j_0(q) - \mu}{\xi}\\
      &= \JB[\mu](\xi)(\alpha_q) \Fdot
   \end{align*}
\end{proof}

\begin{proposition}
   \label{prop:ReduktionKotangentialbuendelAllgemein}

         Es gilt $\FaserG{j_0 - \mu}(C_B^\mu) \subset
         C_{\mathrm{kan}}$ und $\FaserG{j_0 - \mu}\at{C^\mu_B}
         \colon C^\mu_B \to C_{\mathrm{kan}}$ ist ein
         $G$"=äquivarianter Diffeomorphismus, der einen eindeutig
         bestimmten Diffeomorphismus ${\oFaserG{j_0 - \mu}}
         \colon C^\mu_B/G \to C_{\mathrm{kan}}/G$
         induziert, so dass das  Diagramm

         \def\tA[#1]{A_{#1}}
\begin{equation}
   \label{eq:ReduktionKotangentialbuendelAllgemein}
   \begin{tikzpicture}[baseline=(current
      bounding box.center),description/.style={fill=white,inner sep=2pt}]
      \matrix (m) [matrix of math nodes, row sep=3.0em, column
      sep=3.5em, text height=1.5ex, text depth=0.25ex]
      {
        C^\mu_B & C_{\mathrm{kan}}  \\
         C^\mu_B/G & C_{\mathrm{kan}}/G   \\
      }; %

      \path[->] (m-1-1) edge
      node[auto]{$\FaserG{j_0 - \mu}\at{C^\mu_B}$}(m-1-2); %
      \path[->] (m-1-1) edge node[auto]{$\pi^\mu_B$}(m-2-1); %
      \path[->] (m-2-1) edge node[auto]{$\oFaserG{j_0 - \mu}$}(m-2-2); %
      \path[->] (m-1-2) edge node[auto]{$\piKan$}(m-2-2); %
   \end{tikzpicture}
\end{equation}
kommutiert.
\end{proposition}
\begin{proof}

Die erste Aussage folgt direkt aus Proposition \ref{lem:sBeiUns}, der Rest ist klar.
\end{proof}

\begin{bemerkung}
   \label{bem:KotangentialbuendelreduktionSymplekto}
   Man kann weiter noch das Folgende zeigen, wodurch die Struktur des
   reduzierten Phasenraums geklärt wird. Er hat nämlich wieder die
   Struktur eines Kotangentialbündels, mit einer symplektischen Form, die
   Magnetfeldterme aufweist, genauer gilt:
   \begin{bemerkungEnum}
   \item %
      Es gibt eine eindeutig bestimmte geschlossene Zweiform $b$ auf
      $Q/G$ mit $\varpi^* b = B + d\Gamma_{j_0 - \mu}$.
   \item %
      Der reduzierte Phasenraum $((T^*Q)_{\mathrm{red},\mu} =
      C^\mu_B/G,\omega_{\mathrm{red},\mu})$ ist symplektomorph zu
      $(T^*(Q/G), ({\omegaKan})_{\mathrm{red}} +
      \overline{\mathfrak{p}}^* b)$ vermöge der Abbildung $\overline{u}
      \circ {\oFaserG{j_0 - \mu}} \colon C^\mu_B/G \to T^*(Q/G)$, dabei
      bezeichnet $\overline{\mathfrak{p}} \colon T^*(Q/G) \to Q/G$ die
      Fußpunktprojektion.
         \begin{equation}   \label{eq:ReduktionKotangentialbuendelAllgemein2}
   \begin{tikzpicture}[baseline=(current
      bounding box.center),description/.style={fill=white,inner sep=2pt}]
      \matrix (m) [matrix of math nodes, row sep=3.0em, column
      sep=3.5em, text height=1.5ex, text depth=0.25ex]
      {
         C^\mu_B & C_{\mathrm{kan}} &  \\
         C^\mu_B/G & C_{\mathrm{kan}}/G & T^*(Q/G)  \\
      }; %

      \path[->] (m-1-1) edge
      node[auto]{$\FaserG{j_0 - \mu}\at{C^\mu_B}$}(m-1-2); %
      \path[->] (m-1-1) edge node[auto]{$\pi^\mu_B$}(m-2-1); %
      \path[->] (m-2-1) edge node[auto]{${\oFaserG{j_0 - \mu}}$}(m-2-2); %
      \path[->] (m-1-2) edge node[auto]{$\piKan$}(m-2-2); %
      \path[->] (m-1-2) edge node[auto]{$u$}(m-2-3); %
      \path[->] (m-2-2) edge node[auto]{$\overline{u}$}(m-2-3); %
   \end{tikzpicture}
\end{equation}
Diese Aussage ist für die Konstruktion von Kowalzig, Neumaier und Pflaum
\cite{kowalzig.neumaier.pflaum:2005a} wichtig, für den von uns
angestrebten Vergleich mit dieser Arbeit spielt sie jedoch nur eine
untergeordnete Rolle. Für einen Beweis dieser Tatsache verweisen wir auf
\cite[Thm. 3.2]{kowalzig.neumaier.pflaum:2005a} sowie auf die dort
angegebenen Referenzen.
   \end{bemerkungEnum}
\end{bemerkung}

Wir wählen nun für $C_{\mathrm{kan}}$ die globale Tubenumgebung aus
Proposition \ref{prop:globaleTubenabbildung}, die davon induzierte
geometrische $G$"=äquivariante Homotopie sei wieder mit $\hKan$
bezeichnet und die induzierte $G$"=äquivariante Prolongation mit $\prolKan$.

Mit Proposition \ref{lem:sBeiUns} ist klar, dass mit den Wahlen $M = M' = T^*Q$, $J =
\JB[\mu]$, $J' = \JKan$ und $s = \FaserG{j_0 - \mu} \colon T^*Q \to
T^*Q$ die Voraussetzungen von Satz
\ref{satz:ZurueckZiehenVonGeomHomotopieDaten} erfüllt sind. Wir
bezeichnen die so induzierte geometrische Homotopie mit $\hB[\mu]$
und die induzierte geometrische Prolongation mit $\prolB[\mu]$.

Um den Vergleich mit der Arbeit \cite{kowalzig.neumaier.pflaum:2005a} zu
erleichtern, wollen wir zunächst einige Symbole definieren. Dazu sei
$\{e_i\}$ eine fest gewählte Basis von $\lieAlgebra$.
   \begin{align}
      \label{eq:SymboleVonNikolai}
      \mathrm{h}_{j_0 - \mu} &:=  \mathrm{h}
      \circ \FaserG[-]{j_0 - \mu}^*\at{\mathcal{P}(Q)[[\lambda]]} \colon
      \mathcal{P}(Q)[[\lambda]] \to \mathcal{P}(Q)[[\lambda]] \notag \\
      \mathrm{r}^i_{j_0 - \mu} &:= \FaserG{j_0 - \mu}^* \circ
      \mathrm{r}^i \circ \FaserG[-]{j_0 -
        \mu}^*\at{\mathcal{P}(Q)[[\lambda]]} \colon
      \mathcal{P}(Q)[[\lambda]] \to \mathcal{P}(Q)[[\lambda]] \quad
      \forall i \in \{1,\dots,\dim G\} \notag \\
      \Delta_{\mu,\star}(F) &:= \frac{1}{\I \lambda} \sum_{i=1}^{\dim G}
      \left(\mathrm{r}^i_{j_0 - \mu}(F)\JB[\mu](e_i) - \mathrm{r}^i_{j_0
           - \mu}(F) \star \qJ[B]^\mu(e_i)\right) \in \mathcal{P}(Q)[[\lambda]] \quad \forall F \in
      \mathcal{P}(Q)[[\lambda]]
   \end{align}
   Man beachte, dass oben eingeführte Abbildungen wegen
   $\FaserG{j_0 - \mu}^*(\mathcal{P}(Q)[[\lambda]]) \subset
   \mathcal{P}(Q)[[\lambda]]$ und aufgrund der Bemerkungen nach Gleichung
   \eqref{eq:DeltaNikolai} wohldefiniert sind.

Wie in Abschnitt \ref{sec:SpezialFallImpulsMagnetfeldNull} wollen wir
nun unter Ausnutzung der dortigen Resultate die Bausteine des nach der
Quanten"=Koszul"=Methode konstruierten Sternprodukts auf dem
reduzierten Phasenraum in an die Kotangentialbündelsituation angepassten
Termen schreiben, um dann einen Vergleich mit
\cite{kowalzig.neumaier.pflaum:2005a} durchführen zu können. Mit Hilfe
von Satz~\ref{satz:ZurueckZiehenVonGeomHomotopieDaten} erhalten wir dabei
die folgenden Propositionen.

\begin{proposition}
   \label{prop:HomotopieMitMagnetfeld}
   Für alle $F \in \mathcal{P}(Q)$ gilt
   \begin{align}
      \hB[\mu](F) =  \mathrm{r}^i_{j_0 - \mu}(F) \otimes e_i \Fdot
   \end{align}
\end{proposition}
\begin{proof}
   Mit Satz \ref{satz:ZurueckZiehenVonGeomHomotopieDaten} und Lemma
   \ref{lem:HomotopiePolynomial} folgt sofort für $F \in \mathcal{P}(Q)$
   \begin{align*}
      \hB[\mu](F) = (\FaserG{j_0 - \mu}^* \circ \hKan \circ \FaserG[-]{j_0 -
        \mu}^*)(F) = \FaserG{j_0 - \mu}^*  \mathrm{r}^i
      (\FaserG[-]{j_0 - \mu}^*(F)) \otimes e_i = \mathrm{r}^i_{j_0 -\mu}(F) \otimes
      e_i \Fdot
   \end{align*}
\end{proof}

Damit ist klar, dass für $F \in \mathcal{P}(Q)$ die Gleichung
\begin{align}
   \label{eq:DeltaMitMagnetfeld}
   \Delta_{\mu,\star}(F) = \frac{1}{\I \lambda}(\kKoszulB[\mu] -
   \qKoszulB[\mu])(\hB[\mu](F))
\end{align}
gilt.

\begin{proposition}
   \label{prop:klassischeEinschraenkungMitMagnetfeld}
Für jedes $F \in \mathcal{P}(Q)^G$ ist die Gleichung
\begin{align}
\label{eq:klassischeEinschraenkungMitMagnetfeld}
{\kInB[\mu]}^*(F) = (\piB[\mu]^* \circ (\overline{u} \circ {\oFaserG{j_0
    - \mu}})^* \circ l^{-1} \circ \mathrm{h}_{j_0 - \mu})(F)
\end{align}
richtig.
\end{proposition}
\begin{proof}
   Sei $F \in \mathcal{P}(Q)^G$, dann erhält man mit Gleichung
   \eqref{eq:InklusionenZurueckZiehen}, Proposition
   \ref{prop:EinschraenkungFuerKotangentialbuendel} und Diagramm \eqref{eq:ReduktionKotangentialbuendelAllgemein2}
   \begin{align*}
      {\kInB[\mu]}^*(F) &= ((\FaserG{j_0 - \mu}\at{C_B^\mu})^* \circ
      \kInKan^* \circ \FaserG[-]{j_0 -
        \mu}^*)(F)\\
      &= ((\FaserG{j_0 - \mu}\at{C_B^\mu})^*\circ u^* \circ l^{-1} \circ
      \mathrm{h} \circ \FaserG[-]{j_0 -
        \mu}^*)(F)\\
      &= (\piB[\mu]^* \circ (\overline{u} \circ {\oFaserG{j_0 - \mu}})^*
      \circ l^{-1} \circ \mathrm{h}_{j_0 - \mu})(F)\Fdot
   \end{align*}
\end{proof}

\begin{proposition}
   \label{lem:l}
   Sei $j \colon Q \to \lieAlgebra^*$ glatt, dann gilt für jedes $F \in
   \mathrm{h}(\mathcal{P}(Q))$ schon $\FaserG{j}^*(F) = F$
   und somit auch $\FaserG{j}^* \circ l = l$.
\end{proposition}
\begin{proof}
   Sei $Y \in \Gamma^\infty(HQ)$. Dann gilt für $q \in Q$ und $\alpha_q
   \in T_q^*Q$
   \begin{align*}
      (\FaserG{j}^*\mathsf{P}(Y))(\alpha_q) &=
      \mathsf{P}(Y)(\alpha_q + \Gamma_{j}(q)) = \mathsf{P}(Y)(\alpha_q)
      + \dPaar{Y(q)}{\Gamma_{j}(q)}\\ &= \mathsf{P}(Y)(\alpha_q) +
      \dPaar{\gamma(Y(q))}{j(q)} = \mathsf{P}(Y)(\alpha_q) \Fdot
   \end{align*}
   Da auch $\FaserG{j}^*(\mathsf{P}(\chi)) = \mathsf{P}(\chi)$ für jedes
   $\chi \in C^\infty(Q)$ und da $\FaserG{j}^*$ ein
   Algebrahomomorphismus ist, folgt die Aussage.
\end{proof}

\begin{proposition}
   \label{prop:PropPiMitMagnetfeld}
   Für alle $f \in \oFaserG{j_0 -
     \mu}^*(\overline{u}^*(\mathcal{P}(Q/G)))$ ist
   \begin{align}
      \label{eq:PropPiMitMagnetfeld}
      (\prolB[\mu] \circ \piB[\mu]^*)(f) = (l \circ ({\overline{u}^{-1}}^*)
      \circ {\oFaserG[-]{j_0 - \mu}^*})(f) \Fdot
   \end{align}
\end{proposition}
\begin{proof}
   Sei $f \in \oFaserG{j_0 - \mu}^*(\overline{u}^*(\mathcal{P}(Q/G)))$,
   dann gilt
    \begin{align*}
       (l \circ ({\overline{u}^{-1}}^*) \circ {\oFaserG[-]{j_0 -
           \mu}^*})(f) &= (\prolKan \circ \pi_{\mathrm{kan}}^* \circ
       {\oFaserG[-]{j_0 -
           \mu}^*})(f) \\
       &= (\prolKan \circ (\FaserG[-]{j_0 - \mu}\at{C_{\mathrm{kan}}})^*
       \circ
       \piB[\mu]^*)(f)\\
       &= (\FaserG[-]{j_0 - \mu}^* \circ \prolB[\mu] \circ
       \piB[\mu]^*)(f) \Fdot
    \end{align*}
Dabei wurde im ersten Schritt Proposition
\ref{prop:ProlongationFuerKotangentialbuendel} verwendet, im zweiten
Diagramm \eqref{eq:ReduktionKotangentialbuendelAllgemein} und im letzten
Schritt Gleichung \eqref{eq:ZurueckGezogenesProl}.
Nun gilt nach Proposition \ref{lem:l} $\FaserG{j_0 - \mu}^* \circ l = l$, woraus die Behauptung folgt.
\end{proof}

Durch Zusammenfügen der Resultate dieses Abschnittes erhält man
schließlich unter Beachtung der Injektivität von $\piB[\mu]^*$ für alle
$f,f' \in \oFaserG{j_0 - \mu}^*(\overline{u}^*(\mathcal{P}(Q/G)))$
folgende Formel für das reduzierte Sternprodukt \eqref{eq:SternproduktAufDemReduziertenPhasenraum3}:

   \begin{align}
    f \star_{\mathrm{red},\mu} f' = \left ((\oFaserG{j_0 - \mu}^*
       \overline{u}^*) l^{-1} \mathrm{h}_{j_0 - \mu} \frac{\id}{\id - \I
         \lambda\Delta_{\mu,\star}}\right)\left (l(\oFaserG{j_0 - \mu}^*
       \overline{u}^*)^{-1} f \star l(\oFaserG{j_0 - \mu}^*
       \overline{u}^*)^{-1} f'\right ) \Fdot
    \end{align}
Dies führt uns unmittelbar auf das Hauptresultat dieses Kapitels.
\begin{satz}
   \label{satz:VergleichMitNikolaisArbeitAllgemein}

   Unter den in Abschnitt \ref{sec:KotangentialBuendelMitMagnetfeld}
   erklärten Voraussetzungen und den in Abschnitt
   \ref{sec:SpezialFallImpulsMagnetfeldNull} und
   \ref{sec:nichtVerschwindendeImpulswerteInAnwesenheitVonMagnetfeldern}
   eingeführten Bezeichnungen gilt für das durch den Symplektomorphismus
   $(\overline{u} \circ {\oFaserG{j_0 - \mu}})^{-1} \colon
   T^*(Q/G) \to C_B^\mu/G$ zurückgezogene Sternprodukt $\tilde
   \star_{\mathrm{red},\mu} := ((\overline{u} \circ {\oFaserG{j_0 -
       \mu}})^{-1})^*\star_{\mathrm{red},\mu} $ und alle $f,f' \in \mathcal{P}(Q/G)$ die
   Beziehung
   \begin{align}
      \label{eq:VergleichMitNikolaiAllgemein}
      f \tilde \star_{\mathrm{red},\mu} f' = \left ( l^{-1} \mathrm{h}_{j_0 - \mu} \frac{\id}{\id - \I
         \lambda\Delta_{\mu,\star}}\right)\left (l (f) \star l (f')\right ).
   \end{align}
\end{satz}
Dies bedeutet, dass auch für den Fall mit Magnetfeldern und
$G$"=invarianten Impulswerten ungleich Null die in
\cite[Thm. 3.5]{kowalzig.neumaier.pflaum:2005a} konstruierten
Sternprodukte einen Spezialfall des hier betrachteten
Quanten-Koszul-Reduktions"=Schemas darstellen.  Kowalzig, Neumaier und
Pflaum erlauben in ihrer Arbeit noch zusätzlich Quantenimpulswerte, was
jedoch in die Quantenimpulsabbildung absorbiert werden kann und somit an
dieser Stelle keine weitere Verallgemeinerung bedeutet.

Zusammengefasst konnten wir zeigen, dass die von Kowalzig, Neumaier und
Pflaum in \cite{kowalzig.neumaier.pflaum:2005a} vorgeschlagene Methode,
im Falle von Kotangentialbündeln mit Magnetfeldern ein Sternprodukt auf
dem reduzierten Phasenraum zu konstruieren, ein Spezialfall des
allgemeineren Quanten"=Koszul"=Reduktions"=Schemas ist. Genauer ergibt
sich, wie wir gesehen haben, die zunächst sehr technisch erscheinende
Methode aus \cite{kowalzig.neumaier.pflaum:2005a} durch naheliegende
Wahlen aus dem natürlicher wirkenden Quanten"=Koszul-Reduktions"=Schema
Insbesondere handelt es sich auch um einen Spezialfall der in
\cite{bordemann.herbig.waldmann:2000a} vorgestellten
BRST"=Reduktionsmethode. Die Bedeutung der Methode von Kowalzig,
Neumaier und Pflaum besteht sicherlich unter anderem darin, dass in
diesem Rahmen Resultate zur Frage, wann Quantisierung mit Reduktion
vertauscht, gezeigt werden konnten. Die Quanten"=Koszul-Methode besticht
hingegen durch ihre konzeptionelle Klarheit. Somit wagen wir zu
behaupten, eine wichtige Brücke zwischen zwei Methoden geschlagen zu
haben.

\chapter{Koszul-Reduktion in Stufen}
\label{cha:Koszul-Reduktion_in_Stufen}

In diesem Kapitel beschäftigen wir uns mit der Frage, ob man sowohl im
klassischen als auch im Quanten"=Fall die in Kapitel
\ref{cha:Koszul-Reduktion} vorgestellten Methoden der
Phasenraumreduktion auch in zwei Stufen durchführen kann, und wie die so
gewonnenen Observablenalgebren zu denen in Beziehung stehen, die man
erhält, wenn man die vorhandenen Symmetrien auf einmal eliminiert.

Die klassische Situation wurde in der Literatur schon eingehend
untersucht, siehe etwa \cite[Ch. 6.7]{ortega.ratiu:2004} sowie
\cite{marsden.misiolek:2007a} für eine ausführliche
Darstellung. Betrachtet man Phasenraumreduktion nur für verschwindende
Impulswerte, so stellt sich heraus, dass sich die Observablenalgebra
beziehungsweise der Phasenraum, den man durch Reduktion nach der
gesamten Symmetriegruppe gewinnt, von dem, den man durch Reduktion in
zwei Stufen enthält, nur durch einen Isomorphismus unterscheidet. Für
nicht verschwindende Impulswerte wird die Situation wesentlich
komplizierter und man kann unter Annahme einiger technischer Bedingungen
an die Impulswerte (vgl.\ \cite{marsden.misiolek:2007a}) ein analoges
Resultat zeigen.  Wir benötigen jedoch nur die wesentlich einfachere
Theorie für Impulswert $0$, welche wir in Abschnitt
\ref{sec:symplektischeReduktionInStufen} genauer darstellen werden.

Für die Quanten"=Koszul-Reduktion wurde die oben angesprochene Frage in
der Literatur noch nicht untersucht. Es ist eines der Hauptresultate
dieser Arbeit, dass unter bestimmten Annahmen an die Struktur der
Symmetriegruppe auch bei der Quanten"=Koszul-Reduktion die
Observablenalgebra, die man durch Reduktion in zwei Stufen erhält, zu
der in einer Stufe gewonnenen isomorph ist. Dies werden wir  in
\ref{sec:QuantenKoszulInStufen} behandeln.

Die in diesem Kapitel im Fließtext eingeführten Bezeichnungen gelten für
den Rest dieser Arbeit.

\section{Symplektische Reduktion in Stufen}
\label{sec:symplektischeReduktionInStufen}

Sei $(M,\omega)$ eine symplektische Mannigfaltigkeit, $G$ eine
zusammenhängende Lie"=Gruppe, die via $\Phi^{\scriptscriptstyle{G}} \dpA
G \times M \to M$ stark Hamiltonsch, frei und eigentlich auf $M$
operiere und $\gIn \dpA G_1 \hookrightarrow G$ eine normale,
zusammenhängende Unter"=Lie-Gruppe von $G$; insbesondere ist $G_1$ dann
abgeschlossen in $G$ (vgl.\ \cite[Prop. 5.4.1]{lee:2003a}). Man bemerke,
dass dies via $\Phi^{\scriptscriptstyle{G_1}}(g_1,p) :=
\Phi^{\scriptscriptstyle{G}}(g_1,p)$ für alle $g_1 \in G_1$, $p \in M$
eine symplektische, freie und eigentliche Wirkung
$\Phi^{\scriptscriptstyle{G_1}} \colon G_1 \times M \to M$ von $G_1$ auf
$M$ induziert. Sei weiter $J \dpA M \to \lieAlgebra[g]^*$  eine  $G$"=äquivariante Impulsabbildung.

Es stellt sich dann die naheliegende Frage, ob es egal ist, ob man
Marsden"=Weinstein-Reduktion bezüglich der großen Gruppe $G$ betreibt
oder in einem ersten Schritt bezüglich $G_1$ und in einem zweiten
bezüglich $G/G_1$. Dies ist tatsächlich der Fall, wie wir in diesem
Abschnitt in Satz \ref{satz:redPhasenraeumeSymplektomorph} für
verschwindende Impulswerte positiv präsentieren. Zunächst stellen wir
jedoch den detaillierten Rahmen des Problems bereit. Dazu überlegen wir
uns als erstes, dass die adjungierte Wirkung bezüglich $G$ eine
$G$"=Wirkung auf $\lieAlgebra_1^*$ induziert. Dies benötigt man, um zu
sehen, dass die für den ersten Reduktionsschritt im Anschluss daran
definierte Impulsabbildung $J_1$ sogar $G$"=äquivariant ist, womit auch
die durch $J_1$ definierte Zwangsfläche $G$"=invariant sein wird.

Da $G_1 \subset G$ normal ist, induziert für jedes $g \in G$ die
Abbildung $\conjG_g \dpA G \to G$, $h \mapsto ghg^{-1}$ eine eindeutige
glatte Abbildung $\tconjG_{g}\dpA G_1 \to G_1$, so dass das Diagramm
\begin{center}
   \begin{tikzpicture}[description/.style={fill=white,inner sep=2pt}]
      \matrix (m) [matrix of math nodes, row sep=3.0em, column
      sep=3.5em, text height=1.5ex, text depth=0.25ex] %
      { G & G \\
        G_1 & G_1 \\}; %
      \path[right hook->] (m-2-1) edge node[auto] {\gIn} (m-1-1);%
      \path[right hook->] (m-2-2) edge node[auto] {\gIn} (m-1-2);%
      \path[->] (m-1-1) edge node[above=-2.5pt] {\small $\conjG_g$}
      (m-1-2);%
      \path[->] (m-2-1) edge node[above=-2.5pt] {\small $\tconjG_g$}
      (m-2-2);%
   \end{tikzpicture}
\end{center} kommutiert.
Wegen $\AdG_g = T_e \conjG_g$ kommutieren mit $\tAdG_g := T_e \tconjG_g$
aufgrund der Kettenregel auch die folgenden beiden Diagramme.
\begin{center}
   \begin{tikzpicture}[description/.style={fill=white,inner sep=2pt}]
      \hspace*{-3cm} \matrix (m) [matrix of math nodes, row sep=3.0em,
      column sep=3.5em, text height=1.5ex, text depth=0.25ex]
      { \lieAlgebra[g] & \lieAlgebra[g] \\
        \lieAlgebra[g]_1 & \lieAlgebra[g]_1\\};%
      \path[right hook->] (m-2-1) edge node[auto] {$T_e \gIn$} (m-1-1);%
      \path[right hook->] (m-2-2) edge node[right] {$T_e \gIn$}
      (m-1-2);%
      \path[->] (m-1-1) edge node[above=-2.5pt] {\small $\AdG_g$}
      (m-1-2);%
      \path[->] (m-2-1) edge node[above=-2.5pt] {\small $\tAdG_g$}
      (m-2-2);%
      \hspace{6cm} %
      \matrix (n) [matrix of math nodes, row sep=3.0em, column
      sep=3.5em, text height=1.5ex, text depth=0.25ex]
      { \lieAlgebra[g]^* & \lieAlgebra[g]^* \\
        \lieAlgebra[g]_1^* & \lieAlgebra[g]_1^* \\};%
      \path[->] (n-1-1) edge node[left] {$(T_e\gIn)^*$} (n-2-1);%
      \path[->] (n-1-2) edge node[right] {$(T_e\gIn)^*$} (n-2-2);%
      \path[->] (n-1-1) edge node[above=-2.5pt] {\small$(\AdG_g)^*$}
      (n-1-2);%
      \path[->] (n-2-1) edge node[above=-2.5pt] {\small$(\tAdG_g)^*$}
      (n-2-2);%
   \end{tikzpicture}
\end{center}
Damit ist  unmittelbar klar, dass $\tAdG \colon G \times \lieAlgebra[g]_1 \to
\lieAlgebra[g]_1$, $(g,\xi_1) \mapsto \tAdG_g\xi_1$ und $G \times
\lieAlgebra[g]^*_1 \to \lieAlgebra[g]^*_1$, $(g,\mu_1) \mapsto
(\tAdG_{g^{-1}})^*\mu_1$ lineare $G$"=Wirkungen auf $\lieAlgebra[g]_1$
bzw.\ $\lieAlgebra[g]^*_1$ sind. Ferner gilt offensichtlich für $g_1 \in
G_1$ per Definition $\tAdG_{g_1} = \AdG[G_1]_{g_1}$, also auch
$(\tAdG_{g_1})^* = (\AdG[G_1]_{g_1})^*$.

Wir haben demnach auf natürliche Weise eine $G$"=Wirkung auf
$\lieAlgebra[g]^*_1$ vorliegen, bezüglich der $(T_e\gIn)^*$
$G$"=äquivariant wird. Da diese auf $G_1 \times \lieAlgebra[g]^*_1$ mit
der koadjungierten Wirkung $(\AdG[G_1]_{\cdot^{-1}})^*$ bezüglich $G_1$ übereinstimmt, ist der im Folgenden oft verwendete Notationsmissbrauch gerechtfertigt, für $g \in G$ und $\mu_1 \in \lieAlgebra[g]_1^*$ statt
$(\tAdG_{g^{-1}})^*\mu_1$ auch hier
wieder einfach $g \mu_1$ zu schreiben.

Man sieht nun sofort, dass
\begin{align}
   J_1 := (T_e \gIn)^* \circ J \dpA M \to
   \lieAlgebra[g]_1^*\label{eq:klassischesJ1} \Fcom
\end{align}
d.\,h.\  $\dPaar{J_1(p)}{\xi_1} := \dPaar{J(p)}{T_e \gIn
  \xi_1}$ für $p \in M$, $\xi_1 \in \lieAlgebra[g]_1$ eine Impulsabbildung
für die $G_1$"=Wirkung definiert. In der Tat rechnet man dies für $p \in
M$ und $\xi \in \lieAlgebra[g]_1$ leicht nach,
\begin{align*}
   X_{J_1(\xi_1)}(p) = X_{J(T_e\gIn \xi_1)}(p) = (T_e\gIn
   \xi_1)^{\scriptscriptstyle{G}}_M(p) =
   (\xi_1)^{\scriptscriptstyle{G_1}}_M(p) \textrm{,}
\end{align*}
wobei sich die letzte Gleichheit sofort aus der Definition und der
Kettenregel ergibt. Per Definition ist $J_1$ wegen der $G$"=Äquivarianz von $(T_e
\gIn)^*$ und $J$ sogar $G$"=äquivariant -- und nicht etwa nur
$G_1$"=äquivariant. Aus der $G$"=Äquivarianz folgt insbesondere, dass $C_1 :=
J_1^{-1}(0)$ stabil unter $\Phi^{\scriptscriptstyle{G}}$ ist.  Die eingeschränkten Wirkungen
$\Phi^{\scriptscriptstyle{G}}_{\scriptscriptstyle{C_1}} :=
\Phi^{\scriptscriptstyle{G}}\at{G \times C_1} \dpA G \times C_1 \to C_1$
und $\Phi^{\scriptscriptstyle{G_1}}_{\scriptscriptstyle{C_1}} :=
\Phi^{\scriptscriptstyle{G_1}}\at{G_1 \times C_1} \dpA G_1 \times C_1
\to C_1$ sind offensichtlich wieder frei und eigentlich.

$G/G_1$ ist bekanntlich ebenfalls eine Lie"=Gruppe
(vgl.\ \cite[Thm. 9.22]{lee:2003a}), da $G_1$ eine normale Untergruppe von
$G$ ist.

Die kanonischen Projektionen $\wp \dpA G \to G/G_1$ und $\pi_1 \dpA C_1 \to C_1/G_1$
sind surjektive Submersionen, also auch $\wp \times
\pi_1$, vgl.\ Proposition \ref{prop:ProduktVonSubmersionen}. Offensichtlich gilt
für alle $c_1 \in C_1$ und $g \in G$
\begin{align*}
   (\wp\times\pi_1)(g,c_1) = (\wp\times\pi_1)(g',c_1') &\implies \exists g_1,
   h_1 \in G_1 \: \text{mit}\; g_1g(h_1c_1) = g'c_1'\\ &\implies \pi_1(gc_1) =
   \pi_1(g'c_1') \textrm{,}
\end{align*}
wobei bei der letzten Implikation ausgenutzt wurde, dass $G_1$ normal in
$G$ ist. Somit gibt es eine eindeutig bestimmte glatte Abbildung
$\Phi^{\scriptscriptstyle{G/G_1}} \dpA G/G_1 \times C_1/G_1 \to C_1/G_1$,
so dass das Diagramm

\begin{equation}
   \label{eq:gedrueckteWirkung}
   \begin{tikzpicture}[baseline=(current
      bounding box.center),description/.style={fill=white,inner sep=2pt}]
      \matrix (m) [matrix of math nodes, row sep=3.0em, column
      sep=3.5em, text height=1.5ex, text depth=0.25ex]
      { G \times C_1& C_1 \\
        G/G_1 \times C_1/G_1 & C_1/G_1 \\}; %
      \path[->] (m-1-1) edge node[auto] {$\wp \times \pi_1$} (m-2-1); %
      \path[->] (m-1-2) edge node[auto] {$\pi_1$} (m-2-2);%
      \path[->] (m-1-1) edge node[above=-2.5pt] {\small
        $\Phi^{\scriptscriptstyle{G}}_{\scriptscriptstyle{C_1}}$}(m-1-2);%
      \path[->] (m-2-1) edge node[above=-2.5pt] {\small $
        \Phi^{\scriptscriptstyle{G/G_1}}$} (m-2-2); %
   \end{tikzpicture}
\end{equation}
kommutiert. Offensichtlich ist $\Phi^{\scriptscriptstyle{G/G_1}}$ eine
Wirkung von $G/G_1$ auf $C_1/G_1$. Notieren wir Äquivalenzklassen in eckigen Klammern,
so nimmt die Wirkung in unserer abkürzenden Notation folgende Gestalt an
\begin{align}
   \label{eq:gedrueckteWirkungAbk}
   [g][c_1] = \Phi^{\scriptscriptstyle{G/G_1}}([g],[c_1]) = [gc_1], \, g \in
   G,c_1 \in C_1 \Fdot
\end{align}
\begingroup
\emergencystretch=0.8em
\begin{proposition}
   $\Phi^{\scriptscriptstyle{G/G_1}}$ ist symplektisch.
\end{proposition}
\begin{proof}
   Sei $\omega_{\text{red}_1}$ die reduzierte symplektische Form auf $\Mred[1] =
   J_1^{-1}(0)/G$ und $\pi_1 \colon J_1^{-1}(0) \to \Mred[1]$ die
   Projektion.
   Dann gilt für alle $g \in G$
   \begin{align*}
      \pi_1^* ({\Phi^{\scriptscriptstyle{G/G_1}}_{[g]}})^*\omega_{\text{red}_1} =
      ({\Phi^{\scriptscriptstyle{G}}_{\scriptscriptstyle{C_1}})}_g^* \pi_1^*
      \omega_{\text{red}_1}=
      ({\Phi^{\scriptscriptstyle{G}}_{\scriptscriptstyle{C_1}}})_g^*
      \kRes_1 \omega_{\text{red}_1} = \kRes_1
      ({\Phi^{\scriptscriptstyle{G}}_g})^* \omega = \kRes_1
      \omega  =  \pi_1^* \omega_{\text{red}_1}\Fcom
   \end{align*}
   womit wegen der Submersivität von $\pi_1$ das Behauptete folgt.
\end{proof}
\endgroup
\begin{lemma}[{\cite[Lem. 5.2.1]{marsden.misiolek:2007a}}]
   \label{lem:QuotientWirktEigentlich}
   $\Phi^{\scriptscriptstyle{G/G_1}}$ ist frei und eigentlich.
\end{lemma}

\begin{proof}
   Sei $[g][c_1] = [c_1]$, dann gibt es ein $g_1 \in G_1$ mit $g_1c_1 = c_1$. Da
   $\Phi^{\scriptscriptstyle{G}}_{\scriptscriptstyle{C_1}}$ frei ist,
   folgt $g_1g = e$, d.\,h.\ $g = g_1^{-1} \in G_1$, womit wir wiederum $[g]
   = [e]$ erhalten. Somit haben wir gezeigt, dass
   $\Phi^{\scriptscriptstyle{G/G_1}}$ frei ist.

   Sei nun $([g_n])$ eine Folge in $G/G_1$ und $([{c_1}_n])$ eine
   konvergente Folge in $C_1/G_1$ mit Grenzwert $[c_1] \in C_1/G_1$ so dass
   $([g_n][{c_1}_n])$ gegen ein $[c_1'] \in C_1/G_1$ konvergiert.  Da $C_1 \to
   C_1/G_1$ ein Hauptfaserbündel ist, können wir einen lokalen Schnitt
   $\chi \dpA U_{[c_1]} \to C_1$ über einer Umgebung $U_{[c_1]}$ von $[c_1]$
   wählen mit $\chi([c_1]) = c_1$ und dabei o.\,E.\ annehmen, dass $[{c_1}_n] \in
   U_{[c_1]}$ $\forall n \in \mathbb{N}$. Aus Stetigkeitsgründen
   konvergiert $\chi([{c_1}_n])$ gegen $c_1 = \chi([c_1])$. Sei nun für jedes $n
   \in \mathbb{N}$ ein Element $h_{1,n} \in G_1$ mit $\chi([{c_1}_n]) =
   h_{1,n} {c_1}_n$ gewählt, so gilt $h_{1,n} {c_1}_n \to {c_1}$. Analog können wir
   annehmen, dass es Elemente $\{k_{1,n}\}_{n \in \mathbb{N}} \subset
   G_1$ gibt, so dass $k_{1,n} g_n {c_1}_n \to {c_1}'$, also auch
   $k_{1,n}g_nh^{-1}_{1,n}(h_{1,n}{c_1}_n) \to {c_1}'$ für $n \to \infty$. Da
   $\Phi$ eigentlich ist, gibt es also eine konvergente Teilfolge von
   $(k_{1,n}g_nh^{-1}_{1,n})$ und aus Stetigkeitsgründen also auch eine
   konvergente Teilfolge von $([g_n])$ in $G/G_1$.
\end{proof}

Um den zweiten Reduktionsschritt durchführen zu können, wollen wir
zeigen, dass es eine $G/G_1$"=äquivariante Impulsabbildung $J_2 \dpA
\Mred[1] = C_1/G_1 \to \Lie[G/G_1]^* $ gibt, welche für alle $c_1 \in
C_1, \xi \in \lieAlgebra[g]$ die Gleichung
\begin{align}
   \label{eq:DefZweiteImpulsAbbildung}
   \dPaar{J_2([c_1])}{T_e\wp \xi} = \dPaar{J(c_1)}{\xi}
\end{align}
erfüllt.
Um dies zu erreichen, benötigen wir die beiden folgenden elementaren
Propositionen.

\begin{proposition}
   \label{prop:AdProp}
   Für $g_1 \in G_1$ und $\xi \in \lieAlgebra[g]$ ist $\Ad_{g_1}\xi -
   \xi \in T_e\gIn \lieAlgebra[g]_1$.
\end{proposition}
\begin{proof}

      Sei $\tilde{I}_{g_1} \dpA G \times G \to G$, $(g,h) \mapsto g_1g
   g_1^{-1}h^{-1}$.  Da $G_1 \subset G$ normal, ist $I_{g_1} \dpA G
   \to G_1$, $g \mapsto g_1gg_1^{-1}g^{-1}$ wohldefiniert. Ist
   $\Delta \dpA G \to G \times G$ die Diagonalabbildung $g \mapsto
   (g,g)$, so gilt $\tilde{I}_{g_1} \circ \Delta = \gIn \circ I_{g_1}
   $ und somit, wie man leicht mit Hilfe von Proposition \ref{prop:produkte}~\refitem{item:DiagonaleKomponiert} nachprüft
   \begin{align*}
      \Ad_{g_1}\xi - \xi = T_e\conj_{g_1} \xi - T_e(\id_G) \xi = T_e\gIn
      T_eI_{g_1} \xi \Fcom
   \end{align*}
   womit die Behauptung klar ist.
\end{proof}

\begin{proposition}
   \label{prop:AdVomQuotienten}
   Für die adjungierte Wirkung $\AdG[G/G_1]$ auf dem Quotienten gilt
   \begin{align}
      \label{eq:AdVomQuotienten}
      \AdG[G/G_1]_{\wp(g)} T_e\wp \xi = T_e\wp \AdG[G]_g \xi
   \end{align}
   für alle $\xi \in \lieAlgebra[g]$ und $g \in G$.
\end{proposition}

\begin{proof}

   Wir sehen, dass
   \begin{align*}
      \conjG[G/G_1]_{\wp(g)}\circ\wp(h) &= \wp(g) \wp(h) \wp(g)^{-1} \\
      &= \wp(ghg^{-1}) = \wp \circ \conjG[G]_g(h) \quad \text{für alle} \,
      g,h \in G
   \end{align*}
   gilt, woraus die Behauptung durch Anwendung der Kettenregel folgt.

\end{proof}

\begin{lemma}
   \label{lem:zweiteImpulsabbildung}
   Durch Gleichung \eqref{eq:DefZweiteImpulsAbbildung} wird eine
   $G/G_1$"=äquivariante Impulsabbildung $J_2 \dpA \Mred[1] \to
   \Lie[G/G_1]^* $ definiert.
\end{lemma}

\begin{proof}

   Zuerst sehen wir ein, dass $J_2$ wohldefiniert ist. Nach Definition
   von $J_1$ und $C_1$ gilt für $\xi_1 \in \lieAlgebra[g]_1$ und ${c_1}
   \in C_1$ schon $\dPaar{J({c_1})}{T_e\gIn \xi_1} =
   \dPaar{J_1({c_1})}{\xi_1} = 0$. Um die Unabhängigkeit
   von $[{c_1}]$ einzusehen, reicht es zu zeigen, dass der Ausdruck
   $\dPaar{J(g_1 c_1)}{\xi}$ für alle $c_1 \in C_1$ unabhängig von $g_1
   \in G_1$ ist. Für beliebige, aber fest gewählte $c_1 \in C_1$, $\xi
   \in \lieAlgebra[g]$ ist die Bedingung $\dPaar{J(g_1{c_1})}{\xi} =
   \dPaar{J(c_1)}{\xi}$ für alle $g_1 \in G_1$ wegen der
   $G_1$"=Äquivarianz von $J$ gleichbedeutend mit
   $\dPaar{J(c_1)}{\Ad_{g_1^{-1}} \xi -\xi} = 0$ für alle $g_1 \in G_1$.
   Letzteres ist aber  nach Proposition \ref{prop:AdProp} und Definition
   von $J_1$ und $C_1$ klar.

   Als nächstes erinnern wir an das kommutative Diagramm
   \eqref{eq:gedrueckteWirkung}. Durch Anwenden der Kettenregel erhalten
   wir daraus die folgende, für die nächste Rechnung hilfreiche,
   Relation für die fundamentalen Vektorfelder.
   \begin{align*}
      (T_e\wp \xi)_{C/G_1}(\pi_1(c_1)) &= T_{[e]}
      \Phi^{\scriptscriptstyle{G/G_1}}(\cdot,\pi_1(c_1)) T_e \wp \xi \\
      &= T_e(\Phi^{\scriptscriptstyle{G/G_1}}(\wp(\cdot),\pi_1(c_1)) \xi \\
      &= T_e(\pi_1 \circ \Phi^{\scriptscriptstyle{G}}(\cdot,c_1)) \xi \\
      &= T_{c_1} \pi_1 T_e\Phi^{\scriptscriptstyle{G}}(\cdot,c_1) \xi \\
      &= T_{c_1} \pi_1 \xi_{C_1}(c_1) \Fcom
   \end{align*}
   für alle $c_1 \in C_1$ und $\xi \in \lieAlgebra[g]$.
   \jot=3pt

   Bezeichne $\omega_{\text{red}_1}$ die durch die Beziehung $\kInE^*
   \omega = \pi_1^* \omega_{\text{red}_1}$ eindeutig bestimmte
   symplektische Form auf $\Mred[1]$.  Als nächstes zeigen wir, dass
   $J_2$ tatsächlich eine Impulsabbildung ist. Dazu sei ${c_1} \in C_1$, $v
   \in T_pC_1$ und $\xi \in \lieAlgebra[g]$, dann gilt

   \jot=5pt
   \begin{align*}
      \omega_{\text{red}_1} (\pi_1({c_1})) ((T_e\wp
      \xi)_{\Mred[1]}(\pi_1({c_1})), T_{c_1} \pi_1 v) &=
      \omega_{\text{red}_1} (\pi_1({c_1})) (T_p\pi_1\xi_{C_1}({c_1}), T_{c_1}
      \pi_1
      v) \\
      &= (\pi_1^* \omega_{\text{red}_1})({c_1})(\xi_{C_1}({c_1}),v) \\
      &= (\kInE^* \omega)({c_1})(\xi_{C_1}({c_1}),v) \\
      &= \omega({c_1})(T_{c_1} \kInE \xi_{C_1}({c_1}),T_{c_1}\kInE v) \\
      &= \omega({c_1})(\xi_M({c_1}),T_{c_1}\kInE v) \\
      &= \dPaar{d_{c_1}J(\xi)}{T_{c_1}\kInE v} \\
      &= \dPaar{d_{c_1}(J(\xi) \circ \kInE)}{v} \\
      &= \dPaar{d_{c_1}(J_2(T_e\wp\xi) \circ \pi_1)}{v} \\
      &= \dPaar{d_{[{c_1}]}(J_2(T_e\wp\xi))}{T_{c_1}\pi_1 v} \Fdot
   \end{align*}
Wir wenden uns nun der Äquivarianz zu. Für ${c_1} \in C_1$, $g \in G$ und
   $\xi \in \lieAlgebra[g]_1$ sieht man leicht folgendes ein:
   \begin{align*}
      \dPaar{J_2([g][{c_1}])}{T_e\wp \xi} &= \dPaar{J(g{c_1})}{\xi} = \dPaar{gJ({c_1})}{\xi} = \dPaar{J({c_1})}{g^{-1} \xi} \\
      &= \dPaar{J_2([{c_1}])}{T_e \wp (g^{-1} \xi)}
      \stackrel{\hidewidth\eqref{eq:AdVomQuotienten}\hidewidth}{=}
      \dPaar{J_2([{c_1}])}{[g^{-1}] T_e \wp\xi} = \dPaar{[g]^{-1}J_2([{c_1}])}{T_e \wp
        \xi} \Fdot
   \end{align*}
   Somit ist auch die $G/G_1$"=Äquivarianz von $J_2$ gezeigt.

\end{proof}

Für den weiteren Verlauf dieser Arbeit seien $C_2 := J_2^{-1}(0)$ und
$\Mred[2] := C_2/(G/G_1)$, versehen mit der reduzierten
symplektischen Form $\omega_{\text{red}_2}$, welche durch die Beziehung
$ \pi_2^* \omega_{\text{red}_2} = \imath_2^* \omega_{\text{red}_1}$
eindeutig bestimmt ist. Dabei ist $\imath_2 \dpA C_2 \hookrightarrow
\Mred[1]$ die Inklusion und $\pi_2 \dpA C_2 \to
\Mred[2]$ die kanonische Projektion auf den Quotienten. Ebenso seien $C :=
J^{-1}(0)$ und $\Mred[] := C/G$, versehen mit der reduzierten
 symplektischen Form $\omega_{\text{red}}$, die durch die Gleichung
$\pi^* \omega_{\text{red}} = \kIn^* \omega$ eindeutig bestimmt ist,
wobei $\kIn \dpA C \hookrightarrow M$ die Inklusion ist und $\pi \dpA C
\to \Mred[]$ die kanonische Projektion auf den Quotienten. Schließlich sei
$\kjIn \dpA C \hookrightarrow C_1$ die Inklusion von $C$ in $C_1$.

\begin{bemerkung}
   \label{bem:PiAufC}
   Ist $c \in C \subset C_1$, so ist $\pi_1(c) \in C_2$, denn für jedes
   $\xi \in \lieAlgebra[g]$ gilt $\dPaar{J_2(\pi_1(c))}{T_e \wp \xi} =
   \dPaar{J(c)}{\xi} = 0$. Somit gibt es eine glatte Abbildung
   $\widetilde \varsigma \dpA C \to C_2$ mit $\tilde \varsigma(c) =
   \pi_1(c)$ für $c \in C$. Falls $\pi(c) = \pi(c')$ für $c,c' \in C$
   gilt, gibt es ein $g \in G$ mit $gc = c'$ und deshalb gilt nach
   Definition der $G/G_1$"=Wirkung (vgl.\ Gleichung
   \eqref{eq:gedrueckteWirkungAbk}) $[g] \pi_1(c) = \pi_1(c')$, also
   $\pi_2(\pi_1(c)) = \pi_2(\pi_1(c'))$. Folglich gibt es eine glatte
   Abbildung $\varsigma \dpA \Mred[] \to \Mred[2]$, so dass das Diagramm

   \begin{equation}
      \begin{tikzpicture}[baseline=(current
         bounding box.center),description/.style={fill=white,inner sep=2pt}]
         \matrix (m) [matrix of math nodes, row sep=3.0em, column
         sep=3.5em, text height=1.5ex, text depth=0.25ex]
         { C & C_2 \\
           \Mred[] & \Mred[2] \\}; %
         \path[->] (m-1-1) edge node[auto] {$\pi$} (m-2-1); %
         \path[->] (m-1-2) edge node[auto] {$\pi_2$} (m-2-2);%
         \path[->] (m-1-1) edge node[above=-2.5pt] {$\widetilde
           \varsigma$}(m-1-2);%
         \path[->] (m-2-1) edge node[above=-2.5pt] {$\varsigma$}
         (m-2-2); %
      \end{tikzpicture}
   \end{equation}
   kommutiert.
\end{bemerkung}
\begin{proposition}
   \label{prop:tildeVarsigmaSurjektiv}
   Die Abbildung $\tilde \varsigma \colon C \to C_2$ ist surjektiv.
\end{proposition}
\begin{proof}
   Sei $c_2 \in C_2$, d.\,h.\ $J_2(c_2) = 0$. Es gibt nun ein $c_1 \in C_1$
   mit $c_2 = \pi_1(c_1)$, so dass $\dPaar{J(c_1)}{\xi} =
   \dPaar{J_2(\pi_1(c_1))}{T_e\wp \xi} = 0$ für alle $\xi \in
   \lieAlgebra$. Somit gilt aber  schon $J(c_1) = 0$, also $c_1 \in
   C$, d.\,h.\ auch  $c_2 = \pi_1\at{C}(c_1) = \tilde \varsigma(c_1)$.
\end{proof}

Insgesamt haben wir  folgendes wichtige kommutative Diagramm:

\newcommand{\BB}{}
\begin{equation}
   \label{eq:Notation}
   \begin{tikzpicture}[baseline=(current
      bounding box.center),description/.style={fill=white,inner sep=2pt}]
      \matrix (m) [matrix of math nodes, row sep=2.0em, column
      sep=2.3em, text height=1.5ex, text depth=0.25ex]
      {  J^{-1}(0) & \BB & J_1^{-1}(0)  &  \BB & M & \BB & \lieAlgebra[g]^* \\
        \BB & \BB & \BB & \BB & \lieAlgebra[g]_1^* & \BB & \BB \\
        \Mred[] & J_2^{-1}(0) & \Mred[1] & \Lie[G/G_1]^* & \BB & \BB & \BB \\
        \BB & \BB & \BB & \BB & \BB & \BB & \BB \\
        \BB & \Mred[2] & \BB & \BB & \BB & \BB & \BB \\ };
      \path[right hook->] (m-1-1) edge node[auto] {\kjIn} (m-1-3);
      \path[right hook->] (m-1-3) edge node[auto] {\kInE} (m-1-5);
      \path[->] (m-1-5) edge node[auto] {$J$} (m-1-7);

      \path[right hook->] (m-3-2) edge node[auto] {$\imath_2$} (m-3-3);%
      \path[->] (m-3-3) edge node[auto] {$J_2$} (m-3-4);%
      \path[->] (m-1-1) edge node[auto] {$\pi$} (m-3-1); %
      \path[->] (m-1-1) edge node[auto] {$\tilde\varsigma$} (m-3-2); %
      \path[->] (m-1-3) edge node[auto] {$\pi_1$} (m-3-3); %
      \path[->] (m-1-5) edge node[auto] {$J_1$} (m-2-5); \path[->]
      (m-3-2) edge node[auto] {$\pi_2$} (m-5-2); %
      \path[->] (m-3-1) edge node[auto] {$\varsigma$} (m-5-2.north
      west); %
      \path[right hook->] (m-1-1) edge[bend left=20] node[auto] {\kIn}
      (m-1-5);%
      \path[<-] (m-2-5) edge[] node[below,sloped] {$(T_e \gIn)^*$}
      (m-1-7);%
      \path[->] (m-3-4.east) edge[bend right=40] node[below,sloped]
      {$(T_e\wp)^*$} (m-1-7);%

   \end{tikzpicture}
\end{equation}

Schon aus Dimensionsgründen ist klar, dass $\widetilde \varsigma$ im Allgemeinen nicht immersiv ist. Ebenfalls aus Dimensionsgründen hat
$\varsigma$ jedoch die Chance ein Diffeomorphismus zu sein. Dies ist
tatsächlich der Fall. Der folgende zentrale Satz zeigt, dass $\varsigma$
sogar ein Symplektomorphismus ist.

\begin{satz}
   \label{satz:redPhasenraeumeSymplektomorph}
   Sei $(M,\omega)$ eine symplektische Mannigfaltigkeit und eine
   Lie"=Gruppe $G$ wirke auf $M$ stark Hamiltonsch, frei und
   eigentlich mit $G$"=äquivarianter Impulsabbildung
   $J\colon M \to \lieAlgebra^*$ und $C := J^{-1}(0) \neq
   \emptyset$. Sei $\gIn \colon G_1 \hookrightarrow G$ eine normale
   Unter"=Lie"=Gruppe von $G$ und $J_1 := (T_e\kIn)^*\circ J\colon M \to
   \lieAlgebra_1^*$ die zugehörige Impulsabbildung. Ferner sei $\Mred =
   C/G$ der reduzierte Phasenraum bezüglich der $G$"=Wirkung und
   $\Mred[1] = C/G_1$ derjenige bezüglich der $G_1$"=Wirkung. Dann wirkt
   $G/G_1$ frei und eigentlich auf $\Mred[1]$ und es gibt eine
   $G/G_1$"=äquivariante Impulsabbildung $J_2$, die Gleichung
   \eqref{eq:DefZweiteImpulsAbbildung} genügt. Unter diesen
   Voraussetzungen sind $\Mred$ und der bezüglich der $G/G_1$"=Wirkung
   reduzierte Phasenraum $\Mred[2] := J_1^{-1}(0)/(G/G_1)$
   symplektomorph. Explizit ist durch die in Bemerkung \ref{bem:PiAufC}
   eingeführte Abbildung
   \begin{align}
      \varsigma \dpA \Mred[] \to \Mred[2]
   \end{align}
    ein Symplektomorphismus gegeben.
\end{satz}

\begin{proof}
   Zuerst zeigen wir, dass $\varsigma$ injektiv ist. Angenommen
   $\varsigma(\pi(c)) = \varsigma(\pi(c'))$ für $c,c' \in C$, dann gilt
   $\pi_2(\pi_1(c)) = \pi_2(\pi_1(c'))$. Es gibt somit ein $g \in G$ mit
   $\pi_1(gc) = [g] \pi_1(c) = \pi_1(c')$. Also gibt es ein $g_1 \in
   G_1$ mit $g_1gc = c'$, womit schon $\pi(c) = \pi(c')$ gilt, was zu
   zeigen war.

   Die Surjektivität von $\varsigma$ folgt unmittelbar aus der von
   $\pi$, $\pi_2$ und $\tilde \varsigma$.
   Mit Hilfe des kommutativen Diagramms \eqref{eq:Notation}
   und der Definition der reduzierten symplektischen Formen rechnet man
   leicht nach:
   \begin{align*}
      \pi^* \varsigma^* \omega_{\text{red}_2} = {\widetilde \varsigma}^*
      \pi_2^* \omega_{\text{red}_2} &= {\widetilde \varsigma}^* \imath_2^*
      \omega_{\text{red}_1} = \kjIn^* \pi_1^* \omega_{\text{red}_1} =
      \kjIn^* \kInE^* \omega = \kIn^* \omega = \pi^* \omega_{\text{red}}\Fdot
   \end{align*}
   Demnach ist $\varsigma$ symplektisch. Insbesondere ist $\varsigma$
   immersiv.  Aus Dimensionsgründen ($\dim \Mred[] = \dim \Mred[2]$) und
   dem Satz über die Umkehrfunktion ist $\varsigma$ ein
   Diffeomorphismus.

\end{proof}

\section{Quanten-Koszul-Reduktion in Stufen}
\label{sec:QuantenKoszulInStufen}

Gegenstand dieses Kapitels ist es, zu untersuchen, ob und wie sich die
in Abschnitt \ref{sec:symplektischeReduktionInStufen} präsentierte
klassische Reduktion in Stufen auf die Quanten-Koszul-Reduktion
übertragen lässt. Zunächst stellt sich dabei die Frage, ob das in der
ersten Stufe durch Reduktion erhaltene Sternprodukt überhaupt die
Voraussetzungen für eine weitere Reduktion erfüllt. Dies werden wir im
ersten Abschnitt dieses Kapitels betrachten. Anschließend wenden wir uns
im zweiten Abschnitt der Fragestellung zu, wie das in einer Stufe reduzierte
Sternprodukt mit dem in zwei Stufen reduzierten in Beziehung steht.

Sei $\star$ ein $G$"=invariantes Sternprodukt auf $(M,\omega)$ und $\qJ
\dpA \lieAlgebra \to C^\infty(M)[[\lambda]]$ eine $G$"=äquivariante
Quantenimpulsabbildung.

Das heißt für $\xi \in \lieAlgebra[g]$, $g \in G$, $p \in M$ und $f \in \CM[M]$ gelten
die folgenden Bedingungen.
\begin{enumerate}[\itshape i.)]
\item $\qJ(\xi)$ ist eine Deformation von $J(\xi)$.

\item
   \label{item:quantenInvarianteImpulsabb}
   $\I \lambda{} \{J(\xi),f\} =
   [\qJ(\xi),f]_{\star}$.
\item $\qJ(gp) = g\qJ(p)$.
\end{enumerate}
Es sei weiter daran erinnert, dass $\star$ genau dann stark invariant
bezüglich $G$ ist, wenn die klassische Impulsabbildung $J$ auch eine
Quantenimpulsabbildung ist.

Ist $\lieAlgebra[g]_2 \subset \lieAlgebra[g]$ mit $\lieAlgebra[g] = T_e
\gIn \lieAlgebra[g]_1 \oplus \lieAlgebra[g]_2$, so induziert dies eine
Zerlegung $\lieAlgebra^* = (\lieAlgebra^*)_1 \oplus (\lieAlgebra^*)_2$,
wobei die beiden Summanden $(\lieAlgebra^*)_1$ und $(\lieAlgebra^*)_2$
durch $(\lieAlgebra^*)_1: = \{\mu \in \lieAlgebra^* \mid
\mu\at{\lieAlgebra_2} = 0\}$ und $(\lieAlgebra^*)_2: = \{\mu \in
\lieAlgebra^* \mid \mu\at{T_e\gIn\lieAlgebra_1} = 0\}$ gegeben
sind. Aufgrund der $G$"=Invarianz von $\lieAlgebra_1$ folgt direkt
dieselbige von $(\lieAlgebra^*)_2$. Wir denken uns im Folgenden
$\lieAlgebra_2$ gewählt. Offensichtlich ist $(T_e
\gIn)^*\at{(\lieAlgebra^*)_1} \dpA (\lieAlgebra^*)_1 \to
\lieAlgebra[g]_1^*$ ein Isomorphismus von Vektorräumen.

\begin{definition}[Reduktiver homogener Raum]
   \label{def:ReduktiverHomogenerRaum}
   Der homogene Raum $G/G_1$ heißt \neuerBegriff{reduktiv}, falls es ein
   $G_1$"=invariantes Vektorraumkomplement $\lieAlgebra_2$ von $T_e\gIn
   \lieAlgebra_1$ gibt.
\end{definition}
\begin{bemerkung}
   \label{bem:ReduktiverHomogenerRaum}
   Die $G_1$"=Invarianz von $\lieAlgebra_2$ ist, da $G_1$
   zusammenhängend ist, äquivalent dazu, dass $[\xi_1,\xi_2] \in
   \lieAlgebra_2$ für alle $\xi_1 \in T_e\gIn\lieAlgebra_1,\xi_2 \in
   \lieAlgebra_2$. Da $G_1$ in $G$ normal und somit $\lieAlgebra_1$ ein
   Lie-Ideal in $\lieAlgebra$ ist, ist dies weiter äquivalent zur
   Bedingung $[\xi_1,\xi_2] = 0$ für alle $\xi_1 \in T_e\gIn
   \lieAlgebra_1,\xi_2 \in \lieAlgebra_2$. Genauso sieht man, dass
   $\lieAlgebra_2$ genau dann $G$"=invariant ist, wenn $\lieAlgebra_2$
   ein Lie"=Ideal in $\lieAlgebra$ ist.
\end{bemerkung}

\begin{bemerkung}
   \label{bem:invariantesKomplement}
   Dass $\lieAlgebra[g]_2$ ein $G$"=invariantes Komplement von $T_e \gIn
   \lieAlgebra[g]_1$ ist, bedeutet, falls -- wie wir es immer annehmen
   -- $G$ und $G_1$ zusammenhängend sind, dass es eine diskrete Gruppe
   $\Gamma \subset G$ und einfach zusammenhängende Gruppen $\tilde G_1$
   und $\tilde G_2$ gibt, so dass $G \simeq (\tilde G_1 \times \tilde
   G_2)/\Gamma$. Dies ist insbesondere immer dann erfüllt, wenn $G$
   kompakt ist, siehe etwa \cite[Thm. 20.21]{lee:2003a}.

\end{bemerkung}

\begin{beispiel}
   \label{bsp:heisenbergLieAlgebra}
   Sei $\lieAlgebra$ die Heisenberg-Lie-Algebra, die durch die
   Generatoren $P,Q$ und $R$ über $\mathbb{R}$ erzeugt wird und deren
   Lie-Klammer durch die Beziehungen $[P,Q] = R$ und $[P,R] = [Q,R] = 0$
   festgelegt ist. Weiter sei $\lieAlgebra_1 := \mathbb{R} R$ das
   Zentrum von $\lieAlgebra$. Auf Gruppenniveau entspricht dies der Wahl
   von $G$ als Heisenberg-Gruppe und der von $G_1$ als
   $\mathbb{R}\setminus \{0\}$.  In diesem Fall gibt es zwar ein
   $G_1$"=invariantes Komplement von $\lieAlgebra_1$ in $\lieAlgebra$,
   etwa $\lieAlgebra_2 := \mathbb{R}\textrm{-}\mathrm{Span}\{P,Q\}$,
   aber kein $G$"=invariantes Komplement, denn wie man leicht nachrechnet,
   gilt $[\lieAlgebra,\lieAlgebra] \subset \lieAlgebra_1$.

\end{beispiel}

\begin{proposition}
   \label{prop:p2Invariant}
   Ist $\lieAlgebra_2$ $H$"=invariant (mit $H \in \{G_1,G\}$), so ist die
   kanonische Projektion $\mathrm{p}_2 \colon \lieAlgebra^* \to
   (\lieAlgebra^*)_2$ $H$"=äquivariant.
\end{proposition}
\begin{proof}
   Seien $\mu_1 \in (\lieAlgebra^*)_1$ und $\mu_2 \in
   (\lieAlgebra^*)_2$ sowie $h \in H$. Dann gilt
   \begin{align*}
      h \mathrm{p}_2(\mu_1 + \mu_2) = h \mu_2 = \mathrm{p}_2(h \mu_2) = \mathrm{p}_2(h \mu_1 + h
      \mu_2) = \mathrm{p}_2(h(\mu_1 + \mu_2))\Fdot
   \end{align*}
   Dabei wurde im zweiten Schritt die $H$"=Invarianz von
   $(\lieAlgebra^*)_2$ verwendet. Im dritten Schritt geht die
   Voraussetzung ein, denn die $H$"=Invarianz von $\lieAlgebra_2$
   impliziert sofort die von $(\lieAlgebra^*)_1$.
\end{proof}

\subsection{Invarianzeigenschaften eines nach einer Untergruppe
  reduzierten Sternprodukts}
\label{sec:InvarianzDesErstenReduziertenSternprodukts}
In diesem Abschnitt sehen wir analog zur klassischen Situation,
dass $\qJ$ eine $G$"=äquivariante Quantenimpulsabbildung $\qJ[1]$ für den
ersten Reduktionsschritt und eine Quantenimpulsabbildung $\qJ[2]$
induziert. Da letztere zum zweiten Reduktionsschritt dienen soll, muss
sie auch $G/G_1$"=äquivariant sein. Um dies sicherzustellen, werden wir
weitere Bedingungen an die Gruppen stellen müssen. Wir nehmen $G/G_1$ in
diesem Abschnitt immer als reduktiv an und denken uns ein
$G_1$"=invariantes Komplement $\lieAlgebra_2$ gewählt.

Wir definieren
\begin{align}
   \label{eq:J1Definition}
   \qJ[1] \colon \lieAlgebra_1 \to C^\infty(M)[[\lambda]]
\end{align}
durch
\begin{align}
\qJ[1](\xi_1) := \qJ(T_e \gIn \xi_1) \quad \text{ für }\xi_1 \in
\lieAlgebra[g]_1\label{eq:J1Definition2}
\end{align}
und erhalten die folgende Proposition.

\begin{proposition}
   \label{prop:FettJEImpulsabb}
   $\qJ[1]$ ist eine $G$"=äquivariante Quantenimpulsabbildung.
\end{proposition}
\begin{proof}
   Es ist klar, dass $\qJ[1]$ eine Deformation von $J_1$ ist. Weiter
   gilt  für $f \in \CM[M]$ und $\xi_1 \in
   \lieAlgebra[g]_1$
   \begin{align*}
      \I \lambda \{J_1(\xi_1),f\} = \I \lambda \{J(T_e \gIn \xi_1),f\} =
      [\qJ(T_e \gIn \xi_1),f]_{\star} = [\qJ[1]( \xi_1),f]_{\star} \Fdot
   \end{align*}
   Wegen der $G$"=Äquivarianz von $(T_e \gIn)^*$ und der von $\qJ$ ist auch
   klar, dass $\qJ[1] = (T_e\gIn)^* \circ \qJ$ $G$"=äquivariant ist.
\end{proof}

\begin{proposition}
   \label{prop:GWirkungaufkleing1RespektiertLieAlgebra}
   Die Wirkung $\tAdG$ respektiert die Lie"=Algebra-Struktur auf
   $\lieAlgebra[g]_1$, d.\,h.\ $\tAdG_g$ ist für alle $g \in G$ ein
   Morphismus von Lie"=Algebren.
\end{proposition}
\begin{proof}
   Dies ist unmittelbar klar, da für $g \in G$ die Abbildung $\tconjG_g
   \dpA G_1 \to G_1$ ein Morphismus von Lie"=Gruppen ist, was
   bekanntermaßen (vgl.~\cite[Thm. 4.25]{lee:2003a}) impliziert, dass
   $\tAdG_g = T_e \tconjG_g$ ein Morphismus von Lie"=Algebren ist.
\end{proof}

Im Weiteren sei $\qRes_1\dpA \CM[M][[\lambda]] \to \CM[C_1][[\lambda]]$
die von $\qJ[1]$ und einer gewählten $G$"=äquivarianten
Prolongationsabbildung $\prol[1] \dpA \CM[C_1] \to \CM[M]$ induzierte
$G$"=äquivariante Quanteneinschränkung. Die geometrischen
Homotopie"=Daten von denen $\prol[1]$ induziert werden seien an dieser
Stelle noch beliebig und $G$"=invariant. Die hier vorliegende
Wahlfreiheit werden wir später ausnutzen. Analog zur klassischen
Situation wollen wir nun zeigen, dass es eine Quantenimpulsabbildung
$\qJ[2] \colon \Mred[1] \to \Lie[G/G_1]_{\mathbb{C}}^*[[\lambda]]$ gibt,
so dass für alle $\xi_2 \in \lieAlgebra_2$ und $c_1 \in C_1$
\begin{align}
   \dPaar{\qJ[2](\pi_1(c_1))}{T_e \wp \xi_2} = \qResE \dPaar{
     \qJ}{\xi_2}(c_1)
\end{align}
ist. Anders als in der klassischen Situation kann man aber obige
Gleichung nicht für alle $\xi \in \lieAlgebra$ sondern nur für $\xi_2
\in \lieAlgebra_2$ fordern, da $\qRes_1 \qJ$ im Allgemeinen Anteile aus
$(\lieAlgebra_{\mathbb{C}}^*)_1$ enthalten könnte. Die folgende
Proposition klärt die Situation. Dazu sei $\tqJ[1] \colon \lieAlgebra
\to C^\infty(M)[[\lambda]]$ durch
\begin{align}
   \label{eq:tJ1a}
   \tqJ[1](T_e\gIn \xi_1) = \qJ[1](\xi_1) = \qJ(T_e \gIn \xi_1) \quad
   \text{für $\xi_1 \in \lieAlgebra_1$}
\end{align}
und
\begin{align}
   \label{eq:tJ1b}
   \tqJ[1](\xi_2) = 0 \quad \text{ für $\xi_2 \in \lieAlgebra_2$} \Fdot
\end{align}
 gegeben.
\begin{proposition}
   \label{prop:QuantenJ2}
      Es gibt eine eindeutig bestimmte Abbildung
      \begin{align}
         \qJ[2] \colon \Mred[1] \to \Lie[G/G_1]_{\mathbb{C}}^*[[\lambda]]\label{eq:QuantenJ21}
      \end{align}
      mit
      \begin{align}
         \label{eq:QuantenJ2}
         \dPaar{\qJ[2](\pi_1(c_1))}{T_e \wp \xi} = \qResE\dPaar{
           \qJ - \tqJ[1]}{\xi}(c_1) \quad \forall c_1 \in
         C_1, \xi \in \lieAlgebra \Fdot
      \end{align}
\end{proposition}
\begin{proof}
         Der Beweis geht analog zum klassischen Fall. Die Eindeutigkeit
         ist klar. Wir definieren $\qJ[2]$ durch Gleichung
         \eqref{eq:QuantenJ2} und müssen zeigen, dass dies tatsächlich
         wohldefiniert ist. Zunächst sieht man,
         dass die rechte Seite von Gleichung \eqref{eq:QuantenJ2} für
         $\xi = T_e\gIn \xi_1 \in T_e\gIn \lieAlgebra_1$ Null ergibt, denn
         \begin{align*}
            \qRes_1\dPaar{\qJ - \tqJ[1]}{T_e \gIn \xi_1} &=
            \qRes_1\dPaar{\qJ}{T_e\gIn \xi_1} -
            \qRes_1\dPaar{\tqJ[1]}{T_e\gIn \xi_1} \\ &=
            \qRes_1 \dPaar{\qJ[1]}{\xi_1} - \qRes_1
            \dPaar{\qJ[1]}{\xi_1} = 0\Fcom
         \end{align*}
         wobei im vorletzten Schritt die Definitionen von $\qJ[1]$ und
         $\tqJ[1]$ verwendet wurden.  Als nächstes zeigen wir, dass die
         rechte Seite von Gleichung \eqref{eq:QuantenJ2} auch unter
         einer Ersetzung von $c_1$ durch $g_1c_1$ mit $g_1 \in G_1$
         unverändert bleibt, d.\,h.\ dass für alle $\xi \in \lieAlgebra$
         der Ausdruck $\qRes_1\dPaar{\qJ - \tqJ[1]}{\xi}(g_1c_1)$
         unabhängig von $g_1$ ist. Um dies zu sehen, bemerken wir
         zunächst, dass $\tqJ[1]$ $G_1$"=äquivariant ist, da
         $G/G_1$"=reduktiv ist und schon $\qJ[1]$ $G_1$"=äquivariant
         ist. Da außerdem auch $\qJ$ und $\qRes_1$ $G_1$"=äquivariant
         sind, gilt
         \begin{align*}
            \qRes_1 \dPaar{\qJ - \tqJ[1]}{\xi} (g_1 c_1) &=
            (g_1^{-1} \qRes_1 \dPaar{\qJ - \tqJ[1]}{\xi} )( c_1)\\
            &= \qRes_1 (g_1^{-1} \dPaar{\qJ - \tqJ[1]}{\xi} )( c_1)\\
            &= \qRes_1 (\dPaar{\qJ - \tqJ[1]}{g_1^{-1} \xi})(c_1) \Fcom
         \end{align*}
         womit wir sehen, dass es analog zum klassischen Fall genügt zu
         zeigen, dass $\qRes_1\dPaar{\qJ - \tqJ[1]}{g_1^{-1}\xi - \xi} =
         0$ gilt. Dies ist aber mit Proposition \ref{prop:AdProp}
         und den obigen Ausführungen klar.
\end{proof}

Wir bemerken, dass die Abbildung $\chi\colon \mathrm{Lie}(G/G_1)^* \to
(\lieAlgebra^*)_2$, $\alpha \mapsto (T_e\wp)^*\alpha$ offensichtlich
bijektiv ist, notieren $\mathrm{p}_2 \colon \lieAlgebra^* \to (\lieAlgebra^*)_2 $
für die kanonische Projektion auf den zweiten Faktor und kommen zum
ersten wichtigen Satz dieses Abschnittes.

\begin{samepage}
\begin{satz}[Invarianzeigenschaften des reduzierten Sternprodukts]
   \label{thm:ReduziertesSternproduktInvariant}

   \begin{satzEnum}

   \item %
      Das Sternprodukt $\starred[1]$ ist $G/G_1$"=invariant.
   \item %
      $\qJ[2]$ ist eine Quantenhamiltonfunktion und, falls
      $\lieAlgebra_2$ $G$"=invariant ist, auch $G/G_1$"=äquivariant,
      d.\,h.\ eine Quantenimpulsabbildung.
   \item %
      \label{item:MitAlpha}
      Falls $\qJ = J + \alpha$ für ein $\alpha \in
      \lambda \lieAlgebra[g]_{\mathbb{C}}^*[[\lambda]]$ und $\qResE J = \kResE J
      $ erfüllt ist, gilt $\qJ[2] = J_2 + \chi^{-1}(\mathrm{p}_2(\alpha))$
      und $\starred[1]$ ist stark invariant bezüglich der
      $G/G_1$"=Wirkung.
   \end{satzEnum}
\end{satz}
\end{samepage}
\begin{proof}
   \begin{beweisEnum}
   \item %
      Dies folgt unmittelbar aus der $G$"=Äquivarianz von $\qRes_1$,
      $\prol_1$, der Definition der $G/G_1$"=Wirkung sowie Gleichung
      \eqref{eq:SternproduktAufDemReduziertenPhasenraum3} und der
      Submersivität von $\pi_1$.
   \item  %
      Es ist offensichtlich, dass $\qJ[2]$ eine Deformation von $J_2$ ist.
      Sei $f \in \CM[{\Mred[1]}]$, dann gilt für alle $\xi_2 \in \lieAlgebra_2$
      \begin{align*}
         \lefteqn{\pi_1^*(\I \lambda \{f,J_2(T_e\wp \xi_2)\}_{\text{red}_1})}\\
         &= \I \lambda \kResE \{ \prol[1] \pi_1^*
         f,\prol[1] \pi_1^* (J_2(T_e\wp \xi))\} && \eAnn{nach Gleichung \eqref{eq:ReduziertePoissonKlammer} }\\
         &= \I \lambda \kResE \{\prol[1] \pi_1^*
         f,\prol[1] \kResE J({\xi_2})\} && \eAnn{per Definition von $J_2$}\\
        &= \I \lambda \kResE \{\prol[1] \pi_1^* f,J({\xi_2})\}
         && \eAnn{$(\id - \prol[1] \kResE)(\CM[M]) \subset \kIdeal[1]$,}[] \\
         &\phantom{=} && \eAnn[]{ $\prol[1]
           \pi_1^*f \in \kbIdeal[1]$ nach
           Prop. \ref{prop:FunktionenAufDemReduziertenPhasenraum}~\refitem{item:CharakterisierungDesKlassischenIdealisators}} \\
         &= \I \lambda \kResE {\xi_2}_M(\prol[1]\pi_1^* f) \\
         &= \I \lambda \kResE \prol[1] {\xi_2}_{C_1}(\pi_1^* f)
         &&\eAnn{$\prol[1]$
           $G$-invariant}\\
         &= \I \lambda \qResE \prol[1]{\xi_2}_{C_1}(\pi_1^* f)
         &&\eAnn{da $\kResE \prol[1] = \id = \qResE \prol[1]$}\\
         &= \I \lambda \qResE {\xi_2}_M(\prol[1] \pi_1^*f)
         &&\eAnn{$\prol[1]$ $G$-invariant}\\
         &= \I \lambda \qResE \{\prol[1] \pi_1^*f,J({\xi_2})\} \\
         &= \qResE [\prol[1]\pi_1^* f,\qJ({\xi_2})]_{\star}
         &&\eAnn{da $\qJ$ Quantenimpulsabbildung}\\
         &= \qResE [\prol[1]\pi_1^* f,\prol[1] \qResE \qJ({\xi_2})]_{\star} && \eAnn{$(\id - \prol[1] \qResE)(\CM[M][[\lambda]]) \subset \qIdeal[1]$,}[] \\
         &\phantom{=} && \eAnn[]{ $\prol[1]
           \pi_1^*f \in \qbIdeal[1]$ nach Lem. \ref{lem:CharakterisierungQuantenLieIdealisator}} \\
         &= \qResE [\prol[1]\pi_1^*
         f,\prol[1] \pi_1^* \qJ[2](T_e\wp {\xi_2})]_{\star} && \eAnn{per Definition von $\qJ[2]$ } \\
         &= \pi_1^*[f,\qJ[2](T_e\wp {\xi_2})]_{\starred[1]}\Fdot
      \end{align*}
      Damit folgt wegen der Injektivität von $\pi_1^*$ sofort
      \begin{align*}
         \I \lambda \{f,J_2(T_e\wp \xi_2)\}_{\text{red}_1} =
         [f,\qJ[2](T_e\wp {\xi_2})]_{\starred[1]} \Fdot
      \end{align*}
      Die $G/G_1$"=Äquivarianz sieht man wie im Beweis zu Lemma
      \ref{lem:zweiteImpulsabbildung}.
      Für $c_1 \in C_1$, $g \in G$ und $\xi_2 \in \lieAlgebra[g]_2$ gilt
      nämlich:
      \begin{align*}
         \dPaar{\qJ[2]([g][c_1])}{T_e\wp \xi_2} &= \qResE \dPaar{
           \qJ}{\xi_2}(gc_1) = (g^{-1} \qResE \dPaar{\qJ}{\xi_2})(c_1) = \qResE g^{-1} \dPaar{\qJ}{\xi_2}(c_1)\\
         &= \qResE\dPaar{\qJ}{g^{-1} \xi_2}(c_1) =
         \dPaar{\qJ[2]([c_1])}{T_e
           \wp (g^{-1} \xi_2)} \\
         &\stackrel{\hidewidth\eqref{eq:AdVomQuotienten}\hidewidth}{=}
         \dPaar{\qJ[2]([c_1])}{[g]^{-1} T_e \wp\xi_2}\\
         &= \dPaar{[g]\qJ[2]([c_1])}{T_e \wp \xi_2} \Fdot
      \end{align*}
      Dabei wurde im fünften Schritt verwendet, dass $\lieAlgebra_2$
      $G$"=invariant ist.
   \item %

      Für $\xi_2 \in \lieAlgebra_2$ und $c_1 \in C_1$ gilt
      \begin{align*}
         \dPaar{\qJ[2](\pi_1(c_1))}{T_e\wp \xi_2}
         &= \qRes_1 \dPaar{\qJ}{\xi_2}(c_1) \\
         &= \qRes_1 \dPaar{J}{\xi_2}(c_1) + \dPaar{\alpha}{\xi_2}\\
         &= \kRes_1 \dPaar{J}{\xi_2}(c_1) +
         \dPaar{\mathrm{p}_2(\alpha)}{\xi_2}\\
         &=
         \dPaar{J_2(\pi_1(c_1))}{T_e \wp \xi_2} +
         \dPaar{\chi^{-1}(\mathrm{p}_2(\alpha))}{T_e \wp \xi_2} \Fcom
      \end{align*}
      und somit $\qJ[2] = J_2 + \chi^{-1}(\mathrm{p}_2(\alpha))$. Damit
      ist klar, dass $\star_{\mathrm{red}_1}$ stark invariant ist, da
      Quantenkommutatoren mit konstanten Funktionen verschwinden. Das
      heißt mit $\qJ[2]$ ist auch $J_2 = \qJ[2] -
      \chi^{-1}(\mathrm{p}_2(\alpha))$ eine Quantenimpulsabbildung.
   \end{beweisEnum}
\end{proof}
\begin{bemerkung}
   \label{bem:Quantenhamiltonfunktion}
   \begin{bemerkungEnum}
   \item
      Für die $G/G_1$"=Invarianz von $\starred[1]$ benötigt offenbar
      nicht dass es ein $G_1$"=invariantes Komplement von
      $T_e\gIn\lieAlgebra_1$ gibt.
   \item %
      Für die Aussage, dass $\qJ$ eine Quantenhamiltonfunktion ist,
      kommt man mit der Voraussetzung aus, dass es ein
      $G_1$"=invariantes Komplement von $T_e\gIn\lieAlgebra_1$ gibt, was
      sichergestellt ist, wenn $G/G_1$ reduktiv ist. Um weiter
      reduzieren zu dürfen, benötigt man jedoch eine
      Quantenimpulsabbildung. Diese erhalten wir, falls wir die echt
      stärkere Bedingung, nämlich dass es ein $G$"=invariantes
      Komplement von $T_e\gIn\lieAlgebra_1$ gebe, fordern.
   \end{bemerkungEnum}
\end{bemerkung}

Falls $\star$ stark invariant bezüglich der $G$"=Wirkung ist, kann man
die geometrischen Homotopie"=Daten und damit $\qRes_1$ so wählen, dass
tatsächlich $\kRes_1 J = \qRes_1 J$ gilt.  Um dies zeigen
zu können benötigen wir jedoch noch einige Vorbereitungen.

\begin{lemma}

   \label{lem:LinAlg}
   Seien $U$, $V$, $V'$ und $W$ endlich dimensionale $\mathbb{R}$"=Vektorräume
   mit $\dim V = \dim V'$ und $f \dpA U \oplus V \to W \oplus V'$ eine
   lineare Abbilddung mit $f(U) \subset W$. Dann gilt
   \begin{align*}
      f \text{surjektiv} \implies f\at{U} \dpA U \to W \text{surjektiv}\Fdot
   \end{align*}
\end{lemma}
\begin{proof}
   Zunächst bemerken wir, dass es wegen der Surjektivität von $f$ für $x
   \in W\oplus V'$, $u \in U$ und $v \in V$ gibt mit $x = f(u+v) = f(u) +
   f(v)$. Deshalb ist $W \oplus V' = f(U) + f(V)$. Ferner gilt offenbar
   $\dim(f(U)) \leq \dim W$, da nach Voraussetzung $f(U) \subset W$
   ist. Demnach erhält man
   \begin{align*}
      \dim W + \dim V \; \,
      &\stackrel{\hidewidth\text{Vor.}\hidewidth}{=}\; \, \dim W + \dim
      V' =
      \dim( W \oplus V' )\stackrel{f \text{surj.}} = \dim(f(U) + f(V)) \\
      &= \; \, \dim(f(U)) + \dim(f(V)) - \dim(f(U) \cap f(V)) \\
      &\leq \; \, \underbrace{\dim(f(U))}_{\leq \dim W} +
      \underbrace{\dim(f(V))}_{\leq \dim V} \leq \dim W + \dim V \Fdot
   \end{align*}
   Somit gilt in obiger Ungleichungskette sogar überall
   Gleichheit. Insbesondere ist die folgende Beziehung richtig.
   \begin{align*}
      \dim(f(U)) + \dim(f(V)) = \dim W + \dim V \Fdot
   \end{align*}
   Da aber $\dim (f(U)) \leq \dim W$ und $\dim (f(V)) \leq \dim V$,
   ergibt sich insbesondere schon $\dim (f(U)) = \dim W$, woraus die
   Behauptung folgt.
\end{proof}

\begin{proposition}
   \label{prop:Jg2}
   Es gilt $J(C_1) \subset (\lieAlgebra[g]^*)_2$.
\end{proposition}
\begin{proof}
   Sei $\xi_1 \in \lieAlgebra_1$ und $c_1 \in C_1$. Dann gilt
   offensichtlich
   \begin{align*}
      \dPaar{J(c_1)}{T_e\gIn \xi_1} = \dPaar{J_1(c_1)}{\xi_1} = 0
   \end{align*}
\end{proof}

\begin{lemma}

   \label{lem:JCSubmersiv}
   $J\at{C_1} \dpA C_1 \to (\lieAlgebra[g]^*)_2$ ist submersiv.
\end{lemma}
\begin{proof}
   Sei $c_1 \in C_1$ beliebig aber fest gewählt. Weiter sei $K_{c_1}$
   ein Komplement von $T_{c_1}\kInE T_{c_1}C_1$ in $T_{c_1}M$, d.\,h.\
   \begin{align*}
      T_{c_1}\kInE T_{c_1}C_1 \oplus K_{c_1} = T_{c_1}M \Fdot
   \end{align*}

   Wir definieren
   \begin{align*}
      U &:= T_{c_1}\kInE T_{c_1}C_1, \quad V := K_{c_1} \Fcom \\
      W &:= T_{J(c_1)}(\lieAlgebra[g]^*)_2, \quad V' :=
      T_{0}(\lieAlgebra[g]^*)_1
   \end{align*}
   und
   \begin{align*}
      f := T_{c_1}J \dpA U \oplus V \to W \oplus V'\Fdot
   \end{align*}
   Es gilt
   \begin{align*}
      \dim K_{c_1} = \dim M - \dim C_1 = \dim G_1 \Fdot
   \end{align*}
   Also ist mit der eben eingeführten Notation $\dim V = \dim V'$. Nach
   Lemma \ref{lem:JSurjektiv} ist $f$ surjektiv. Ferner wissen wir mit
   Proposition \ref{prop:Jg2}, dass $J(C_1) \subset (\lieAlgebra[g]^*)_2$ gilt
   und damit auch
   \begin{align*}
      f(U) &= T_{c_1}J(T_{c_1}\kInE T_{c_1}C_1) = T_{c_1}(J \circ
      \kInE)(T_{c_1}C_1) \subset T_{J(c_1)}(\lieAlgebra[g]^*)_2 = W\Fdot
   \end{align*}
   Somit können wir Lemma \ref{lem:LinAlg} anwenden und erhalten
   \begin{align*}
      \im (T_{c_1}J\at{C_1}) = f(U) = W =
      T_{J(c_1)}(\lieAlgebra[g]^*)_2\Fdot
   \end{align*}
\end{proof}

Das folgende Lemma stellt eine interessante Verallgemeinerung des
klassischen Tubensatzes (vgl.\ Def. \ref{def:Inv.Tubenumgebung})
dar. Berücksichtigt man keine Gruppenwirkungen, so ist der Satz in der
Literatur wohlbekannt, insbesondere im Kontext von Kontrolldaten
stratifizierter Räume. Meist wird für den Beweis auf die Arbeit von
Mather \cite[Prop. 6.2]{mather:1970} verwiesen. Der dort geführte Beweis
wirkt jedoch nicht besonders überzeugend. Zwei andere Beweise finden
sich in der Arbeit von Pflaum \cite[3.5.3]{pflaum:2000} und in
\cite[Ch. II. Thm. 1.6]{gibson:1976}. Letzterer scheint die dem Problem
angemessene Sprache zu verwenden und eignet sich gut, den Satz auch
unter zusätzlicher Berücksichtigung von Gruppenwirkungen zu formulieren
und zu beweisen.

\begin{lemma}[kompatible $G$-invariante Tubenumgebung]
   \label{lem:kompatibleTuben}
   Seien $M$ und $N$ Mannigfaltigkeiten und $G$ eine Lie"=Gruppe, die auf
   $M$ und $N$ wirke. Die Wirkung auf $M$ sei eigentlich. Weiter sei $C
   \subset M$ eine $G$"=invariante Untermannigfaltigkeit und $f\dpA M \to
   N$ eine $G$"=äquivariante, glatte Abbildung, so dass $f\at{C}\dpA C
   \to N$ submersiv ist. Dann gibt es eine $G$"=invariante Tubenumgebung
   $(\pi\dpA E \to C,U_E,U_M,\tau)$ von $C$ in $M$ mit $f(r(p)) = f(p)$
   für alle $p \in U_M$. Dabei bezeichnet $r := \pi \circ \tau^{-1} \dpA
   U_M \to C$ die induzierte $G$"=äquivariante Retraktion.
\end{lemma}
Der Beweis von Lemma \ref{lem:kompatibleTuben} befindet sich im Anhang,
siehe Satz \ref{satz:GkompatibleTuben}.

\begin{definition}
   \label{def:kompatibleTube}
   Eine Tubenumgebung wie in Lemma \ref{lem:kompatibleTuben} heißt auch
   \neuerBegriff{mit $f$ kompatible $G$"=invariante Tubenumgebung}.
\end{definition}

\begin{lemma}
   \label{lem:NormalformJ}
   Es gibt eine $G_1$"=invariante \tn{gute} Tubenumgebung $\Psi_1\dpA
   U_1 \to V_1 \subset C_1 \times \lieAlgebra[g]_1^*$ von $C_1$, so dass
   $\dPaar{J(\Psi^{-1}_1(c_1,\mu_1))}{\xi_2}$ für alle $c_1 \in C_1$ und
   $\xi_2 \in \lieAlgebra[g]_2$ unabhängig von $\mu_1 \in
   \projFaktor_2(V_1)$ ist. Dabei bezeichnet $\projFaktor_2 \dpA C_1
   \times \lieAlgebra[g]^*_1 \to \lieAlgebra[g]^*_1$ die Projektion auf
   den zweiten Faktor. Ist $\lieAlgebra[g]_2$ $G$"=invariant, so kann
   $\Psi_1$ sogar $G$"=äquivariant gewählt werden.
\end{lemma}
\begin{proof}
   Sei $\mathrm{p}_2 \dpA \lieAlgebra[g]^* \to (\lieAlgebra[g]^*)_2$ die
   kanonische Projektion. Da $\lieAlgebra_2$ $G_1$"=invariant ist, ist
   $\mathrm{p}_2$ nach Proposition \ref{prop:p2Invariant}
   $G_1$"=äquivariant und sogar $G$"=äquivariant, falls $\lieAlgebra_2$
   $G$"=invariant ist. Diese Äquivarianzeigenschaften übertragen sich
   entsprechend auf die Abbildung $\mathrm{p}_2 \circ J \dpA M \to
   (\lieAlgebra[g]^*)_2$. Offensichtlich ist die Aussage des Lemmas
   äquivalent dazu, dass es eine $G_1$- (bzw.\ $G$)"=invariante gute
   Tubenumgebung $\Psi_1 \colon U_1 \to V_1 \subset C_1 \times
   \lieAlgebra_1^* $ von $C_1$ gibt, so dass $\mathrm{p}_2 \circ J
   (\Psi^{-1}_1(c_1,\mu_1))$ für alle $c_1 \in C_1$ unabhängig von
   $\mu_1 \in \projFaktor_2(V_1)$ ist. Ist also $r_1 \dpA U_1 \to C_1$
   die von $\Psi_1$ induzierte Retraktion, so ist dies weiter äquivalent
   dazu, dass für alle $p \in U_1$ schon $\mathrm{p}_2 \circ J(p) =
   \mathrm{p}_2 \circ J \circ r_1(p)$ gilt. Da nun aber nach Lemma
   \ref{lem:JCSubmersiv} $\mathrm{p}_2 \circ J \at{C_1} = J\at{C_1} \dpA
   C_1 \to (\lieAlgebra[g]^*)_2$ submersiv ist, folgt die Behauptung
   unmittelbar aus Lemma \ref{lem:kompatibleTuben} und Satz
   \ref{satz:CxgA}.
\end{proof}

\begin{satz}

   \label{satz:klassgleichquantenEinschraufJ}

   Sei $\Psi_1\dpA U_1 \to V_1 \subset C_1 \times \lieAlgebra[g]_1^*$
   wie in Lemma \ref{lem:NormalformJ} $G$"=invariant gewählt und $\qResE$
   die von $\Psi_1$ induzierte Quanteneinschränkung. Dann gilt
   \begin{align}
      \qResE J = \kResE J \Fdot
   \end{align}
\end{satz}
\begin{proof}
   Sei $\{\tilde{e}_i\}_{i = 1, \dots, \dim G_1}$ eine Basis von
   $\lieAlgebra[g]_1$ und $\{e_i\}_{i=\dim G_1 +1, \dots, \dim G}$ eine
   Basis von $\lieAlgebra_2$. Weiter definieren wir für $i \in
   \{1,\dots,\dim G_1\}$ $e_i := T_e\gIn \tilde{e}_i$, womit
   $\{e_i\}_{i=1, \dots, \dim G}$ eine Basis von $\lieAlgebra[g]$
   ist. Seien weiter $\{\tilde{e}^i\}$ bzw.\ $\{e^i\}$ die dualen Basen
   zu $\{\tilde{e}_i\}$ bzw.\ $\{e_i\}$.
   Sei nun $(c_1,\mu_1) \in V_1$, dann gilt

   \begin{align*}
      J(e_i)(\Psi^{-1}_1(c_1,\mu_1)) = J_1(\tilde
      e_i)(\Psi^{-1}_1(c_1,\mu_1)) = \mu_1(\tilde e_i) \quad \text{für $i
        \in \{1,\dots,\dim G_1\}$}
   \end{align*}

und $J(e_i)(\Psi^{-1}_1(c_1,\mu_1))$ ist nach Lemma
\ref{lem:NormalformJ} unabhängig von $\mu_1 \in \lieAlgebra_1^*$ für $i \in
\{\dim G_1 + 1,\dots,\dim G\}$.
   Also erhält man
   \begin{align*}
      \partial_{\tilde{e}^{\alpha}} (J(e_i) \circ \Psi_1^{-1})
      (c_1,\tilde{\mu}_1) =
      \begin{cases}
         \delta^{\alpha}_i &\text{für $i \in \{1,\dots,\dim G_1\}$} \\
         0 &\text{sonst}
      \end{cases} \Fdot
   \end{align*}
   Ist $\hE$ die von $\Psi_1$ induzierte Homotopie, so ergibt sich damit
   \begin{align*}
      \hE J(\xi) = \sum_{\alpha = 1}^{\dim G_1} \xi^{\alpha}
      \tilde{e}_{\alpha} \quad \text{für $\xi = \xi^i e_i \in
        \lieAlgebra$} \Fdot
   \end{align*}
   Damit folgt unmittelbar für alle $\xi \in \lieAlgebra$
   \begin{align*}
      (\qkoszulE - \kkoszulE) \hE J(\xi) = 0
   \end{align*}
   und somit
   \begin{align*}
      \qResE J(\xi) = \frac{\id}{\id + (\qkoszulE - \kkoszulE)\hE}
      J(\xi) &= \kResE J(\xi) + 0 =  \kResE J(\xi) \Fdot
   \end{align*}

\end{proof}

\begin{bemerkung}
   Zusammen mit Satz
   \ref{thm:ReduziertesSternproduktInvariant}~\refitem{item:MitAlpha}
   bedeutet die Aussage von Satz
   \ref{satz:klassgleichquantenEinschraufJ}, dass man  die
   geometrischen Wahlen derart treffen kann, dass im Falle eines stark
   invarianten Sternprodukts das reduzierte Sternprodukt wieder stark
   invariant ist. Siehe dazu auch Bemerkung
   \ref{bem:PhysikalicheBedeutung}. Falls $\star$ zusätzlich Hermitesch
   ist und man $\qJ = J + \frac{1}{2} \I \lambda \Delta $ mit modularer
   Einsform $\Delta$ wählt, so zeigt
   \ref{thm:ReduziertesSternproduktInvariant}~\refitem{item:MitAlpha}
   weiter, dass $\qJ[2] = J_2 + \frac{1}{2} \I \lambda \Delta_2$ gilt,
   wobei $\Delta_2$ die modulare Einsform von $\Lie[G/G_1]$ ist. In der
   Tat, da $\lieAlgebra_2$ $G_1$"=invariant ist, folgt, wie man mit
   Bemerkung \ref{bem:ReduktiverHomogenerRaum} leicht
   nachrechnet $\chi^{-1}(\mathrm{p}_2)(\Delta) = \Delta_2$.
\end{bemerkung}

\subsection{Zweite Reduktionsstufe}
\label{sec:VergleichDerSternprodukte}

Wir wollen nun das Sternprodukt $\starred[1]$ bezüglich der Gruppe
$G/G_1$ weiter reduzieren und das so erhaltene Sternprodukt
$\starred[2]$ mit dem bezüglich der $G$"=Wirkung reduzierten $\starred$
vergleichen. Im Folgenden setzen wir voraus, dass $\lieAlgebra_2$
$G$"=invariant ist, $\prol[1]$ $G$"=äquivariant gewählt ist, also $\qJ[2]$
nach Satz \ref{thm:ReduziertesSternproduktInvariant} eine
$G/G_1$"=äquivariante Quantenimpulsabbildung ist.

\subsubsection{Konstruktion kompatibler geometrischer Prolongationen}
\label{sec:KonstruKompGeomProl}
In diesem Abschnitt konstruieren wir Tubenumgebungen, so dass die davon
induzierten Prolongationen naheliegende Verträglichkeitsbedingungen
erfüllen. Wir wollen im Folgenden solche von bestimmten geometrischen
Homotopie"=Daten für einen Reduktionsschritt wie in Kapitel
\ref{cha:Koszul-Reduktion} notieren und mit einem entsprechenden Index
den Reduktionsschritt andeuten. Wir schreiben also
$\prol$, $\prol_1$, $\prol_2$ usw.

\begingroup
\emergencystretch=0.1em
\begin{proposition}
   \label{prop:KompositionVonRetraktionen}
   Sei $\tilde r \dpA \tilde U \to C$ eine $G$"=äquivariante Retraktion
   von einer $G$"=invarianten, offenen Umgebung $\tilde U$ von $C$ in
   $C_1$ auf $C$ und $r_1 \dpA U_1 \to C_1$ eine $G$"=äquivariante
   Retraktion von einer $G$"=invarianten, offenen Umgebung $U_1$ von $C_1$ in $M$
   auf $C_1$. Sei $U := r_1^{-1}(\tilde U)$. Dann ist $r := \tilde r
   \circ r_1\at{U}\dpA U \to C$ eine $G$"=äquivariante Retraktion von der
   in $M$ offenen, $G$"=invarianten Umgebung $U$ von $C$ auf $C$. Es gilt
   insbesondere die wichtige Verträglichkeitsbedingung
   \begin{align}
    r \circ r_1\at{U} = r \Fdot
   \end{align}
\end{proposition}

\endgroup

\begin{proof}
   Die Verträglichkeitsbedingung folgt sofort aus der Definition von $r$, in
   der Tat gilt
      \begin{align*}
      r \circ r_1\at{U} = \tilde r \circ r_1\at{U} \circ r_1\at{U} =
      \tilde r \circ r_1\at{U} = r \Fdot
   \end{align*}
   Der Rest ist klar.
\end{proof}

\begin{bemerkung}
   Die Retraktionen $r_1$ und $\tilde r$ in Proposition
   \ref{prop:KompositionVonRetraktionen} könnten insbesondere von zwei
   \tn{guten} Tubenabbildungen $\Psi_1$ und $\tilde \Psi$ herrühren.
   Somit kann man die vorherige Proposition zusammen mit Satz
   \ref{satz:CxgA} so interpretieren, dass die Wahl der Tuben $\Psi_1$ und
   $\tilde \Psi$ eine \tn{gute} Tubenabbildung $\Psi$ von $C$ in $C_1$
   induziert, deren zugehörige Retraktion $r$ die
   Verträglichkeitsbedingung
   \begin{align}
      r \circ r_1\at{U} = r
   \end{align}
   erfüllt.
\end{bemerkung}

Wir wenden uns nun der Konstruktion einer speziellen \tn{guten} Tubenumgebung für
$C_2$ in $\Mred[1]$ zu.

\begin{lemma}
   \label{lem:RetraktionRunterD}
   Sei $\tilde r \dpA C_1 \supset \tilde U \to C$ eine $G$"=äquivariante,
   glatte Retraktion auf $C$. Dann gibt es eine glatte,
   $G/G_1$"=äquivariante Retraktion $r_2 \dpA U_2 \to C$ von der
   $G/G_1$"=invarianten, offenen Umgebung $U_2 := \pi_1(U_2)$ von $C_2$
   in $\Mred[1]$ auf $C_2$, so dass das folgende Diagramm kommutiert.
   \def\tA[#1]{A_{#1}}
\begin{equation}
      \begin{tikzpicture}[baseline=(current
    bounding box.center),description/.style={fill=white,inner sep=2pt}]
         \matrix (m) [matrix of math nodes, row sep=3.0em, column
         sep=3.5em, text height=1.5ex, text depth=0.25ex]
         {
\tilde U & C \\
U_2 & C_2 \\
}; %

 \path[->] (m-1-1) edge node[left] {$\pi_1\at{\tilde U}$} (m-2-1); %
 \path[->] (m-1-2) edge node[auto]{$\pi_1\at{C}$}(m-2-2); %
 \path[->] (m-1-1) edge node[auto]{$\tilde r$}(m-1-2); %
\path[->] (m-2-1) edge node[auto]{$r_2$}(m-2-2); %
 \end{tikzpicture}
\end{equation}
\end{lemma}
\begin{proof}
   Klar.

\end{proof}

Im Folgenden wollen wir die von $\tilde r$ bzw.\ $r_2$ induzierten guten
Tubenabbildungen mit $\tilde \Psi$ bzw.\ $\Psi_2$ bezeichnen.

\begin{bemerkung}
\label{bem:RetraktionRunterD}
Da in Lemma \ref{lem:RetraktionRunterD} $\tilde U$ insbesondere
$G_1$"=invariant ist, gilt schon $\pi_1^{-1}(U_2) =
\pi_1^{-1}(\pi_1(\tilde U)) = G_1 \tilde U = \tilde U$.
\end{bemerkung}

Wir verwenden nun die oben konstruierten Tubenumgebungen verwenden,
Prolongationen zu bauen, die in noch zu präzisierender Weise kompatibel
sein werden.

Sei $\tilde O$ eine $G$-invariante, offene Umgebung von $C$ in $C_1$ mit
$\abschluss{\tilde O} \subset \tilde U$, $\tilde W := C_1
\setminus \abschluss{\tilde O}$ und $\{\tilde \psi_{\tilde U}, \tilde
\psi_{\tilde W}\}$ eine der $G$"=invarianten, offenen Überdeckung
$\{\tilde U,\tilde W\}$ von $C_1$ untergeordnete glatte $G$"=invariante
Zerlegung der Eins. Weiter sei $O_1$ eine $G$"=invariante, offene
Umgebung von $C_1$ in $M$ mit $\abschluss{O_1} \subset U_1$, $W_1 :=
M\setminus \abschluss{O_1}$ und $\{\psi_{U_1},\psi_{W_1}\}$ eine der
$G$"=invarianten, offenen Überdeckung $\{U_1,W_1\}$ von $M$ untergeordnete
$G$"=invariante, glatte Zerlegung der Eins.

Wir definieren dann

\begin{align}
   \label{eq:defqrol}
   \qrol \dpA \CM[C] &\to \CM[C_1], \quad
   \qrol(f)(x) :=
   \begin{cases}
      \tilde \psi_{\tilde U}(x) f(\tilde r(x)) &\text{für } x \in
      \tilde U \\
      0 &\text{sonst}
   \end{cases} \quad \forall f \in C^\infty(C)\Fdot
\end{align}

Nach Konstruktion gilt
\begin{align}
   \label{eq:qrolEinschr}
   \kjRes \qrol = \id \Fdot
\end{align}

$O := r_1^{-1}(\tilde O) \cap O_1 = r_1\at{O_1}^{-1}(\tilde O)$ ist
offensichtlich offen und $G$"=invariant und es gilt $\abschluss{O}
\subset U := r_1^{-1}(\tilde U)$, denn
\begin{align}
   \label{eq:RechungMitAbschlussVonO}
   \abschluss{O} &= \abschluss{r_1^{-1}(\tilde O) \cap O_1}
   \subset \abschluss{r_1^{-1}(\tilde O)} \cap \abschluss{O_1}
   \notag\\
   &\subset \abschluss{r_1^{-1}(\tilde O)} \cap U_1 \subset
   r_1^{-1}(\abschluss{\tilde O}) \cap U_1 \subset r_1^{-1}(\tilde U)
   \cap U_1 = U.
\end{align}

Bei der vorletzten Inklusion wurde die Stetigkeit von $r_1$
verwendet. Man beachte, dass der Schnitt des Abschlusses bezüglich $M$ einer
Teilmenge von $U_1$ mit $U_1$ mit dem Abschluss bezüglich der
Unterraumtopologie in $U_1$ derselben Menge und damit auch mit diesem
geschnitten mit $U_1$ übereinstimmt. Für diese simple topologische
Feinheit sei auf \cite[Prop. 2.1.1]{engelking:1989} verwiesen.

Die Abbildung
\begin{align}
   \psi_U :=
   \begin{cases}
      \psi_{U_1}\at{U} \cdot (\tilde \psi_{\tilde U} \circ r_1\at{U}) &
      \text{auf
        $U$} \\
      0 &\text{sonst}
   \end{cases}
\end{align}
ist nach Konstruktion glatt und $G$"=invariant.  %

Weiter gilt $\supp \psi_U \subset U$, denn zum einen folgt aus der
Stetigkeit von $r_1$
\begin{align*}
   \supp \psi_{U} \cap U_1 = \abschluss{{\carr \psi_{U}}} \cap U_1 \subset
   \abschluss{r_1^{-1}(\carr \tilde \psi_{\tilde U})} \cap U_1 \subset r_1^{-1}(\supp
   \tilde \psi_{\tilde U}) \subset r_1^{-1}(\tilde U) = U
\end{align*}
und zum anderen gilt offenbar $\supp \psi_{U} \subset \supp \psi_{U_1}
\subset U_1$.  Sei $W := M \setminus \abschluss{O}$ und $\psi_W := 1 -
\psi_U$, dann gilt $\supp(\psi_W) \subset W$. Dies sieht man wie folgt ein. Die Menge $\tilde O' :=
C_1\setminus \supp \tilde \psi_{\tilde W}$ ist eine offene Umgebung von
$C$ in $C_1$ mit $\abschluss{\tilde O} \subset \tilde O' \subset \tilde
U$ und $\tilde \psi_{\tilde U}\at{\tilde O'} = 1$. Ebenso ist $O_1' :=
M\setminus \supp \psi_{W_1}$ eine offene Umgebung von $C_1$ in $M$ mit
$\abschluss{O_1} \subset O_1' \subset U_1$ und $\psi_{U_1}\at{O_1'} =
1$. Dann ist $O' := r_1^{-1}(\tilde O') \cap O_1' \subset U$ eine
offene Umgebung von $C$ in $M$ mit $\abschluss{O} \subset O' \subset U$
und $\psi_U\at{O'} = 1$, wobei man die erste Inklusion mit
Rechnung \eqref{eq:RechungMitAbschlussVonO} einsieht. Somit ist
\begin{align*}
   \supp \psi_{W} = \abschluss{\carr \psi_W} \subset
   \abschluss{M\setminus{O'}} = M \setminus{O'} \subset M \setminus
   \abschluss{O} = W
\end{align*}
klar.

Also ist $\{\psi_U,\psi_W\}$ eine der offenen Überdeckung $\{U,W\}$ von $M$
untergeordnete glatte, $G$"=invariante Zerlegung der Eins.

Wir betrachten nun die von dieser Zerlegung der Eins und der
Tubenumgebung $\Psi$

induzierte Prolongationsabbildung $\prol$ sowie die von der Zerlegung
der Eins $\{\psi_{U_1},\psi_{W_1}\}$ und der Tubenabbildung $\Psi_1$
induzierte Prolongationsabbildung $\prol[1]$. Die folgende Proposition
zeigt, dass die so konstruierten geometrischen Prolongationsabbildungen
auf naheliegende Weise verträglich sind.

\begin{proposition}
   \label{prop:prop1undprolKompatibel}
   Es gilt
   \begin{align}
      \prol = \prol[1] \kResE \prol \Fdot
   \end{align}
\end{proposition}
\begin{proof}
   Wir zeigen  $\prol = \prol[1] \qrol$, woraus leicht das
   Behauptete folgt, denn
   \begin{align*}
      \prol = \prol[1] \qrol = \prol[1] \kResE \prol[1] \qrol = \prol[1]
      \kResE \prol \Fdot
   \end{align*}
   Sei also $x \in U = r_1^{-1}(\tilde U) \subset U_1$ und $f \in
   C^\infty(C)$. Dann gilt

   \begin{align*}
      \prol[1](\qrol(f))(x) &\stackrel{x \in U_1}{=} \psi_{U_1}(x)
      \qrol(f)(r_1(x)) \stackrel{r_1(x) \in \tilde{U}}{=} \psi_{U_1}(x) \cdot
      \psi_{\tilde U}(r_1(x)) \cdot
      f(\tilde{r}(r_1(x)))\\
      &\stackrel{\phantom{x \in U_1}}{=} \psi_{U}(x) \cdot f(r(x)) =
      (\prol f)(x) \Fdot
   \end{align*}
   Falls $x \notin U$ gilt offensichtlich
   \begin{align*}
      \prol[1](\qrol f)(x) = 0 = (\prol f)(x) \Fdot
   \end{align*}
\end{proof}

Die Aussage von Proposition \ref{prop:prop1undprolKompatibel} ist für
eine beliebige Wahl von $\tilde O$ und $\{\tilde \psi_{\tilde U},\tilde
\psi_{\tilde W}\}$ richtig. Diese Wahlfreiheit wollen wir nun ausnutzen
und konstruieren $\tilde O$ und die Zerlegung der Eins $\{\tilde
\psi_{\tilde U},\tilde \psi_{\tilde W}\}$ ausgehend von einer beliebig
gewählten $G/G_1$"=invarianten, offenen Teilmenge $O_2$ von $C_2$ in
$\Mred[1]$ mit $\abschluss{O_2} \subset U_2$ und einer beliebig
gewählten $G/G_1$"=invarianten glatten Zerlegung der Eins
$\{\psi_{U_2},\psi_{W_2}\}$, die $\{U_2,W_2\}$ untergeordnet sei, wobei
$W_2 := \Mred[1] \setminus \abschluss{O_2}$. Dann definieren wir die
offene Menge $\tilde O := \pi_1^{-1}(O_2) \subset C_1$ und $\tilde W :=
C_1 \setminus \abschluss{\tilde O}$. Wegen der Stetigkeit von $\pi_1$
gilt offenbar $\abschluss{\tilde O} = \abschluss{\pi_1^{-1}(O_2)}
\subset \pi_1^{-1}(\abschluss{O_2}) \subset \pi_1^{-1}(U_2) = \tilde U$,
wobei sich die letzte Gleichheit aus Bemerkung
\ref{bem:RetraktionRunterD} ergibt. Weiter ist $\tilde O$ per Definition der
Quotientenwirkung $G$"=invariant. Nun definieren wir $\tilde \psi_{\tilde
  U} := \pi_1^* \psi_{U_2}$ und $\tilde \psi_{\tilde W} := \pi_1^*
\psi_{W_2}$.

\begingroup
\emergencystretch=0.8em
\begin{proposition}
   \label{prop:zerlderEinsUntenHochgezogen}
   $\{\tilde \psi_{\tilde U},\tilde \psi_{\tilde W}\}$ ist eine der
   offenen, $G$"=invarianten Überdeckung $\{\tilde U,\tilde W\}$ von
   $C_1$ untergeordnete $G$"=invariante, glatte Zerlegung der Eins.
\end{proposition}
\endgroup
\begin{proof}
   Offensichtlich gilt
   \begin{align*}
      \tilde \psi_{\tilde U} + \tilde \psi_{\tilde W} = \pi_1^*
      \psi_{U_2} +
      \pi_1^* \psi_{W_2} = \pi_1^*(\psi_{U_2} + \psi_{W_2}) = 1
   \end{align*}
   und
   \begin{align*}
      \tilde \psi_{\tilde U}(g c_1) &= \pi_1^* \psi_{U_2}(gc_1) =
      \psi_{U_2}(\pi_1(gc_1)) \\
      &= \psi_{U_2}(\wp(g) \pi_1(c_1)) = \psi_{U_2}(\pi_1(c_1)) = \tilde
      \psi_{\tilde U}(c_1) \quad \text{für $c_1 \in C_1$ und $g \in G$} \Fdot
   \end{align*}
   Entsprechend folgt die $G$"=Invarianz von $\tilde \psi_{\tilde
     W}$. Ferner sieht man
   \begin{align*}
      \supp \tilde \psi_{\tilde U} = \supp(\pi_1^* \psi_{U_2}) =
      \abschluss{\pi_1^{-1}(\carr \psi_{U_2})} \subset \pi_1^{-1}(\supp
      \psi_{U_2}) \subset \pi_1^{-1}(U_2) = \tilde U \Fdot
   \end{align*}
   Schließlich gilt auch
   \begin{align*}
      \supp \tilde \psi_{\tilde W} \subset  \tilde W \Fdot
   \end{align*}
   Dies sieht man wie folgt ein. Die Menge $O_2' := \Mred[1]\setminus
   \supp \psi_{W_2}$ ist eine offene Umgebung von $C_2$ in $\Mred[1]$
   mit $\abschluss{O_2} \subset O_2' \subset U_2$ und $\supp \psi_{W_2}
   = \Mred[1]\setminus O_2'$. Weiter gilt $\abschluss{\tilde O} \subset
   \pi_1^{-1}(\abschluss{O_2}) \subset \pi_1^{-1}(O_2')$.
   Dann erhält man
   \begin{align*}
      \supp \tilde \psi_{\tilde W} \subset \pi_1^{-1}(\supp \psi_{W_2}) =
      \pi_1^{-1}(\Mred[1]\setminus O_2') = \pi_1^{-1}(\Mred[1])\setminus
      \pi_1^{-1}(O_2') \subset C_1\setminus {\abschluss{\tilde O}} =
      \tilde W \Fdot
   \end{align*}
   Also ist $\{\tilde \psi_{\tilde U},\tilde \psi_{\tilde W}\}$ eine der
   offenen, $G$"=invarianten Überdeckung $\{\tilde U,\tilde W\}$ von
   $\tilde U$ untergeordnete $G$"=invariante, glatte Zerlegung der Eins,
   wie gewünscht.
\end{proof}

Wir können nun einige Verträglichkeitsbedingungen zwischen den
Prolongationsabbildungen formulieren.
\begin{proposition}

   \label{prop:vertrZerlDerEins}
   Sind die Zerlegungen der Eins wie oben beschrieben konstruiert, so
   gilt
   \begin{align}
      \label{eq:vertrZerlDerEins}
      \psi_U\at{\tilde U} = \psi_{U_2} \circ \pi_1\at{\tilde U}  \Fdot
   \end{align}
\end{proposition}
\begin{proof}
   Für $c_1 \in \tilde U = U \cap C_1$ sieht man
   \begin{align*}
      \psi_U(c_1) = \psi_{U_1}(c_1) \cdot \tilde \psi_{\tilde
        U}(\underbrace{r_1(c_1)}_{= c_1}) = \tilde \psi_{\tilde U} (c_1)
      = \psi_{U_2}(\pi_1(c_1)) \Fdot
   \end{align*}

\end{proof}

\begin{satz}
   \label{satz:prol2Komp}
   Seien die Zerlegungen der Eins und die Tubenabbildungen wie oben
   kompatibel gewählt. Dann gelten für die induzierten
   Prolongationsabbildungen folgende Beziehungen.
   \begin{satzEnum}
   \item %
      \label{item:prol2Kompa} %
      $\pi_1^* \prol[2] = \kResE \prol \tilde \varsigma^*$.
   \item%
      \label{item:prol2Komp}%
      $\prol[1]\pi_1^*\prol[2] \pi_2^* = \prol \pi^* \varsigma^*$.
   \item %
      \label{item:hprolKomp2} %
      $\h \prol[1]\pi_1^*\prol[2] = 0$.
   \item %
      \label{item:fettUndduennEinschProl} %
      $\kRes \prol[1] \pi_1^* \prol[2] = \qRes \prol[1] \pi_1^*
      \prol[2]$.
   \end{satzEnum}
\end{satz}
\begin{proof}
   \begin{beweisEnum}
   \item %
      Für $c_1 \in C_1$ und $f \in \CM[C_2]$ gilt
      \begin{align*}
         (\pi_1^* \prol[2] f)(c_1) &= \prol[2] f (\pi_1(c_1))\\
         &=
         \begin{cases}
            \psi_{U_2}(\pi_1(c_1)) f(r_2(\pi_1(c_1))) & \text{für $\pi_1(c_1) \in U_2$}
            \\
            0 & \text{sonst}
         \end{cases} \\
         &=
         \begin{cases}
            \psi_{U_2}(\pi_1(c_1)) f(\pi_1(r(c_1)) & \text{für $ c_1 \in  U$} \\
            0 & \text{sonst}
         \end{cases} \\
         &=
         \begin{cases}
            \psi_{U}(c_1) f(\pi_1(r(c_1))) & \text{für $c_1 \in U$} \\
            0 & \text{sonst}
         \end{cases} \\
         &=
         \begin{cases}
            \psi_{U}(c_1)(\tilde \varsigma^*f)(r(c_1)) & \text{für $c_1 \in U$} \\
            0 & \text{sonst}
         \end{cases} \\
         &= (\prol \tilde \varsigma^* f)(c_1) = (\kResE \prol \tilde
         \varsigma^* f)(c_1) \Fdot
      \end{align*}
      Dabei wurden im dritten Schritt die Definitionen von $r_2$ und $r$
      verwendet:
      \begin{align*}
         r_2(\pi_1(c_1)) = \pi_1(\tilde r(c_1)) \stackrel{c_1 \in C_1}{=}
         \pi_1(\tilde r(r_1(c_1))) = \pi_1(r(c_1)) \Fdot
      \end{align*}
   \item Mit~\refitem{item:prol2Kompa} und Proposition
      \ref{prop:prop1undprolKompatibel} erhält man
      \begin{align*}
         \prol[1] \pi_1^* \prol[2] \pi_2^* = \prol[1] \kResE \prol \tilde
         \varsigma^* \pi_2^* = \prol \tilde \varsigma^* \pi_2^* = \prol \pi^* \varsigma^*\Fdot
      \end{align*}
   \item Wie im Beweis zu~\refitem{item:prol2Komp} sieht man sofort
      \begin{align*}
         \h \prol[1]\pi_1^*\prol[2] = \h \prol[1]\kResE\prol \tilde
         \varsigma^* = \h \prol \tilde \varsigma^* = 0 \Fdot
      \end{align*}
      Dabei wurde im letzten Schritt Gleichung
      \eqref{eq:globalisierteHomotopie3} verwendet.
   \item Dies folgt unmittelbar aus~\refitem{item:hprolKomp2} und der
      expliziten Formel für $\qRes$.
   \end{beweisEnum}
\end{proof}

Auf für die induzierten Quanteneinschränkungen können wir
Verträglichkeitsbedingungen formulieren.

\begin{lemma}
   \label{lem:RelationenMitQuantenj}
   Es gibt ein $\qjRes \dpA \CM[C_1][[\lambda]] \to
   \CM[C][[\lambda]]$ mit
   \begin{lemmaEnum}
   \item
      \label{item:QuantenJEigenschaft1}
      $\qjRes \qResE = \qRes$ und
   \item\label{item:QuantenJEigenschaft2} $\kjRes \pi_1^* \prol[2]
      \qRes_2 = \qjRes\pi_1^*$.
   \end{lemmaEnum}
   Die Abbildung $\qjRes$ ist durch Eigenschaft~\refitem{item:QuantenJEigenschaft1} schon eindeutig festgelegt und
   von der Form $\qjRes = \qRes \prol[1]$.
\end{lemma}
\begin{proof}
   Wir zeigen zuerst die Eindeutigkeit und dann die Existenz.
   \begin{itemize}
   \item[Eindeutigkeit] \hfill \\
      Angenommen es gibt ein $\qjRes \dpA \CM[C_1][[\lambda]] \to
      \CM[C][[\lambda]]$ mit der Eigenschaft~\refitem{item:QuantenJEigenschaft1}, dann gilt
      \begin{align*}
         \qjRes = \qjRes \underbrace{\qResE \prol[1]}_{= \id}
         \stackrel{~\refitem{item:QuantenJEigenschaft1}}{=} \qRes
         \prol[1] \Fdot
      \end{align*} %
   \item[Existenz] \hfill \\
      Wir definieren $\qjRes = \qRes \prol[1]$ und zeigen, dass die
      Eigenschaften~\refitem{item:QuantenJEigenschaft1} und~\refitem{item:QuantenJEigenschaft2} erfüllt sind. Sei $\{\tilde
      e_\alpha\}_{\alpha = 1, \dots, \dim G_1}$ eine Basis von
      $\lieAlgebra_1$ mit zugehöriger dualer Basis
      $\{\tilde e^\alpha\}$ und $e_\alpha := T_e \gIn \tilde e_\alpha$ für $\alpha \in
      \{1,\dots,\dim G_1\}$. Weiter wählen wir eine Basis
      $\{e_\alpha\}_{\dim G_1 +1,\dots,\dim G}$ von
      $\lieAlgebra_2$. Dies induziert eine Basis $\{T_e\wp
      e_\alpha\}_{\alpha = \dim G_1 + 1,\dots,\dim G}$ von $\Lie[G/G_1]$, deren duale Basis wir mit
      $\{(T_e\wp e_\alpha)^*\}$ bezeichnen wollen. Insbesondere ist
      $\{e_\alpha\}_{\alpha = 1, \dots, \dim G}$ eine Basis von
      $\lieAlgebra$ mit dualer Basis $\{e^\alpha\}$.
      \begin{beweisEnum}
      \item Per Definition von $\qIdeal$ gilt
         \begin{align*}
            \qIdeal =  \left\{\sum_{\alpha = 1}^{\dim G} f^\alpha \star
            \qJ(e_\alpha) \mid f^\alpha \in \CM[M][[\lambda]]\right\} \Fdot
         \end{align*}
         Unter Beachtung von $\qJ[1](\tilde e_\alpha) = \qJ(e_\alpha)$ für $\alpha
         \in \{1,\dots,\dim G_1\}$ erhält man damit
         \begin{align*}
            \qIdeal[1] = \left\{\sum_{\alpha = 1}^{\dim G_1} f^\alpha \star
            \qJ(e_\alpha) \mid f^\alpha \in \CM[M][[\lambda]]\right\} \Fdot
         \end{align*}

         Also ist insbesondere $\qIdeal[1] \subset \qIdeal$. Somit gilt
         für alle $f \in \CM[M]$
         \begin{align*}
            & \prol[1] \qResE f -f \in \qIdeal[1] \subset \qIdeal \\
            &\implies \qRes(\prol[1] \qResE f - f) = 0\\
            &\implies \qjRes \qResE f = \qRes \prol[1] \qResE f = \qRes
            f \Fdot
         \end{align*}
         Daraus folgt $\qjRes \qResE = \qRes$.
      \item Wir formulieren Bedingung~\refitem{item:QuantenJEigenschaft2} zunächst etwas um.
         \begin{align*}
            \lefteqn{\kjRes \pi_1^* \prol[2] \qRes_2 = \qjRes \pi_1^*}\\
            &\iff \kRes \prol[1] \pi_1^* \prol[2] \qRes_2 = \qRes
            \prol[1] \pi_1^*\\
            &\iff \qRes \prol[1] \pi_1^* \prol[2] \qRes_2 = \qRes
            \prol[1] \pi_1^*   \\
            &\iff \qRes \prol[1] \pi_1^* (\id - \prol[2] \qRes_2) = 0
         \end{align*}
         Dabei wurde bei der ersten Äquivalenz die Gleichungen $\kjRes =
         \kRes \prol[1]$ und $\qjRes = \qRes \prol[1]$ ausgenutzt. Bei der zweiten
         Äquivalenz wurde Satz \ref{satz:prol2Komp}~\refitem{item:fettUndduennEinschProl} verwendet.
\begingroup
\emergencystretch=0.8em

         Da nach Proposition \ref{prop:QuantenAugmentierung} die
         Abbildung $\id - \prol[2] \qRes_2$ aber auf $\qIdeal[2] \subset
         \CM[{\Mred[1]}][[\lambda]]$ projiziert, genügt es die letzte
         Gleichung auf $\qIdeal[2]$ zu zeigen. Mit der
         Homotopieeigenschaft
         \begin{align*}
          \id =   \qkoszulZ\hZ + \hZ \qRes_2
         \end{align*}
         gilt für $\tilde{f} \in \qIdeal[2]$ die Gleichung
         \begin{align*}
            \tilde{f} = \qkoszulZ \hZ (\tilde{f}) =
            \dPaar{\hZ(\tilde{f})}{(T_e \wp e_\alpha)^*} \starred[1] \qJ[2](T_e\wp
            e_\alpha)\Fdot
         \end{align*}
\endgroup
         Mit der Abkürzung $\tilde{f}^\alpha :=
         \dPaar{\hZ(\tilde{f})}{T_e( \wp e_\alpha)^*} \in \CM[{\Mred[1]}]$ berechnen wir
         \begin{align*}
            \qRes \prol[1] \pi^*_1 \tilde{f} &= \qRes \prol[1] \pi^*_1 (
            \tilde{f}^\alpha \starred[1]
            \qJ[2](T_e\wp e_\alpha))  \\
            &= \underbrace{\qRes \prol[1] \qResE }_{= \qjRes \qResE =
              \qRes}(\underbrace{\prol[1] \pi_1^* \tilde{f}^\alpha}_{=:
              f^\alpha} \star \prol[1] \pi_1^* \qJ[2](T_e \wp e_\alpha))
            \\
            &= \qRes (f^\alpha \star \prol[1] \pi_1^* \qJ[2](T_e \wp
            e_\alpha))\\
            &= \qRes(f^\alpha \star \prol[1]\qResE \qJ(e_\alpha)) -
            \underbrace{\qRes(f^\alpha
              \star \qJ(e_\alpha))}_{=0}\\
            &= \qRes(\underbrace{f^\alpha \star\underbrace{
                ((\prol[1]\qResE \qJ(e_\alpha)) -\qJ(e_\alpha))}_{ \in
                \qIdeal[1] \subset \qIdeal}}_{\in \qIdeal} ) \\
            &= 0 \Fdot
         \end{align*}
      \end{beweisEnum}

   \end{itemize}
\end{proof}

\subsubsection{Vergleich der Sternprodukte}
\label{sec:VergleichDerSternproducteSUB}

Wir kommen nun zum Hauptresultat der vorliegenden Arbeit.
\begin{satz}[Quanten-Koszul-Reduktion in Stufen]
   \label{thm:UebereinstimmungDerSternprodukte}

   Es seien dieselben Voraussetzungen wie in Satz
   \ref{satz:redPhasenraeumeSymplektomorph} gegeben. Zusätzlich gebe es
   ein $G$"=invariantes Komplement $\lieAlgebra_2$ von
   $T_e\gIn\lieAlgebra_1$. Weiter sei $\star$ ein $G$"=invariantes
   Sternprodukt für $(M,\omega)$, für das eine Quantenimpulsabbildung
   $\qJ$ existiere und $\qJ[1]$ die durch Gleichung
   \eqref{eq:J1Definition2} davon induzierte Quantenimpulsabbildung für
   die $G_1$"=Wirkung sowie $\qJ[2]$ die gemäß Gleichung
   \eqref{eq:QuantenJ2} induzierte Quantenimpulsabbildung für die
   $G/G_1$"=Wirkung. Dann kann man die geometrischen Homotopie"=Daten
   derart wählen, dass
   \begin{align}
      \varsigma^* \starred[2] = \starred
   \end{align}
   gilt. Dabei bezeichnet $\starred$ das bezüglich der $G$"=Wirkung nach
   dem Quanten"=Koszul"=Schema mit diesen geometrischen Homotopie"=Daten
   und $\qJ$ konstruierte Sternprodukt für
   $(\Mred,\omega_{\mathrm{red}})$. Ist $\star_1$ das entsprechende
   Sternprodukt für $(\Mred[1],\omega_{{\mathrm{red}}_1})$ bezüglich
   $\qJ[1]$, so ist $\starred[2]$ das bezüglich $\starred[1]$ und
   $\qJ[2]$ reduzierte Sternprodukt für
   $(\Mred[2],\omega_{\mathrm{red}_2})$.
\end{satz}
\begin{proof}
   Für $c \in C = J^{-1}(0)$ ist mit Bemerkung \ref{bem:PiAufC}
   $\pi_1(c) \in C_2$, also gilt für $f \in \CM[C_2]$ schon
   \begin{align*}
      f(\pi_1(c)) = \prol[2] f (\pi_1(c))\Fdot
   \end{align*}
   Damit folgt
   \begin{align*}
      \tilde{\varsigma}^* f(c) = f(\pi_1(c)) = \prol[2] f (\pi_1(\kjIn(c))) =
      \kjRes \pi_1^* \prol[2] f(c) \Fcom
   \end{align*}
   also für $\phi \in \CM[{\Mred[2]}]$
   \begin{align*}\label{loc:Stern}
      \pi^*\varsigma^* \phi = \tilde{\varsigma}^* \pi_2^* \phi = \kjRes
      \pi_1^* \prol[2] \pi_2^* \phi\Fdot \tag{$*$}
   \end{align*}
   Seien $\phi_a,\phi_b \in \CM[{\Mred[2]}]$, dann gilt:
   \begin{align*}
      \lefteqn{\pi^*(\phi_a (\varsigma^* \starred[2]) \phi_b)} \\
      &= \pi^* \varsigma^*({\varsigma^{-1}}^* \phi_a \starred[2]
      {\varsigma^{-1}}^*
      \phi_b)\\
      &= \kjRes \pi_1^* \prol[2] \pi_2^*({\varsigma^{-1}}^* \phi_a
      \starred[2] {\varsigma^{-1}}^* \phi_b) && \eAnn{nach \ref{loc:Stern}} \\
      &= \kjRes \pi_1^* \prol[2] \qRes_2 (\prol[2]
      \pi_2^*{\varsigma^{-1}}^* \phi_a \starred[1] \prol[2]
      \pi_2^*{\varsigma^{-1}}^*
      \phi_b) && \eAnn{nach Gleichung \eqref{eq:SternproduktAufDemReduziertenPhasenraum3}}\\
      &= \qjRes \pi_1^* (\prol[2] \pi_2^*{\varsigma^{-1}}^* \phi_a
      \starred[1] \prol[2] \pi_2^*{\varsigma^{-1}}^* \phi_b) &&
      \eAnn{nach Lemma
        \ref{lem:RelationenMitQuantenj}~\refitem{item:QuantenJEigenschaft2}} \\
      &= \qjRes \qResE (\prol[1] \pi_1^* \prol[2]
      \pi_2^*{\varsigma^{-1}}^* \phi_a \star \prol[1] \pi_1^* \prol[2]
      \pi_2^*{\varsigma^{-1}}^*
      \phi_b)   && \eAnn{nach Gleichung \eqref{eq:SternproduktAufDemReduziertenPhasenraum3}} \\
      &= \qRes (\prol \pi^* \phi_a \star \prol \pi^* \phi_b) &&
      \eAnn{nach Lemma
        \ref{lem:RelationenMitQuantenj}~\refitem{item:QuantenJEigenschaft1}}[]
      \\
      &\phantom{=} &&\eAnn[]{und Satz
        \ref{satz:prol2Komp}~\refitem{item:prol2Komp}} \\
      &= \pi^*(\phi_a \starred \phi_b) && \eAnn{nach Gleichung
        \eqref{eq:SternproduktAufDemReduziertenPhasenraum3}} \Fdot
   \end{align*}

\end{proof}
\begin{bemerkung}
   \label{bem:Hauptresultat}
   Es sei darauf hingewiesen, dass Satz
   \ref{thm:UebereinstimmungDerSternprodukte} nicht nur etwa besagt,
   dass $\varsigma^* \starred[2]$ und $\starred$ äquivalent sind,
   sondern dass sie wirklich übereinstimmen. Dies ist auch von
   physikalischem Interesse, da zwei äquivalente Sternprodukte
   physikalisch inäquivalent sein können, wie schon in Kapitel
   \ref{cha:Deformationsquantisierung} bemerkt wurde. Ferner ist der
   Symplektomorphismus aus Satz
   \ref{thm:UebereinstimmungDerSternprodukte} genau derselbe wie der
   klassische aus Satz \ref{satz:redPhasenraeumeSymplektomorph}.  Die
   Einschränkung in Satz \ref{thm:UebereinstimmungDerSternprodukte} im
   Vergleich zur klassischen Situation besteht in den zusätzlichen
   Forderungen an die Gruppenstruktur, die im Falle kompakter Gruppen
   jedoch automatisch erfüllt sind, vgl.\ Bemerkung
   \ref{bem:invariantesKomplement}. Ansonsten gibt es keine weiteren
   Obstruktionen, insbesondere keine die auf Quanteneffekte schließen
   lassen.
\end{bemerkung}

\cleardoublepage
\begin{appendices}

\renewcommand{\appendixname}{}

\cleardoublepage

\numberwithin{theorem}{chapter}
\chapter{Elementare mathematische Anmerkungen}
\label{cha:elementMathAnmerkungen}

In diesem Kapitel wollen wir einige einfache mathematische Bemerkungen
zusammenstellen, die für das Verständnis des Haupttextes und der
weiteren Anhänge nützlich sein können, aber in der Literatur oft nicht
ausführlich thematisiert werden oder nur als Übungsaufgaben vorkommen.

\begin{proposition}
      \label{prop:produkte}
      Seien $M_i$, $N_i$, $i=1,2$, $M$ und $N$ Mannigfaltigkeiten,
      $\projFaktor_i \dpA M_1 \times M_2 \to M_i$ die Projektion auf den
      $i$"=ten Faktor, $\pi_i \dpA TM_i \to M_i$ ($i = 1,2$) und $\pi
      \dpA TM \to M$ die Fußpunktprojektion sowie $\Delta \dpA M \to M
      \times M $, $p \mapsto (p,p)$ die Diagonalabbildung.
      \begin{propositionEnum}
      \item
         \begin{align*}
            \tangentIso \dpA T(M_1 \times M_2) \to TM_1 \times TM_2,
            \quad X \mapsto (T\projFaktor_1 X, T\projFaktor_2 X)
         \end{align*}
         ist ein Vektorbündelisomorphismus (über $\id$).

         Wir schreiben auch $\TIso{X} = I(X)$ sowie $\TIsoI{(X_1,X_2)} =
         I^{-1}((X_1,X_2))$ für $X \in T(M_1 \times M_2)$
         bzw.\ $(X_1,X_2) \in TM_1 \times TM_2$.
      \item \label{item:KommaKurve} Sei $J \subset \mathbb{R}$ offen und
         $\gamma_i \dpA J \to M_i$ ($i=1,2$) eine glatte Kurve, dann
         gilt
         \begin{align*}
            \ddt (\gamma_1(t),\gamma_2(t)) = \TIsoI{(\dot \gamma_1(0),
              \dot \gamma_2(0))} \Fdot
         \end{align*}
      \item%
         \label{item:ProdukteKomma} %
         Ist $\Psi = (\Psi_1,\Psi_2) \dpA M \to M_1 \times M_2$ glatt,
         so gilt für $X \in TM$
         \begin{align*}
            T\Psi X = \TIsoI{(T\Psi_1 X, T\Psi_2 X)} \Fdot
         \end{align*}
      \item%
         \label{item:vonKomma}%
         Ist $\Psi \dpA M_1 \times M_2 \to M$ glatt, so gilt
         \begin{align*}
            T\Psi \TIsoI{(X_1,X_2)} = T\Psi(\cdot,\pi_2(X_2))X_1 +
            T\Psi(\pi_1(X_1),\cdot)X_2
         \end{align*}
         für $X_i \in TM_i$, $i=1,2$.
      \item %
         \label{item:ProductAbbildung} %
         Ist $\Psi_i \dpA M_i \to N_i$, $i=1,2$, glatt, so gilt
         \begin{align*}
            I \circ T(\Psi_1 \times \Psi_2) \circ I^{-1} = T\Psi_1
            \times T\Psi_2 \Fdot
         \end{align*}
         oder anders geschrieben:
       \begin{align*}
          T(\Psi_1 \times \Psi_2)\TIsoI{(X_1,X_2)} = \TIsoI{(T\Psi_1
            X_1, T\Psi_2 X_2)}
       \end{align*}
       für $(X_1,X_2) \in TM_1 \times TM_2$.
    \item%
       \label{item:Diagonale}%
       Es gilt $T\Delta X = \TIsoI{(X,X)}$ für $X \in TM$.
    \item %
       \label{item:DiagonaleKomponiert}%
       Ist $\Psi \dpA M \times M \to N $ glatt, so gilt für $X \in TM$
       \begin{align*}
          T (\Psi \circ \Delta) X = T\Psi(\cdot,\pi(X))X +
          T\Psi(\pi(X),\cdot)X \Fdot
       \end{align*}
    \end{propositionEnum}
 \end{proposition}
 \begin{proofklein}
    \begin{beweisEnum}

    \item Die Surjektivität folgt aus der Surjektivität der einzelnen
       Komponenten, aus Dimensionsgründen folgt deswegen auch die
       Bijektivität, der Rest ist klar. %

    \item %
       \begin{align*}
          \TIso {\left(\ddt (\gamma_1,\gamma_2)\right)} &=
          (T\projFaktor_1 \ddt
          (\gamma_1,\gamma_2), T\projFaktor_2 \ddt (\gamma_1,\gamma_2)) \\
          &= (\ddt(\projFaktor_1 \circ (\gamma_1,\gamma_2),\ddt
          (\projFaktor_2 \circ
          (\gamma_1,\gamma_2))) \\
          &= (\dot \gamma_1(0), \dot \gamma_2(0)) \Fdot
       \end{align*}
    \item %
       Für eine glatte Kurve $\gamma \dpA \mathbb{R} \to M$ gilt mit
       Teil~\refitem{item:KommaKurve}
       \begin{align*}
          \TIso{(T\Psi \dot \gamma(0))} &= \TIso{\left(\ddt (\Psi \circ \gamma)\right)} \\
          &= \TIso{\left (\ddt (\Psi_1\circ \gamma, \Psi_2 \circ \gamma)\right)}\\
          &= \left(\ddt\Psi_1\circ \gamma, \ddt \Psi_2 \circ \gamma\right) \\
          &= (T\Psi_1 \dot \gamma(0), T\Psi_2 \dot \gamma(0)) \Fdot
       \end{align*}
    \item %
       Sei $\gamma_i \dpA \mathbb{R} \to M_i$ ($i=1,2$) eine Kurve mit
       $\dot\gamma_i (0) = X_i$.  Es sei dann $p_i := \pi_i(X_i)$ und
       $\gamma_{1,p_2}, \gamma_{2,p_1} \dpA \mathbb{R} \to M_1 \times
       M_2$ definiert durch $\gamma_{1,p_2}(t) = (\gamma_1(t),p_2)$
       bzw.\ $\gamma_{2,p_1}(t) = (p_1,\gamma_2(t))$ für $t \in
       \mathbb{R}$.
      \begin{align*}
        T\Psi \TIsoI{(X_1,X_2)} &= T\Psi \TIsoI{(X_1,0)} +
        T\Psi \TIsoI{(0,X_2)} \\
        &\stackrel{\refitem{item:KommaKurve}}{=} T\Psi\dot
        \gamma_{1,p_2}(0) + T\Psi\dot\gamma_{2,p_1}(0) \\
        &= \ddt (\Psi(\gamma_1(t),p_2)) + \ddt (\Psi(p_1,\gamma_2(t))) \\
        &= T\Psi(\cdot,\pi_2(X_2))X_1 + T\Psi(\pi_1(X_1),\cdot)X_2
        \Fdot
      \end{align*}
   \item%
      \begin{align*}
         T(\Psi_1 \times \Psi_2)\TIsoI{(X_1,X_2)} &\stackrel{\refitem{item:vonKomma}}{=}
         T(\Psi_1 \times \Psi_2)(\cdot,\pi_2(X_2))X_1 + T(\Psi_1 \times
         \Psi_2)(\pi_1(X_1),\cdot)X_2 \\
         &= T(\Psi_1,\Psi_2(\pi_2(X_2)))X_1 + T(\Psi_1(\pi_1(X_1)),\Psi_2)X_2 \\
         &\stackrel{\refitem{item:ProdukteKomma}}{=} \TIsoI{(T\Psi_1 X_1,0)} + \TIsoI{(0,T\Psi_2 X_2)}  \\
         &= \TIsoI{(T\Psi_1X_1,T\Psi_2X_2)} \Fdot
      \end{align*}
   \item%
      \begin{align*}
         T\Delta X \stackrel{\refitem{item:ProdukteKomma}}{=} \TIsoI{(T
           \id X , T\id X)} = \TIsoI{(X,X)} \Fdot
      \end{align*}
   \item%
      \begin{align*}
         T(\Psi \circ \Delta) X = T\Psi T\Delta X
         \stackrel{\refitem{item:Diagonale}}{=} T\Psi \TIsoI{(X,X)}
         \stackrel{\refitem{item:vonKomma}}{=} T\Psi(\cdot,\pi(X))X +
         T\Psi(\pi(X),\cdot)X \Fdot
      \end{align*}
   \end{beweisEnum}
\end{proofklein}

\begin{proposition}

   \label{prop:ProduktVonSubmersionen}
   Seien $f_i \dpA M_i \to N_i$, ($i=1,2$), glatte Submersionen zwischen
   den Mannigfaltigkeiten $M_i, N_i$, dann ist auch $f_1 \times f_2 \dpA
   M_1 \times M_2 \to N_1 \times N_2$ eine glatte Submersion.
\end{proposition}
\begin{proofklein}
   Sei $(p_1,p_2) \in M_1 \times M_2$ und $\TIsoI{(Y_1,Y_2)} \in
   T_{(f_1(p_1),f_2(p_2))}(N_1 \times N_2)$. Dann gibt es nach
   Voraussetzung $X_i \in T_{p_i}M_i$ mit $T_{p_i}f_i X_i = Y_i$
   ($i=1,2$) und es gilt nach Proposition \ref{prop:produkte}~\refitem{item:ProductAbbildung}
   \begin{align*}
      T_{(p_1,p_2)}(f_1 \times f_2)\TIsoI{(X_1,X_2)} =
      \TIsoI{(T_{p_1}f_1 X_1, T_{p_2}f_2 X_2)} = \TIsoI{(Y_1,Y_2)} \Fdot
   \end{align*}
\end{proofklein}

\begin{definition}
   \label{def:InduzierteBuendelAbbildung}

   Sei $k \in \mathbb{N}$ und seien $E_i \to M$ und $F \to M$
   Vektorbündel über der Mannigfaltigkeit $M$ für $i \in
   \{1,\dots,k\}$. Sei weiter
   \begin{align*}
      \varphi \dpA E_1\oplus \dots \oplus E_k \to F
   \end{align*}
   ein Vektorbündelmorphismus über $\id$. Dann wollen wir die
   $\CM$"=lineare Abbildung
   \begin{align*}
      \Gamma^\infty(E_1)\oplus \dots \oplus \Gamma^\infty(E_k) &\to
      \Gamma^\infty(F)\\%
      (s_1,\dots,s_k) &\mapsto (M \ni p \mapsto
      \varphi(s_1(p),\dots, s_k(p)) \in F)
   \end{align*}
   die \neuerBegriff{von $\varphi$ induzierte Abbildung} nennen und den
   Notationsmissbrauch begehen, sie wieder mit $\varphi$ zu bezeichnen.
\end{definition}

\begin{proposition}
   \label{prop:RechenregelnSchnitteUndSummen}
   Sei $n \in \mathbb{N}$ und $M$ eine Mannigfaltigkeit. Sei $E
   \to M$, $F \to M$ und $E_i \to M$ für $i \in \{1,\dots,n\}$ jeweils ein
   Vektorbündel über $M$. Dann sind die folgenden Aussagen richtig.
   \begin{propositionEnum}
   \item %
      \label{item: RechenregelnSchnitteUndSummen} %
      Die Abbildung
      \begin{align*}
         \bigoplus_{i = 1}^n \Gamma^\infty(E_i) &\to
         \Gamma^\infty(\bigoplus_{i = 1}^n E_i), \\ (s_1,\dots,s_n)
         &\mapsto (M \ni p \mapsto (s_1(p),\dots,s_n(p)))
      \end{align*}
      ist ein Isomorphismus von $\CM$"=Moduln.
   \item %
      \label{item: RechenregelnTensor} %
      Seien  $\iota_E \dpA E \hookrightarrow E \oplus F$ und $\iota_F
      \dpA F \hookrightarrow E \oplus F$ die kanonischen Inklusionen und
      sei $k \in \mathbb{N}$. Dann ist für $l \in \{0,\dots,k\}$ die
      Abbildung
      \begin{align*}
         \iota_{lk}:= \vee_{i=0}^{k-l}\iota_E \vee \vee_{i=0}^{l} \iota_F \dpA \Bigvee^{k-l} E \otimes \Bigvee^{l} F  \to
         \Bigvee^k(E \oplus F)
       \end{align*}
       ein injektiver Vektorbündelmorphismus und
      \begin{align*}
         \oplus_{l=0}^k\iota_{lk} \dpA \bigoplus_{l=0}^k(\Bigvee^{k-l} E \otimes \Bigvee^l F) \to
         \Bigvee^k(E \oplus F)
      \end{align*}
      ist ein Vektorbündelisomorphismus.
   \end{propositionEnum}

\end{proposition}
\begin{proofklein}
Klar.
\end{proofklein}

\cleardoublepage

\chapter{Differentialoperatoren}
\label{cha:Differentialoperatoren}

In diesem Kapitel definieren wir den Begriff des Differentialoperators
und stellen ein paar wenige Propositionen dazu bereit, soweit sie für
die vorliegende Arbeit von Belang sind.

\begin{definition}[Differentialoperator]
   \label{def:Differentialoperator}
   Sei $\mathcal{A}$ eine assoziative, kommutative Algebra über einem
   Körper $\mathbb{K}$ und sei $\mathsf{L}_a \colon \mathcal{A} \in b
   \mapsto ab \in \mathcal{A}$ für alle $a \in \mathcal{A}$. Dann werden
   die \neuerBegriff{Differentialoperatoren}
   $\mathrm{DiffOp}^k(\mathcal{A})$ der Ordnung $k \in \mathbb{Z}$ auf
   $\mathcal{A}$ induktiv durch
   \begin{align}
      \label{eq:DiffopNiedrigeOrdnung}
      \mathrm{DiffOp}^k(\mathcal{A}) = \{0\}
   \end{align}
für $k < 0$ und
   \begin{align}
      \label{eq:DiffopHoehereOrdnung}
    \mathrm{DiffOp}^k(\mathcal{A}) = \left\{  D \in
           \operatorname{\mathrm{End}}_{\mathbb{K}}(\mathcal{A})  \middle
            | \; [D,\mathsf{L}_a] \in \mathrm{DiffOp}^{k-1}(\mathcal{A}
           )
            \quad \text{für alle $a \in \mathcal{A}$} \right\}
   \end{align}
für $k \geq 0$ definiert.
\end{definition}
\begin{bemerkung}
   \label{bem:Multidifferentialoperator}
   Analog definiert man \neuerBegriff{Multidifferentialoperatoren}, siehe \cite[Anhang A.5]{waldmann:2007a}.
\end{bemerkung}

\begin{proposition}
   \label{prop:EigenschaftenVonDiffops}
   Sei $\mathcal{A}$ eine assoziative, kommutative Algebra über einem
   Körper $\mathbb{K}$. Dann sind folgende Aussagen richtig.
   \begin{propositionEnum}
      \item %
         Es ist $\id \in \mathrm{DiffOp}^0(\mathcal{A})$ und für alle $a
         \in \mathcal{A}$ auch $\mathsf{L}_a \in
         \mathrm{DiffOp}^0(\mathcal{A})$.
      \item %
         Die Differentialoperatoren der Ordnung $k \in \mathbb{Z}$ sind
         auf natürliche Weise ein Bimodul, wobei $a \cdot D \cdot b$ für
         $a,b \in \mathcal{A}$  und $D \in
         \mathrm{DiffOp}^k(\mathcal{A})$ durch
         \begin{align}
            \label{eq:DiffOpBimodul}
            a \cdot D \cdot b = \mathsf{L}_a \circ D \circ \mathsf{L}_b
         \end{align}
         definiert ist.
      \item %
         Die Komposition zweier Differentialoperatoren ist wieder ein
         Differentialoperator. Genauer gilt für $k,l \in \mathbb{Z}$ und
         $D \in \mathrm{DiffOp}^k(\mathcal{A})$ und $D' \in
         \mathrm{DiffOp}^l(\mathcal{A})$
         \begin{align}
            \label{eq:KompositionVonDiffops}
            D \circ D' \in \mathrm{DiffOp}^{k+l}(\mathcal{A}) \Fdot
         \end{align}
   \end{propositionEnum}
\end{proposition}
\begin{proofklein}
   Siehe \cite[Prop. A.2.4]{waldmann:2007a}
\end{proofklein}

Ein wichtiges Beispiel für die assoziative, kommutative Algebra
$\mathcal{A}$, sind die glatten Funktionen $C^\infty(M)$ auf einer
Mannigfaltigkeit $M$. Die nächste Proposition, vgl.\ \cite[Satz A.3.5,
Prop. A.3.6]{waldmann:2007a}, gibt eine Charakterisierung von
Differentialoperatoren auf $C^\infty(M)$ und liefert gleichzeitig den
Anschluss an die elementare analytische Definition von
Differentialoperatoren auf dem $\mathbb{R}^n$.

\begin{proposition}
   \label{prop:CharDiffop}
   Sei $M$ eine Mannigfaltigkeit und sei $D \colon C^\infty(M) \to
   C^\infty(M)$ eine lineare Abbildung. Dann gilt genau dann $D \in
   \mathrm{DiffOp}^k(C^\infty(M))$, wenn es einen Atlas von $M$ gibt,
   so dass für jede Karte $(U,x)$ dieses Atlas Funktionen $D_U^{i_1\dots
     i_r}\in C^\infty(U)$ für $r = 0,\dots,k$ derart existiert, dass
   \begin{align}
      \label{eq:CharDiffop}
      Df\at{U} = \sum_{r = 0}^k \frac{1}{r!} D_U^{i_1\dots i_r}
      \frac{\partial^r f\at{U}}{\partial x^{i_1}\dotsm \partial x^{i_r}}
   \end{align}
   gilt.
\end{proposition}

\begin{definition}[Lokale Abbildung]
   \label{def:LokaleAbbildung}
   Ist $M$ eine Mannigfaltigkeit, so heißt eine lineare Abbildung $D
   \colon C^\infty(M) \to C^\infty(M)$ \neuerBegriff{lokal}, falls für
   alle $f \in C^\infty(M)$
   \begin{align}
      \label{eq:LokaleAbbildung}
      \supp Df \subset \supp f
   \end{align}
   erfüllt ist.
\end{definition}

\begin{proposition}
   \label{prop:DiffopsLokal}
   Sei $M$ eine Mannigfaltigkeit, dann ist jeder Differentialoperator
   auf $C^\infty(M)$  lokal.
\end{proposition}
\begin{proofklein}
   Siehe \cite[Lemma A.3.2]{waldmann:2007a}.
\end{proofklein}

\begin{proposition}[Einschränkbarkeit lokaler Abbildungen]
   \label{prop:einschrLokalerAbbbildungen}
   Ist $M$ eine Mannigfaltigkeit und $D \colon C^\infty(M) \to
   C^\infty(M)$ eine lokale lineare Abbildung, so gibt es für jedes offene $U
   \subset M$ eine eindeutig bestimmte lokale lineare Abbildung $D_U \colon
   C^\infty(U) \to C^\infty(U)$ mit $D_M = D$ und
   \begin{align}
      \label{eq:einschrLokalerAbbildungen}
      D_V f\at{V} = (D_U f)\at{V}
   \end{align}
   für jedes offene $V \subset U$ und alle $f \in C^\infty(U)$. Ist $D
   \in \mathrm{DiffOp}^k(C^\infty(M))$ sogar ein Differentialoperator, so
   folgt $D_U \in \mathrm{DiffOp}^k(C^\infty(U))$ für alle offenen
   Teilmengen $U \subset M$.
\end{proposition}
\begin{proofklein}
   Siehe \cite[Prop. A. 3.3]{waldmann:2007a}
\end{proofklein}

\begin{proposition}[Fortsetzung von Differentialoperatoren]
   \label{prop:VortsetzungVonDiffops}
   Sei $M$ eine Mannigfaltigkeit und $U \subset M$ offen, $k \in
   \mathbb{Z}$ und $D_U \in \mathrm{DiffOp}^k(C^\infty(U))$ ein
   Differentialoperator. Weiter sei $\chi_U \colon M \to \mathbb{R}$ eine
   glatte Funktion mit $\supp \chi_U \subset U$ und $\kIn_U \colon U
   \hookrightarrow M$ bezeichne die kanonische Inklusion. Dann ist
   $\chi_U D_U \circ \kIn_U^* \in \mathrm{DiffOp}^k(C^\infty(M))$.
\end{proposition}
\begin{proofklein}
   Wir zeigen die Aussage durch Induktion nach $k$. Für $k < 0$ ist
   sie trivial. Es gelte die Aussage also für ein $k-1$ mit $ k \in
   \mathbb{N}$. Wir zeigen, dass daraus folgt, dass sie auch schon für
   $k$ gilt. Dies ist aber klar nach der Definition von
   Differentialoperatoren, denn seien $f,g \in C^\infty(M)$, dann gilt
   \begin{align*}
      [\chi_U D_U \circ \kIn_U^*,\mathsf{L}_f]g &= \chi_U
      D_U(\kIn_U^*f \cdot
      \kIn_U^* g) - f \chi_U D_U(\kIn_U^* g) \\
      &= \chi_U(D_U(\kIn_U^* f  \cdot \kIn_U^*g)- \kIn_U^*f D_U(\kIn_U^*g)) =
      \chi_U([D_U,\mathsf{L}_{\kIn_U^* f}]\kIn_U^* g) \Fdot
   \end{align*}
   Damit folgt schon die Behauptung, denn es gilt
   $[D_U,\mathsf{L}_{\kIn_U^* f}] \in
   \mathrm{DiffOp}^{k-1}(C^\infty(U))$ nach Voraussetzung und

     $ \chi_U[D_U,\mathsf{L}_{\kIn_U^* f}]\kIn_U^* \in
     \mathrm{DiffOp}^{k-1}(C^\infty(M))$ nach
   Induktionsannahme.
\end{proofklein}

\cleardoublepage

\chapter{Hauptfaserbündel}
\label{cha:Hauptfaserbuendel}

In diesem Kapitel möchten wir in aller Kürze die für das Verständnis
dieser Arbeit notwendigen Grundlagen über Hauptfaserbündel
zusammenstellen. Für Details und weiterführende Betrachtungen verweisen
wir auf die Literatur, etwa \cite{kobayashi.nomizu:1963a},
\cite{baum:2009}, \cite{sharpe1997differential}, \cite{michor:2007}.
Interessant im Hinblick auf Anwendung in der Physik sind auch
\cite{schottenloher:1995a}, \cite{bleecker1981gauge}, \cite{naber:1997a}
und \cite{naber:2000}.

\begin{definition}[Hauptfaserbündel]
   \label{def:Hauptfaserbuendel}
   Ein \neuerBegriff{(Links-)Hauptfaserbündel} $(P,\pi,B,G)$ mit
   \neuerBegriff{Totalraum} $P$, \neuerBegriff{Basis} $B$,
   \neuerBegriff{Fußpunktprojektion} $\pi$ und
   \neuerBegriff{Strukturgruppe} $G$ besteht aus zwei
   Mannigfaltigkeiten $P$ und $B$, einer Wirkung einer
   Lie-Gruppe $G$ auf $P$ sowie einer glatten Abbildung $\pi \colon P \to B$,
   so dass folgendes gilt:
   \begin{definitionEnum}
   \item %
      Für jedes $p \in P$ und $g \in G$ gilt $\pi(gp) = \pi(p)$.
   \item %
      Für jedes $b \in B$ gibt es eine offene Umgebung $U$ und einen
      $G$"=äquivarianten Diffeomorphismus $\Psi \colon \pi^{-1}(U) \to U
      \times G$, der eine lokale Trivialsisierung von $\pi$ über $U$ mit
      typischer Faser $G$ ist, d.\,h.\ dass das Diagramm
      \def\tA[#1]{A_{#1}}
\begin{equation}
      \begin{tikzpicture}[baseline=(current
    bounding box.center),description/.style={fill=white,inner sep=2pt}]
         \matrix (m) [matrix of math nodes, row sep=3.0em, column
         sep=3.5em, text height=1.5ex, text depth=0.25ex]
         {
\pi^{-1}(U) & U \times G \\
U &  \\
}; %

 \path[->] (m-1-1) edge node[left] {$\pi$} (m-2-1); %
\path[->] (m-1-1) edge node[auto]{$\Psi$}(m-1-2); %
\path[->] (m-1-2) edge node[auto]{$\mathrm{pr}_1$}(m-2-1); %
 \end{tikzpicture}
\end{equation}
      kommutiert. Dabei operiert $G$ auf $\pi^{-1}(U)$ durch
      Einschränkung der gegebenen $G$"=Wirkung und auf $U \times G$ durch
      $g(p,h) := (p,gh)$, d.\,h.\ durch Linksmultiplikation auf dem zweiten
      Faktor. Die Abbildung $\mathrm{pr}_1$ ist die kanonische
      Projektion auf den ersten Faktor. Das Paar $(U,\Psi)$ heißt auch
      \neuerBegriff{lokale Trivialsierung} des Hauptfaserbündels.
   \end{definitionEnum}
\end{definition}

\begin{bemerkung}
   \label{bem:Hauptfaserbuendel}
   \begin{bemerkungEnum}
      \item %
         Analog zum Links"=Hauptfaserbündel werden auch
         \neuerBegriff{Rechts"=Haupt\-fa\-ser\-bün\-del} definiert. Diese sind in
         der Literatur geläufiger.
      \item %
         Statt von Hauptfaserbündel spricht man auch von
         \neuerBegriff{Prinzipalfaserbündel}
         (vgl.\ \cite{schottenloher:1995a}) oder
         \neuerBegriff{Prinzipalbündel} (vgl.\ \cite{tomdieck:2000a})
         möchte man die Strukturgruppe $G$ betonen, sagt man auch
         $G$"=Hauptfaserbündel oder kurz $G$"=Bündel
         (vgl.\ \cite{naber:2000}).
      \item %
         In der Literatur wird oft noch gefordert, dass die Wirkung auf
         $P$ faserweise transitiv sei, oder auch, dass sie frei
         sei. Beides sind jedoch Eigenschaften, die es genügt in einer
         lokalen Trivialsierung zu prüfen und man erkennt unschwer, dass
         sie schon aus der gegebenen Definition folgen, vgl.\ auch
         \cite{naber:1997a},\cite{greub1973connections}.
      \item %
         Ähnlich wie bei Vektorbündeln gibt es auch bei Hauptfaserbündeln
         konstruktive Charakterisierungen vermöge Kozykeln. Siehe etwa
         \cite{baum:2009} für drei äquivalente Definitionen von
         Hauptfaserbündel. Man beachte, dass diese Autorin fordert, dass
         die Wirkung einfach transitiv auf den Fasern sei. Dies ist aber
         äquivalent dazu, dass sie frei ist und transitiv auf den Fasern.
   \end{bemerkungEnum}
\end{bemerkung}

Die nächste Proposition und der darauffolgende Satz klären die wichtige
Beziehung zwischen Hauptfaserbündeln und Quotienten freier und
eigentlicher Gruppenwirkungen.

\begin{proposition}
   \label{prop:HauptfaserBuendelFreiUndEigentlich}
   Sei $(P,\pi,B,G)$ eine Hauptfaserbündel. Dann ist die $G$"=Wirkung auf
   $P$ frei und eigentlich.

\end{proposition}
\begin{proofklein}
   In Bemerkung \ref{bem:Hauptfaserbuendel} wurde schon gesagt, dass
   die Wirkung frei ist. Dass die Wirkung eigentlich ist, ergibt sich
   durch eine einfache Anwendung des Folgenkriteriums für eigentliche
   Wirkungen, siehe \cite[Ch. 4, Prop. 2.3]{sharpe1997differential}.

\end{proofklein}

\begingroup
\emergencystretch=0.8em
\begin{satz}
   \label{satz:FreiEigentlichHauptfaserBuendel}
   Sei $P$ eine Mannigfaltigkeit, $G$ eine Lie-Gruppe, die auf $M$
   frei und eigentlich wirke. Dann ist $P/G$ eine topologische
   Mannigfaltigkeit der Dimension $\dim P - \dim G$ und besitzt eine
   eindeutig bestimmte glatte Struktur, so dass die kanonische
   Projektion $\pi \colon P \to P/G $
   auf den Quotienten  eine glatte Submersion wird.
   Ferner ist $(P,\pi,P/G,G)$ ein Hauptfaserbündel.
\end{satz}
\begin{proofklein}
   Für einen Beweis dieses nicht ganz trivialen Satzes verweisen wir auf
   \cite[Satz 3.3.18]{waldmann:2007a}, \cite[Thm. 7.10]{lee:2003a} sowie
   \cite[Prop. 4.1.23]{abraham.marsden:1985a} und \cite[Appendix E,
   Thm. 4.2.4]{sharpe1997differential}.

\end{proofklein}
\endgroup

\begin{definition}[Vertikalbündel]
   \label{def:Vertikalbuendel}
   Ist $(P,\pi,B,G)$ ein Hauptfaserbündel, so heißt das
   Untervektorbündel $VP := \ker T \pi \subset TP$ des Tangentialbündels
   $TP$ von $P$ \neuerBegriff{Vertikalbündel}. Die Fasern von $VP$
   heißen \neuerBegriff{Vertikalräume}. Ein Vektorfeld $Y \in
   \Gamma^\infty(TP)$ auf $P$ heißt \neuerBegriff{vertikal}, falls $Y(p)
   \in V_pP$ für alle $p \in P$ gilt, die bedeutet $Y \in \Gamma^\infty(VP)$.
\end{definition}

\begin{definition}[Hauptfaserbündelzusammenhang]
   \label{def:Horizontalbuendel}

   Ist $(P,\pi,B,G)$ ein Hauptfaserbündel, so nennt man ein
   Untervektorbündel $HP$ von $TP$ \neuerBegriff{Horizontalbündel}
   oder auch \neuerBegriff{Hauptfaserbündelzusammenhang}, falls es ein
   Komplement von $VP$ in $TP$ ist, d.\,h.\ wenn $TP = VP \oplus HP$ gilt
   und falls es invariant unter der gelifteten $G$"=Wirkung ist. Ist
   $\Phi \colon G \times P \to P$ die $G$"=Wirkung auf $P$, so bedeutet
   dies, dass $T_p\Phi_g H_pP \subset H_{gp}P$ gilt für alle $p \in P$
   und $g \in G$. Ein Vektorfeld $Y \in \Gamma^\infty(TP)$ auf $P$
   heißt \neuerBegriff{horizontal}, falls $Y(p) \in H_pP$ für alle $p
   \in P$ gilt, dies bedeutet kurz $Y \in \Gamma^\infty(HP)$.
\end{definition}

Natürlich stellt sich unmittelbar die Frage ob auch immer ein
Hauptfaserbündelzusammenhang existiert. Dies beantwortet die folgende
Proposition.

\begin{proposition}[Existenz von Hauptfaserbündelzusammenhängen]
   \label{prop:ExistenzHauptfaserbuendelzusammenhang}
   Für jedes Hauptfaserbündel existiert ein Hauptfaserbündelzusammenhang.
\end{proposition}
\begin{proofklein}
   Man konstruiert den Zusammenhang lokal und globalisiert dann unter
   Verwendung einer Zerlegung der Eins.  Siehe \cite[Thm. II
   2.1]{kobayashi.nomizu:1963a} oder auch \cite[Thm. 3.1.7]{naber:2000}
   für Details.
\end{proofklein}

Eine für diese Arbeit relevante Charakterisierung von
Hauptfaserbündelzusammenhängen lässt sich in Termen Lie-Algebra-wertiger
Einsformen geben.

\begin{definition}[Zusammenhangseinsform]

   \label{def:Zusammenhangseinsform}
   Sei $(P,\pi,B,G)$ ein Hauptfaserbündel und sei $\lieAlgebra$ die
   Lie-Algebra von $G$.
   Eine Lie"=Algebra-wertige Einsform
   \begin{align*}
      \gamma \colon TP \to \lieAlgebra
   \end{align*}
   heißt \neuerBegriff{Zusammenhangseinsform}, falls die folgenden
   Bedingungen gelten.
   \begin{definitionEnum}
   \item %
      Für alle $\xi \in \lieAlgebra$ und $p \in P$ gilt
      $\gamma(\xi_P(p))=\xi$.
   \item %
      $\gamma$ ist $G$"=äquivariant, d.\,h.\ es gilt $\gamma(T_p\Phi_g v) =
      \Ad_g(\gamma(v))$ für alle $g \in G$, $p \in P$ und $v \in T_pP$.
   \end{definitionEnum}
\end{definition}

\begin{proposition}
   \label{prop:HauptfaserBuendelZusammenhang}
   Sei $(P,\pi,B,G)$ ein Hauptfaserbündel, dann sind die folgenden
   Aussagen richtig.
   \begin{propositionEnum}
      \item %
         Sei $HP \subset TP$ ein Hauptfaserbündelzusammenhang für
         $P$. Dann gibt es eine eindeutig bestimmte
         Zusammenhangseinsform $\gamma_{HP} \colon TP \to \lieAlgebra$ mit
         \begin{align}
            \label{eq:ZusammenhangZuZusammenhangseinsform}
            \gamma_{HP}(\xi_P(p) + Y) = \xi \quad \forall p \in P, \xi \in
            \lieAlgebra, Y \in T_pH \Fdot
         \end{align}
         \item %
            Sei $\gamma \colon TP \to \lieAlgebra$ eine
            Zusammenhangseinsform. Dann definiert $HP_{\gamma} := \ker \gamma$
            ein Hauptfaserbündelzusammenhang.
         \item %
            Ist $HP$ ein Hauptfaserbündelzusammenhang, so gilt
            $HP_{\gamma_{HP}} = HP$. Ist $\gamma$ eine
            Zusammenhangseinsform, so gilt $\gamma_{HP_{\gamma}} = \gamma$.
   \end{propositionEnum}
\end{proposition}
\begin{proofklein}
   Siehe \cite[Satz 3.2]{baum:2009}
\end{proofklein}
Aufgrund der letzten Proposition sind die Konzepte
Hauptfaserbündelzusammenhang und Zusammenhangseinsform äquivalent,
so dass wir die Indices aus Proposition
\ref{prop:HauptfaserBuendelZusammenhang} weglassen wollen und den von
einer Zusammenhangseinsform  $\gamma$ induzierten
Hauptfaserbündelzusammenhang einfach mit $HP$ bezeichnen statt mit
$HP_{\gamma}$, genauso umgekehrt.

Mit Hilfe eines Hauptfaserbündelzusammenhangs ist es möglich,
Vektorfelder von der Basismannigfaltigkeit auf den Totalraum zu heben.

\begin{definition}[Horizontale Hebung]
   \label{def:horizontaleHebung}
   Sei $\gamma \colon TP \to \lieAlgebra$ eine Zusammenhangseinsform
   eines Hauptfaserbündels $(P,\pi,B,G)$ und $X \in \Gamma^\infty(TB)$
   ein Vektorfeld auf der Basis $B$. Das eindeutig bestimmte horizontale
   Vektorfeld $X^{\mathsf{h}} \in \Gamma^\infty(TP)$, das zu $X$
   $\pi$"=verwandt ist, heißt \neuerBegriff{horizontale Hebung} oder auch
   \neuerBegriff{horizontaler Lift} von $X$. Es gilt also
   \begin{align}
      \label{eq:horizontaleHebung}
      \gamma \circ X^{\mathsf{h}} &= 0 \quad \text{und} \\
      T\pi \circ X^{\mathsf{h}} &= X \circ \pi \Fdot
   \end{align}
\end{definition}

\begin{lemma}
   \label{lem:Schnitte}
   Sei $\pi \colon E \to M$ ein Vektorbündel mit Faserdimension $k >
   0$. Weiter seien für $n \in \mathbb{N}$ $s_1,\dots,s_n \in
   \Gamma^\infty(E)$ glatte Schnitte, so dass
   $\mathbb{R}\textrm{-}\mathrm{Span}\{s_1(p),\dots,s_n(p)\} = E_p$ für
   alle $p \in M$. Dann ist auch
   $C^\infty(M)\textrm{-}\mathrm{Span}\{s_1,\dots,s_n\} =
   \Gamma^\infty(E)$.
\end{lemma}
\begin{proofklein}
   Sei $p \in M$ und $U$ eine offene Umgebung von $p$ auf der es lokale
   Basisschnitte $e_1,\dots,e_k \in \Gamma^\infty(E\at{U})$ gebe. Dann
   gibt es also für alle $r \in \{1,\dots,n\}$ glatte Funktionen $A_{rs}
   \colon U \to \mathbb{R}$ mit $s \in \{1,\dots,k\}$ so dass für alle
   $i \in \{1,\dots,n\}$ die Gleichung
   \begin{align*}
      s_i\at{U} = \sum_{s = 1}^k A_{is} e_s
   \end{align*}
   gilt.
   Wir können nach dem Basisauswahlsatz o.\,E.\ annehmen, dass
   $s_1(p),\dots,s_k(p)$ linear unabhängig sind. Dann ist die Matrix
   $A(p) \in \mathbb{R}^{k \times k}$ invertierbar, wobei die glatte
   matrixwertige Funktion $A$ durch
   \begin{align*}
      A \colon M \ni q \mapsto (A_{rs}(q))_{r,s=1,\dots,k} \in
      \mathbb{R}^{k \times k} \Fdot
   \end{align*}
   gegeben ist.  Da dies aufgrund der Stetigkeit der
   Determinantenfunktion eine offene Bedingung ist, gibt es eine offene
   Umgebung $V \subset U$ von $p$, sodass $A(q)$ auch für alle $q \in V$
   invertierbar ist. Dies bedeutet jedoch wiederum, dass die Vektoren
   $s_1(q),\dots,s_k(q)$ für alle $q \in V$ linear unabhängig sind. Ist
   nun $s \in \Gamma^\infty(E)$ ein glatter Schnitt, so folgt daraus,
   dass es glatte Funktionen $f^V_i \colon V \to \mathbb{R}$, $i =
   1,\dots,n$ gibt mit $s\at{V} = \sum_{i = 1}^n f^V_i s_i$, wobei
   $f_i^V = 0$ für $i > k$. Da $p$ beliebig war, gibt es also eine
   offene Überdeckung $\{V_j\}_{j \in J}$ derartiger offener Mengen
   $V_i$ mit zugehörigen Funktionen $f^{V_j}_i$ wie oben. Sei nun
   $\{\chi_j\}_{j \in J}$ eine Zerlegung der Eins mit $\supp \chi_j
   \subset V_j$. Dann gilt
   \begin{align*}
      s = \sum_{j \in J}\chi_j s\at{V_j} = \sum_{j \in J}\chi_j \sum_{i
        = 1}^n f^{V_j}_i s_i\at{V_j} = \sum_{i = 1}^n \sum_{j \in J}
      \chi_{j}f^{V_j}_i s_i = \sum_{i = 1}^n f_i s_i \Fcom
   \end{align*}
   wobei die glatten Funktionen $f_i$ durch $f_i := \sum_{j \in J}
   \chi_{j}f^{V_j}_i $ definiert sind. Damit ist die Aussage klar.
 \end{proofklein}

\begin{proposition}
   \label{prop:AngepassteBasen}
   Sei $(P,\pi,B,G)$ ein Hauptfaserbündel und $\gamma$ eine
   Zusammenhangseinsform. Dann sind die folgenden Aussagen richtig.
   \begin{propositionEnum}
   \item %
      Sei $\{e_i\}$ eine Basis von $\lieAlgebra$, dann sind die
      fundamentalen Vektorfelder $(e_i)_P$ globale Basisschnitte für
      $\Gamma^\infty(VP)$.

   \item %
      Für alle für alle $p \in P$ gilt $H_pP = \{X^{\mathsf{h}}(p) \mid
      X \in \Gamma^\infty(TB)\}$.
   \item %
      Es gilt $\Gamma^\infty(HP)^G = \{X^{\mathsf{h}} \mid X \in
        \Gamma^\infty(TB)\}$. Dabei bezeichnet $\Gamma^\infty(HP)^G :=
        \{s \in \Gamma^\infty(HP) \mid gs = s\}$ die $G$"=invarianten
        Schnitte auf dem Horizontalbündel, wobei $G$ auf
        $\Gamma^\infty(HP)$ vermöge $(gs)(p) := g(s(g^{-1}p))$ für alle
        $g \in G$ operiert.
   \end{propositionEnum}
\end{proposition}
\begin{samepage}

\begin{proofklein}
   \begin{beweisEnum}
   \item %
      Einerseits ist die Dimension jeder Faser gleich $\dim G$, wie man
      in lokaler Trivialisierung erkennt, sowie gleich der punktweisen
      Dimension des Kerns von $T\pi$ (vgl.\
      \cite[Lem. 8.15]{lee:2003a}), andererseits ist $\lieAlgebra \ni
      \xi \mapsto \xi_P \in \Gamma^\infty(TP)$ injektiv, da die
      $G$"=Wirkung frei ist (vgl.\ Prop. \ref{prop:lokalFrei}), d.\,h.\
      auch die Dimension des von den fundamentalen Vektorfeldern
      punktweise aufgespannten Raumes ist gleich $\dim G$. Da nach der
      Kettenregel offensichtlich für alle $p \in P$ $\xi_P(p) \in \ker
      T_p\pi$ gilt, spannen die fundamentalen Vektorfelder $\xi_P$
      punkteweise jede Faser von $\Gamma^\infty(VP)$ auf. Siehe auch
      \cite[vor Prop. I 5.1]{kobayashi.nomizu:1963a}. Da für jedes $p
      \in P$ die Vektoren $(e_i)_P(p)$ linear unabhängig sind und aus
      Dimensionsgründen demnach schon $V_pP$ aufspannen, ist die
      Behauptung mit Lemma \ref{lem:Schnitte} klar.
   \item %
      Dies ist aus Dimensionsgründen ebenfalls klar, da für $p \in M$
      $\dim B = \dim P - \dim G = \dim T_pP - \dim V_pP = \dim T_pH$ und
      da für jedes $X \in \Gamma^\infty(TB)$ die Beziehung $T_p\pi
      X^{\mathsf{h}}(p) = X(\pi(p))$ gilt. Siehe auch \cite[vor Prop. II
      1.2]{kobayashi.nomizu:1963a}.
   \item %
      Klar. Siehe auch \cite[vor Prop. II 1.2]{kobayashi.nomizu:1963a}.
   \end{beweisEnum}

\end{proofklein}
\end{samepage}
\begin{definition}
   \label{def:HebungvonWeg}
   Sei $(P,\pi,B,G)$ ein Hauptfaserbündel, $HP$ ein
   Hauptfaserbündelzusammenhang und $ \alpha \colon I \to B$ ein
   Weg. Ein Weg $\tilde \alpha \colon I \to B$ heißt
   \neuerBegriff{horizontale Hebung} (oder auch
   \neuerBegriff{horizontaler Lift}) von
   $\alpha$, falls für alle $t \in I$
   \begin{align}
      \label{eq:HebungWeg1}
      \pi(\tilde \alpha(t)) = \alpha(t)
   \end{align}
und
\begin{align}
   \label{eq:HebungWeg2}
   \dot {\tilde\alpha}(t) \in HP
\end{align}
   gilt.
\end{definition}
Derartige horizontale Hebungen existieren immer, wie die folgende
Proposition zeigt, siehe \cite[Satz 3.7]{baum:2009}.
\begin{proposition}
   \label{prop:HebungVonWegExImmer}
   Seien die Voraussetzungen wie in Definition \ref{def:HebungvonWeg}
   gegeben. Dann existiert zu jedem $t_0 \in I$ und $p_0 \in
   \pi^{-1}({\alpha(t_0)})$ genau eine horizontale Hebung $\tilde
   \alpha$ von $\alpha$ mit $\tilde \alpha(t_0) = p_0$.
\end{proposition}
\begin{bemerkung}
   \label{bem:HorizontaleHebung}
   In der Situation von Proposition \ref{prop:HebungVonWegExImmer}
   sagt man auch, $\tilde \alpha$ sei eine \neuerBegriff{horizontale
     Hebung durch $p_0$}.
\end{bemerkung}

\cleardoublepage

\numberwithin{theorem}{section}
\chapter{\(G\)-invariante Strukturen auf \(G\)-Mannigfaltigkeiten}
\label{cha:invariante_Strukturen}

In diesem Kapitel wollen wir zeigen, wie man bestimmte Strukturen auf
einer Mannigfaltigkeit, auf der eine Lie"=Gruppe $G$ eigentlich wirkt
unter $G$ invariant bzw.\ äquivariant machen kann. Insbesondere zeigen
wir, wie man aus einem gegebenen Spray einen $G$"=äquivarianten
konstruieren kann. Wir werden hier die meisten Resultate beweisen, da
sie in der Literatur teilweise schwer zu finden  und teilweise neu
sind. Dabei gehen wir etwas über das für die Lektüre dieser Arbeit
benötigte Maß hinaus.

\section{Topologische Vorbereitungen}

In diesem Abschnitt stellen wir einige für die später folgenden
Betrachtungen wichtige topologische Grundlagen zusammen. Die meisten
hier dargestellten Begriffe und Sätze sind findet man in der einschlägigen
Lehrbuchliteratur über mengentheoretische Topologie, siehe etwa
\cite{dugundji:1970} oder \cite{engelking:1989}.

Ist $X$ eine Menge, so bezeichnen wir die Potenzmenge von $X$ mit
\potMenge[X]. Ist $X$ ein topologischer Raum und $A \subset X$, so
schreiben wir $\abschluss{A}$ für den topologischen Abschluss und
$\inneres{A}$ für das Innere von $A$ in $X$. Ist $G$ eine Gruppe, und
$\phi \dpA G \times X \to X$ eine Linkswirkung, so sprechen wir im Folgenden einfach von Wirkung und schreiben gelegentlich $gx :=
\phi_g(x) := \phi(g,x)$ für $g \in G$ und $x \in X$.

\begin{definition}
   \label{def:1}
   Sei $X$ ein topologischer Raum.
   \begin{definitionEnum}
   \item $\mengenSystem[U] \subset \potMenge[X]$ heißt \neuerBegriff{lokal
        endlich}, falls jedes Element $x$ von $X$ eine Umgebung besitzt, die
      nur endlich viele Mengen aus \mengenSystem[U] schneidet. Eine
      derartige Umgebung von $x$ wollen wir \neuerBegriff{zeugende
        Umgebung} von $x$  (für $\mengenSystem[U]$) nennen.
   \item Ist $\mengenSystem[U] \subset \potMenge[X]$, so sagen wir,
      $\mengenSystem[U]$ habe die Eigenschaft $\mathcal{E}$, falls jedes $U
      \in \mengenSystem[U]$ die Eigenschaft $\mathcal{E}$ besitzt.
   \item Sei $\mengenSystem[U] = \{U_i\}_{i \in I} \subset \potMenge[X]$ eine
      \neuerBegriff{Überdeckung} von $X$, d.h $X = \bigcup_{i \in
        I}U_i$. Eine andere Überdeckung $\mengenSystem[V] = \{V_j\}_{j \in
        J}$ von $X$ heißt \neuerBegriff{Verfeinerung} von
      \mengenSystem[U], wenn es für jedes $V \in \mengenSystem[V]$ ein $U
      \in \mengenSystem[U]$ gibt mit $V \subset U$. Wir nennen
      $\mengenSystem[V]$ \neuerBegriff{präzise}, wenn $I = J$ und $V_i
      \subset U_i$ für jedes $i \in I$.
    \item $X$ heißt \neuerBegriff{parakompakt}, wenn jede offene
      Überdeckung von $X$ eine lokal endliche, offene Verfeinerung
      besitzt.
   \end{definitionEnum}
\end{definition}

\begin{bemerkung}
   \label{bem:kannLeerSein}
   Ist $X$ ein topologischer Raum und $\{U_i\}_{i \in I}$ eine
   Überdeckung von $X$, so kann es $i \in I$ geben, mit $U_i =
   \emptyset$. Insbesondere wenn $\{U_i\}_{i \in I}$ eine lokal endliche
   Verfeinerung einer anderen offenen Überdeckung von $X$ ist, darf dies
   der Fall sein.
\end{bemerkung}

Das folgende Resultat wollen wir nur zitieren, der Beweis verläuft ähnlich
wie der von Proposition \ref{prop:reg_verfeinerung}.

\begin{proposition}
   \label{prop:praezise_verfeinerung}
   Sei $X$ ein topologischer Raum. Falls die (offene) Überdeckung
   $\mengenSystem[U]$ von $X$ eine lokal endliche (offene) Verfeinerung
   besitzt, so gibt es auch eine präzise lokal endliche Verfeinerung.
\end{proposition}
\begin{proofklein}
   Siehe \cite[Seite~161]{dugundji:1970}.
\end{proofklein}

\begin{bemerkung}
   \label{bem:1}
   Es ist in der Literatur auch üblich, in der Definition von
   Parakompaktheit noch hausdorffsch zu fordern. Siehe
   \cite{dugundji:1970} und \cite{engelking:1989} sowie für unsere
   Konvention \cite{lee:2003a}.
\end{bemerkung}

Die nächste Proposition gibt ein nützliches Kriterium an, wann der das
Bilden des topologisches Abschlusses mit dem Bilden beliebiger
Vereinigungen vertauscht, sieh auch \cite[Thm.~1.1.11]{engelking:1989}.

\begin{proposition}
   \label{prop:1}
   Sei $X$ ein topologischer Raum und $\mengenSystem[U] \subset
   \potMenge[X]$ lokal endlich. Dann gilt Folgendes.
   \begin{propositionEnum}
      \item \label{item:lokalendlichStabilUnterAbschluss}
        $\abschluss{\mengenSystem[U]} := \{\abschluss{U} \mid U \in
         \mengenSystem[U]\}$ ist lokal endlich.
      \item \label{item:lokalendlichUndVereinigungen}
        $\abschluss{\bigcup_{U \in \mengenSystem[U]} U} = \bigcup_{U
          \in \mengenSystem[U]} \abschluss{U}$.
   \end{propositionEnum}
\end{proposition}
\begin{proofklein}
   \begin{beweisEnum}
    \item Sei $x \in X$ und $\zeugendeU{x}$ eine offene zeugende Umgebung von
      $x$ für $\mengenSystem[U]$. Dann gilt für $U \in
      \mengenSystem[U]$ mit $U \cap \zeugendeU{x} = \emptyset$, dass
      $U \subset X\setminus \zeugendeU{x}$ und somit, da $X\setminus
      \zeugendeU{x}$ abgeschlossen ist, auch dass $\abschluss{U}
      \subset X\setminus \zeugendeU{x}$, d.\,h.\ $\abschluss{U} \cap
      \zeugendeU{x} = \emptyset$.
    \item Unabhängig von der lokalen Endlichkeit gilt natürlich
      $\abschluss{V} \subset \abschluss{\bigcup_{U \in
          \mengenSystem[U]} U}$ für $V \in \mengenSystem[U]$ und damit
      also $\bigcup_{U \in \mengenSystem[U]} \abschluss{U} \subset
      \abschluss{\bigcup_{U \in \mengenSystem[U]} U}$.  Für die
      umgekehrte Inklusion reicht es offensichtlich zu zeigen, dass $B
      := \bigcup_{U \in \mengenSystem[U]} \abschluss{U}$ abgeschlossen
      ist. Dazu sei $x \notin B$ beliebig. Nach
      Teil~\refitem{item:lokalendlichStabilUnterAbschluss} gibt es
      eine offene Umgebung $\zeugendeU{x}$ von $x$, die nur von
      endlich vielen Elementen
      $\abschluss{U_1},\ldots,\abschluss{U_n}$ von
      $\abschluss{\mengenSystem[U]}$ geschnitten wird. Dann ist aber
      $\zeugendeU{x} \cap \bigcap_{i= 1}^n( X\setminus
      \abschluss{U_i})$ eine offene Umgebung von $x$, die $B$ nicht
      schneidet. Also ist $X\setminus B$ offen und alles ist gezeigt.
   \end{beweisEnum}
\end{proofklein}

\begin{lemma}
   \label{lem:1}
   Sei $X$ parakompakt und hausdorffsch, $x \in X$ und $U$ eine offene
   Umgebung  von $x$ in $X$. Dann gibt es eine offene Umgebung $V$ von $x$
   mit $x \in V \subset \abschluss{V} \subset U$.
\end{lemma}
\begin{proofklein}
   Sei $A := X\setminus U$. Da $X$ hausdorffsch ist, gibt es zu jedem
   $a \in A$ eine offene Umgebung $U_a$ von $a$ in $X$ und eine offene
   Umgebung $W_a$ von $x$ mit $W_a \cap U_a = \emptyset$ und somit $x
   \notin X\setminus W_a \supset \abschluss{U_a}$. Da $\{U_a \mid a
   \in A\} \cup \{U\}$ eine offene Überdeckung von $X$ ist und $X$
   parakompakt, können wir nach Proposition
   \ref{prop:praezise_verfeinerung} eine präzise lokal endliche
   Verfeinerung wählen, finden also zu $a \in A$ ein
   $F_a$ mit $F_a \subset U_a$, so dass $\{F_a\}_{a \in A} \cup \{U\}$
   immer noch $X$ überdeckt. Insbesondere ist $\{F_a\}_{a \in A}$
   lokal endlich und $F := \bigcup_{a \in A} F_a$ ist eine
   Umgebung von $A$ in $X$. Nach Proposition
   \ref{prop:1}~\refitem{item:lokalendlichUndVereinigungen} gilt
   $\abschluss{F} = \bigcup_{a \in A} \abschluss{F_a}$. Für $a \in A$
   gilt nun $x \notin \abschluss{U_a} \supset \abschluss{F_a}$,
   d.\,h.\ $x \notin \abschluss{F}$. Damit ist $X\setminus \abschluss{F}$
   eine offene Umgebung von $x$ in $X$ und wir haben nach Konstruktion
   $X\setminus\abschluss{F} \subset X\setminus F \subset X \setminus A
   = U$ und damit $\abschluss{X\setminus\abschluss{F}} \subset
   X\setminus F \subset U$.
\end{proofklein}

\begin{bemerkung}
   \label{bem:regulaer}
   In der Sprache der mengentheoretischen Topologie haben wir in Lemma
   \ref{lem:1} gezeigt, dass $X$ \neuerBegriff{regulär} ist im Sinne von
   \cite[Ch.~VII Def.~2.1]{dugundji:1970}. Dies ist äquivalent dazu, dass
   sich jeder beliebig gewählte Punkt und jede abgeschlossene Menge, die
   diesen Punkt nicht enthält durch offene Umgebungen trennen lassen. Da
   wir nicht zu viele topologische Begriffe einführen wollen, und es in
   der Literatur für die Verwendung des Begriffes regulär unterschiedliche
   Konventionen gibt, werden wir diesen Begriff im Folgenden nicht weiter
   verwenden.
\end{bemerkung}

\begin{proposition}
   \label{prop:reg_verfeinerung}
   Sei $X$ parakompakt und hausdorffsch. Sei $\{U_i\}_{i \in I}$ eine
   offene Überdeckung von $X$. Dann gibt es zu jedem $i \in I$ eine
   offene Menge $V_i$, so dass $\abschluss{V_i} \subset U_i$ und
   $\{V_i\}_{i\in I}$ eine lokal endliche, offene Überdeckung von  $X$ ist.
\end{proposition}

\begin{proofklein}
   Nach Lemma \ref{lem:1} gibt es zu $x \in U_i$ eine offene Umgebung
   $W_i^x$ von $x$ in $X$ mit $\abschluss{W_i^x} \subset U_i$. Sei nun
   $\{F_j\}_{j \in J}$ eine lokal endliche, offene Verfeinerung von
   $\{W_i^x\}_{i\in I, x \in X}$. Definiere eine Abbildung $\alpha: J
   \to I$ die zu $j \in J$ ein $\alpha(j) \in I$ auswählt, so dass $F_j
   \subset U_{\alpha(j)}$, dies geht nach dem Auswahlaxiom. Setze dann
   für $i \in I$ $V_i := \bigcup_{j\in J \mid \alpha(j) = i} F_j$. $V_i$
   ist als Vereinigung offener Mengen natürlich offen in $X$. Ferner
   gilt wegen der lokalen Endlichkeit von $\{F_j\}_{j \in J}$ und
   Proposition \ref{prop:1} $\abschluss{V_i} = \bigcup_{j\in J \mid
     \alpha(j) = i} \abschluss{F_j} \subset U_i$. Nach Konstruktion ist
   klar, dass $\{V_i\}_{i \in I}$ immer noch eine lokal endliche, offene
   Überdeckung von $X$ ist.
\end{proofklein}

\begin{bemerkung}
   \label{bem:3}
   Ist $\{U_i\}_{i \in I}$ in Proposition \ref{prop:reg_verfeinerung} schon lokal endlich, so braucht man
   im Beweis das Auswahlaxiom offenbar nicht.
\end{bemerkung}

\begin{bemerkung}
   \label{bem:parakomp_T_2_zerlegung_der_eins}
   Eine wichtige Konsequenz daraus, dass ein topologischer Raum $X$
   parakompakt und hausdorffsch ist, besteht darin, dass es zu jeder
   offenen Überdeckung eine dieser untergeordnete stetige Zerlegung
   der Eins gibt.
\end{bemerkung}

\begin{proofklein}
   Siehe \cite[Thm.~5.1.9]{engelking:1989}.
\end{proofklein}

\begin{definition}[Lokal endliches
   reguläres Überdeckungspaar]
   \label{def:reg_Ueberdeckungs_Paar}
   Seien die Bezeichnungen wie in Proposition \ref{prop:reg_verfeinerung}
   und $\{U_i\}_{i \in I}$ schon lokal endlich. Dann wollen wir das
   Paar $(\{U_i\}_{i\in I}, \{V_i\}_{i \in I})$ ein \neuerBegriff{lokal endliches
   reguläres Überdeckungspaar} des topologischen Raumes $X$ nennen.
\end{definition}

\begin{definition}[Gute zeugende Umgebung]
   \label{def:guteZeugendeUmgebung}
   Sei $RP := (\{U_i\}_{i\in I}, \{V_i\}_{i \in I})$ ein lokal
   endliches reguläres Überdeckungspaar von $X$, und $\zeugendeU{x}$
   eine zeugende Umgebung von $x$ zu $\{V_i\}_{i \in I}$, d.\,h.\ $J :=
   \{j \in I \mid V_j \cap \zeugendeU{x}\neq \emptyset \}$ ist
   endlich. Wir wollen $\zeugendeU{x}$ \neuerBegriff{gute} zeugende
   Umgebung von $x$ (zu $RP$) nennen, falls für alle $j \in J$
   folgendes gilt:
   \begin{definitionEnum}
      \item $x \in  \abschluss{V_{j}}$. \label{item:xinGuteZeugendeUmgebung}
      \item $\zeugendeU{x} \subset
         U_{j}$. \label{item:guteZeugendeUmgebungEnthaltenIn}
   \end{definitionEnum}
\end{definition}

\begin{proposition}
   \label{prop:guteZeugendeUmgebung}
   Ist $RP = (\{U_i\}_{i\in I}, \{V_i\}_{i \in I})$ ein lokal endliches
   reguläres Überdeckungspaar eines topologischen Raumes $X$, so gibt es
   zu jedem $x \in X$ eine gute zeugende Umgebung von $x$ zu $RP$.
\end{proposition}

\begin{proofklein}
   Sei $\widetilde{\zeugendeU{x}}$ eine zeugende Umgebung von $x$ zu
   $\{V_i\}_{i \in I}$, also insbesondere $E := \{i \in I \mid V_i \cap
   \widetilde{\zeugendeU{x}} \neq \emptyset \}$ endlich. Sei $J$ die
   maximale Teilmenge von $E$, so dass $x \in \abschluss{V_j} \subset
   U_j$ für $j \in J$. Dann definiere
 \begin{align*}
    \zeugendeU{x} := \bigcap_{j \in E\setminus J}
    \left(\widetilde{\zeugendeU{x}}\setminus \abschluss{V_j} \right)
    \cap \bigcap_{i \in J} U_j \Fdot
 \end{align*}
 $\zeugendeU{x}$ ist offenbar eine gute zeugende Umgebung von $x$. Falls
 $i \in I$ mit $\zeugendeU{x} \cap V_i \neq \emptyset$ folgt zunächst $i
 \in E$ und nach Konstruktion von $\zeugendeU{x}$ sogar $i \in J$. Ist
 umgekehrt $j \in J$, so ist $x \in \abschluss{V_j}$ und damit
 $\zeugendeU{x} \cap V_j \neq \emptyset$, da $\zeugendeU{x}$, d.\,h.\ $J =
 \{j \in I \mid V_j \cap \zeugendeU{x} \neq \emptyset\}$. Die
 Eigenschaften~\refitem{item:xinGuteZeugendeUmgebung} und~\refitem{item:guteZeugendeUmgebungEnthaltenIn} sind dann klar nach
 Konstruktion von $\zeugendeU{x}$.
\end{proofklein}

\begin{klein}
   \begin{bemerkung}
      \label{bem:guteZeugendeUmgebung_zuvieldesguten}
      Ist $RP = (\{U_i\}_{i\in I}, \{V_i\}_{i \in I})$ ein lokal endliches
      reguläres Überdeckungspaar eines topologischen Raumes $X$, so könnte
      man denken, dass man zu $x \in X$ eine zeugende Umgebung von $x$ für
      $\{U_i\}$ mit folgenden Eigenschaften erhalten könnte:
      \begin{bemerkungEnum}
      \item $x \in \abschluss{V_{j}}$
      \item $\zeugendeU{x} \subset U_{j}$
      \end{bemerkungEnum}
      für alle Elemente $j$ der endlichen Menge $J := \{i \in I \mid U_i \cap \zeugendeU{x}\neq \emptyset \} $.

      Dies ist i.\,Allg.\ falsch. Man betrachte z.\,B.\ $X = \mathbb{R}$, $x = 0$,
      $RP = (\{U_1,U_2\},\{V_1,V_2\})$ mit
      \begin{align*}
         U_1 &= (-\infty,3) & U_2 &= (0,\infty) \\
         V_1 &= (-\infty,2) & V_2 &= (1,\infty)\Fdot
      \end{align*}

      Hier sind offensichtlich weder die erste noch die zweite Eigenschaft
      erfüllbar. Demnach ist es bei der Definition von guter zeugender Umgebung
      erstens wichtig  ein Überdeckungspaar zu haben und zweitens zu fordern, dass
      $\zeugendeU{x}$ eine zeugende Umgebung für $\{V_i\}_{i \in I}$ ist und nicht
      etwa auch für $\{U_i\}_{i \in I}$.
   \end{bemerkung}
\end{klein}

\begin{definition}[Invarianz und Äquivarianz]
   \label{def:aequivariant_etc}
   Seien $X$ und $Y$ Mengen, $G$ eine Gruppe und $\theta \dpA G \times X
   \to X$ sowie $\phi \dpA G \times Y \to Y$ Wirkungen von
   $G$ auf $X$ bzw.\ $Y$.
   \begin{definitionEnum}
    \item $U \subset X$ heißt \neuerBegriff{$G$"=invariant}, falls
      $\theta_g(x) \in U$ für alle $x \in U$ und $g \in G$.

      \item Eine Abbildung $F: X \to Y$ heißt \neuerBegriff{($G$-)äquivariant}
         (bezüglich $\theta$ und $\phi$), falls
         \begin{align*}
            \phi_g \circ F = F \circ \Theta_g
         \end{align*}
         für alle  $g \in G$.

   \end{definitionEnum}
\end{definition}

\begin{bemerkung}
   \label{bem:aequivariantesInverses}
   Seien die Voraussetzungen wie in Definition
   \ref{def:aequivariant_etc}, $F\colon X \to Y$ $G$"=äquivariant und
   bijektiv. Dann ist offensichtlich auch $F^{-1}$ $G$"=äquivariant.
\end{bemerkung}

\section{Eigentliche Gruppenwirkungen}

Wir erinnern an dieser Stelle noch an einen wichtigen Begriff wichtigen
Begriff der eigentlichen Gruppenwirkung und stellen eine Reihe von
Charakterisierungen bereit. Vergleiche \cite[Seite~216]{lee:2003a},
\cite[Seite~319]{tomdieck:2000a}.

\begin{definition}[Eigentliche Wirkungen]
   \label{def:eigentlicheWirkung}
   Sei $G$ eine topologische Gruppe und $\phi \dpA G \times X \to
   X$ eine stetige Wirkung von $G$ auf dem topologischen Raum $X$.
   \begin{definitionEnum}
    \item Die Wirkung $\phi$ heißt \neuerBegriff{eigentlich}, falls
      $\Phi \dpA G \times X \to X \times X$, $(g,x) \mapsto
      (x,\phi_g(x))$ eigentlich ist, d.\,h.\ $\Phi^{-1}(K)$ kompakt ist
      für jedes kompakte $K \subset X \times X$.
   \item Wir nennen $\phi$ \neuerBegriff{folgeneigentlich} falls
     folgende Bedingung erfüllt ist.

      Ist $(x_n)$ eine konvergente Folge in $X$ und $(g_n)$ eine
      Folge in $G$, so dass $(g_n x_n)$ konvergiert. Dann gibt es eine
      konvergente Teilfolge von $(g_n)$.
   \end{definitionEnum}

\end{definition}

\begin{bemerkung}
   \label{bem:eigentlichKonventionen}
   In der Literatur gibt es verschiedene übliche Konventionen für
   eigentliche Gruppenwirkungen. Etwa \cite{tomdieck:2000a} definiert
   \glqq{}eigentlich\grqq{}  anders als wir. Falls zum Beispiel $G$
   und $X$ hausdorffsch sind und $X$ lokal kompakt ist, fallen seine
   und unsere Definition zusammen.
\end{bemerkung}

Die folgende Proposition liefert einige äquivalente Charakterisierungen
eigentlicher Wirkungen im Falle von Lie"=Gruppen und
Mannigfaltigkeiten.

\begin{lemma}
   \label{lem:eigentlicheWirkungCharakterisierung}
   Sei $M$ eine Mannigfaltigkeit und $G$ eine Lie"=Gruppe, die stetig auf
   $M$ wirke. Dann sind äquivalent
   \begin{lemmaEnum}
      \item $G$ wirkt eigentlich auf $M$. \label{item:aeqs_eigentlich}
      \item \label{item:eigentlich_folgenkriterium} \label{item:aeqs_folgeneigentlich}
         $G$ wirkt folgeneigentlich auf $M$.
      \item \label{item:aeq_eigentlich_verschiedene_kompakta} Für jedes
         Paar $K$ und $L$ von kompakten Teilmengen von $M$ ist $\{g \in G
         \mid g K \cap L\ \neq \emptyset\}$ kompakt.
      \item Für jedes kompakte $K \subset M$ ist $\{g \in G \mid g K \cap
         K \neq \emptyset\}$
         kompakt. \label{item:aeq_eigentlich_gleiche_kompakta}
      \item \label{item:aeq_eigentlich_kompakte_widerkehr} Für jedes $(x,y) \in
         M \times M$ gibt es Umgebungen $V_x$ von $x$ und $V_y$ von $y$,
         so dass $\{g \in G \mid gV_x \cap V_y \neq \emptyset\}$ in einer
         kompakten Teilmenge von $G$ liegt.
   \end{lemmaEnum}
\end{lemma}
\begin{proofklein}
   \begin{beweisEnum}
   \item[\refitem{item:aeqs_eigentlich} $\Rightarrow$~\refitem{item:aeqs_folgeneigentlich}] Sei $(x_n)$ eine Folge in
      $M$ und $(g_n)$ eine Folge in $G$, so dass $(x_n,g_n x_n)$ in $M
      \times M$ konvergiert. Dann gibt es jedoch ein Kompaktum $L
      \subset M \times M$ mit $(x_n,g_n x_n) \in L$ für alle $n \in
      \mathbb{N}$. Folglich verläuft aber schon die Folge $(g_n,x_n)$ in
      dem Kompaktum $\Phi^{-1}(L) \subset G \times M$, womit $(g_n,x_n)$
      und damit auch $(g_n)$ eine konvergente Teilfolge besitzt.
   \item[\refitem{item:aeqs_folgeneigentlich} $\Rightarrow$~\refitem{item:aeq_eigentlich_verschiedene_kompakta}] Sei $(g_n)$
      eine Folge in $\{g \in G \mid g K \cap L \neq \emptyset\}$ und für
      $n \in \mathbb{N}$ sei $x_n \in K$ mit $g_n x_n \in L$. Da $K$ und
      $L$ kompakt sind, besitzen $(x_n)$ und $(g_n x_n)$ in konvergente
      Teilfolgen und damit nach Voraussetzung auch $(g_n)$. Da $L$
      kompakt, insbesondere, da $M$ hausdorffsch, abgeschlossen ist,
      liegt ihr Grenzwert wieder in $\{g \in G \mid g K \cap L \neq
      \emptyset\}$, womit letztere Menge kompakt ist.
   \item[~\refitem{item:aeq_eigentlich_verschiedene_kompakta}
      $\Rightarrow$ ~\refitem{item:aeq_eigentlich_gleiche_kompakta}] Klar.
   \item[~\refitem{item:aeq_eigentlich_gleiche_kompakta} $\Rightarrow$~\refitem{item:aeq_eigentlich_kompakte_widerkehr}] Da $M$ eine
      Mannigfaltigkeit ist, gibt es eine ausschöpfende Folge von Kompakta
      und somit ein Kompaktum $K \subset M$ mit $x,y \in K$, das
      zugleich Umgebung von $x$ wie von $y$ ist.

   \item[~\refitem{item:aeq_eigentlich_kompakte_widerkehr} $\Rightarrow$~\refitem{item:aeqs_eigentlich}] Sei $L \subset M \times M$ ein
      beliebiges Kompaktum. Zu jedem $(x,y) \in L$ gibt es eine offene
      Umgebung $U_x$ von $x$ und eine offene Umgebung $V_y$ von $y$,
      so dass $S_{(x,y)} := \{g \in G \mid gU_x \cap V_y \neq
      \emptyset\}$ in einer kompakten Menge liegt. Da $L$ kompakt ist,
      gibt es ein $m \in \mathbb{N}$ und $(x_i,y_i) \in L$ für
      $i=1,\ldots,m$, so dass $L \subset \bigcup_{i=1}^m U_{x_i} \times
      V_{y_i}$. Nach Konstruktion gibt es ein Kompaktum $S$ mit
      $\bigcup_{i=1}^m S_{(x_i,y_i)} \subset S \subset G$. Sei nun
      $\projFaktor_2 \dpA G \times M \to M$ die Projektion auf die zweite
      Komponente. Ist $(g,x) \in \Phi^{-1}(L)$, so gilt $(x,g x) \in L$,
      also insbesondere $ x \in \projFaktor_2(L)$ und nach Konstruktion
      gibt es ein $i \in \{1,\ldots,m\}$ mit $x \in U_{x_i}$ und $g x
      \in V_{y_i}$, also $g \in S_{(x_i,y_i)} \subset S$. Damit liegt
      $(g,x)$ aber schon in der kompakten Menge $S \times
      \projFaktor_2(L)$. Folglich gilt $\Phi^{-1}(L) \subset S \times
      \projFaktor_2(L)$. Da $M$ hausdorffsch ist und aufgrund der
      Kompaktheit von $L$ ist $L$ abgeschlossen, womit $\Phi^{-1}(L)$
      wegen der Stetigkeit von $\Phi$ auch abgeschlossen ist und als
      Teilmenge des Kompaktums $S \times \projFaktor_2(L)$ somit selbst
      wieder kompakt ist.

   \end{beweisEnum}
\end{proofklein}

\begin{bemerkung}
   \label{bem:eigentlicheDinge}
   \begin{bemerkungEnum}
   \item Für die Bedingung~\refitem{item:aeq_eigentlich_kompakte_widerkehr}
      aus Proposition \ref{lem:eigentlicheWirkungCharakterisierung} sagt
      man auch $G$ wirkt mit \neuerBegriff{kompakter Wiederkehr} auf $M$
      (siehe \cite[20.7]{tomdieck:2000a}).
   \item Da $M$ hausdorffsch ist, ist in Bedingung~\refitem{item:aeq_eigentlich_kompakte_widerkehr} in Proposition
      \ref{lem:eigentlicheWirkungCharakterisierung} die Aussage, dass
      $S_X := \{g \in G \mid gV_x \cap V_y \neq \emptyset\}$ in einer
      kompakten Teilmenge von $G$ liege äquivalent dazu, dass ihr
      Abschluss kompakt sei. Ist $G$ diskret, so ist dies weiter
      gleichwertig dazu, dass $S_X$ endlich ist. In diesem Fall spricht
      man wohl aus historischen Gründen auch davon, dass die $G$"=Wirkung
      \neuerBegriff{eigentlich diskontinuierlich} sei. Siehe
      \cite[Seite~225]{lee:2003a} für einige kritische Bemerkungen zu
      dieser Wortwahl.
   \end{bemerkungEnum}
\end{bemerkung}

\begin{klein}
   Zur Abrundung und um den Anschluss an eine übliche Charakterisierung
   von (frei und) eigentlich diskontinuierlich zu finden, zeigen wir noch
   folgende Proposition.

   \begin{proposition}
      \label{prop:eigentlich_diskontinuierlich}
      Die diskrete Lie"=Gruppe $G$ wirke auf einer Mannigfaltigkeit
      $M$. Dann ist die $G$"=Wirkung genau dann eigentlich, wenn folgende
      Bedingung gilt:
      \begin{propositionEnum}
      \item \label{item:eigentlich_diskontinuierlich_1} Ist $p \in M$, so
         gibt es eine Umgebung $U$ von $p$ so dass $g U \cap U = \emptyset$
         für alle bis auf endlich viele $g \in G$.
      \item Liegen $p,q \in M$ nicht auf derselben $G$"=Bahn, so gibt es
         eine Umgebung $U_p$ von $p$ und eine Umgebung $V_q$ von $q$
         so dass $g U_p \cap V_q = \emptyset$ für alle $g \in G$.
      \end{propositionEnum}
   \end{proposition}

   \begin{proofklein}
      \begin{beweisEnum}
      \item[\glqq{}$\Rightarrow$\grqq{}] Die erste Bedingung ist
         trivial. Für die zweite seien also $p,q \in M$ beliebig. Nach
         Voraussetzung gibt es eine Umgebung $\widetilde{U}_p$ von $p$
         und eine Umgebung $\widetilde{V}_q$ von $q$, so dass es ein $m
         \in \mathbb{N}$ gibt und $g_1,\ldots,g_m \in G$ mit
         $\{g_1,\ldots,g_m\} = \{g \in G \mid g\widetilde{U}_p \cap
         \widetilde{V}_q \neq \emptyset\}$. Da $p$ und $q$ nicht auf
         derselben Bahn liegen, die Wirkung stetig und $M$ hausdorffsch
         ist, können wir ohne Einschränkung nach eventuellem Verkleinern
         annehmen, dass $p \notin \abschluss{g_i^{-1}\widetilde{V}_q}$
         für $i=1,\ldots,m$. Dann setze $U_p := \widetilde{U}_p
         \setminus (\abschluss{g_1^{-1}\widetilde{V}_q} \cup \ldots \cup
         \abschluss{g_m^{-1}\widetilde{V}_q})$.
      \item[\glqq $\Leftarrow$\grqq{}] Ist klar.
      \end{beweisEnum}
   \end{proofklein}

   \begin{bemerkung}
      \label{bem:frei_eigentlich_diskontinuierlich}
      Wirkt eine diskrete Lie"=Gruppe $G$ eigentlich auf einer
      Mannigfaltigkeit $M$, so ist die $G$"=Wirkung offenbar genau dann
      frei, wenn man in Bedingung~\refitem{item:eigentlich_diskontinuierlich_1} in Proposition
      \ref{prop:eigentlich_diskontinuierlich} \glqq{}für alle bis auf
      endlich viele $g \in G$\grqq{} durch \glqq{}für alle $g \in G$ außer
      $g = e$\grqq{} ersetzen kann.

   \end{bemerkung}

   \begin{proofklein} Eine Richtung ist trivial, auch die für die andere ist
      simpel und verläuft analog zu der im Beweis von Proposition
      \ref{prop:eigentlich_diskontinuierlich} \glqq
      $\Rightarrow$\grqq{}. Siehe \cite[Cor. 12.10]{lee:2000}.
   \end{proofklein}
\end{klein}

\begin{bemerkung}
   \label{bem:G_kompakt_eigentlich}
   Es ist offensichtlich, dass jede Wirkung einer kompakten Lie"=Gruppe auf einer
   Mannigfaltigkeit folgeneigentlich und damit eigentlich ist. Siehe auch
   \cite[Cor. 9.14]{lee:2003a}.
\end{bemerkung}

Eigentliche Gruppenwirkungen sind \glqq{}ansteckend\grqq{} im Sinne der
folgenden Proposition.

\begin{proposition}
   \label{prop:eigentlichAnsteckend}
   Sei $G$ eine Lie"=Gruppe und $M$, $N$ Mannigfaltigkeiten, auf denen $G$
   wirke und $F \dpA M \to N$ stetig und $G$"=äquivariant. Ist die
   $G$"=Wirkung auf $N$ eigentlich, so auch die auf $M$.
\end{proposition}
\begin{proofklein}
   Sei $(p_n)$ eine konvergente Folge in $M$ und $(g_n)$ eine Folge in $G$,
   so dass $(g_n p_n)$ konvergiere. Dann konvergieren wegen der Stetigkeit
   von $F$ auch $(F(p_n))$ und $(F(g_n p_n))$ in $N$. Wegen der
   Äquivarianz ist aber $g_nF(p_n) = F(g_n p_n)$ für alle $n \in
   \mathbb{N}$. Da die $G$"=Wirkung auf $N$ eigentlich ist, gibt es somit
   eine konvergente Teilfolge von $(g_n)$.
\end{proofklein}

Wir stellen nun einige topologische Eigenschaften von $M/G$ zusammen für
den Fall, dass die Lie"=Gruppe $G$ auf der Mannigfaltigkeit $M$ eigentlich
wirkt.

\begin{definition}[$G$-Vektorbündel]
   \label{def:hochhebungEinerLieGruppenWirkung}
   Sei $\phi \dpA G \times M \to M$ eine Wirkung einer Lie"=Gruppe $G$
   auf einer Mannigfaltigkeit $M$. Ist $E \to M$ ein Vektorbündel über
   $M$, so heißt eine $G$"=Wirkung $\Phi \dpA G \times E \to E$ auf $E$
   \neuerBegriff{Hochhebung} von $\phi$ oder \neuerBegriff{Lift} von
   $\phi$, falls $\Phi_g$ für alle $g \in G$ ein
   Vektorbündelmorphismus über $\phi_g$ ist. Ist dies der Fall, so
   wollen wir von einem \neuerBegriff{$G$"=Vektorbündel} (über $M$)
   sprechen.
\end{definition}

\begin{definition}[Klein und kompakt in Faserrichtung]
   \label{def:kleinundkompaktinFaserrichtung}
   Es wirke eine Lie"=Gruppe $G$ auf einer Mannigfaltigkeit $M$. Sei $\pi \dpA M
    \to M/G$ die kanonische Projektion auf den Quotienten.
   \begin{definitionEnum}
   \item $F \subset M$ heißt \neuerBegriff{klein}, falls jedes $x \in M$ eine
      Umgebung $U$ besitzt, so dass $\{g \in G \mid gU \cap F \neq \emptyset \}$
      relativ kompakt ist, d.\,h.\ in einem Kompaktum in $G$
      liegt. Vgl.\ \cite{palais:1961a}.
   \item $F \subset M$ heißt \neuerBegriff{(relativ) kompakt in
        Faserrichtung}, falls $F \cap \pi^{-1}(K) $ (relativ) kompakt
      ist für jedes kompakte $K \subset M/G$, vgl.\ \cite[Ch. VII 5,
      7.11]{greub1972connections} und
      \cite[Gl.
      (94)]{gutt2010involutions}.
   \end{definitionEnum}
\end{definition}

\begin{proposition}
   \label{prop:klein-kompaktInFaserrichtung}
   Es wirke eine Lie"=Gruppe $G$ auf einer Mannigfaltigkeit $M$. Sei $\pi \dpA M
    \to M/G$ die kanonische Projektion auf den Quotienten. Dann gelten
    folgende Implikationen.
    \begin{propositionEnum}

    \item  $F \subset M $
       klein $\Rightarrow$ $F$ ist relativ kompakt in Faserrichtung.
    \item Falls die Wirkung eigentlich ist, gilt auch die Umkehrung: $F \subset M$ ist relativ kompakt in Faserrichtung $\Rightarrow$ $F$ ist klein.
    \end{propositionEnum}
\end{proposition}
\begin{proofklein}
   \begin{beweisEnum}

    \item Ist $F \subset M$ nicht relativ kompakt in Faserrichtung, so
      gibt es ein kompaktes $K \subset M/G$, so dass $\pi^{-1}(K) \cap
      F$ nicht relativ kompakt ist.  Es gibt also eine Folge $(x_n)$
      in $\pi^{-1}(K) \cap F$, die keine in $M$ konvergente Teilfolge
      enthält, so dass $(\pi(x_n))$ in $K$ konvergiert. (Falls nicht,
      können wir wegen der Kompaktheit von $K$ eine konvergente
      Teilfolge von $(\pi(x_n))$ auswählen.) Sei nun $x \in
      \pi^{-1}(K)$ mit $\pi(x_n) \to \pi(x)$ und $\{U_n\}_{n \in
        \mathbb{N}}$ eine Umgebungsbasis von $x$. Nach Wahl einer
      Teilfolge von $(x_n)$, können wir o.\,E.\ annehmen, dass $\pi(x_n)
      \in \pi(U_n)$ für alle $n \in \mathbb{N}$.
      \begin{klein}
         In der Tat ist $\{\pi(U_n)\}_{n \in \mathbb{N}}$ eine
         Umgebungsbasis von $\pi(x)$, denn ist $V$ eine offene
         Umgebung von $\pi(x)$ so ist $\pi^{-1}(V)$ eine Umgebung von
         $x$, d.\,h.\ es gibt ein $n \in \mathbb{N}$ mit $U_n \subset
         \pi^{-1}(V)$ und damit $\pi(U_n) \subset V$.
      \end{klein}
      Es gibt also für jedes $n \in \mathbb{N}$ ein $y_n \in U_n$ und
      ein $g_n \in G$, so dass $x_n = g_ny_n$. Für jedes $k \in
      \mathbb{N}$ ist $g_n \in \{g \in G \mid gU_k \cap F \neq \emptyset
      \} =: S_k$ für alle $n > k$. $\{g_n\}_{n > k}$ hat aber keine
      konvergente Teilfolge in $G$, denn hätte sie eine, so auch
      $(x_n)$, da nach Konstruktion $(y_n)$ (gegen $x$) konvergiert. Dies
      wäre jedoch ein Widerspruch. Damit ist $S_k$ für alle $k \in
      \mathbb{N}$ nicht relativ kompakt und $F$ demnach nicht klein, was
      zu zeigen war.
    \item Sei $F \subset M$ relativ kompakt in Faserrichtung, $x \in
      M$ und $U_x$ eine kompakte Umgebung von $x$ in $M$. Aufgrund der
      Stetigkeit von $\pi$, ist auch $\pi(U_x) \subset M/G$
      kompakt. Und nach Voraussetzung gibt es ein Kompaktum $L \subset
      M$ mit $F \cap \pi^{-1}(\pi(U_x)) \subset L$.  Da $G$
      eigentlich wirkt, ist damit
      \begin{align*}
         \{g \in G \mid g U_x \cap F \neq \emptyset \} &= \{g \in G \mid g U_x \cap
         (F \cap \pi^{-1}(\pi(U_x))) \neq
         \emptyset \} \\
         &\subset \{g \in G \mid g \underbrace{U_x}_{\text{kompakt}} \cap
         \underbrace{L}_{\text{kompakt}} \neq \emptyset \}
      \end{align*}

      relativ kompakt. Demnach ist $F$ klein.
   \end{beweisEnum}
\end{proofklein}

\section{Mittelung von Schnitten auf Vektorbündeln}
\label{sec:MittelungVonSchnitten}

In diesem Abschnitt wollen wir zeigen, wie man Schnitte von
Vektorbündeln bezüglich einer eigentlichen $G$"=Wirkung mitteln und sie
so $G$"=äquivariant machen kann. Als Anwendung dieser Technik
konstruieren wir eine invariante Zerlegung der Eins, invariante
Riemannsche Fasermetriken und äquivariante Sprays.

Für die folgenden Überlegungen benötigen wir an vielen Stellen
linksinvariante Volumendichten auf einer Lie"=Gruppe, weshalb wir
zunächst die folgende Proposition zitieren,
vgl.\ \cite[Ch. 3.13]{duistermaat.kolk:2000a},
\cite[Prop. 14.24]{lee:2003a}.

\begin{proposition}(Linksinvariante Volumendichte)
   \label{prop:linksinvarianteVolumendichte}
   Sei $G$ eine Lie"=Gruppe. Dann gibt es eindeutig bestimmte
   eine linksinvariante Volumendichte $\Omega$ auf $G$. Ist $G$ kompakt,
   so gibt es eine eindeutig bestimmte Volumendichte $\Omega$ auf $G$
   mit
   \begin{align}
      \label{eq:NormierungVolumendichte}
      \int_G \Omega = 1
   \end{align}
   \Fdot
\end{proposition}

\begin{proposition}
 \label{prop:integrieren_inv_funkt}
 Sei $E \to M$ ein $G$"=Vektorbündel, $\Omega$ eine linksinvariante
 Volumendichte auf $G$ und $s\dpA M \to E$ ein glatter Schnitt. Ist
 der Träger $\supp{s}$ von $s$ klein, so ist für $x\in M$ das Integral
   \begin{align*}
      s^G(x) := \int_G g (s(g^{-1} x)) \Omega(g)
   \end{align*}
   wohldefiniert und $M \to E$, $x \mapsto s^G(x)$ definiert einen
   glatten Schnitt $s^G$, der $G$"=äquivariant ist, d.\,h.\
   \begin{align}
      \label{eq:AequivaranzSchnitt}
      g s^G(x) = s^G(gx)
   \end{align}
   für alle $g \in G$ und $x \in M$.
\end{proposition}

\begin{bemerkung}
   \label{bem:intergriern_inv_funkt}
   Die Voraussetzung von Proposition \ref{prop:integrieren_inv_funkt} ist nach
   Proposition \ref{prop:klein-kompaktInFaserrichtung} erfüllt, wenn $G$
   eigentlich wirkt und der Träger von $s$ kompakt in Faserrichtung ist, also
   insbesondere dann, wenn er kompakt ist.
\end{bemerkung}

\begin{bemerkung}
   \label{bem:InvariantAequivariant}
   Statt $G$"=äquivariant könnte man den Schnitt $s^G$ mit gleicher
   Berechtigung auch $G$"=invariant nennen, denn vermöge $G \times
   \Gamma^\infty(E) \to \Gamma^\infty(E)$, $(g,s) \mapsto
   g(s(g^{-1}\cdot))$ wird auf den glatten Schnitten von $E$ eine
   $G$"=Wirkung erklärt, bezüglich der ein Schnitt genau dann $G$"=invariant
   ist, wenn er als Abbildung $M \to E$ $G$"=äquivariant ist.
\end{bemerkung}

\begin{proofklein}
   Wir folgen der Argumentation von \cite[Prop.~1.2.6]{palais:1961a}
   (siehe auch \cite[Lem.~6.27]{michor:2007}).  Sei $x_0 \in M$
   beliebig. Da $\supp(s)$ klein ist, gibt es eine kompakte Umgebung
   $K$ von $x_0$, so dass $D_K := \{g \in G \mid g K \cap \supp(s) \neq
   \emptyset \}^{-1} = \{g \in G \mid g^{-1} K \cap \supp(s) \neq
   \emptyset\}$ relativ kompakt (die Gruppeninversion ist ein
   Homöomorphismus!) ist.

   Damit ist $s^G$ wohldefiniert und es gilt
   \begin{align*}
      s^G\at{K} = \int_{D_K} g s(g^{-1} \cdot) \Omega(g)\Fdot
   \end{align*}
   Man sieht also, dass $s^G$ tatsächlich glatt ist.

   Falls $x \notin G \supp(s)$, folgt sofort $s(gx) = 0$ für alle $g \in
   G$, also $s^G(x) = 0$. Wir müssen nur noch die Äquivarianz von $s^G$
   nachrechnen. Seien dazu $h \in G$  und $x \in M$ beliebig. Außerdem
   schreiben wir
   $\phi$ für die $G$"=Wirkung auf $M$, $\theta$ für die auf $E$ und $L_g$
   für die Linksmultiplikation in $G$ mit $g \in G$.
   \begin{align*}
      h s^G(x) &= \theta(h, s^G(x)) = \theta\left(h,\int_G
      \theta(\cdot,s(\phi(\cdot^{-1},x))) \Omega\right)\\
      &= \int_G \theta(h,\theta(\cdot,s(\phi(\cdot^{-1},x)))) \Omega \\
      &= \int_G \theta(L_h(\cdot),s(\phi((L_h(\cdot))^{-1},\phi(h,x))))
      L_h^*(\Omega)
      \\
      &= \int_G L_h^*(\theta(\cdot,s(\phi(\cdot^{-1},\phi(h,x))))
     \Omega) \\
     &= \int_G \theta(\cdot,s(\phi(\cdot^{-1},\phi(h,x)))) \Omega \\
     &= s^G(\phi(h,x)) = s^G(h x) \Fdot
   \end{align*}
   Dabei wurde im dritten Schritt ausgenutzt, dass $E_p \to E_{gp}$, $v
   \mapsto g v$ linear ist. Ferner haben wir verwendet, dass $L_g^*
   \Omega = \Omega$ gilt für $g \in G$, sowie den Transformationssatz
   der Intergralrechnung (siehe etwa
   \cite[Prop.~14.6~(d)]{lee:2003a}).

\end{proofklein}

Über die bekannte Korrespondenz zwischen glatten vektorwertigen
Funktionen auf einer Mannigfaltigkeit und den Schnitten auf dem
trivialen Vektorbündel erhält man nun sofort das folgende Korollar.

\begin{korollar}
   \label{kor:integrieren_inv_funkt}
   Sei $G$ eine Lie"=Gruppe, die eigentlich auf einer Mannigfaltigkeit $M$
   wirke. $G$ wirke ferner auf einem endlichdimensionalen Vektorraum $V$,
   so dass $V \to V$, $v \mapsto g v$ linear ist für $g \in
   G$. Sei schließlich $\Omega$ eine linksinvariante Volumendichte auf $G$.
   Ist $f \dpA M \to V$ glatt mit kompaktem Träger, so ist
   \begin{align*}
      f^G \dpA M \to V, \quad x \mapsto \int_G g f(g^{-1} x) \Omega(g)
   \end{align*}
   eine glatte, $G$"=äquivariante Abbildung mit $f^G(x) = 0$ für $x \notin G
   \supp(f)$.
\end{korollar}

Der folgende Satz gibt Aufschluss über wichtige topologische
Eigenschaften des Quotienten einer glatten Mannigfaltigkeit nach einer
eigentlich operierenden Lie"=Gruppe.

\begin{satz}
   \label{satz:Quotient_Parakompakt_T_2}
   Ist $M$ eine Mannigfaltigkeit und wirke eine Lie"=Gruppe $G$ eigentlich
   auf $M$, dann ist  $M/G$
   \begin{satzEnum}
      \item hausdorffsch
         und
      \item parakompakt.
   \end{satzEnum}
\end{satz}
\begin{proofklein}
   Sei $\pi \dpA M \to M/G$ die kanonische Projektion auf den Quotienten.
   \begin{beweisEnum}
      \item Siehe \cite[Prop.~3.3.17]{waldmann:2007a} oder \cite[Bew. zu
         Thm.~9.16]{lee:2003a}.

      \item Wir orientieren uns grob an der Darstellung von Michor
         \cite[Cor. 6.28]{michor:2007} Sei $\{U_i\}_{i \in I}$ eine
         offene Überdeckung von $M/G$. Mit dem Auswahlaxiom erhalten wir
         eine Abbildung $\alpha \dpA M/G \to I$ mit $y \in
         U_{\alpha(y)}$.
         \begin{itemize}
         \item[1.Schritt] Zuerst zeigen wir, dass es für jedes $y \in M/G$
            eine offene Umgebung $V_y$ von $y$ gibt mit $\abschluss{V_y}
            \subset U_{\alpha(y)}$. Sei dazu $y \in M/G$ und $x \in M$ mit
            $\pi(x) = y$. Wir wählen eine glatte Funktion $f_y \dpA M
            \to [0,\infty)$ mit $f_y\at{M\setminus
              \pi^{-1}(U_{\alpha(y)})} = 0$ und $f_y(x) = 1$ sowie eine
            linksinvariante Volumendichte $\Omega$ auf $G$ und definieren wie in
            Proposition \ref{kor:integrieren_inv_funkt} $f_y^G := \int_G
            f_y(g^{-1} \cdot) \Omega(g)$.  Dies induziert eine eindeutig
            bestimmte stetige Funktion $\bar{f}_y \dpA M/G \to
            [0,\infty)$ mit $\bar{f}_y \circ \pi = f_y^G$.  Es gilt dann
            offensichtlich aus Stetigkeitsgründen $\epsilon_y :=
            \bar{f}_y(y) = f_y^G(x) > 0$.

            Für $\epsilon_y > \delta_y > 0$ ist $y \in
            \bar{f}_y^{-1}((\delta_y,\infty)) \subset U_{\alpha(y)}$, denn ist
            $\bar{f}_y(\pi(\tilde{x})) > \delta_y$, so gibt es wenigstens
            ein $g \in G$ mit $f_y(g\tilde{x}) > 0$, also $g\tilde{x} \in
            \pi^{-1}(U_{\alpha(y)})$, d.\,h.\ $\pi(\tilde{x}) = \pi(g\tilde{x}) \in
            U_{\alpha(y)}$.  Damit ist klar, dass für $V_y :=
            \bar{f}_y^{-1}((\frac{\epsilon_y}{2},\infty))$ wegen der
            Stetigkeit von $\bar{f}_y$ die folgende Inklusion gilt
         \begin{align*}
           y \in  \abschluss{V_y}
            =\abschluss{\bar{f}_y^{-1}((\frac{\epsilon_y}{2},\infty))}
            \subset  \bar{f}_y^{-1}([\frac{\epsilon_y}{2},\infty)) \subset U_{\alpha(y)}\Fdot
         \end{align*}
      \item[2.Schritt] Im zweiten Schritt wählen wir von $\{V_y\}_{y \in
           M/G}$ eine abzählbare Teilüberdeckung aus. Sei dazu
         $\{\hat{K}_n\}_{n \in \mathbb{N}}$ eine Familie von Kompakta in
         $M$ mit $\bigcup_{n \in \mathbb{N}} \hat{K}_n = M$. Wegen der
         Stetigkeit von $\pi$ ist auch $K_n := \pi(\hat{K}_n)$ kompakt für
         $n \in \mathbb{N}$ und es gilt offenbar $\bigcup_{n \in
           \mathbb{N}} K_n = M/G$.  Da für $n \in \mathbb{N}$ $K_n$
         kompakt ist, gibt es ein endliches $E_n \subset M/G$ mit $K_n
         \subset \bigcup_{y\in E_n} V_{y}$.  Setze $E := \bigcup_{n
           \in \mathbb{N}} E_n$. $\{V_{y}\}_{y\in E}$ ist dann die
         gesuchte abzählbare Teilüberdeckung von $\{V_y\}_{y \in M/G}$.
      \item[3.Schritt] Nun konstruieren wir eine lokal endliche, offene
         Verfeinerung von $\{U_i\}_{i \in I}$. Dazu sei $\mathbb{N}
         \to E$, $k \mapsto y_k$ eine Abzählung von $E$. Definiere
         ferner für $l \in \mathbb{N}$
         \begin{align*}
            W_l := {U_{\alpha(y_l)}}\setminus (\abschluss{V_{y_1}} \cup
            \ldots \cup \abschluss{V_{y_{l-1}}}) \Fdot
         \end{align*}
         Dann ist klarerweise $\{W_l\}_{l \in \mathbb{N}}$ eine offene
         Überdeckung
         von $M/G$ und auch eine Verfeinerung von $\{U_i\}_{i \in
           I}$. Schließlich sehen wir, dass sie auch lokal endlich ist,
         denn $V_{y_k} \cap W_l = \emptyset$ für $l > k$: ist $y \in M/G$
         beliebig, so gibt es ein $k \in \mathbb{N}$ mit $y \in V_{y_k}$,
         d.\,h.\ $V_{y_k}$ ist eine zeugende Umgebung von $x$.
         \end{itemize}

   \end{beweisEnum}
\end{proofklein}

\begin{bemerkung}
   In \cite[Conj. 1]{antonyan:2009} wird die Vermutung aufgestellt, dass
   $M/G$ auch parakompakt ist unter den  schwächeren Voraussetzungen,
   dass $M$ ein parakompakter Hausdorffraum sei und $G$ eine
   hausdorffsche, lokalkompakte topologische Gruppe, die mit kompakter
   Wiederkehr auf $M$ wirke.
\end{bemerkung}

\begin{bemerkung}
   \label{bem:StratifizierterRaum}
   $M/G$ ist im Allgemeinen zwar keine Mannigfaltigkeit mehr, jedoch ein
   stratifizierter Raum, siehe etwa \cite{ortega.ratiu:2004} oder
   \cite{pflaum:2001a}.
\end{bemerkung}

Wir wollen nun die Parakompaktheit von $M/G$ nur durch Betrachten von
Überdeckungen auf $M$ beschreiben.

\begin{proposition}
   \label{prop:inv.lok.endl}
   Es wirke eine topologische Gruppe $G$ auf dem topologischen Raum $X$
   und sei $\pi \dpA X \to X/G$ die kanonische Projektion auf den
   Quotienten.
   \begin{propositionEnum}
   \item \label{item:lokEndlQuot} $\{\widetilde{V}_i\}_{i \in I} \subset
      \potMenge[X/G]$ ist genau dann eine lokal endliche Familie von
      offenen Mengen in $X/G$, wenn $\{\pi^{-1}(\widetilde{V}_i)\}_{i \in I}$
      eine lokal endliche Familie von $G$"=invarianten, offenen Mengen in
      $X$ ist.
   \item \label{item:invZeugendeUmgebung}Ist $\{V_i\}_{i \in I}$ eine
     lokal endliche Familie aus $G$"=invarianten,
      offenen Mengen in $X$ und $x \in X$, so gibt es eine $G$"=invariante
      zeugende Umgebung von $x$.
   \item \label{item:X/G_parakompakt_1} $X/G$ ist genau dann
      parakompakt, wenn jede $G$"=invariante, offene Überdeckung von $X$
      eine $G$"=invariante, offene, lokal endliche Verfeinerung
      besitzt. Die Verfeinerung kann auch präzise gewählt werden.
   \item \label{item:X/G_parakompakt_2} Sei $X/G$ parakompakt,
      hausdorffsch und $\{U_i\}_{i\in I}$ eine $G$"=invariante, lokal
      endliche, offene Überdeckung von $X$. Dann gibt es für jedes $i\in I$
      eine $G$"=invariante, offene Menge $V_i$, so dass $\abschluss{V_i}
      \subset U_i$ und $\{V_i\}_{i \in I}$ eine lokal endliche Überdeckung
      von $X$ ist.
    \item \label{item:guteInvZeugendeUmgebung} In der Situation von~\refitem{item:X/G_parakompakt_2} existiert zu jedem $x \in X$
      eine $G$"=invariante, gute zeugende Umgebung für $RP =
      (\{U_i\}_{i\in I}, \{V_i\}_{i \in I})$.
   \end{propositionEnum}
\end{proposition}

\begin{proofklein}
   \begin{beweisEnum}
   \item Sei $\{\widetilde{V}_i\}_{i \in I}$ also eine lokal endliche
      Familie von offenen Mengen in $X/G$ und $x \in X$ beliebig. Nach
      Voraussetzung gibt es eine zeugende Umgebung
      $\widetilde{\zeugendeU{x}}$ von $\pi(x)$, sie werde o.\,E.\ nur von
      $\widetilde{V}_{i_1},\ldots,\widetilde{V}_{i_r}$ geschnitten. Für
      $i \in I$ setze $V_i := \pi^{-1}(\widetilde{V}_i)$ und
      $\zeugendeU{x} := \pi^{-1}(\widetilde{\zeugendeU{x}})$. Dann folgt
      aus $V_i \cap \zeugendeU{x} \neq \emptyset$ schon $\widetilde{V}_i
      \cap \widetilde{\zeugendeU{x}} \neq \emptyset$ und somit $i \in
      \{i_1,\ldots,i_r\}$. Man bemerke, dass $\zeugendeU{x}$ per
      Definition $G$"=invariant ist.  Sei umgekehrt
      $\{\pi^{-1}(\widetilde{V}_i)\}_{i \in I}$ eine lokal endliche
      Familie $G$"=invarianter, offener Mengen in $X$, $x \in X$ und
      $\zeugendeU{x}$ eine zeugende Umgebung von $x$, die o.\,E.\ nur von
      $\pi^{-1}(\widetilde{V}_{i_1}),\ldots,\pi^{-1}(\widetilde{V}_{i_r})$
      geschnitten werde. Sei $\widetilde{\zeugendeU{x}} :=
      \pi(\zeugendeU{x})$ (offen in $X/G$, da die Wirkung stetig ist und
      $\pi$ damit offen). Angenommen $\widetilde{V}_i \cap
      \widetilde{\zeugendeU{x}} \neq \emptyset$, dann ist aber wegen der
      $G$"=Invarianz von $\pi^{-1}(\widetilde{V}_i)$ schon
      $\pi^{-1}(\widetilde{V}_i) \cap \zeugendeU{x} \neq \emptyset$,
      d.\,h.\ $i \in \{i_1,\ldots,i_r\}$.
   \item Klar nach dem Beweis von Teil~\refitem{item:lokEndlQuot}.
   \item Klar nach Teil~\refitem{item:lokEndlQuot}.
   \item $\{\pi(U_i)\}_{i \in I}$
     ist eine lokal endliche, offene Überdeckung von $X/G$. Nach
     Proposition~\ref{prop:reg_verfeinerung} gibt es für jedes $i \in
     I$ ein offenes $\widetilde{V}_i$ in $X/G$ mit
     $\abschluss{\widetilde{V}_i} \subset \pi(U_i)$, so dass
     $\{\widetilde{V}_i\}_{i \in I}$ eine lokal endliche, offene
     Überdeckung von $X/G$ ist. Wegen der Stetigkeit von $\pi$ folgt
     $\abschluss{\pi^{-1}(\widetilde{V}_i)} \subset
     \pi^{-1}(\abschluss{\widetilde{V}_i}) \subset \pi^{-1}(\pi(U_i))
     = U_i$ (vgl.\ \cite[Prop.~1.4.1]{engelking:1989}). Und nach Teil~\refitem{item:lokEndlQuot} ist $\{\pi^{-1}(\widetilde{V}_i)\}_{i
       \in I}$ eine lokal endliche Überdeckung von $X$ aus
     $G$"=invarianten, offenen Mengen.
   \item Man bemerke, dass mit einer $G$"=invarianten Menge $A$ wegen der Stetigkeit
      der Wirkung auch $\abschluss{A}$ $G$"=invariant ist, denn jene impliziert
      sofort $g \abschluss{A} \subset \abschluss{gA}$ für alle $g \in G$. Offensichtlich ist auch der Durchschnitt
      zweier $G$"=invarianter Mengen wieder $G$"=invariant. Nach Teil~\refitem{item:invZeugendeUmgebung} gibt es eine $G$"=invariante zeugende
      Umgebung. Führt man mit dieser die Konstruktion der guten zeugenden Umgebung
      wie im Beweis zu Proposition \ref{prop:guteZeugendeUmgebung} durch, so sieht
      man mit den eben gemachten Bemerkungen sofort, dass die resultierende gute
      zeugende Umgebung $G$"=invariant ist.
   \end{beweisEnum}
\end{proofklein}

\begin{satz}
   \label{satz:inv_zerlegung_der_eins}
   Sei $M$ eine Mannigfaltigkeit und wirke eine Lie"=Gruppe $G$
   eigentlich auf $M$. Weiter sei $\{U_i\}_{i \in I}$ eine
   $G$"=invariante, lokal endliche, offene Überdeckung von $M$. Dann
   gibt es eine $\{U_i\}_{i \in I}$ untergeordnete, $G$"=invariante glatte
   Zerlegung der Eins.
\end{satz}

\begin{proofklein} Mit Satz \ref{satz:Quotient_Parakompakt_T_2} und Proposition
   \ref{prop:inv.lok.endl}~\refitem{item:X/G_parakompakt_2} finden wir zu
   $i \in I$ ein offenes $G$"=invariantes $V_i$ mit $\abschluss{V_i}
   \subset U_i$, so dass $\{V_i\}_{i \in I}$ eine lokal endliche Überdeckung
      von $X$ ist.  Sei $\{\chi_i\}_{i \in I}$ eine nicht unbedingt
   $G$"=invariante Zerlegung der Eins mit kompaktem Träger, die $\{V_i\}_{i \in I}$
   untergeordnet sei. Wähle eine linksinvariante Volumendichte $\Omega$ auf $G$
   und definiere für $i \in I$ $\chi^G_i$ wie in Korollar
   \ref{kor:integrieren_inv_funkt} durch $\chi^G_i(x) := \int_G
   \chi_i(g^{-1} x) \Omega(g)$. Nach  Korollar \ref{kor:integrieren_inv_funkt} ist
   $\chi^G_i$ glatt und $\supp(\chi_i^G) \subset \abschluss{V_i} \subset
   U_i$. Dann definiere für $i \in I$
   \begin{align*}
      \widetilde{\chi}_i := \frac{\chi^G_i}{\sum_j \chi^G_j} \Fdot
   \end{align*}
   Man beachte, dass die Summe im Nenner endlich ist, da es zu $x \in M$
   nur endlich viele $j\in I$ gibt mit $x \in \abschluss{V_i}$ und da $\supp(\chi_i^G)
   \subset \abschluss{V_i}$.

   Aus Stetigkeitsgründen und nach Konstruktion
   ist die Summe auch echt positiv, so dass $\widetilde{\chi}_i$
   wohldefiniert ist. $\{{\widetilde{\chi}_{i}}\}_{i \in I}$ ist dann
   offenbar eine $G$"=invariante glatte Zerlegung der Eins, die $\{U_i\}_{i
     \in I}$ untergeordnet ist.  Siehe auch
   \cite[Thm.~5.2.5]{palais.terng:1988a}.
\end{proofklein}

\begin{korollar}
   \label{kor:invarianteFasermetrik}
   Sei $G$ eine Lie"=Gruppe und $\pi \dpA E \to M$ ein $G$"=Vektorbündel
   über der Mannigfaltigkeit $M$. Falls die $G$"=Wirkung auf $M$
   eigentlich ist, gibt es eine $G$"=invariante Riemannsche Fasermetrik $h^G$ auf
   $E$.
\end{korollar}
\begin{proofklein}
   Wir folgen der Argumentation von \cite[Satz~4.2.4~(6)]{pflaum:2000}.
   Sei $K_0 \subset \inneres{K_1} \subset K_1 \subset \inneres{K_2}
   \subset \ldots $ eine ausschöpfende Folge von Kompakta
   (vgl.\ \cite[Lem.~A.1.2]{waldmann:2007a}).  Bekanntlich, vgl.\ etwa
   \cite[Ch. II.4 Prop. IV]{greub1972connections}, gibt es eine
   Riemannsche Fasermetrik $h$. Zu $i \in \mathbb{N}$ sei $\chi_i\dpA M
   \to [0,1]$ eine glatte Abschneidefunktion mit $\supp{\chi_i} \subset
   \inneres{K_{i+1}}$ und $\chi_i \at{K_i} = 1$. Wir wählen eine
   linksinvariante Volumendichte $\Omega$ auf $G$ und definieren für $p
   \in M$ und $v,w \in E_p$
   \begin{align*}
      h_i(p)(v,w) := \int_G \chi_i(g^{-1}p)
      h(g^{-1}x)(g^{-1}v,g^{-1}w) \Omega(g) \Fdot
   \end{align*}
   Offensichtlich hat $\chi_i h$ kompakten Träger und $h_i = (\chi_i
   h)^G$ im Sinne von Proposition~\ref{kor:integrieren_inv_funkt} ist
   wohldefiniert und $G$"=äquivariant. Ferner ist offensichtlich, dass
   $h_i(p)$ symmetrisch und positiv ist und sogar positiv definit für $p
   \in GK_i$. Mit Satz \ref{satz:inv_zerlegung_der_eins} können wir eine
   der $G$"=invarianten offenen Überdeckung $\{G \inneres{K_i}\}_{i \in \mathbb{N}}$
   untergeordnete $G$"=invariante Zerlegung der Eins $\{\psi_i\}_{i \in
     \mathbb{N}}$ wählen. Dann definieren wir
   \begin{align*}
      h^G := \sum_{i \in \mathbb{N}} {\psi_i h_i}
   \end{align*}
   und sehen ohne Weiteres, dass $h^G$ eine $G$"=invariante Riemannsche Fasermetrik
   auf $E \to M$ liefert.
\end{proofklein}

\begin{bemerkung}
   \label{bem:invarianteFasermetrikAusInvarianterMetrik}
   Ist $G$ eine Lie"=Gruppe und $\pi \dpA E \to M$ ein $G$"=Vektorbündel
   über der Mannigfaltigkeit $M$,  so induziert jede $G$"=invariante
   Riemannsche Metrik $\widetilde{h}$ auf $E$ eine $G$"=invariante
   Riemannsche Fasermetrik $h$ auf $E \to M$.
\end{bemerkung}
\begin{proofklein}  Sei $n \dpA
   M \to E$ der Nullschnitt.
   Wir  definieren für $p \in M$ und $v_p,w_p \in E_p$
   \begin{align*}
      h(p)(v_p,w_p) :&= \widetilde{h}(n(p))\left(\ddt (tv_p),\ddt (tw_p)\right) \\
      &= \widetilde{h}(n(p))(\vbNIso{E \to M}(0,v_p),\vbNIso{E
        \to M}(0,w_p)) \Fcom
   \end{align*}
    wobei
    \begin{align*}
      \vbNIso{E \to M} \dpA TM \oplus E \to
      TE\at{\nSchnitt(M)}, \quad (v,w) \mapsto T \nSchnitt \, v + \ddt
      (tw)
   \end{align*}
   der Vektorbündelmorphismus aus
   Proposition~\ref{prop:3}~\refitem{item:TangentialraumEinesVBsAmNullschnitt}
   ist.  Für $p \in M$ ist $h(p)$ natürlich bilinear und symmetrisch
   und positiv auf $E_p$ und wegen der Injektivität von $E_p
   \hookrightarrow T_pM \oplus E_p$ , $v \mapsto (0,v)$ und $\vbNIso{E
     \to M}$ ist $h(p)$ auch klarerweise positiv definit.  Da auch
   klar ist, dass $h$ glatt ist, haben wir also schon gezeigt, dass $h$
   eine Riemannsche Fasermetrik auf $E$ ist.

   Wir rechnen nach, dass $h$ invariant ist. Dazu sei $\phi$ die
   $G$"=Wirkung auf $M$ und $\Phi$ die dazugehörende geliftete $G$"=Wirkung
   auf $E$ sowie $g \in G$, $p \in M$ und $v_p,w_p \in E_p$.
   \begin{align*}
      h(\phi_g(p))(\Phi_g(v_p),\Phi_g(w_p)) &=
      \widetilde{h}(n(\phi_g(p)))\left(\ddt(t \Phi_g(v_p)), \ddt (t
        \Phi_g (w_p))\right) \\
      &= \widetilde{h}(\Phi_g(n(p)))\left(\ddt(\Phi_g(t v_p)), \ddt (\Phi_g (
      tw_p))\right) \\
      &= \widetilde{h}(\Phi_g(n(p)))\left(T\Phi_g\ddt(t v_p), T\Phi_g \ddt(
      tw_p)\right) \\
      &= \widetilde{h}(n(p))\left(\ddt(t v_p),\ddt(t w_p)\right) \\
      &= h(p)(v_p,w_p) \Fdot
   \end{align*}
\end{proofklein}

\begin{bemerkung}
   \label{bem:invarianteVollstaendigeRiemannscheMetrik}
   Mit einigem Mehraufwand kann man sogar zeigen, dass es auf $M$ eine
   vollständige, $G$"=invariante Riemannsche Metrik gibt. Siehe
   \cite[Thm.~0.2]{illmann_kankaarinta:2000}.
\end{bemerkung}

\begin{korollar}
   \label{kor:InvarianterZusammenhang}
   Sei $M$ eine Mannigfaltigkeit, $\phi \colon G \times M \to M$ die
   Wirkung einer Lie"=Gruppe auf $M$ und $h$ eine $G$"=invariante
   Riemannsche Metrik auf $M$. Dann ist der Levi"=Civita"=Zusammenhang
   $\nabla$ zu $h$ $G$"=invariant, d.\,h.\ es gilt für $X,Y \in
   \Gamma^\infty(TM)$ und $g \in G$
   \begin{align}
      \label{eq:invarianterLeviCivita}
      g (\nabla_XY) = \nabla_{gX}(gY)\Fdot
   \end{align}
   Dabei wirkt $G$ auf $\Gamma^\infty(TM)$ durch $gX :=
   \phi_{g^{-1}}^*X$ für $X \in \Gamma^\infty(TM)$ und $g \in
   G$. Insbesondere existiert auf $M$ immer eine $G$"=invariante,
   torsionsfreie kovariante Ableitung, falls $G$ eigentlich wirkt.
\end{korollar}
\begin{proofklein}
   Seien $X,Y,Z \in \Gamma^\infty(TM)$ und $g \in G$. Man rechnet mit
   den Definitionen leicht folgende Formeln nach: $h(g X,g Y) =
   g(h(X,Y))$, $(gX)(f) = g(X(g^{-1}f))$ für $f \in C^\infty(M)$ und
   $g[X,Y] = [gX,gY]$. Dann folgt die Behauptung unmittelbar aus der
   Koszul"=Formel für  $\nabla$
   \begin{align*}
      2 h(\nabla_XY,Z) &= X(h(Y,Z)) + Y(h(Z,X)) - Z(h(X,Y))\\
      &\phantom{=}- h(Y,[X,Z]) - h(Z,[Y,X]) + h(X,[Z,Y]) \Fcom
   \end{align*}
   vgl.\ \cite[Thm. 5.4]{lee1997riemannian}, und der Nichtausgeartetheit
   von $h$.
\end{proofklein}

\begin{lemma}
   \label{lem:1Ball}
   Sei $G$ eine Lie"=Gruppe und $\pi \dpA E \to C$ ein $G$"=Vektorbündel
   über der Mannigfaltigkeit $C$, wobei die $G$"=Wirkung auf $C$
   eigentlich sei. Ist $W$ eine $G$"=invariante, offene Umgebung des
   Nullschnittes $\nSchnitt(C)$ in $E$, so gibt es eine $G$"=invariante
   Riemannsche Fasermetrik $h$ auf $E$, so dass $B_1(C) := \{v \in E \mid
   h(v,v) \leq 1\} \subset W$.
\end{lemma}
\begin{proofklein}
   \tolerance=500
   Wir übertragen im Groben die Darstellung in \cite[Bew. zu
   Lem.~3.1.2]{pflaum:2000} auf den $G$"=invarianten Fall.  Mit Korollar
   \ref{kor:invarianteFasermetrik} gibt es eine $G$"=invariante Riemannsche
   Fasermetrik $g$ auf $E$. Sei $\rho \dpA E \to \mathbb{R}$
   gegeben durch $\rho(v) := g(v,v)$. $\rho$ ist offensichtlich
   $G$"=invariant.

   \tolerance=200
   \begin{itemize}
   \item[1.Schritt] Ist $c \in C$ und $\{U^n_c\}_{n \in \mathbb{N}}$ eine
      Umgebungsbasis von $c$ (aus offenen Mengen) in $C$ und $V_c$ eine
      offene Umgebung von $n(c)$ mit kompaktem Abschluss
      $\abschluss{V_c}$, so ist $\{W_c^n\}_{n \in \mathbb{N}}$ mit $W_c^n
      := V_c \cap \pi^{-1}(U_c^n) \cap \rho^{-1}([0,\frac{1}{n}))$ eine
      Umgebungsbasis (aus offenen Mengen) von $n(c)$ in $E$. Denn ist dies
      nicht der Fall, so gibt es eine offene Umgebung $Z_c$ von $n(c)$, so dass
      zu $n \in \mathbb{N}$ ein $x_n \in W_c^n \setminus Z_c$
      existiert. Dies ist aber ein Widerspruch. Denn einerseits verläuft
      die Folge $(x_n)$ im Kompaktum $\abschluss{V_c}$ und besitzt somit
      eine konvergente Teilfolge, die wir uns o.\,E.\ schon ausgewählt denken
      und deren Grenzwert wir mit $x$ bezeichnen wollen. Nach Konstruktion
      und der Stetigkeit von $\pi$ und $\rho$ gilt dann aber $\pi(x) =
      \lim_{n \to \infty} \pi(x_n) = c$ und $\rho(x) = \lim_{n \to \infty}
      \rho(x_n) = 0$, also $x = n(c)$. Andererseits gilt $x_n \notin Z_c$
      für alle $n \in \mathbb{N}$, d.\,h.\ $(x_n)$ besitzt keine gegen $n(c)$
      konvergierende Teilfolge.
   \item[2.Schritt] Ist $c \in C$, dann gibt es eine offene Umgebung $U_c$
      von $c$ in $C$ und ein $\lambda_c > 0$ mit $W_c := \pi^{-1}(U_c)
      \cap \rho^{-1}([0,\lambda_c)) \subset W$. Dies folgt unmittelbar aus
      dem 1. Schritt. Damit ist aber wegen der $G$"=Äquivarianz von $\pi$, der
      $G$"=Invarianz von $\rho$, und da $W$ $G$"=invariant ist, schon
      \begin{align*}
         W_c^G := \pi^{-1}(G U_c) \cap \rho^{-1}([0,\lambda_c)) = G W_c
         \subset W \Fdot
      \end{align*}

      Nun existiert nach Satz~\ref{satz:inv_zerlegung_der_eins} eine
      der $G$"=invarianten offenen Überdeckung $\{G U_c\}_{c \in C}$ untergeordnete
      $G$"=invariante Zerlegung der Eins $\{\chi_i\}_{i \in I}$, es
      gibt also insbesondere für $i \in I$ ein $c(i) \in C$ mit $\supp
      \chi_i \subset G U_{c(i)}$. Für $\lambda := \sum_i
      \lambda_{c(i)} \chi_i$ gilt dann für $v_c \in E_c$ die
      Implikation $g(v_c,v_c) \leq \lambda(c) \leq
      \max\{\lambda_{c(i)} \mid i \in I \, \text{mit} \, \chi_i(c)
      \neq 0\} \Rightarrow v_c \in W_c^G \subset W$.
   \item[3.Schritt] Da $\lambda_c > 0$ für alle $c \in C$, ist $h :=
      \frac{1}{\lambda} \, g$ wohldefiniert sowie offensichtlich  eine
      Riemannsche Fasermetrik auf $E$ und für $v \in E$ mit $h(v,v) \leq
      1$, gilt
         \begin{align*}
            g(v,v) = \lambda(\pi(v)) h(v,v) \leq \lambda(\pi(v)) \Fcom
         \end{align*}
         d.\,h.\ $v \in W$. Also ist $B_1(C) \subset W$ wie gewünscht.
         Nach Konstruktion ist $h$ offensichtlich $G$"=invariant.
      \end{itemize}
\end{proofklein}

Wir treffen nun einige Vorbereitungen um aus einem gegebenen Spray einen
$G$"=invarianten Spray konstruieren zu können.

\begin{lemma}
   \label{lem:KleinGVektorbuendel}

   Sei $G$ eine Lie"=Gruppe und $M$, $N$ Mannigfaltigkeiten, auf denen
   $G$ wirke und $F \dpA M \to N$ stetig, surjektiv und $G$"=äquivariant.
   Falls $K \subset M$ klein ist, so ist auch $F^{-1}(K) \subset N$ klein.
\end{lemma}
\begin{proofklein}
   Sei $p \in N$ und $U$ eine Umgebung von $F(p)$, so dass $\{g
   \in G \mid g U \cap K \neq \emptyset\}$ relativ kompakt ist. Dann ist
   $F^{-1}(U)$ eine Umgebung von $p$ in $N$, so dass
   \begin{align*}
      \{g \in G \mid g F^{-1}(U) \cap F^{-1}(K) \neq
      \emptyset \} &= \{g \in G \mid F^{-1}(gU) \cap
      F^{-1}(K) \neq \emptyset \} \\
      &= \{g \in G \mid F^{-1}(gU \cap K) \neq \emptyset\} \\
      &= \{g \in G \mid gU \cap K \neq \emptyset\}
   \end{align*}
   relativ kompakt ist. Somit ist  $F^{-1}(K)$ klein.
\end{proofklein}
\begin{bemerkung}
   \label{bem:GVektorbuendelUndKlein}
   Die Voraussetzungen von Lemma \ref{lem:KleinGVektorbuendel} sind
   insbesondere erfüllt, wenn $F \colon N \to M$ ein
   $G$"=Vektorbündel ist.
\end{bemerkung}

\begin{lemma}[{\cite[Appendix I, Lemma 1]{bourbaki.2004}}]
   \label{lem:SpezFunktion}
   Sei $M$ eine Mannigfaltigkeit, auf der eine Lie"=Gruppe $G$
   operiere. Dann gibt es eine glatte Funktion $\chi \dpA M \to \mathbb{R}$
   mit den folgenden Eigenschaften.
   \begin{lemmaEnum}
   \item $\chi \geq 0$.
   \item Für alle $x \in M$ gibt es ein $g \in G$ mit $f(gx) \neq 0$,
      d.\,h.\ $\chi$ ist nicht identisch Null auf jedem Orbit.
   \item Der Träger $\supp \chi$ ist kompakt in Faserrichtung.
   \end{lemmaEnum}
\end{lemma}
\begin{proofklein}
   Sei $\pi \dpA M \to M/G$ die kanonische Projektion auf den
   Quotienten. Diese ist offen, da die $G$"=Wirkung stetig ist. Für
   $\overline x \in M/G$ sei $\chi_{\overline x} \dpA M \to \mathbb{R}$
   eine glatte, nichtnegative Funktion mit kompaktem Träger $\supp
   \chi_{\overline x}$ und $\chi_{\overline x}(x) > 0$ für ein $x \in
   \pi^{-1}(\overline x)$. Wegen der Stetigkeit von $\chi_{\overline x}$
   ist $\Omega_{\overline x} := \{ x \in M \mid \chi_{\overline x}(x) >
   0\}$ offen in $M$ und damit auch $\pi(\Omega_{\overline x})$ offen in
   $M/G$. Demnach ist $\{G \Omega_{\overline x}\}_{{\overline x} \in
     M/G}$ eine offene, $G$"=invariante Überdeckung von $M$. Wir wählen
   nun eine, dieser Überdeckung untergeordnete, glatte, $G$"=invariante
   Zerlegung der Eins $\{\epsilon_i\}_{i \in I}$ mit lokal endlichen
   Trägern. Für $i \in I$ sei $\overline x(i)$ so gewählt, dass $\supp
   \epsilon_i \subset G \Omega_{\overline x(i)} =: U_i$. Dann setzen wir
   $\chi_i := \epsilon_i \chi_{\overline x(i)}$ und definieren
   \begin{align*}
      \chi := \sum_i \chi_i \Fdot
   \end{align*}
   Dies ist offensichtlich eine wohldefinierte, glatte, nichtnegative
   Funktion. Sei nun $x \in M$, dann gibt es ein $i \in I$, so dass
   $\epsilon_i(x) > 0$, also $x \in G\Omega_{\overline x(i)}$. Es gibt
   also ein $g \in G$ mit $gx \in \Omega_{\overline x(i)}$, d.\,h.\ mit
   $\chi(gx) \geq \epsilon_i(gx) \chi_{\overline x(i)}(gx) =
   \epsilon_i(x) \chi_{\overline x(i)}(gx) > 0$.  Mit $\{\supp
   \epsilon_i\}_{i \in I}$ ist auch $\{\supp \chi_i\}_{i \in I}$ lokal
   endlich, also auch die Familie $\{\carr \chi_i\}_{i \in I}$. Damit
   gilt (vgl.\ Prop. \ref{prop:1})
   \begin{align*}
      \supp \chi = \abschluss{\carr \chi} = \abschluss{\bigcup_i \carr
        \chi_i} = \bigcup_i \abschluss{\carr \chi_i} =
      \bigcup_i \supp \chi_i \Fdot
   \end{align*}
   Seien nun $K \subset M/G$ kompakt. Wegen der lokalen Endlichkeit der
   Familie $\{\supp \chi_i\}_{i \in I}$ ist auch die Familie $\{\pi(\supp
   \chi_i)\}_{i \in I}$ lokal endlich. Aufgrund der Kompaktheit von $K$ gibt es
   damit eine endliche Indexmenge $J \subset I$ mit $\pi(\supp \chi_i)
   \cap K = \emptyset$ und somit $\supp \chi_i \cap \pi^{-1}(K) \subset
   \pi^{-1}(\pi(\supp \chi_i)) \cap \pi^{-1}(K) = \pi^{-1}(\pi(\supp
   \chi_i) \cap K) = \emptyset$ für alle $i \in I\setminus J$. Wegen der
   Kompaktheit von $K$ ist $\pi^{-1}(K)$ abgeschlossen, also ist mit
   $\supp \chi_i$ auch $\pi^{-1}(K) \cap \supp \chi_i$ kompakt für jedes
   $i \in I$. Dann ist aber auch
   \begin{align*}
      \pi^{-1}(K) \cap \supp(\chi) &= \pi^{-1}(K) \cap \bigcup_{i \in
        I} \supp \chi_i\\
      &= \bigcup_{i \in I} (\pi^{-1}(K) \cap
      \supp \chi_i)\\
      &= \bigcup_{i \in J} (\pi^{-1}(K) \cap \supp \chi_i)
   \end{align*}
kompakt als endliche Vereinigung von Kompakta.
\end{proofklein}

\begin{korollar}
   \label{kor:SpezFunktion}
   Sei $\Omega$ eine linksinvariante Volumendichte auf der Lie"=Gruppe
   $G$ und $M$ eine  Mannigfaltigkeit, die mit einer eigentlichen
   $G$"=Wirkung versehen sei. Dann gibt es eine glatte Funktion $\chi
   \dpA M \to \mathbb{R}$ mit in Faserrichtung kompaktem Träger, so dass
   \begin{align}
      \label{eq:SpezFunktionIntegral}
      \int \chi(g^{-1} p) \Omega(g) = 1
   \end{align}
   für alle $p \in M$ gilt.
\end{korollar}
\begin{proofklein}
   Sei $\tilde \chi \dpA M \to \mathbb{R}$ wie in Lemma
   \ref{lem:SpezFunktion}. Dann definieren wir
   \begin{align*}
      \chi \dpA M \to \mathbb{R}, \quad p \mapsto \frac{\tilde{\chi}(p)}{\int_G
        \tilde \chi(g^{-1}p) \Omega(g)} \Fdot
   \end{align*}
   Man beachte dass der Nenner nach Korollar
   \ref{kor:integrieren_inv_funkt} und Proposition
   \ref{prop:klein-kompaktInFaserrichtung} wohldefiniert ist, da $\tilde
   \chi$ in Faserrichtung kompakten Träger hat.

\end{proofklein}

Als nächstes möchten wir kurz an die Definition von Spray-Vektorfeldern,
der Exponentialabbildung und an einige elementare Eigenschaften dieser
erinnern und anschließend die obigen Vorbereitungen nutzen, um aus einem
gegebenen Spray einen $G$"=invarianten Spray zu konstruieren.

\begin{definition}[Spray]
   \label{def:Spray}
   Sei $M$ eine Mannigfaltigkeit.
   \begin{definitionEnum}
      \item
         \label{item:VektorFeldZweiterOrdnung}
         Ein Vektorfeld $\xi \colon TM \to TTM$ heißt
         \neuerBegriff{Spray-Vektorfeld} (oder kurz
         \neuerBegriff{Spray}), falls für alle $p \in M$, $v_p \in
         T_pM$ und $\alpha \in \mathbb{R}$ die folgenden Bedingungen
         erfüllt sind.
         \begin{definitionEnum}
         \item
            $T_{v_p}\pi\xi(v_p) = v_p$ ($\xi$ ist Vektorfeld zweiter Ordnung).
         \item
            $\xi(\alpha v_p) = T_{v_p} \mu_{\alpha} \xi(v_p)$.
         \end{definitionEnum}
         Dabei ist $\mu_\alpha \dpA TM \to TM$, $v \mapsto \alpha v$ die
         faserweise Multiplikation mit $\alpha \in \mathbb{R}$ und $\pi
         \colon TM \to M$ die Fußpunktprojektion.
      \item %
         \label{item:SprayMultiplikationsBedingung}
         Ist $\phi \colon G \times M \to M$ die Wirkung einer
         Lie"=Gruppe $G$ auf $M$ und $\Phi \colon G \times TM \to TM$,
         $(g,v) \mapsto T_{\pi(v)}\phi_g v$ die davon induzierte
         Wirkung auf $TM$, so nennen wir einen Spray $\xi \colon TM
         \to TTM$ \neuerBegriff{$G$"=äquivariant}, falls das Vektorfeld
         $\xi$ für alle $g \in G$ zu sich selbst $\Phi_g$"=verwandt
         ist, d.\,h.\  für alle $p \in M$ und $v_p \in T_pM$

         \begin{align}
            \label{eq:SprayInvariant}
            \xi(\Phi_g v_p) = T_{v_p}\Phi_g \xi(v_p)
         \end{align}
         gilt. Mit anderen Worten bedeutet dies, dass $\xi$
         $G$"=äquivariant ist bezüglich $\Phi$ und der davon induzierten
         Wirkung $G \times TTM \ni (g,w) \mapsto T \Phi_g w$ auf $TTM$.
      \item %
         Sei $\xi \colon TM \to TTM$ ein Spray und $W \subset \mathbb{R}
         \times TM
         $ die größte offene Teilmenge von $\mathbb{R} \times TM$, auf welcher der Fluß $F$ von $\xi$
         definiert ist.

         Weiter sei $\mathcal{O} := \{v \in TM \mid (v,1) \in W \}$. Die
         Abbildung $\exp \colon \mathcal{O} \to M$, $v \mapsto
         \pi(F(v,1)))$ heißt dann (von $\xi$ induzierte)
         \neuerBegriff{Exponentialabbildung} mit maximalem
         Definitionsbereich $\mathcal{O}$.
   \end{definitionEnum}
\end{definition}

\begin{bemerkung}
   \label{bem:Spray}
   Spray"=Vektorfelder werden in vielen Lehrbüchern der
   Differentialgeometrie leider nicht behandelt, weshalb wir an dieser
   Stelle einige Bücher angeben möchten, die sich diesem Thema in
   einführender Weise annehmen. Zu nennen sind hier
   \cite{broecker.jaenich:1990a}, \cite{lang:1999a} und
   \cite{michor:2007}, \cite{brickell1970differentiable}. Wir verweisen
   auch auf die Originalarbeit von Ambrose, Palais und Singer
   \cite{ambrose1960sprays}, auf welche der Begriff Spray zurückgeht.
\end{bemerkung}

Die nächste Proposition gibt eine lokale Charakterisierung eines
Spray"=Vektorfeldes an, vgl.\ \cite[Eq. (10.2.3)]{brickell1970differentiable}.

\begin{proposition}
   \label{prop:lokaleForm}
   Ein Vektorfeld $\xi \colon TM \to TTM$ ist genau dann ein Spray, wenn
   es für jede Karte $(U,x)$ und die davon induzierte Karte
   $(\pi^{-1}(U),(q,v))$ von $TM$ glatte Funktionen $f^i_{jk}\colon
   \mathbb{R}^n \to \mathbb{R}$ gibt mit
   \begin{align}
      \label{eq:SprayLokal}
      \xi\at{\pi^{-1}(U)} = v^i \frac{\partial}{\partial q^i} + (f^i_{jk}\circ
      q)v^iv^k \frac{\partial}{\partial v^i} \Fdot
   \end{align}
\end{proposition}

\begin{proposition}
   \label{prop:Spray}
   Sei $M$ eine Mannigfaltigkeit, $\xi \colon TM \to TTM$ ein Spray und
   $\exp$ die von $\xi$ induzierte Exponentialabbildung mit maximalem
   Definitionsbereich $\mathcal{O}$.
   \begin{propositionEnum}
      \item %
         $\mathcal{O}$ ist eine offene Umgebung des Nullschnittes von
         $M$ in $TM$.
      \item %
         Für alle $t \in \mathbb{R}$ und $v \in \mathcal{O}$ ist
         $tv \in \mathcal{O}$.
      \item %
         Ist $\phi \colon G \times M \to M$ die Wirkung einer Lie"=Gruppe
         $G$ auf $M$ und ist $\xi$ $G$"=äquivariant, so ist
         $\mathcal{O}$ $G$"=invariant und $\exp$ $G$"=äquivariant.
   \end{propositionEnum}
\end{proposition}
\begin{proofklein}
   \begin{beweisEnum}
      \item %
         $\mathcal{O}$ ist offen, da $W$ offen ist.

         Aus der Bedingung~\refitem{item:SprayMultiplikationsBedingung}
         folgt $\xi(0_p) = 0 \in T_{0_p}TM$ für alle $p \in M$, also
         sind Integralkurven von $\xi$, die in $0_p$ starten, für ein $p
         \in M$ auf ganz $\mathbb{R}$ definiert.
      \item %
         Siehe \cite[Lemma 10.5.1]{brickell1970differentiable}.
      \item %
         Dies ist klar, da $\xi$ zu sich selbst $\Phi_g$ verwandt ist für alle $g
         \in G$, wobei $\Phi_g := T\phi_g$. Ist dann $F \colon \mathbb{R} \times TM
         \supset W \to TM$ der Fluß zu $\xi$, so gilt aus diesem Grund bekanntlich,
         dass mit $(t,v) \in W$ auch schon $(t,\Phi_g(v)) \in W$ ist und die
         Gleichung $F_t(\Phi_g(v)) = \Phi_g(F_t(v))$ richtig ist, siehe etwa
         \cite[Lem. 18.4]{lee:2003a}. Also sieht man unmittelbar, dass
         $\mathcal{O}$ $G$"=invariant ist und für alle $g \in G$ und $v \in
         \mathcal{O}$
         \begin{align*}
            \exp(\Phi_g v) = \pi(F_1(\Phi_g v)) = \pi(\Phi_g F_1(v)) =
            \phi_g \pi(F_1(v)) = \phi_g \exp(v)
         \end{align*}
         gilt.
   \end{beweisEnum}
\end{proofklein}

\begin{bemerkung}
   \label{bem:VektorfeldZweiterOrdnung}
   Die Bedingung~\refitem{item:VektorFeldZweiterOrdnung} in Definition
   \ref{def:Spray} ist äquivalent dazu, dass für jede Integralkurve $c \colon I \to
   TM$ von $\xi$
   \begin{align}
      \label{eq:IntergalKurveZweiteOrdnung}
      \ddt(\pi \circ c) = c
   \end{align}
   gilt. Insbesondere ist dann $[0,1] \ni t \mapsto \ddt(\exp(t\nu)) \in TM$ eine
   Integralkurve von $\xi$ für $\nu \in \mathcal{O}$, siehe \cite[Lemma
   3.2.31]{waldmann:2007a}.
\end{bemerkung}

\begin{satz}
   \label{satz:invarianterSpray} Sei $M$ eine Mannigfaltigkeit, die
   eine eigentliche Wirkung einer Lie"=Gruppe $G$ trage. Sei $\pi
   \colon E \to M$ mit $E \subset TM$ ein $G$"=invariantes
   Untervektorbündel und $\xi\dpA TM \to TTM$ ein Spray mit $\xi(E)
   \subset TE \subset TTM$. Weiter sei $\Omega$ eine linksinvariante
   Volumendichte auf $G$ und  $\chi \dpA M \to \mathbb{R}$ wie in
   Korollar~\ref{kor:SpezFunktion}. Dann ist $\supp(\pi^{*}
   \chi)$ klein und
   \begin{align}
      \label{eq:invarianterSpray}
      &\xi^G \dpA TM \to TTM, \\
      &v_p \mapsto \xi^G(v_p) := \int_G \pi^* \chi(g^{-1} v_p) g
      \xi(g^{-1} v_p) \Omega(g)
   \end{align}
   ist ein $G$"=äquivarianter Spray mit $\xi^G(E) \subset TE$.
\end{satz}
\begin{proofklein}
   Zunächst sieht man mit
   Proposition~\ref{prop:klein-kompaktInFaserrichtung} und
   Lemma~\ref{lem:KleinGVektorbuendel} leicht, dass $\supp(\pi^* \chi)
   \subset \pi^{-1}(\supp \chi)$ klein ist. Nach Proposition
   \ref{prop:integrieren_inv_funkt} ist damit $\xi^G$ wohldefiniert und
   ein $G$"=äquivariantes Vektorfeld auf $TM$. Weiter rechnen wir
   \begin{align*}
      T_{v_p} \pi \xi^G(v_p) &= \int_G \pi^*\chi(g^{-1} v_p) T_{v_p}\pi
      g
      \xi(g^{-1} v_p) \Omega(g) \\
      &= \int_G \pi^*\chi(g^{-1} v_p) g T_{g^{-1} v_p} \pi \xi(g^{-1} v_p)
      \Omega(g)
      \\
      &= \int_G \chi(g^{-1}p) g g^{-1} v_p \Omega(g) = v_p \int_G
      \chi(g^{-1} p) \Omega(g) = v_p
   \end{align*}
   nach.  Dabei wurde im zweiten Schritt die Kettenregel ausgenutzt
   sowie die Tatsache, dass die $G$"=Wirkungen auf $TM$ bzw.\ $TTM$
   durch ein bzw.\ zweimaliges Ableiten aus der $G$"=Wirkung auf $M$
   induziert werden. Im dritten Schritt haben wir verwendet, dass
   $\xi$ schon ein Spray ist.
   Als nächstes sehen wir für $v_p \in T_pM$
   \begin{align*}
      \xi^G(\alpha v_p) &= \int_G \chi(g^{-1} p) g \xi(g^{-1} \alpha
      v_p)
      \Omega(g) \\
      &= \int_G \chi(g^{-1} p) g T_{g^{-1} v_p}\mu_\alpha \xi(g^{-1}
      v_p) \Omega(g) \\
      &= T_{v_p} \mu_\alpha \int_G \chi(g^{-1} p) g \xi(g^{-1}v_p) \Omega(g) \\
      &= T_{v_p} \mu_\alpha \xi^G(v_p)\Fcom
   \end{align*}
   wobei $\mu_\alpha \dpA TM \to TM$, $v \mapsto \alpha v$ die faserweise
   Multiplikation mit $\alpha \in \mathbb{R}$ ist.
   Sei schließlich $v_p \in E_p$, dann gilt schon
   \begin{align*}
      \xi^G(v_p) = \int_G \chi(g^{-1} p) g \xi(g^{-1} v_p) \Omega(g) \in
      T_{v_p}E \Fcom
   \end{align*}
   da die Integration faserweise verläuft, denn $\xi(g^{-1}v_p)
   \in T_{g^{-1}v_p}E$, d.\,h.\ $g\xi(g^{-1}v_p) \in T_{v_p}E$ für alle $g
   \in G$.
\end{proofklein}

\cleardoublepage

\chapter{Tubensätze}
\label{cha:Tubensatze}

In diesem Kapitel beschäftigen wir uns zunächst grob gesagt mit der
Frage, unter welchen Voraussetzungen ein lokaler Homöomorphismus in der
Umgebung einer Untermannigfaltigkeit zu einem Homöomorphismus
wird. Dabei sind wir insbesondere an der Situation interessiert, in der
alle auftauchenden Strukturen mit der Wirkung einer topologischen Gruppe
$G$ verträglich.

Auf diesen Überlegungen aufbauend können wir dann den bekannten
Existenzsatz für kompatible Tubenumgebungen \cite[Ch. II,
Thm. 1.6]{gibson:1976} auf die $G$"=invariante Situation verallgemeinern.

\section{Zusammenkleben von Schnitten}

Sind $X$ und $Y$ topologische Räume und ist $F \dpA X \to Y$
ein lokaler Homöomorphismus, so findet man  nach Definition von
lokalem Homöomorphismus zu jedem $x \in X$ eine offene Umgebung $U$ von
$x$ in $X$, so dass $V := F(U)$ eine offene Umgebung von $F(x)$ in $Y$
ist und $F\at{U} \dpA U\to V$ ein Homöomorphismus ist. Hat man
angenommen zusätzlich noch eine (topologische) Gruppe $G$ gegeben, die
sowohl auf $X$ als auch auf $Y$ wirkt und ist $F$ bezüglich dieser äquivariant, so
könnte man naiv denken, dass man $U$ und $V$ auch $G$"=invariant wählen
könnte. Dies ist im Allgemeinen aber -- selbst wenn die Wirkungen eigentlich
sind --  nicht so, wie das folgende Beispiel
zeigt.

\begin{beispiel}
   \label{bsp:aeqiv_lokal_homoe_gegenbsp}
   Sei $X = Y = G = S^1 \subset \mathbb{C}$. $G$ wirke auf $X$ via $(z,x)
   \mapsto z x$ und auf $Y$ via $(z,x) \mapsto z^2 x$. Sei $F \dpA X
   \to Y$, $F(x) := x^2$. Da $G$ kompakt ist, sind beide Wirkungen
    eigentlich. $F$ ist offensichtlich $G$"=äquivariant und ein
   lokaler Homöomorphismus. Die Wirkung auf $X$ ist jedoch transitiv,
   weshalb $X$ die einzige $G$"=invariante Teilmenge von $X$ ist.  $F$ ist
   aber offensichtlich auf $X$ nicht injektiv.
\end{beispiel}

\begin{definition}[$G$"=lokaler Homöomorphimus]
   \label{def:G_lokal_homeo}
   Seien $X$ und $Y$ topologische Räume und $G$ eine Gruppe, die sowohl
   auf $X$ als auch auf $Y$ wirke. Ist $F \dpA X \to Y$ $G$"=äquivariant und gibt
   es zu jedem $x \in X$ eine $G$"=invariante, offene Umgebung $U$ von $x$
   und eine $G$"=invariante, offene Umgebung $V$ von $F(x)$ in $Y$, so dass
   $F\at{U} \dpA U \to V$ ein Homöomorphismus ist, so nennen wir
   $F$ \neuerBegriff{$G$"=lokaler Homöomorphimus}.
\end{definition}

Es ist nun eine naheliegende Frage, unter welchen Bedingungen ein
$G$"=äquivarianter lokaler Homöomorphismus ein $G$"=lokaler Homöomorphismus
ist.

\begin{proposition}
   \label{prop:G_lokal_homeo}
   Seien $X$ und $Y$ topologische Räume und $G$ eine Gruppe, $\phi$ eine
   $G$"=Wirkung auf $X$ und $\theta$ eine $G$"=Wirkung auf $Y$.
   \begin{propositionEnum}
    \item \label{item:G_lokal_homeo1} Ist $F \dpA X \to Y$ ein $G$"=äquivarianter lokaler
      Homöomorphismus und gibt es um jeden Punkt in $X$ eine
      $G$"=invariante, offene Umgebung, auf der $F$ injektiv ist, so
      ist $F$ ein $G$"=lokaler Homöomorphismus.

    \item
       \label{item:G_lokal_homeo2}
       Seien $X$ und $Y$ Mannigfaltigkeiten, $G$ eine Lie-Gruppe,
      $F\dpA X \to Y$ $G$"=äquivariant, $C \subset Y$ eine
      $G$"=invariante Untermannigfaltigkeit und $s \dpA C \to X$ ein
      $G$"=äquivarianter, stetiger Schnitt von $F$ auf $C$. Sei weiter
      $F$ ein lokaler Homöomorphismus auf einer (nicht
      notwendigerweise $G$"=invarianten) Umgebung von $s(C)$ und
      $\theta$ eigentlich. Dann gibt es eine $G$"=invariante, offene
      Umgebung $U$ von $s(C)$ in $X$, so dass $F\at{U} \dpA U \to F(U)$
      ein $G$"=lokaler Homöomorphismus ist.
   \end{propositionEnum}
\end{proposition}
\begin{proofklein}
   \begin{beweisEnum}
      \item Klar.
      \item Sei $c \in C$. Angenommen es gibt keine offene,
         $G$"=invariante Umgebung von $s(c) \in X$, so dass $F$
         eingeschränkt auf diese injektiv ist. Sei $\{U_n\}_{n \in
           \mathbb{N}}$ eine abzählbare Umgebungsbasis von $s(c)$. Dann
         gibt es zu $n \in \mathbb{N}$ $a_n, b_n \in U_n$ und
         $g^a_n,g^b_n \in G$, so dass $g^a_n a_n \neq g^b_n b_n$, aber
         $F(g^a_n a_n) = F(g^b_n b_n)$. d.\,h.\ mit $g_n := (g^a_n)^{-1}
         g^b_n$ ist $a_n \neq g_n b_n$ und, da $F$ äquivariant ist,
         $F(a_n) = g_nF(b_n)$. Nach Konstruktion gilt $a_n,b_n \to s(c)$
         für $n \to \infty$. Und mit der Stetigkeit von $F$ folgt
         $g_nF(b_n) = F(a_n) \to F(s(c)) = c$ und $F(b_n) \to F(s(c)) =
         c$ für $n \to \infty$. Da $\theta$ eigentlich auf $Y$ wirkt,
         gibt es damit eine konvergente Teilfolge von $(g_n)$, die wir
         uns o.\,E.\ ausgewählt denken, mit $g := \lim_{n \to \infty}
         g_n$. Somit haben wir einerseits $g_n b_n \to g s(c)$ und
         andererseits $c = \lim_{n \to \infty} (g_n F(b_n)) = \lim_{n \to
           \infty} g_n \lim_{n \to \infty} F(b_n) = g c$, was wegen der
         Äquivarianz von $s$ die Gleichheit $g s(c) = s(c)$
         impliziert. Somit gibt es ein $n_0 \in \mathbb{N}$, so dass
         $a_{n_0}, g_{n_0} b_{n_0}$ in einer Umgebung von $s(c)$ liegen,
         auf der $F$ injektiv ist. Dies ist aber ein Widerspruch zur
         Wahl von $(a_n)$,$(b_n)$ und $(g_n)$.

        Sei nun $U_c$ eine $G$"=invariante, offene Umgebung von $s(c)$
        in $X$, auf der $F$ injektiv ist. Dann setze $U := \bigcup_{c
          \in C} U_c$. Dies ist dann die gesuchte, $G$"=invariante
        offene Umgebung von $s(C)$. Denn sei $x\in U$, dann gibt es
        ein $c \in C$ mit $x \in U_c$, d.\,h.\ $U_c$ ist eine offene,
        $G$"=invariante Umgebung von $x$ in $U$, auf der $F$ injektiv
        ist. Mit~\refitem{item:G_lokal_homeo1} folgt schließlich die
        Behauptung.
   \end{beweisEnum}
\end{proofklein}

\begin{bemerkung}
   \label{bem:G_lokal_homeo}
   Schaut man sich obigen Beweis genau an, so sieht man, dass man mit
   etwas schwächeren Voraussetzungen auskommt. Es genügt zu fordern, dass
   $X$ das erste Abzählbarkeitsaxiom erfüllt und $\theta$ folgeneigentlich
   ist. $G$ muss nur eine topologische Gruppe sein.
\end{bemerkung}

\begin{lemma}
   \label{lem:kleben}
   Seien $Z$ und $X$ topologische Räume und $G$ eine topologische
   Gruppe, die sowohl auf $Z$ als auch auf $X$ wirke, so dass für jede
   $G$"=invariante offene Umgebung $U$ von $C$ die Menge $U/G$
   parakompakt und hausdorffsch ist.  Weiter sei $\tau \dpA Z \to X$ ein
   surjektiver $G$"=lokaler Homöomorphismus, $C \subset X$ $G$"=invariant
   und $s \dpA C \to Z$ ein stetiger, $G$"=äquivarianter Schnitt von
   $\tau$, d.\,h.\ $\tau \circ s = \id$. Dann gibt eine offene,
   $G$"=invariante Umgebung $V$ von $C$ in $X$ und einen stetigen,
   $G$"=äquivarianten Schnitt $S \dpA V \to Z$ von $\tau$ (d.\,h.\  $\tau
   \circ S = \id$), so dass $S(V) \subset Z$ eine offene, $G$"=invariante
   Umgebung von $s(C)$ ist

   und $S\at{C} = s$, d.\,h.\ $S$ eine Fortsetzung von $s$ ist.
\end{lemma}
\begin{proofklein}
   Zunächst bemerken wir, dass $\tau\at{s(C)}\colon s(C) \to C$ ein
   Homöomorphismus mit stetigem Inversen $s \colon C \to S(C)$
   ist. $\tau\at{s(C)}$ ist nämlich ein Linksinverses zu $s$, d.\,h.\ $s$
   ist injektiv und als Abbildung auf sein Bild surjektiv. Da $\tau$ ein
   $G$"=lokaler Homöomorphismus ist, gibt es für jedes $c \in C$ eine
   offene, $G$"=invariante Umgebung $Z_c$ von $s(c)$ in $Z$, und eine
   offene, $G$"=invariante Umgebung $W_c$ von $c$ in $C$, so dass
   $\tau\at{Z_c} \colon Z_c \to W_c$ ein $G$"=äquivarianter
   Homöomorphismus ist. Wir können o.\,E.\ annehmen, dass für jedes $c' \in
   W_c \cap C$
   \begin{align}
      \label{eq:inverseGleichSchnitt}
      s(c')  =\tau\at{Z_c}^{-1}(c')  \tag{$*$}
   \end{align}
   gilt. Dass man dies in der Tat immer erreichen kann, sieht man wie
   folgt. Da $s \colon C \to Z$ stetig und $G$"=äquivariant ist, gibt es
   eine offene, $G$"=invariante Umgebung $\tilde{W}_c \subset W_c$ von
   $c$ in $X$, so dass $s(\tilde{W}_c \cap C) \subset Z_c$. Sei nun $c'
   \in \tilde{W}_c \cap C$, dann gilt
   \begin{align*}
      \tau\at{Z_c}(s(c')) = \tau(s(c')) = c' =
      \tau\at{Z_c}(\tau\at{Z_c}^{-1}(c'))
   \end{align*}
   und damit auch  $s(c') = \tau\at{Z_c}^{-1}(c')$. Weiter ist wegen
   $\tau \circ s = \id$ die Menge $\tau\at{Z_c}^{-1}(\tilde{W}_c)$ eine
   Umgebung von $s(c)$.

   Ersetzt man also $Z_c$ durch $\tau\at{Z_c}^{-1}(\tilde{W}_c)$ und
   $W_c$ durch $\tilde{W}_c$, so sieht man, dass wir o.\,E.\ Eigenschaft
   \eqref{eq:inverseGleichSchnitt} annehmen können, dies wollen wir im
   Folgenden immer tun. Mit Hilfe von Gleichung
   \eqref{eq:inverseGleichSchnitt} und mit $\tau \circ s = \id$ prüft
   man leicht nach, dass
   \begin{align*}
      Z_c \cap s(C) = s(W_c \cap C)
   \end{align*}
   erfüllt ist.  Außerdem können wir o.\,E.\ annehmen, dass die
   $G$"=invariante, offene Überdeckung $\{W_c\}_{c \in C}$ lokal endlich
   ist. Denn ist dies nicht der Fall, so können wir nach Proposition
   \ref{prop:inv.lok.endl}~\refitem{item:X/G_parakompakt_1} eine
   präzise, $G$"=invariante, offene, lokal endliche Verfeinerung
   $\{\tilde{W}_c\}_{c \in C}$ wählen. Die Menge $\tilde{Z}_c :=
   \tau\at{Z_c}^{-1}(\tilde{W}_c)$ ist dann für jedes $c \in C$
   natürlich $G$"=invariant und offen, enthält aber i.\,Allg.\ $s(c)$ nicht
   mehr. Trotzdem überdeckt die Familie $\{\tilde{Z}_c\}_{c \in C}$
   $s(C)$, denn sei $c' \in C$, dann gibt es ein $c \in C$, so dass $c'
   \in \tilde{W}_c \subset W_c$ mit $c' \in \tilde{W}_c \cap C$ und damit nach
   Gleichung \eqref{eq:inverseGleichSchnitt}
   \begin{align*}
      s(c') = \tau\at{Z_c}^{-1}(c') \in \tilde{Z}_c \cap s(C) \Fdot
   \end{align*}
   Wir wollen im Folgenden annehmen, dass $\{W_c\}_{c \in C}$ eine lokal
   endliche Überdeckung von $C$ sei.

   Betrachte nun $I := \{c \in C \mid  Z_c \cap s(C) \neq
   \emptyset\}$. Dann ist auch $\{Z_i\}_{i \in I}$ eine offene,
   $G$"=invariante Überdeckung von $s(C)$. Wegen $\tau \circ s = \id$
   ist des Weiteren $\{W_i\}_{i \in I}$ ebenfalls eine offene,
   $G$"=invariante, lokal endliche Überdeckung von $C$ mit $W_i \cap C
   \neq \emptyset$ für alle $i \in I$. Im Folgenden schreiben wir $s_i
   := \tau\at{Z_i}$ für alle $i \in I$. Da $(\bigcup_{i \in I} W_i)/G$
   parakompakt und hausdorffsch ist, können wir nach Proposition
   \ref{prop:inv.lok.endl}~\refitem{item:X/G_parakompakt_2} eine offene,
   $G$"=invariante, lokal endliche Verfeinerung $\{V_i\}_{i \in I}$ von
   $\{W_i\}_{i \in I}$ mit $\abschluss{V}_i \subset W_i$ für alle $i \in
   I$ wählen.

   Sei nun $c \in C$ beliebig. Nach Proposition
   \ref{prop:inv.lok.endl}~\refitem{item:guteInvZeugendeUmgebung}
   existiert eine offene, $G$"=invariante, gute zeugende Umgebung
   $\tilde{\Gamma}_c$ von $c$ für $(\{W_i\}_{i \in I},\{V_i\}_{i \in
     I})$. Nach Konstruktion ist $E_c := \{i \in I \mid V_i
   \cap \tilde{\Gamma}_c \neq \emptyset \}$ endlich. Weiter sei
   $D_c \subset Z$ eine offene, $G$"=invariante Umgebung von $s(c)$, so dass
   $\tau\at{D_c}$ injektiv ist. Für $i \in E_c$ ist $s_i^{-1}(D_c) = \{x
   \in W_i \mid s_i(x) \in D_c\} \subset X$ nicht leer, denn da
   $\tilde{\Gamma}_c$ eine gute zeugende Umgebung ist, gilt $c \in
   \tilde{\Gamma}_c \subset W_i$ für alle $i \in E_c$ und damit sowie nach
   Gleichung \eqref{eq:inverseGleichSchnitt} $s_i(c) = s(c) \in D_c$,
   also $c \in s_i^{-1}(D_c)$ für alle $i \in E_c$. Außerdem ist klar,
   dass $s_i^{-1}(D_c)$ für alle $i \in E_c$ $G$"=invariant und offen
   ist. Dann ist offenbar
   \begin{align*}
      \Gamma_c := \tilde{\Gamma}_c \cap \bigcap_{i \in E_c} s_i^{-1}(D_c)
   \end{align*}
   eine $G$"=invariante, offene Umgebung von $c$. Die Menge
   \begin{align*}
      V := \bigcup_{i \in I}V_i \cap \bigcup_{i \in C} \Gamma_c
   \end{align*}
   ist dann eine offene, $G$"=invariante Umgebung von $C$.
   Schließlich definieren wir $S \colon V \to Z$ durch
   \begin{align*}
      S(x) := s_i(x) \quad \text{für $x \in V_i$} \Fdot
   \end{align*}
   Dies ist wohldefiniert, denn sei $x \in V_i \cap V_j \cap V$ mit $i,j
   \in I$, dann gibt es insbesondere ein $c \in C$ mit $x \in \Gamma_c
   \subset \tilde{\Gamma}_c$ und nach Definition von $E_c$ folgt $i,j
   \in E_c$. Da $x \in \Gamma_c$ ist, gilt weiterhin $x \in s^{-1}_k(D_c)$ für
   alle $k \in E_c$ und damit $s_i(x), s_j(x) \in D_c$. Aufgrund der
   Injektivität von $\tau\at{D_c}$ folgt aus $\tau(s_i(x)) = x =
   \tau(s_j(x))$ direkt $s_i(x) = s_j(x)$. Jetzt ist leicht zu sehen,
   dass $S$ ein stetiger, $G$"=invarianter Schnitt von $\tau$ definiert,
   der $s$ fortsetzt und dass $S(V)$ eine offene, $G$"=invariante
   Umgebung von $s(C)$ ist.

\end{proofklein}

\begin{korollar}
   \label{kor:AufUmgebungVonCHomoe}
   Seien $Z$ und $X$ Mannigfaltigkeiten, $G$ eine Lie"=Gruppe, die auf
   $Z$ und $X$ wirke, wobei die Wirkung auf $X$ eigentlich sei. Weiter
   sei $F\dpA Z \to X$ $G$"=äquivariant, $C \subset Z$ eine
   $G$"=invariante Untermannigfaltigkeit und $s \dpA C \to X$ ein
   $G$"=äquivarianter, stetiger Schnitt von $F$ auf $C$. Sei ferner $F$
   ein lokaler Homöomorphismus auf einer (nicht notwendigerweise
   $G$"=invarianten) Umgebung $U'$ von $s(C)$. Dann gibt es eine
   $G$"=invariante, offene Umgebung $U$ von $s(C)$ in $Z$, so dass
   $F\at{U} \dpA U \to F(U)$ ein Homöomorphismus ist.
\end{korollar}
\begin{proofklein}
   Nach Proposition
   \ref{prop:G_lokal_homeo}~\refitem{item:G_lokal_homeo2} können wir
   o.\,E.\ annehmen, dass $U'$ $G$"=invariant und $F\at{U'} \colon U' \to
   F(U')$ ein $G$"=lokaler Homöomorphismus ist. Mit Lemma
   \ref{lem:kleben} folgt, dass es eine $G$"=invariante, offene Umgebung
   $V$ von $C$ in $F(U') \subset X$ gibt und einen stetigen,
   $G$"=äquivarianten Schnitt $S\colon V \to U' \subset Z$ von
   $F\at{U'}$, so dass $S(V)$ eine offene, $G$"=invariante Umgebung von
   $s(C)$ in $U'$ ist und $S\at{C} = s$ gilt. $F\at{S(V)} \colon S(V)
   \to V$ ist ein Linksinverses von $S\colon V \to S(V)$, womit dieses
   injektiv und damit auch bijektiv ist. Also ist auch $F\at{S(V)}$
   bijektiv mit stetigem Inversem $S\colon V \to S(V)$. Mit $U := S(V)$
   folgt die Behauptung.
\end{proofklein}
\section{Existenz von $G$-invarianten kompatiblen Tubenumgebungen}
\label{sec:ExistenzVonGInvariantenKompatiblenTubenumgebungen}

Bei einer Tubenumgebung handelt es sich grob gesagt um eine Umgebung
einer Untermannigfaltigkeit, bei der die Richtungen transversal zu
dieser wie ein Vektorraum aussehen. Die folgende Definition präzisiert
diese Vorstellung.

\begin{definition}[$G$-invariante Tubenumgebung]
   \label{def:Inv.Tubenumgebung}
   Sei $\phi$ eine Wirkung einer Lie"=Gruppe $G$ auf einer Mannigfaltigkeit
   $M$ und $C$ eine $G$"=invariante Untermannigfaltigkeit von $M$. Eine
   \neuerBegriff{$G$"=invariante Tubenumgebung} $(E \to C,U_E,U_M,\tau)$ von
   $C$ in $M$ besteht aus einem Vektorbündel $E \to C$, einer Hochhebung
   $\Phi$ von $\phi$ auf $E$ (d.\,h.\ einem $G$"=Vektorbündel) und aus
   einem $G$"=äquivarianten Diffeomorphismus $\tau \dpA U_E \to U_M$
   von einer offenen, $G$"=invarianten Umgebung $U_E$ von $\nSchnitt(C)$
   in $E$ auf eine offene, $G$"=invariante Umgebung $U_M$ von $C$ in $M$ mit
   $\tau \circ n = \id$. Dabei bezeichnet $\nSchnitt \dpA C \to E$ den
   Nullschnitt. $\tau$ heißt auch \neuerBegriff{Tubenabbildung}. Ist
   $U_E = E$, so wollen wir von einer \neuerBegriff{totalen
   Tubenumgebung} sprechen und schreiben $(E \to C,U_M,\tau)$, in diesem Fall
   heißt $\tau$ \neuerBegriff{totale Tubenabbildung}.
\end{definition}

\begin{bemerkung}
   \label{bem:Tubenumgebungen_partiell_total}
      In unserer Terminologie folgen wir im Wesentlichen
      \cite{lang:1999a}. In der Literatur ist es auch üblich, das, was wir
      ($G$"=invariante) Tubenumgebung nennen, als \neuerBegriff{partielle
        Tubenumgebung} zu bezeichnen und das, was wir totale Tubenumgebung
      nennen, einfach als Tubenumgebung. Siehe etwa \cite[Ch.~4.5]{hirsch:1976a}
      oder \cite[Ch.~4.5]{kankaanrinta:2007}.
\end{bemerkung}

\begin{bemerkung}
   \label{bem:TubenumgebungLokal}
   Ist $U$ eine $G$"=invariante, offene Umgebung von $C$ in $M$, so induziert
   offensichtlich jede $G$"=invariante Tubenumgebung von $C$ in $U$ auf
   naheliegende Weise eine solche für $C$ in $M$.
\end{bemerkung}

Wir wollen uns nun der Frage nach der Existenz von Tubenumgebungen zuwenden.

\begin{proposition}
   \label{prop:3}
   \begin{propositionEnum}
   \item  Sei $\pi \dpA E \to M$ ein Vektorbündel, $\nSchnitt{} \dpA M \to E$
      der Nullschnitt, dann ist
      \begin{align*}
         \vbNIso{E \to M} \dpA TM \oplus E \to TE\at{\nSchnitt(M)},
         \quad (v,w) \mapsto T \nSchnitt \, v + \ddt (tw)
      \end{align*}
       ein Vektorbündelisomorphismus. \label{item:TangentialraumEinesVBsAmNullschnitt}
     \item Seien $\pi_i \dpA E_i \to M_i$, $i=1,2$, Vektorbündel mit
       Nullschnitten $\nSchnitt[i] \dpA M_i \to E_i$ und $\Phi \dpA E_1
       \to E_2$ ein Vektorbündelmorphismus über $\phi \dpA M_1 \to
       M_2$. Dann kommutiert

      \def\tA[#1]{A_{#1}}
\begin{equation}
      \begin{tikzpicture}[baseline=(current
    bounding box.center),description/.style={fill=white,inner sep=2pt}]
         \matrix (m) [matrix of math nodes, row sep=3.0em, column
         sep=5.5em, text height=1.5ex, text depth=0.25ex]
         {
TE_1\at{\nSchnitt[1](M_1)}  & TE_2\at{\nSchnitt[2](M_2)} \\
TM_1 \oplus E_1  &  TM_2 \oplus E_2 \\
}; %

\path[<-] (m-1-1) edge node[auto] {$\vbNIso{1}$} (m-2-1); %
\path[<-] (m-1-2) edge node[auto] {$\vbNIso{2}$} (m-2-2); %
\path[->] (m-1-1) edge node[auto]{$T\Phi$}(m-1-2); %
\path[->] (m-2-1) edge node[auto]{${(T\phi,\Phi)}$}(m-2-2); %
 \end{tikzpicture}
\end{equation}
      Dabei ist $\vbNIso{i} := \vbNIso{E_i \to M_i}$ $i=1,2$.
   \end{propositionEnum}
\end{proposition}

\begin{proofklein}
   \begin{beweisEnum}
      \item Ist $Tn \, v = \ddt (tw)$ für $(v,w) \in TM \oplus E$, dann gilt
         \begin{align*}
            v = T\id \, v &= T (\pi \circ n)v = T\pi T n \, v = T\pi\, \ddt
            (tw)\\
            &= \ddt(\pi(tw)) = \ddt (\pi(w)) = 0 \Fdot
         \end{align*}
         Also gilt für $p \in M$
         \begin{align*}
            \vbNIso{E \to M}(T_pM \times \{0\}) \cap \vbNIso{E
              \to M}(\{0\} \times E_p) = \{0\} \Fdot
         \end{align*}
         Da sowohl $\vbNIso{E \to M}\at{T_pM}$ als auch $
         \vbNIso{E \to M}\at{E_p}$ bekanntlich injektiv sind,
         folgt die faserweise Bijektivität aus Dimensionsgründen.  Der
         Rest ist klar.

      \item Seien $p \in M_1$, $(v,w) \in T_pM_1 \oplus {E_1}_p$
        beliebig, dann gilt
         \begin{align*}
            \vbNIso{2}^{-1} \circ T \Phi \circ \vbNIso{1}(v,w) &=
            \vbNIso{2}^{-1}(T\Phi(T n_1 v + \ddt (t w))) \\
            &= \vbNIso{2}^{-1}(T\underbrace{(\Phi \circ n_1)}_{=n_2 \circ
              \phi}
            v + \ddt (\Phi(tw))) \\
            &= \vbNIso{2}^{-1}(Tn_2(T\Phi v) + \ddt t \Phi(w)) \\
            &= (T\Phi v, \Phi (w)) \Fdot
         \end{align*}
   \end{beweisEnum}
\end{proofklein}

\begin{proposition}
   \label{prop:ExponentialDiffeo}
   Sei $M$ eine Mannigfaltigkeit und $\xi \colon TM \to TTM$ ein Spray
   auf $M$ sowie $\exp$ die davon induzierte Exponentialabbildung mit
   maximalem Definitionsbereich $\mathcal{O}$.
   \begin{propositionEnum}
   \item %
      \label{item:ExponentialabbildungIdentitaet}
      Für alle $p \in M$ und $X_1,X_2 \in T_pM$ gilt
      \begin{align}
         \label{eq:ExponentialabbildungIdentitaet}
         T_{0_p}\exp (\vbNIso{TM \to M}(X_1,X_2)) = X_1 + X_2 \Fdot
      \end{align}
   \end{propositionEnum}
   Sei weiter $\kIn \colon C \hookrightarrow M$ eine
   Untermannigfaltigkeit und $K \to C$ ein Komplementvektorbündel von
   $T\kIn TC$ in $TM\at{C}$.
   \begin{propositionEnum}
      \addtocounter{enumi}{1}
   \item %
      Für alle $c \in C$ ist die Abbildung
      \begin{align}
         \label{eq:ExpIsoKomplement}
         T_{0_c}\exp\at{K} \colon T_{0_c}K \to T_cM
      \end{align}
      ein Vektorraumisomorphismus.
   \item%
      \label{item:expIstLokalDiffeoAufNormalBundel} %
      Die Menge
      \begin{align*}
         U := \{v \in K \cap \mathcal{O} \mid T_v \exp
         \text{ ist invertierbar}\}
      \end{align*}
      ist eine offene Umgebung von $\nSchnitt(C)$ in $K$. Die
      Abbildung $\exp$ ist eingeschränkt auf $U$ ein lokaler
      Diffeomorphismus. Dabei bezeichnet $n \colon C \to K$ den
      Nullschnitt.
   \end{propositionEnum}
\end{proposition}
\begin{proofklein}
   \begin{beweisEnum}
   \item %
      Dies ist klar, da jeweils nach Kettenregel
         \begin{align*}
            T_{o_p} \exp (T_p \nSchnitt \, v) = T_p(\exp \circ \nSchnitt)
            \, v = T_p (\id) \, v = v
         \end{align*}
         und
         \begin{align*}
            T_{0_p} \exp \left(\ddt (tw)\right) = \ddt (\exp(tw)) = w
         \end{align*}
         für $v,w \in T_pM$, $p \in M$.
      \item %
         Sei $\kIn_K \colon K \hookrightarrow TM$ die Inklusion, dann
         schreibt sich $\exp\at{K} = \exp \circ \kIn_K$. Sei nun $c \in
         C$ beliebig. Nach Kettenregel gilt:
         \label{tau_invertierbar}
         \begin{align*}
            T_{0_c} \exp\at{K} = T_{0_c} \exp \circ T_{0_c} \kIn_{K}
         \end{align*}
         Offenbar ist $\kIn_K$ ein Vektorbündelmorphismus
         über $\kIn \colon C \hookrightarrow M$ womit aus
         Proposition \ref{prop:3} und  folgt :
         \begin{eqnarray*}
            \lefteqn{ T_{0_c} \exp\at{K}
              (\vbNIso{K \to C}
              (v,w))}\\
            &=& (T_{0_a} \exp) \circ \vbNIso{TM
              \to M} \circ \vbNIso{TM \to M}^{-1} \circ
            (T_{0_c} \kIn_K)  \circ
            \vbNIso{K \to C}(v,w) \\
            &=& T_{0_c} \exp \circ \vbNIso{TM
              \to M} (T\kIn \, v, w) \\
            &=& T_c\kIn v + w
         \end{eqnarray*}
         für $(v,w) \in T_cA \oplus K_c$.
         Beachtet man, dass die Abbildung $T_cA \oplus K_c \ni (v,w)
         \mapsto T_c\kIn v + w \in T_cM$ offensichtlich bijektiv ist,
         folgt die Behauptung.
      \item %
          Nach~\refitem{tau_invertierbar} gilt $\nSchnitt(C) \subset U$ und
         mit dem Satz über die Umkehrfunktion (siehe
         \cite[Thm.~5.11]{lee:2003a}) gibt es für jedes $v \in U$ eine
         offene Umgebung $U_v$ von $v$ in $K$ so dass
         $\tau(U_v)$ offen ist und $\exp\at{U_v} \dpA U_v \to
         \tau(U_v)$ ein Diffeomorphismus ist. Damit ist der Rest klar.
   \end{beweisEnum}
\end{proofklein}

Als nächstes zeigen wir, wie ein Spray und ein Komplementvektorbündel einer
Untermannigfaltigkeit eine Tubenumgebung dieser induziert.

\begin{lemma}
   \label{lemma:ExistenzVonTubenumgebungen}
   Sei $M$ eine Mannigfaltigkeit, $\phi \dpA G \times M \to M$ eine
   eigentliche Wirkung einer Lie"=Gruppe $G$ auf $M$ und $\kIn \colon C
   \hookrightarrow M$ eine $G$"=invariante Untermannigfaltigkeit. Sei
   $\xi \colon TM \to TTM$ ein $G$"=äquivarianter Spray und $K$ ein
   $G$"=invariantes Komplementbündel von $T\kIn TC$ in $TM$. Dann gibt
   es eine $G$"=invariante Tubenumgebung $(E \to C,U_E,U_M,\tau)$ von $C$ in
   $M$ mit folgenden Eigenschaften.
   \begin{satzEnum}
      \item $E = K$.
      \item $G$ wirkt auf $E$ via
         \begin{align*}
            \Phi \dpA G \times E \to E, \quad (g,v) \mapsto
            T\phi_g \, v \Fdot
         \end{align*}
       \item $U_E$ liegt im Definitionsbereich $\mathcal{O}$ der von
         $\xi$ induzierten Exponentialabbildung $\exp$.
      \item $\tau := \exp\at{U_E}$.
      \item $U_M = \exp(U_E)$.
   \end{satzEnum}
\end{lemma}
\begin{proofklein}

   Die Menge $U := \{v \in K \cap \mathcal{O} \mid T_v
   \exp \text{ ist invertierbar}\}$ ist nach Proposition
   \ref{prop:ExponentialDiffeo}~\refitem{item:expIstLokalDiffeoAufNormalBundel}
   offen in $K$.

   Die Abbildung $\exp\colon \mathcal{O} \to M$ ist $G$"=äquivariant und $\exp\at{U} \dpA U \to \exp(U)$

   ist ein lokaler Homöomorphismus. Weiter ist der Nullschnitt
   $\nSchnitt$ ist $G$"=äquivariant, denn für $g \in G$ und $c \in C$ ist
   \begin{align*}
      \nSchnitt(g c) = 0_{g c} = T \phi_g (0_c) = \Phi_g (\nSchnitt(c)) \Fdot
   \end{align*}
   Damit folgt die Behauptung unmittelbar aus Korollar
   \ref{kor:AufUmgebungVonCHomoe}.

\end{proofklein}

\begin{lemma}
   \label{lem:SubmersionAufC1}
   Seien $M$ und $N$ Mannigfaltigkeiten und $\phi\colon G \times
   M \to M$ sowie $\theta\colon G \times N \to N$ Wirkungen einer
   Lie"=Gruppe $G$ auf $M$ bzw.\ $N$.  Weiter sei $C \subset M$ eine
   $G$"=invariante Untermannigfaltigkeit und $f\dpA M \to N$ eine
   $G$"=äquivariante, glatte Abbildung, so dass $f\at{C}\dpA C \to N$
   submersiv ist. Dann gibt es eine $G$"=invariante, offene Umgebung $U$
   von $C$ in $M$, so dass $\ker Tf\at{U}$ ein Untervektorbündel von $TU$
   ist. Insbesondere hat $Tf\at{U}$ konstanten Rang.
\end{lemma}
\begin{proofklein}
   Da $f\at{C}$ submersiv ist, gilt schon $T_cf T_c\kIn C = T_{f(c)}N$
   für alle $c \in C$, d.\,h.\ der Rang $\mathrm{rang}(T_cf)$ von $T_cf$
   ist voll. Bekanntlich, vgl.\ \cite[Rem. 5.3]{broecker.jaenich:1990a},

   gibt es für jedes $c \in C$ eine offene Umgebung $V_c$ in $M$, so dass
   $\mathrm{rang}(T_pf) \geq \mathrm{rang}(T_cf)$ für alle $p \in
   V_c$. Da der Rang aber voll ist, gilt schon $\mathrm{rang}(T_pf) =
   \mathrm{rang}(T_cf) = \dim N$ für alle $p \in V_c$. Wir setzen dann
   $V := \bigcup_{c \in C}V_c$, es gilt also $\mathrm{rang}(T_pf) = \dim
   N$ für alle $p \in V$. Nun definieren wir weiter $U := GV$. Sei nun
   $p \in V$ und $g \in G$, dann gilt für jedes $v_p \in \ker T_pf$
   schon $T_p \phi_g v_p \in \ker T_{\phi_g(p)}f$, denn mit der
   Kettenregel und der $G$"=Äquivarianz von $f$ erhält man
   \begin{align*}
      T_{\phi_g(p)}f T_p \phi_g v_p = T_p(f \circ \phi_g)v_p =
      T_p(\theta_g \circ f)v_p = T_{f(p)}\theta_gT_pf v_p = 0 \Fdot
   \end{align*}
   So folgt wegen der Injektivität von $T_p \phi_g$ schon $\dim \ker
   T_pf \leq \dim \ker T_{\phi_g(p)}f$. Analog sieht man $ \dim \ker
   T_{\phi_g(p)}f \leq \dim \ker T_pf$, womit  schon $\dim \ker
   T_{\phi_g(p)}f = \dim \ker T_pf$ gilt, d.\,h.\ auch
   $\mathrm{rang}(T_{\phi_g(p)}f) = \dim N$. Damit ist aber nach dem
   Fibrationssatz, siehe \cite[Thm. 3.5.18]{marsden:2002} $\ker
   Tf\at{U}$ ein Untervektorbündel von $TU$.
\end{proofklein}

\begin{lemma}
   \label{lem:SprayErhaeltUnterbuendel}
   Sei $M$ eine Mannigfaltigkeit und $E \to M$ ein integrables
   Untervektorbündel von $\pi \colon TM \to M$. Dann gibt es einen Spray
   $\xi\colon TM \to TTM$ mit $\xi(E) \subset TE$.
\end{lemma}
\begin{proofklein}
   Da $E$ ein integrables Untervektorbündel von $TM$ ist, gibt es nach
   dem Frobenius"=Theorem, vgl.\ \cite[Ch. 19, p. 500]{lee:2003a}, um
   jeden Punkt in $M$ eine Karte $(U,x)$ mit induzierter Karte des
   Tangentialbündels $(\pi^{-1}(U),(q,v))$ so dass für jedes $p \in U$
   \begin{align*}
      \sum_{i=1}^m v_p^i \frac{\partial}{\partial x^i}\At{p} \in
       E_p \iff v_p^{n+1},\dots, v_p^m = 0 \Fdot
   \end{align*}
   Dabei ist $m := \dim M$ und $n := \dim E_p$ für alle $p \in M$.  Ist
   nun $\tilde \gamma\dpA \mathbb{R} \to U$ eine glatte Kurve und $c_i
   \in \mathbb{R}$, so ist
   \begin{align*}
      \gamma \dpA \mathbb{R} \to E\at{U}, \quad t \mapsto \sum_{i=1}^{n}
      c_i \frac{\partial}{\partial x^i}\At{\tilde{\gamma}(t)}
   \end{align*}
   eine wohldefinierte glatte Kurve in $E\at{U}$. Für deren Ableitung
   folgt offenbar
   \begin{align*}
      TE\ni {\dot \gamma}(0) = \sum_{i=1}^{m} \dot{\tilde
        \gamma}^i(0) \frac{\partial}{\partial q^i}\At{\gamma(0)}
   \end{align*}
   mit ${\tilde \gamma}^i := (x^i \circ \tilde \gamma) = (q^i \circ
   \gamma) =: \gamma^i $.  Wir definieren dann auf $\pi^{-1}(U)$ das
   Spray"=Vektorfeld
   \begin{align*}
      \xi_{U} = v^i \frac{\partial}{\partial q^i} \Fcom
   \end{align*}
   vergleiche Proposition \ref{prop:lokaleForm}.

   Für $p \in U$ und $v_p = \sum_{i=1}^nv_p^i \frac{\partial}{\partial
     x^i}\At{p} \in E_p$ gilt damit
   \begin{align*}
      \xi_{U}\left(\sum_{i=1}^nv_p^i \frac{\partial}{\partial
          x^i}\At{p}\right) = \sum_{i=1}^n v_p^i
      \frac{\partial}{\partial q^i}\At{v_p} \Fdot
   \end{align*}
   Um die Inklusion $\xi_{U}(E\at{U}) \subset T (E\at{U})$ zu zeigen,
   genügt es für alle $p \in U$ und $v_p \in E\at{U} $ eine glatte Kurve
   $\gamma$ in $E\at{U}$ mit $\xi(v_p) = \dot \gamma(0)$ und $v_p =
   \gamma(0)$ zu finden. Wählen wir nun $\tilde \gamma$ von oben derart,
   dass $\tilde \gamma(0) = p$ und $\dot {\tilde{ \gamma}}^i(0) = v^i_p$
   für $i=1,\dots,n$ und $\dot {\tilde{ \gamma}}^i(0) = v^i_p = 0$ für
   alle $i = n+1,\dots,m$ gilt, sowie $\gamma$ mit $c^i = v^i_p$ für
   alle $i=1,\dots,n$, folgt die Behauptung.

   Wir können nun $M$ durch
   Karten $(U_i,x_i)$ wie oben überdecken und eine dieser
   Überdeckung untergeordnete Zerlegung der Eins $\{\chi_i\}$
   wählen. Offensichtlich ist dann $\xi := \sum_i \pi^*\chi_i \xi_{U}$
   ein Spray mit $\xi(E) \subset TE$.
\end{proofklein}

\begin{lemma}
   \label{lemma:kompatibleTuben}

   Seien $M$ und $N$ Mannigfaltigkeiten, $ \kIn \colon C \hookrightarrow
   M$ eine Untermannigfaltigkeit und $f\dpA M \to N$ eine glatte
   Abbildung mit konstantem Rang, so dass $f\at{C}\dpA C \to N$ submersiv
   ist. Dann sind folgende Aussagen richtig.
   \begin{lemmaEnum}
   \item %
      \label{item: SprayExists} %
      Es gibt einen Spray $\xi\dpA TM \to TTM$ mit $\xi(\ker Tf) \subset
      T(\ker Tf)$.
   \item %
      \label{item: Komplement} %
      Für alle $c \in C$ gilt für jedes Komplement $K_c \subset \ker
      T_cf$ von $T_c \kIn T_c C \cap \ker T_cf$ in $\ker T_cf$ schon
      $T_cM = K_c \oplus T_c \kIn T_cC$.
   \item %
      \label{item: SprayInduzKompTube} %
      Für jeden Spray $\xi\dpA TM \to TTM$ mit $\xi(\ker Tf) \subset
      T(\ker Tf)$ und jedes Komplementbündel $K$ von $T\kIn TC$ in
      $TM\at{C} \to C$ mit $K \subset \ker Tf$ (z.\,B.\ via einer
      Riemannschen Metrik und nach~\refitem{item: Komplement} induziert)
      gilt für die davon nach Lemma
      \ref{lemma:ExistenzVonTubenumgebungen} induzierte Tubenumgebung
      $(\pi\dpA E \to C,U_E,U_M,\tau)$ von $C$ in $M$, dass $f(r(p)) =
      f(p)$ für alle $p \in U_M$, wobei $r := \pi \circ \tau^{-1} \dpA
      U_M \to C$ die induzierte Retraktion ist.
   \end{lemmaEnum}
\end{lemma}
\begin{proofklein}
    Wir folgen bei der Beweisführung grob der Darstellung von
      Wirtmüller, vgl.\ \cite[Ch. II, Thm. 1.6]{gibson:1976}.
   \begin{beweisEnum}
   \item %
      Nach Lemma \ref{lem:SubmersionAufC1} ist $\ker Tf$ ein
      Untervektorbündel von $TM$. Dieses ist involutiv, denn seien $X$
      und $Y$ zwei Vektorfelder die Punktweise in $\ker Tf$ liegen, so
      sind bedeutet dies, dass sie zum Nullvektorfeld $f$"=verwandt
      sind. Also ist auch ihre Lie"=Klammer zum Nullvektorfeld
      $f$"=verwandt, womit nach dem Frobenius"=Theorem folgt, dass $\ker
      Tf$ integrabel ist. Mit Lemma \ref{lem:SprayErhaeltUnterbuendel}
      ist die Behauptung dann klar.
\item %
   Dies folgt sofort durch einfache Dimensionsbetrachtungen. Da $f \at{C} \dpA C \to N$ submersiv ist, gilt
   \begin{align*}
      \dim \ker T_cf\at{C} = \dim C - \dim N \Fdot
   \end{align*}
   Natürlich ist mit $f\at{C}$ und wegen des konstanten Rangs auch $f$
   selbst submersiv, also haben wir analog
   \begin{align*}
      \dim \ker T_cf = \dim M - \dim N \Fdot
   \end{align*}
   Weiter gilt offensichtlich
   \begin{align*}
      T_c\kIn \ker T_cf\at{C} = \ker T_cf \cap T_c\kIn T_cC \Fdot
   \end{align*}
   Fügen wir diese Informationen zusammen, folgt unmittelbar
   \begin{align*}
      \dim K_c &= \dim \ker T_c f - \dim (T_c \kIn T_c C \cap \ker T_cf)
      \\
      &=
      \dim M - \dim N - \dim (\ker T_cf\at{C})\\
      &= \dim M - \dim N - (\dim C
      - \dim N)\\
      &= \dim M - \dim C \Fdot
   \end{align*}
   Da offensichtlich $K_c \cap T_c\kIn T_cC = \{0\}$ folgt die Behauptung.
\item %

   Sei also
   $\xi\dpA TM \to TTM$ ein Spray mit $\xi(\ker Tf) \subset T(\ker
   Tf)$ und $K$ wie in der Voraussetzung.

   Aus $\xi(\ker Tf) \subset T(\ker Tf)$ folgt, dass jede Integralkurve
   von $\xi$ mit Startpunkt in $\ker Tf$ in $T(\ker Tf)$ ganz in $\ker
   Tf$ verläuft. Insbesondere gilt dann
   \begin{align*}
      T f \frac{\D}{\D t} \exp(t\nu) = 0 \quad \forall c \in C, \nu \in
      K_c \cap \mathcal{O}, t \in [0,1]\Fdot
   \end{align*}
   Dies ist aber gleichbedeutend mit
   \begin{align*}
      f(\exp(t \nu)) = f(c)\quad \forall c \in C, \nu \in
      K_c \cap \mathcal{O}, t \in [0,1]\Fcom
   \end{align*}
   was für jede von $\xi$ und $K$ nach Lemma
   \ref{lemma:ExistenzVonTubenumgebungen} induzierte Tubenumgebung mit
   Retraktion $r$ die Gleichung
   \begin{align*}
      f = f \circ r
   \end{align*}
   impliziert.
\end{beweisEnum}
\end{proofklein}

\begin{satz}[Existenz $G$"=invarianter, kompatibler Tubenumgebungen]
   \label{satz:GkompatibleTuben}
   Seien $M$ und $N$  Mannigfaltigkeiten und $G$ eine Lie"=Gruppe,
   die auf $M$ und $N$ wirke. Die Wirkung auf $M$ sei eigentlich. Weiter
   sei $C \subset M$ eine $G$"=invariante Untermannigfaltigkeit und $f\dpA M \to N$
   eine $G$"=äquivariante, glatte Abbildung, so dass $f\at{C}\dpA C \to N$
   submersiv ist. Dann gibt es eine $G$"=invariante Tubenumgebung
   $(\pi\dpA E \to C,U_E,U_M,\tau)$ von $C$ mit $f(r(p)) = f(p)$ für
   alle $p \in U_M$. Dabei bezeichnet $r := \pi \circ \tau^{-1} \dpA U_M
   \to C$ die induzierte $G$"=äquivariante Retraktion.
\end{satz}
\begin{proofklein}
   Nach Lemma \ref{lem:SubmersionAufC1} und Bemerkung
   \ref{bem:TubenumgebungLokal} können wir o.\,E.\ annehmen, dass $f$
   konstanten Rang hat.  Nach Satz \ref{satz:invarianterSpray} können
   wir den Spray aus Lemma \ref{lemma:kompatibleTuben} $G$"=invariant
   wählen, ebenso können wir $K$ von einer $G$"=invarianten Riemannschen
   Metrik auf $M$ induzieren lassen, d.\,h.\ auch $K$ $G$"=invariant
   wählen. Damit ist die Aussage klar.
\end{proofklein}

\end{appendices}

\cleardoublepage

\addcontentsline{toc}{chapter}{Literaturverzeichnis}

\newcommand{\etalchar}[1]{$^{#1}$}

\cleardoublepage

\chapter*{Danksagung}
\label{Danksagung}
\addcontentsline{toc}{chapter}{Danksagung}

Ich danke allen Personen, die zum Gelingen der vorliegenden Arbeit
beigetragen haben.\vspace{5pt}

\noindent Großer Dank gebührt meinem Betreuer PD Dr. Nikolai Neumaier, der
leider die Fertigstellung dieser Arbeit nicht mehr miterleben
durfte. Ich danke ihm für die Freiheit, die er mir bei meiner Arbeit
gelassen hat und für die, trotz schwerer Krankheit, immerwährende
Bereitschaft geduldig auf Fragen einzugehen sowie für die vielen
großartigen Diskussionen, die ich mit ihm führen konnte. Besonders
dankbar bin ich außerdem für seine außergewöhnlichen Übungsgruppen, denen ich
als Student beiwohnen durfte.
\vspace{5pt}

\noindent Ebenso großer Dank gebührt auch meinem
zweiten Betreuer apl. Prof. Dr. Stefan Waldmann, der von Anfang an immer
ein offenes Ohr für Fragen hatte und später die Betreuung der Arbeit
übernommen hat. Ihm danke ich auch besonders für die vielen glänzenden
Vorlesungen, bei denen ich sehr viel über Physik, Mathematik und deren
Schnittstelle lernen konnte.
\vspace{5pt}

\noindent Weiterer Dank gilt auch Prof. Dr. Hartmann Römer, dem es in
seinen Vorlesungen immer wieder gelang, enge Gebietsgrenzen zu sprengen
und die großen Verbindungen innerhalb der Physik, aber auch zu anderen
Gebieten wie Mathematik oder Philosophie aufzuzeigen. Ich bin dankbar,
in ihm einen akademischen Lehrer gefunden zu haben, der zeigt, dass es
auch heute noch möglich ist, tiefes fachliches Wissen und Verständnis
auf einer äußerst breiten Ebene zu erlangen.  \vspace{5pt}

\noindent Auch Prof. Dr. Bernd Siebert möchte ich an dieser Stelle
meinen Dank aussprechen, dafür, dass er es in den ersten zwei Semestern
geschafft hat, mein Interesse am Abstrakten zu wecken.
\vspace{5pt}

\noindent Ferner möchte ich
mich bei allen bedanken, die diese Arbeit korrekturgelesen haben, bei
Svea Beiser, Patricia Calon, Sascha Fröschl, Maximilian Hanusch und
Alexander Held.
\vspace{5pt}

\noindent Der größte Dank gebührt meiner Freundin Patricia. Im
Hinblick auf diese Arbeit danke ich ihr insbesondere für die unzähligen
tiefen fachlichen Diskussionen während der gesamten Studienzeit und für
die moralische Unterstützung während und vor allem in der Endphase der
Diplomarbeit.
\vspace{5pt}

\noindent Nicht zuletzt möchte ich mich noch besonders herzlich bei
meinen Eltern bedanken, die mir das Studium ermöglicht und mich in
jeglicher Hinsicht die ganze Zeit unterstützt haben.

\end{document}
